\documentclass[aps, prb, superscriptaddress, two column, preprintnumbers, floatfix, nofootinbib]{revtex4-2}
\usepackage[utf8]{inputenc}

\usepackage{tikz}
\usetikzlibrary {decorations.pathreplacing, calc}

\usepackage{color}
\usepackage{dsfont}
\usepackage{bm}
\usepackage{mathtools}
\usepackage{amsthm}
\usepackage{amsfonts, amssymb, amsmath}
\usepackage[colorlinks=true]{hyperref}
\usepackage{physics}
\usepackage{multirow}
\usepackage{graphicx}

\usepackage{verbatim}

\usepackage{quantikz}
\usepackage{marvosym} 

\newtheorem{theorem}{Theorem}

\newtheorem{definition}{Definition}

\definecolor{zx_red}{RGB}{232, 165, 165}
\definecolor{zx_green}{RGB}{216, 248, 216}

\begin{document}

\title{2D Quon Language: Unifying Framework for Cliffords, Matchgates, and Beyond}

\author{Byungmin Kang}
\affiliation{Dpartment of Physics, Massachusetts Institute of Technology, Cambridge, MA 02139, USA}

\author{Chen Zhao}
\affiliation{QuEra Computing Inc., 1284 Soldiers Field Road, Boston, MA, 02135, USA}

\author{Zhengwei Liu}
\thanks{liuzhengwei@mail.tsinghua.edu.cn}
\affiliation{Yau Mathematical Sciences Center and Department of Mathematical Sciences, Tsinghua University, Beijing 100084, China}
\affiliation{Yanqi Lake Beijing Institute of Mathematical Sciences and Applications, Beijing 101408, China}

\author{Xun Gao}
\thanks{xun.gao@colorado.edu}
\affiliation{JILA and Department of Physics, CU Boulder, Boulder, CO, USA}

\author{Soonwon Choi}
\thanks{soonwon@mit.edu}
\affiliation{Dpartment of Physics, Massachusetts Institute of Technology, Cambridge, MA 02139, USA}

\preprint{MIT-CTP/5871}

\begin{abstract}

Simulating generic quantum states and dynamics is practically intractable using classical computers. However, certain special classes---namely Clifford and matchgate circuits---permit efficient computation. They provide invaluable tools for studying many-body physics, quantum chemistry, and quantum computation. While both play foundational roles across multiple disciplines, the origins of their tractability seem disparate, and their relationship remain unclear. A deeper understanding of such tractable classes could expand their scope and enable a wide range of new applications. In this work, we make progress toward the unified understanding of the Clifford and matchgate---these two classes are, in fact, distinct special cases of a single underlying structure. Specifically, we introduce the 2D Quon language, which combines Majorana worldlines with their underlying spacetime topology to diagrammatically represent quantum processes and tensor networks. In full generality, the 2D Quon language is universal---capable of representing arbitrary quantum states, dynamics, or tensor networks---yet they become especially powerful in describing Clifford and matchgate classes. Each class can be efficiently characterized in a visually recognizable manner using the Quon framework. This capability naturally gives rise to several families of efficiently computable tensor networks introduced in this work: punctured matchgates, hybrid Clifford-matchgate-MPS, and ansatze generated from factories of tractable networks. All of these exhibit high non-Cliffordness, high non-matchgateness, and large bipartite entanglement entropy. We discuss a range of applications of our approach, from recovering well-known results such as the Kramers-Wannier duality and the star-triangle relation of the Ising model, to enabling variational optimization with novel ansatz states. 
\end{abstract}

\maketitle


\section{Introduction}
Simulating generic quantum states and quantum dynamics is practically intractable for classical computers, due to the exponential growth of Hilbert space dimensions with increasing system sizes. This intractability poses significant challenges both for understanding quantum many-body systems and for developing applications in areas such as quantum algorithms~\cite{dalzell2023quantum}, quantum chemistry~\cite{RevModPhys.92.015003}, and material design~\cite{dalzell2023quantum, RevModPhys.92.015003}. Despite these difficulties, rapid developments in quantum information science over the past several decades have provided new insights into addressing these challenges in certain cases. Specifically, there exists classes of quantum states or dynamics that are efficiently simulable using classical computers or even allow exact solutions. Notable examples include Clifford circuits~\cite{gottesman1998heisenberg, nielsen2001quantum}, matchgates~\cite{valiant2001quantum, PhysRevA.65.032325}, and matrix-product states (MPS)~\cite{schollwock2011density}. The Clifford circuits have proven useful in designing quantum error correcting codes and fault tolerant computation~\cite{nielsen2001quantum}. The matchgates allow one to solve free-fermion systems and have been extensively used as variational ansatze for finding grounds states of quantum chemistry problems and quantum many-body systems~\cite{RevModPhys.92.015003}, under the name of the Hartree-Fock approximation. The MPS, an architecture in which low-rank tensors are arranged in an efficiently contractable manner, lies at the core of the density-matrix renormalization group (DMRG) algorithm, one of the leading numerical tools for finding the ground state for low-dimensional systems, particularly in (1+1)-dimensional quantum systems~\cite{schollwock2011density, PhysRevLett.69.2863}. 

Despite the great success of those techniques, their landscape is highly heterogeneous, lacking any systematic understanding of what are the general class of quantum states or processes that allow efficient classical computations. This naturally brings us to a question---can we understand efficiently simulable classes of quantum systems within a unified framework?

In this work, we address a corner of this question using a pictorial framework for understanding the quantum states and processes~\cite{jaffe2018holographic, liu2017quon, jaffe2017planar}. Our paper introduces and utilizes a graphical language called the \textit{two-dimensional (2D) Quon language}, which is the 2D version of the previously introduced three-dimensional (3D) Quon language~\cite{liu2017quon, liu2019quon}. As a graphical language, the key elements are called \textit{2D Quon diagrams}, where each diagram provides a semantic interpretation representing a quantum state, a quantum process, or a tensor network. Moreover, there exist diagrammatic rewriting rules that often drastically modify the diagram while preserving its semantics. We refer to Refs.~\onlinecite{liu20233, liu2023alterfold, liu2024functional, liu2024alterfold} for recent advances in the generalization of Quon in mathematics, and Refs.~\onlinecite{liu2019quantized, shao2024variational} for an application of Quon in the design of quantum error correcting codes. 

Using the 2D Quon language, we provide four key results. First, we introduce the two dimensional (2D) Quon language and prove that it is \textit{universal}, meaning that any tensor network\footnote{Since any quantum states or quantum processes can be represented using tensor networks, our results is quite generic.} can be represented by a 2D Quon diagram, with an efficient conversion between the two. Second, we present simple \textit{pictorial} characterizations of Clifford and matchgate tensor networks in terms of 2D Quon diagrams. Not only do these pictorial characterizations allow us to put Cliffords and matchgates---two tractable classes with seemingly distinct origins of simulability---under a common Quon framework, but they also allow us to push the boundaries of matchgates beyond their standard definition, leading to a new class of \textit{tractable} tensor networks called \textit{punctured matchgates}, which contain matchgates as a strict subset. Using the concept of the punctured matchgates, we can prove the \textit{decomposition theorem}: any tensor network can be decomposed into a Clifford part and a matchgate part, allowing one to reexpress the network as a tensor network consisting of a single Clifford tensor and a single matchgate tensor. As an added benefit, we introduce a new class of hybrid tensor network, called \textit{hybrid Clifford-matchgate-MPS}. Our hybrid tensor network combines Clifford and matchgate components in a geometrically separated manner---similar to previous tensor networks in the literature~\cite{shi2018variational, mishmash2023hierarchical, PhysRevLett.133.190402, projansky2024extending}---while our tensor network also capturing the tractability of a special class of non-planar dimer models~\cite{vazirani1989nc, straub2016counting, curticapean2014counting, likhosherstov2020tractable}. A state represented by hybrid Clifford-matchgate-MPS generally exhibits high non-Cliffordness, high non-matchgateness, and high bipartite entanglement---factors that generally make the classical simulation difficult---yet the state remains efficiently simulable due the choice in the geometric arrangement. We note that this hybrid tensor network does not exploit the diagrammatic nature of Quon, although it is motivated by the concepts from Quon. Third, as a first nontrivial application of Quon, we provide a way to generate a large class of tensor networks with controlled tractability, by combining three simple diagrammatic moves that fully utilize the diagrammatic nature of Quon. These three moves are sufficiently versatile that, with only a few manipulations, the punctured matchgates can be readily generated. Moreover, they are capable of creating a wide variety of tractable tensor networks that are much more convoluted than Cliffords and punctured matchgates. The resulting tractable tensor networks generally exhibit high non-Cliffordness and high non-matchgateness, with the Clifford and matchgate components intertwine geometrically. See FIG.~\ref{fig:venndiagram} for an illustration of how newly introduced ansatze are positioned in the space of tensor networks. Finally, as a second application of Quon, we provide simple diagrammatic proofs of the Kramers-Wannier duality and the star-triangle relation for the two-dimensional classical Ising model, therefore showcasing the flexibility and power of our diagrammatic approach.

\begin{figure}[t]
\centering
\includegraphics[width=0.45\textwidth]{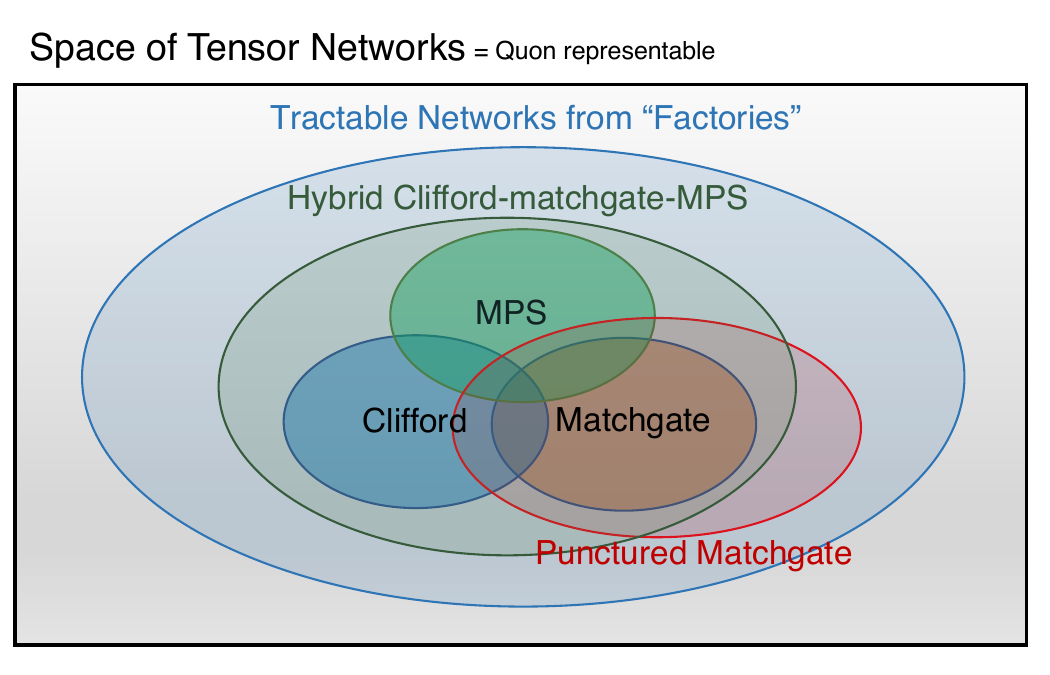}
\caption{A Venn diagram illustrating the space of tensor networks and tractable tensor networks considered in our work. The 2D Quon language is universal, capable of representing any arbitrary tensor network using a 2D Quon diagram. In this work, we focus on three known tractable classes of tensor networks---Cliffords, matchgates, and MPS---as well as three newly introduced classes of tensor networks---punctured matchgates, hybrid Clifford-matchgate-MPS, and tractable networks from ``factories,'' where newly introduced class are described in Sec.~\ref{sec:quon-universal} and Sec.~\ref{subsec:factory}. While tensor networks described by punctured matchgates and hybrid Clifford-matchgate-MPS exhibit a clear geometric separation of the Clifford and matchgate components, those obtained from ``factories'' generally do not have a clear geometric separation of the Clifford and matchgate parts.}
\label{fig:venndiagram}
\end{figure}

\section{Summary of Results}
Our four key results can be summarized in the following. Additionally, we highlight the relationship and differences between previous works and our approach. 

\subsection{Universality of 2D Quon Language}
In Sec.~\ref{sec:Quon-diagrams}, we formally introduce the two-dimensional (2D) Quon language as the 2D version of the previously introduced the three-dimensional (3D) Quon language~\cite{liu2017quon, liu2019quon}. In the 2D Quon language, we employ 2D Quon diagrams to pictorially represent quantum states, quantum circuits, or generally tensor networks. One may view a 2D Quon diagram as a form of tensor network; however unlike in a conventional tensor network where each tensor leg represents a qubit degree of freedom, our approach utilizes ``fractionalized'' Majorana degrees of freedom to visualize the content of the tensor, offering additional structure and flexibility that clearly illuminates important information when it comes to tractable tensor networks. 

\begin{figure*}[t]
\centering
\includegraphics[width=\textwidth]{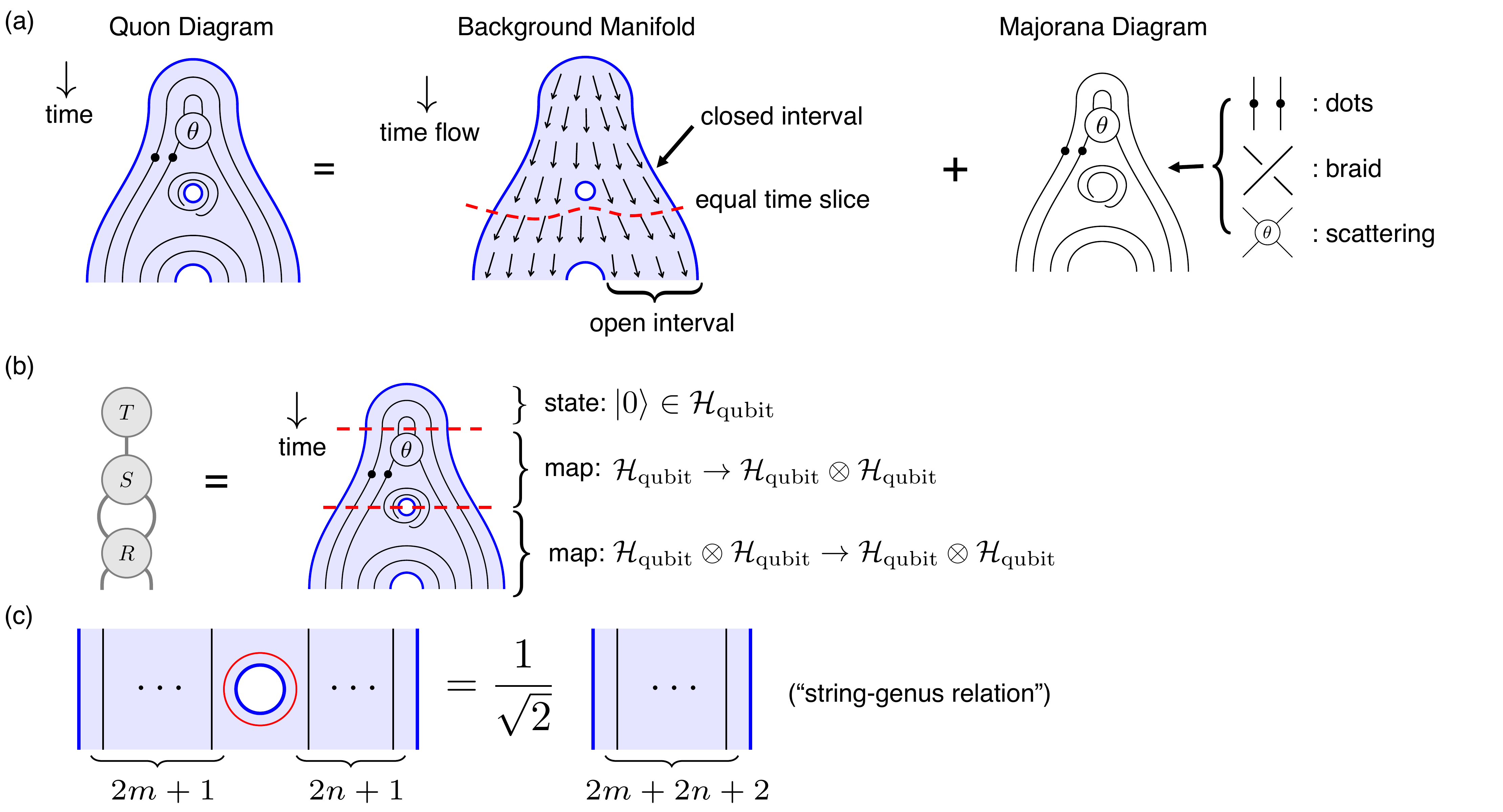}
\caption{(a) Example of a \textit{2D Quon diagrm}, which consists of a \textit{background manifold} and a \textit{Majorana Diagram}. The background manifold has a temporal and a spatial direction; in this specific example, the arrows in the bulk indicate the time direction at each point. In most cases, we simply adopt the convention of a global time direction, indicated as an arrow outside the bulk (shown together with ``time flow''). With this global time direction, an equal time slice is given by a horizontal cut. If the global time direction is not shown, we assume that time flows from top to bottom. Furthermore, the background manifold imposes the parity constraints (see main text for details.). The background manifold generally has multiple \textit{boundary components} and each boundary component consists of \textit{open boundaries} and \textit{closed boundaries}. Closed boundaries are highlighted in thick blue to distinguish them from open boundaries. The Majorana diagram illustrates the worldlines of Majoranas as thin black lines, together with \textit{dots}, \textit{braids}, and \textit{scattering elements}. Dots indicate Majorana operators actions, braids indicate the braiding process between two Majoranas, and the scattering element indicates the scattering (or a free-fermion evolution) between two Majoranas with the scattering angle $\theta$ shown inside the circle. A braid is a special type of the scattering element, with the scattering angle $\theta = \pm \frac{\pi}{2}$. A scattering element is called \textit{generic} if the scattering angle is not an integer multiple of $\pi/2$. (b) A 2D Quon diagram can be interpreted in two ways: as a dynamical quantum process, or as providing an internal structure of a tensor network. In the quantum process interpretation, this example describes a process with the creation of Majorana pairs, ultimately resulting in eight Majorana modes in a certain quantum state. Alternatively, the example can be viewed of as the contractions of three tensors $T$, $S$, and $R$, corresponding to three 2D Quon diagrams (separated by two red dashed lines) glued together. (c) The \textit{string-genus relation}, an example of a diagrammatic rewriting rule, amounts to removing a hole together with its enclosing Majorana loop. This example shows how the topology of the background manifold can be altered in the Quon diagram without changing its semantic content.}
\label{fig:quon-example}
\end{figure*}

\emph{Quon Diagram---}A 2D Quon diagram consists of two components: the \textit{background manifold} and the \textit{Majorana diagram}, as outlined in FIG.~\ref{fig:quon-example} (a). The \textit{background manifold} is a two-dimensional manifold embedded in a plane, with one direction corresponding to time and the other to space. As a topological manifold embedded in a plane, it may have holes, and thus it can be topologically non-trivial. In principle, each point in the bulk of the background manifold carries its own time-direction, as indicated by arrows in FIG.~\ref{fig:quon-example} (a). Importantly, only the relative time ordering between two points matters, reflecting the \textit{topological} nature of Quon. Therefore to avoid clutter, we often adopt a global time direction (flowing from top-to-bottom) to represent the time flow of the manifold. In this case, an equal-time contour line is given by a line perpendicular to the time direction. A background manifold consists of a bulk and a boundary, where the boundary consists of multiple boundary components when the bulk contains one or a few holes. Each boundary component is a disjoint union of \textit{open} and \textit{closed} intervals, where the latter are depicted as thick blue line, as shown in FIG.~\ref{fig:quon-example} (a). As we will see later, each \textit{open interval} corresponds to an \textit{open leg} in a conventional tensor network. The \textit{Majorana diagram} is a diagrammatic representation of the worldlines of Majorana zero modes, where the worldlines are usually depicted as thin black or red lines, as shown in FIG.~\ref{fig:quon-example} (a). It describes how Majoranas are pair-created or pair-annihilated, as well as how Majoranas scatter in a pair-wise fashion. Similar to the global time direction of the background manifold, we usually assume that the time flows from top to bottom in the Majorana diagram. The Majorana worldlines can form closed strings, or they can be open, in which case always an even number of endpoints appear near one another. When a Majorana diagram is embedded in a background manifold, the Majorana worldlines always terminate on the open intervals of the boundary of the background manifold. We are primarily interested in the case where four Majoranas terminate on each open interval of the background manifold. Specifically, we encode one qubit using four Majoranas by utilizing the parity-even subspace of the two-qubit Hilbert space, spanned by two complex fermions formed from these four Majoranas, a method referred to as the \textit{sparse encoding} in the literature~\cite{sarma2015majorana}, or equivalently as the Kitaev encoding in the Kitaev's honeycomb model~\cite{kitaev2006anyons}. This sparse nature of encoding adds flexibility to the Quon language, enabling a rich set of diagrammatic rewriting rules. In fact, the role of the background manifold is to impose the parity-even projections to the embedded Majorana diagram on every (connected part of an) equal-time slice. For an open interval, this amounts to imposing $i^2 \gamma_1 \gamma_2 \gamma_3 \gamma_4 = 1$, where $\gamma_i$ denotes the $i$th Majorana operator. Therefore, the background manifold also tracks the parity-even projections applied to the embedded Majorana diagram, in addition to providing the time ordering. 

Furthermore, Quon diagrams can be glued together to form a bigger diagram, corresponding to compositions or tensor contractions, as illustrated in FIG.~\ref{fig:quon-example} (b). 

Finally, the topology of the manifold can be altered, notably by the \textit{string-genus} relation, a two-dimensional version of the 3D string-genus relation~\cite{liu2017quon, liu2019quon}, which is shown in FIG.~\ref{fig:quon-example} (c). Through the string-genus relation, a hole in the background manifold can be completely removed if it is enclosed by an isolated Majorana loop, which is also removed. This relation also demonstrates the flexibility of the Quon language, enabled by the use of the fractionalized degrees of freedom.

\begin{figure}
\centering
\includegraphics[width=0.4\textwidth]{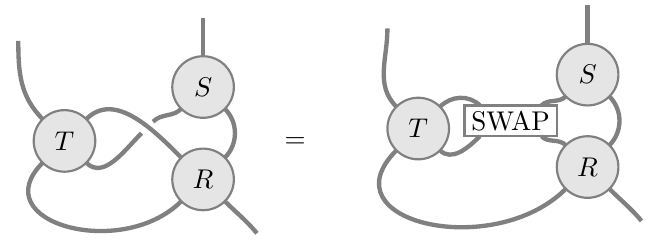}
\caption{A non-planar tensor network can always be reduced to a planar one at the expense of using the SWAP gates.}
\label{fig:planar-TN-via-SWAP}
\end{figure}

\emph{Universality of Quon---}Our first technical result is the \textit{universality of the Quon language}, presented in Sec.~\ref{sec:quon-universal}. By universality, we mean that any arbitrary tensor network\footnote{As previously mentioned, one can always view a quantum state or process as a tensor network, therefore our universality result is quite general.} can be represented by a single 2D Quon diagram, and vice versa, where the conversion between the two can be done efficiently. The key aspects of the universality proof is as follows. First, we construct the 2D Quon diagrams for tensors in a universal generating set, consisting of tensors of rank at most $3$. This set is universal in the sense that any arbitrary tensor, including the SWAP gate which plays an important role later, can be generated from it using tensor operations such as tensor contractions. Next, we observe that any tensor network can be rewritten as a planar tensor network, at the expense of using the SWAP gates as described in FIG.~\ref{fig:planar-TN-via-SWAP}. Given the planar tensor network, we first construct a 2D Quon representation for each tensor in the network using the universal generating set and its Quon representations, and then glue these diagrams together according to the tensor contractions in the network, thereby obtaining a 2D Quon representation of the entire network. Our construction is efficient in the following sense: if the tensor network is specified by tensors each with an efficient representation, e.g., (potentially large-ranked) Clifford and matchgate tensors, and tensors with small ranks, then a representing 2D Quon diagram can always be efficiently constructed. Note that such a representing 2D Quon diagram is not unique due to the diagrammatic rewriting rules. Conversely, given a 2D Quon diagram, compiling it into a conventional tensor network can also be done efficiently by parsing the diagram. 

\begin{figure*}[t]
\centering
\includegraphics[width=\textwidth]{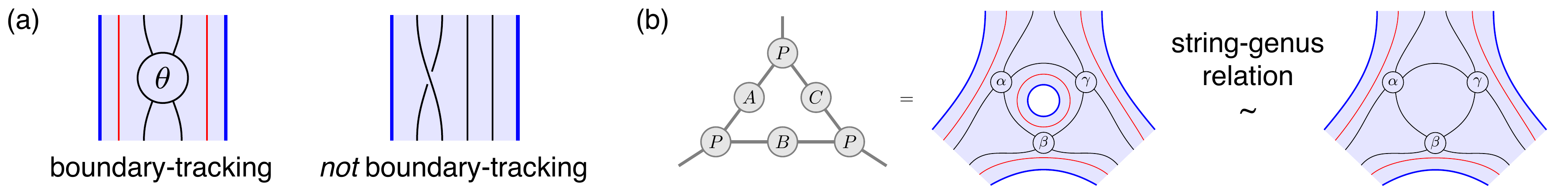}
\caption{(a) Examples of 2D Quon diagrams satisfying or violating the boundary-tracking property. The boundary-tracking Majoranas are colored in red for better visibility. (b) Example of a matchgate tensor network with a loop in the network, where $A$, $B$, $C$, and $P$ represent matchgate tensors. Tensor contractions of the network introduces a hole in the background manifold. However, due to a Majorana loop originated from the boundary-tracking Majorana lines, the hole can be removed via the string-genus relation.}
\label{fig:matchgate-example}
\end{figure*}

\subsection{Pictorial characterizations of Cliffords, matchgates, and punctured matchgates}
Using the universality of Quon, one can always find a 2D Quon diagrammatic representation for an arbitary tensor network, including Clifford and matchgate tensor networks. It turns out that for Clifford and matchgate tensor networks, simple pictorial characterizations exist, as presented in Sec.~\ref{sec:clifford-matchgate-quon-diagrams}. 

\emph{Pictorial characterization of Clifford---}The pictorial characterization of Clifford tensor networks were previously presented in Ref.~\onlinecite{liu2017quon} using the 3D Quon language. In our work, we present the corresponding characterizations using the 2D Quon language. A Clifford tensor network can be represented by a 2D Quon diagram in which the Majorana diagram contains no generic scattering elements---only braids are present. Moreover, the converse is true: if a 2D Quon diagram contains no generic scattering elements, then it represents a Clifford tensor network. 

\emph{Pictorial characterizations of matchgate---}We provide pictorial characterizations of matchgate tensor networks in terms of 2D Quon diagrams. Specifically, we prove that any arbitrary matchgate tensor network has a 2D Quon diagrammatic representation with two properties: (i) the boundary-tracking property and (ii) the background manifold contains no holes. The \textit{boundary-tracking property} asserts that along every closed interval in the boundary of the background manifold, there always exists a Majorana string spatially adjacent to it, tracking the closed interval without scattering or braiding with other Majoranas. See FIG.~\ref{fig:matchgate-example} (a) for examples of 2D Quon diagrams satisfying or violating the boundary-tracking property. We show that properties (i) and (ii) are kept upon planar tensor contractions between matchgate tensors, with the help of the string-genus relation, as demonstrated in FIG.~\ref{fig:matchgate-example} (b). Similar to the pictorial characterization of Clifford tensor networks, we also show that the converse is true: a 2D Quon diagram having properties (i) and (ii) represents a matchgate tensor network. In essence, the matchgate tensor networks correspond to networks of Majorana lines featuring free-fermion scatterings, which are effectively encoded in the parity-even subspace in the 2D Quon diagrammatic representations satisfying properties (i) and (ii). 

\begin{figure}[t]
\centering
\includegraphics[width=0.4\textwidth]{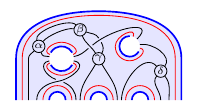}
\caption{A 2D Quon diagram representing a punctured matchgate network. Note that the diagram satisfies the boundary-tracking property, with the boundary-tracking Majorana lines highlighted in red, but it contains holes.}
\label{fig:punctured-matchgate-example}
\end{figure}

\emph{Introduction of a new class I: punctured matchgates---}In addition to the pictorial characterizations of the matchgate tensor networks, we newly introduce and provide a pictorial characterization of the \textit{punctured matchgates}. The punctured matchgates, introduced in Sec.~\ref{sec:planar-TN}, contain matchgates as a strict subset, while remaining tractable for computing their components. The pictorial characterization of the punctured matchgates can be obtained from that of the matchgates by requiring only the boundary-tracking property, potentially having irremovable holes in the background manifold. See FIG.~\ref{fig:punctured-matchgate-example} for an example of a 2D Quon diagram representing a punctured matchgate tensor. 

\emph{Topological interpretation of non-matchgateness---}Our pictorial characterization suggests a topological interpretation of \textit{non-matchgateness}, or equivalently, non-Gaussianity, in terms of \textit{irremovable holes in the background manifold}. In order to promote matchgate tensor networks to a universal one, one can supplement with a \textit{magic gate}, similar to how $T$-gate serves as a magic gate for Clifford computation~\cite{PhysRevA.71.022316}. While it is common to use the SWAP gate or the Hadamard gate as a magic gate for matchgate computation~\cite{PhysRevA.84.022310}, it is known that any pure non-matchgate state can serve as a magic gate via the gate injection for the matchgate quantum computation~\cite{PhysRevLett.123.080503}. We note that both the SWAP and Hadamard gates violate the boundary-tracking property, often leading to irremovable holes in the Quon diagram. We therefore propose \textit{the number of holes} in the background manifold as a measure for non-matchgateness/non-Gaussianity. Indeed, one can design an algorithm to perform a brute-force summation over $2^\textrm{\# of holes}$ matchgate Quon diagrams to evaluate an arbitrary 2D Quon diagram; see earlier literature for mathematical aspects of this algorithm~\cite{cimasoni2007dimers, cimasoni2008dimers, cimasoni2009dimers, Dijkgraaf2009dimer, bravyi2009contraction}. 

\emph{Decomposition theorem---}Using the concept of the punctured matchgates, we prove the \textit{decomposition theorem}: any tensor network can be rewritten as a tensor network consisting of a single Clifford tensor and a single matchgate tensor, presented as Theorem~\ref{thm:Clifford-matchgate-decomposition}. The theorem can be achieved by first properly identifying the Clifford part and the matchgate part in the network, and then contracting these parts separately to form a single Clifford and a single matchgate tensor. Once decomposed, the individual contractions of  the Clifford tensor network and matchgate tensor networks, separately, can be done efficiently. However, the contraction of the closed legs of the resulting (potentially large) Clifford and matchgate tensors are generally intractable as the brute-force contractions scale exponentially with the number of closed legs. 

\emph{Introduction of a new class II: hybrid Clifford-matchgate-MPS---}The decomposition theorem naturally leads to a new tractable tensor network architecture called the \textit{hybrid Clifford-matchgate-MPS}, as presented in Sec.~\ref{sec:quon-universal}. See FIG.~\ref{fig:venndiagram} to see where the hybrid Clifford-matchgate-MPS fits within the space of tensor networks. This class is a matrix-product state (MPS) type of tensor network---or more generally a tree structure is possible---where each tensor in the network is either a Clifford, a matchgate, or a low-rank tensor. While the number of open tensor legs attached to each Clifford and matchgate can reach $\mathrm{O}(n)$, where $n$ is the total number of external legs, the number of closed tensor legs between nearby tensors is kept small to ensure tractability; otherwise, our decomposition theorem implies that using two tensors---a Clifford and a matchgate---would be powerful enough to represent an arbitrary tensor network. By design, our tensor network spatially separates the regions corresponding to the Clifford part and the matchgate part in a way that ensures tractability. Therefore, one can easily understand that the contraction of the network is computationally tractable despite they exhibit high non-Cliffordness, high non-matchgateness, and large amount of entanglement entropy. Perhaps more surprisingly, this class offers an intuitive understanding behind the tractability of a special class of non-planar dimer models defined on so-called \textit{single-crossing minor-free graphs}, studied in computer-science literature~\cite{vazirani1989nc, straub2016counting, curticapean2014counting, likhosherstov2020tractable}.

\subsection{Applications: Factories for tractable tensor networks and Diagrammatic proofs of the Kramers-Wannier duality and the star-triangle relation}
\emph{Introduction of a new class III: ansatze from factories for tractable networks---}In addition to the punctured matchgates and hybrid Clifford-matchgate-MPS, we introduce a novel diagrammatic approach to generate a large class of tensor networks with controlled tractability. Specifically, we introduce three diagrammatic “moves” called \textit{stretching}, \textit{inserting}, and \textit{switching}, where each move modifies a 2D Quon diagram to produce new tensor networks. The first move is stretching, where a segment of Majorana string is stretched to a target location; during stretching, the string segment passes over other diagrammatic elements, introducing several braids into the Quon diagram. Notably, one can even stretch the Majorana segment to the ``outside'' of the Quon diagram, in which case it can increase the total number of tensor legs. The second move is inserting, which locally inserts an arbitrary Quon diagram. By itself, the inserting move does not alter the tensor-network content; however, when combined with the stretching move, the two moves are versatile enough to generate the entire class of the punctured matchgates and beyond. Crucially, neither the stretching nor the inserting move increases the tractability of the tensor network despite their generative power. The third and final move is switching, which replaces a local component of a Quon diagram, such as a braid, with another, for instance the opposite braid. Within the switching move, multiple specific replacements are available, where all except one preserve the tractability of the network. The exception, replacing a braid with an arbitrary scattering element, increases the computational cost by at most a factor of $2$. We track the number of such changes as the \textit{number of transformed scatterings}, which plays a role analogous to the $T$-count in $T$-gate-doped Clifford circuits or the bond dimension in MPS~\cite{PhysRevA.71.022316, schollwock2011density}. This analogy holds because (i) as the number of transformed scatterings increases without bound, the entire space of tensor networks is eventually spanned, and (ii) the computational cost scales exponentially with this number. We remark that different combinations of these three simple moves can generate drastically different tensor networks even when starting from the same 2D Quon diagram; hence we term our method ``factories'' for tractable tensor networks. 

Finally, using the 2D Quon language, we provide diagrammatic proofs of the \textit{Kramers-Wannier duality} and the \textit{star-triangle relation} of the two-dimensional (2D) classical Ising model, both of which highlight the core features of the 2D classical Ising model. 

\emph{Kramers-Wannier duality---}The Kramers-Wannier (KW) duality~\cite{PhysRev.60.252, baxter2016exactly} relates the partition function of the classical Ising model defined on a planar graph to that on its dual lattice, with modified coupling constants. The KW duality highlights key properties of the 2D classical Ising model, notably identifying the self-dual point in the square lattice, where the critical phenomena occur. Our diagrammatic proof of the KW duality follows from the following observations. First, the partition function of the 2D classical Ising model can be expressed as a matchgate tensor network. Next, by applying a series of the string-genus relation and the space-time duality of the Majorana scattering, we change the topology of the 2D Quon diagram and finally obtain the partition function of the 2D classical Ising model defined on its dual lattice. During the process, we re-group the Majoranas, which naturally leads to the half-translation of the KW duality~\cite{seiberg2024non}. 

\emph{Star-Triangle Relation---}The star-triangle relation of the 2D classical Ising model~\cite{PhysRev.65.117, baxter2016exactly} relates the partition function of a local ``star''-shaped region to that of a ``triangle''-shaped region of two-dimensional Ising model, and vice versa. In essence, it captures the integrability of the underlying model~\cite{baxter2016exactly}. Using 2D Quon language, we derive the star-triangle relation via the Yang-Baxter equation of the scatterings, which is equivalent the Euler decomposition of \textsc{SO(3)} with the spinor representation in terms of Majoranas. 

\subsection{Relation to earlier results}
\emph{Comparison between Quon approach and topological quantum computation---}The idea of using fractionalized particles, or anyons, to perform quantum computation has been extensively studied in the context of topological quantum computation~\cite{nayak19962n, freedman2003topological, sarma2015majorana, RevModPhys.80.1083, wang2010topological}. Broadly speaking, our Quon language corresponds to the topological quantum computation using \textit{Majorana zero modes} (MZMs) with additional details involving scattering between Majoranas and projections imposed by the background manifold. A MZM can be realized at the core of a point-like topological defect in a topological ordered state, such as on the boundary of $(1+1)$d $p$-wave superconductor, at the vortex core of $(2+1)$d $(p+ip)$-wave superconductor, on the surface of $(3+1)$d topological insulator proximitized by a superconductor, or as a twist defect in a surface code~\cite{PhysRevB.61.10267, kitaev2001unpaired, PhysRevB.63.224204, PhysRevB.95.235305, 10.21468/SciPostPhys.8.6.091, nayak19962n, PhysRevX.7.021029}. With multiple MZMs, there exist degenerate eigenstates in which quantum state can be encoded. By adiabatically exchanging two localized MZMs in a spatially separated manner, i.e., braiding them, one can perform a quantum gate as well. Mathematically, this process can be captured by the braiding of Ising anyons in the Ising anyon theory~\cite{nayak19962n, freedman2003topological, RevModPhys.80.1083, wang2010topological}. While the Ising anyon theory also includes fermions, which appear when fusing multiple Ising anyons, in our setting, fermions are condensed~\cite{aasen2019fermion} so that the usual fermion lines that appear in the fusion process do not appear, but leaving only their presence at the end points of the fermion lines, depicted as dots~\cite{jaffe2017planar, liu2019quon, aasen2019fermion}. Here, we instead construct the 2D Quon language purely starting from MZMs, extending the construction presented in Refs.~\onlinecite{jaffe2017planar, jaffe2018holographic}. We emphasize that our work focuses numerical and analytic methods for understanding quantum states or processes, rather than proposing implementations on topological quantum computing platforms. Experimental realization of operations such as parity projections requires post-selections, which demands a large amount of resource in experiments, but not necessarily for our purposes. Throughout the manuscript, we will simply refer to MZMs as Majoranas. 

\emph{Previous works on Clifford, matchgate, and MPS hybrids---}Previously, several lines of work have focused on creating hybrids of Clifford, matchgate, and MPS, particularly for constructing variational ansatze. Notable examples include a matchgate unitary acted on a non-Gaussian unitary---specifically, time-evolution under a density-density interaction---applied to a matchgate state, or its variants including bosonic generalizations~\cite{shi2018variational}, Clifford augmented MPS structure~\cite{PhysRevLett.133.190402}, Clifford augmented matchgate (Hartree-Fock) state~\cite{mishmash2023hierarchical, projansky2024extending}, and Clifford conjugated matchgates~\cite{projansky2024extending} have been explored. These states share some similarities with our hybrid Clifford-matchgate-MPS states, in that they allow a clear geometric separatation of the Clifford and the matchgate components, and computing the expectation value of an local observable remains tractable. As previously mentioned, crucially, our hybrid ansatz captures the underlying tractability of the dimer models defined on single-crossing minor-gree graphs~\cite{vazirani1989nc, straub2016counting, curticapean2014counting, likhosherstov2020tractable}. In contrast, punctured matchgates and states from ``factories'' have a distinct notion of tractability: while computing each component remains tractable, the evaluation of the expectation value of a local observable may become intractable. We note that recent neural-network-based variational wavefunctions~\cite{PhysRevLett.122.226401, acevedo2020vandermonde, von2022self, chen2023exact, geier2025attention} share some similarities with punctured matchgates. By generalizing the Hartree-Fock wavefunction (or matchgate), each component of a neural network wavefunction can be expressed as a single determinant of a matrix whose entries depends on multi-particle coordinates, rather than on a single-particle orbitals as in traditional Hartree-Fock. A component of a punctured matchgate can likewise be expressed as now a Pfaffian of a matrix, with its entries encoding correlations. However, we emphasize punctured matchgates possess additional structure, namely, a geometrical interpretation inherited from matchgates. We also note that tractable tensor networks generated from factories generally lack a simple geometric separation of the Clifford and the matchgate components; while one may invoke the decomposition theorem to rewrite the tensor network as consisting of a single Clifford and a single matchgate, these two tensors are typically connected by a large number of closed tensor legs. 

\emph{Relationship between ZX-calculus and Quon---}Several graphical calculi for quantum information--often referred to as categorical tensor networks~\cite{biamonte2011categorical}---have been developed, most notably the ZX-calculus and its variants~\cite{coecke2007graphical, coecke2008interacting, coecke2011interacting, coecke2010compositional, van2020zx}, for reasoning about tensor networks and quantum information. In the ZX-calculus, a tensor network is expressed in terms of the generators called Z-and X-spiders. Equipped with a set of rewriting rules, the ZX-calculus is universal and complete~\cite{ng2017universal, ng2018completeness, hadzihasanovic2018two, jeandel2020completeness, wang2022completeness}, similar to the Quon language. The crucial difference is that the Quon language employs fractionalized degrees of freedom, rather than the bare qubit degrees of freedom used in the ZX-calculus, thereby offering additional flexibilities compared to the ZX-calculus. For example, an arbitrary Clifford tensor network can be represented by the ZX-calculus using a diagram in which the phases of the Z- and X- spiders are multiples of $\pi/2$~\cite{backens2014zx}. However, representing matchgate tensor networks using the ZX-calculus is not as natural, and rewriting them tends to be generally complicated. Alternatively, one can choose a different set of generators, e.g., Z- and W-spiders to form the ZW-calculus, which is also a universal and complete graphical calculus~\cite{hadzihasanovic2015diagrammatic, hadzihasanovic2018two}. Using the ZW-calculus, its planar fragment, known as the pW-calculus, is complete for matchgate tensor networks and allows for efficient computations~\cite{carette2023compositionality}. On the other hand, it is complicated to manipulate Clifford tensor networks within the ZW-calculus framework. The Quon language can be considered a graphical calculus unifying the ZX- and ZW-calculus, in that the rewriting rules of the ZX- and ZW-calculus can be derived from more fundamental ones in the Quon language~\cite{liu2017quon}. We note that the ZXW-calculus has also been proposed, with its original motivation rooted in quantum machine learning~\cite{koch2022quantum}. While one may treat the Clifford and matchgate tensor networks on an equal footing under the ZXW-calculus, its rewriting rules~\cite{poor2023completeness} are much more complex than those of Quon, making the Quon language a more natural choice.

\section{Quon Diagrams}
\label{sec:Quon-diagrams}
In this section, we introduce two-dimensional (2D) Quon language. The main objects in the Quon language are pictorial elements called \textit{Quon diagrams}. A Quon diagram consists of two components: the \textit{Majorana diagram} and the \textit{background manifold}. The Majorana diagram depicts the worldlines of Majorana zero modes and is embedded the background manifold, where the background manifold imposes additional constraints on the Majorana diagram. In the remainder of this section, we outline all the diagrammatic rules that are relevant for our study. 

\subsection{Majorana Diagram and Diagrammatic Rewriting Rules}
\label{sec:Majorana-diagrams}
We first introduce the Majorana diagram. A Majorana diagram is a two-dimensional diagram of lines, or equivalently strings, depicting the spacetime trajectories, or worldlines, of Majoranas. Here, we only highlight key aspects of the Majorana diagrams relevant for our results, following the convention used in Ref.~\onlinecite{jaffe2018holographic, liu2017quon, jaffe2017planar}, where additional details, except the scattering element, can be found. 

As briefly mentioned, a string in a Majorana diagram depicts the worldline of a Majorana; by convention, time flows from top to bottom, unless explictly stated otherwise. At each time-slice (a horizontal cut), an even number of Majorana lines is present, forming a fermionic Fock space. When there are no Majorana line in a time slice, it is interpreted as a $1$-dimensional vector space $\mathbb{C}$. Additionally, if we consider two time slices, the intermediate section of the diagram, bounded by these slices, can be interpreted as a map from the Fock space associated with the upper time-slice to that with the lower time-slice. Therefore, the composition of maps is realized by vertically gluing diagrams, and the tensor product is realized by horizontally gluing diagrams. For example, the following Majorana diagram can be parsed into quantum processes between Majoranas: 
\begin{equation}
\begin{tikzpicture}[scale=1.2]
\draw (0, 0.8)--(0, -0.8);

\draw (0.4, 0.8)--(0.4, 0.4) to[out=-90, in=130] (0.55, 0.06);
\draw (0.65, -0.06) to[out=-50, in=90] (0.8, -0.4);

\draw (0.4, -0.8)--(0.4, -0.4) to[out=90, in=-90] (0.8, 0.4) arc (180:0:0.2)--(1.2, -0.4) arc (0:-180:0.2);

\draw [decorate, decoration = {brace}] (1.4, 0.8)--(1.4, 0.45);
\draw [decorate, decoration = {brace}] (1.4, 0.35)--(1.4, -0.35);
\draw [decorate, decoration = {brace}] (1.4, -0.45)--(1.4, -0.8);

\node at (3.6, 0.65) {\scriptsize pair-creation of Majoranas $3$ and $4$};
\node at (3.4, 0.02) {\scriptsize braiding of Majoranas $2$ and $3$};
\node at (3.86, -0.62) {\scriptsize pair-annihilation of Majoranas $3$ and $4$};
\end{tikzpicture} \nonumber
\end{equation}
With this \textit{semantic} meaning in mind, we present in TABLE~\ref{tab:majorana-diagrams} a dictionary relating Majorana states and operators to their corresponding elements in Majorana diagrams. In addition, we provide the corresponding states and operators in a qubit Hilbert space, obtained via the Jordan-Wigner transformation. This process, encoding one qubit per pair of Majoranas, is also known as the \textit{dense encoding} in the literature~\cite{bravyi2002fermionic, sarma2015majorana}. 

Let us explain diagrammatic elements presented in TABLE~\ref{tab:majorana-diagrams}. To begin, the identity evolution is denoted by parallel lines, with each line corresponds to a Majorana mode. Suppose that a horizontal (equal-time) cut supports $2n$ Majorana modes. We typically order them linearly from the left to right, labeling the corresponding Majorana operators as $\gamma_1, \gamma_2 \ldots, \gamma_{2n}$, which satisfy the anticommutation relation $\{ \gamma_j, \gamma_k \} = \delta_{j, k}$. When a Majorana operator $\gamma_j$ acts at a specific time, we diagrammatically represent that action as a \textit{dot} (also called a \textit{charge}) on the $j$th string at that time. Due to anticommuting nature of Majoranas, the relative time ordering of dots is important. Therefore, in principle, no more than one dot can appear on each time slice. However, it is convenient to introduce the following diagrammatic notation, in which two dots appear simultaneously: 
\begin{equation}
\label{eq:pair-of-dots}
\raisebox{-0.4cm}{\tikz{
\draw (-0.2, 0)--(-0.2, 1);
\draw (0, 0)--(0, 1);
\draw (1, 0)--(1, 1);
\draw (1.2, 0)--(1.2, 1);

\node at (-0.2, 0.5) [circle, fill, inner sep=1pt] {};
\node at (0.5, 0.5) {$\cdots$};
\node at (1.2, 0.5) [circle, fill, inner sep=1pt] {};
}} \,\, := i \,\, \raisebox{-0.4cm}{\tikz{
\draw (-0.2, 0)--(-0.2, 1);
\draw (0, 0)--(0, 1);
\draw (1, 0)--(1, 1);
\draw (1.2, 0)--(1.2, 1);

\node at (-0.2, 0.25) [circle, fill, inner sep=1pt] {};
\node at (0.5, 0.5) {$\cdots$};
\node at (1.2, 0.75) [circle, fill, inner sep=1pt] {};
}} \,\, , 
\end{equation}
which also appears in TABLE~\ref{tab:majorana-diagrams} as the \textit{parity operator} between two Majoranas. More generally, we allow an even number of dots to appear simultaneously, with the convention that dots are ``paired up'' one by one from left to right. Hence, the \textit{global parity} operator in TABLE~\ref{tab:majorana-diagrams} equals 
\begin{equation}
\label{eq:global-fermion-parity}
\hat{P}_{2n} = \,\, \raisebox{-0.4cm}{\tikz{
\draw (-0.6, 0)--(-0.6, 1);
\draw (-0.4, 0)--(-0.4, 1);
\draw (-0.2, 0)--(-0.2, 1);
\draw (0, 0)--(0, 1);
\draw (1, 0)--(1, 1);
\draw (1.2, 0)--(1.2, 1);

\node at (-0.6, 0.5) [circle, fill, inner sep=1pt] {};
\node at (-0.4, 0.5) [circle, fill, inner sep=1pt] {};
\node at (-0.2, 0.5) [circle, fill, inner sep=1pt] {};
\node at (0, 0.5) [circle, fill, inner sep=1pt] {};
\node at (0.525, 0.5) {$\cdots$};
\node at (1, 0.5) [circle, fill, inner sep=1pt] {};
\node at (1.2, 0.5) [circle, fill, inner sep=1pt] {};
}} \,\, = \,\, \raisebox{-0.4cm}{\tikz{
\draw (-0.6, 0)--(-0.6, 1);
\draw (-0.4, 0)--(-0.4, 1);
\draw (-0.2, 0)--(-0.2, 1);
\draw (0, 0)--(0, 1);
\draw (1, 0)--(1, 1);
\draw (1.2, 0)--(1.2, 1);

\node at (-0.6, 0.2) [circle, fill, inner sep=1pt] {};
\node at (-0.4, 0.2) [circle, fill, inner sep=1pt] {};
\node at (-0.2, 0.4) [circle, fill, inner sep=1pt] {};
\node at (0, 0.4) [circle, fill, inner sep=1pt] {};
\node at (0.525, 0.6) {$\cdots$};
\node at (1, 0.8) [circle, fill, inner sep=1pt] {};
\node at (1.2, 0.8) [circle, fill, inner sep=1pt] {};

}}  \,\, . 
\end{equation}
Note that this ``pairing'' convention is crucial as 
\begin{equation}
\label{eq:four-dots-ambiguity}
\raisebox{-0.33cm}{
\begin{tikzpicture}
\draw (0, 0)--(0, 0.8);
\draw (0.3, 0)--(0.3, 0.8);
\draw (0.6, 0)--(0.6, 0.8);
\draw (0.9, 0)--(0.9, 0.8);
\draw (1.2, 0)--(1.2, 0.8);
\draw (1.5, 0)--(1.5, 0.8);

\node at (0.3, 0.4) [circle, fill, inner sep=1pt] {};
\node at (0.9, 0.4) [circle, fill, inner sep=1pt] {};
\node at (1.2, 0.4) [circle, fill, inner sep=1pt] {};
\node at (1.5, 0.4) [circle, fill, inner sep=1pt] {};
\end{tikzpicture}
} = 
\raisebox{-0.33cm}{
\begin{tikzpicture}
\draw (0, 0)--(0, 0.8);
\draw (0.3, 0)--(0.3, 0.8);
\draw (0.6, 0)--(0.6, 0.8);
\draw (0.9, 0)--(0.9, 0.8);
\draw (1.2, 0)--(1.2, 0.8);
\draw (1.5, 0)--(1.5, 0.8);

\node at (0.3, 0.25) [circle, fill, inner sep=1pt] {};
\node at (0.9, 0.25) [circle, fill, inner sep=1pt] {};
\node at (1.2, 0.55) [circle, fill, inner sep=1pt] {};
\node at (1.5, 0.55) [circle, fill, inner sep=1pt] {};
\end{tikzpicture}
} \ne \raisebox{-0.33cm}{ 
\begin{tikzpicture}
\draw (0, 0)--(0, 0.8);
\draw (0.3, 0)--(0.3, 0.8);
\draw (0.6, 0)--(0.6, 0.8);
\draw (0.9, 0)--(0.9, 0.8);
\draw (1.2, 0)--(1.2, 0.8);
\draw (1.5, 0)--(1.5, 0.8);

\node at (0.3, 0.25) [circle, fill, inner sep=1pt] {};
\node at (0.9, 0.55) [circle, fill, inner sep=1pt] {};
\node at (1.2, 0.25) [circle, fill, inner sep=1pt] {};
\node at (1.5, 0.55) [circle, fill, inner sep=1pt] {};
\end{tikzpicture}
}  . 
\end{equation}
We extensively use the following diagrammatic notation, in which we place a number to the left of a dot to indicate the total number of dots: 
\begin{equation}
\raisebox{-0.45cm}{\tikz{
\draw (0, 0)--(0, 1);

\coordinate[label = left:$k$] (A) at (0, 0.5);
\node at (A) [circle, fill, inner sep=1pt] {};
}} \,\, := \begin{cases}
\,\, \raisebox{-0.45cm}{\tikz{
\draw (0, 0)--(0, 1);
}} \qquad \textrm{if $k=0$,} \\
\,\, \raisebox{-0.45cm}{\tikz{
\draw (0, 0)--(0, 1);

\node at (0, 0.5) [circle, fill, inner sep=1pt] {};
}} \qquad \textrm{if $k=1$.}
\end{cases}
\end{equation}
Using $2n$ Majorana modes, one can form $n$ complex fermions via $c_k = (\gamma_{2k-1} - i \gamma_{2k})/2$ and $c_k^\dagger = (\gamma_{2k-1} + i \gamma_{2k})/2$. The vacuum state $\vert \textrm{vac} \rangle$ in TABLE~\ref{tab:majorana-diagrams} precisely denotes the normalized state satisfying $c_k \vert \textrm{vac} \rangle = 0$ for all $k$. The \textit{caps} and \textit{cups}, used in representing the vacuum state, can be understood as the pair-creation and pair-annihilation processes of Majoranas, respectively. Note that a factor $2^{-n/4}$, which multiplies the $n$-cap diagram, is called the \textit{scalar coefficient}; generally, a linear superposition of diagrams with associated scalar coefficients is possible. There are special linear combinations of diagrams, which we denote with distinct diagrammatic notations---most notably for the \textit{braids}\footnote{We use a braid to represent the braiding process between two Majoranas. This might seem odd, since, in our context, Majoranas are confined to one spatial dimension, a setting in which physical braiding is impossible. One could regard a braid as a special case of a scattering element. Alternatively, a braid can be interpreted as the two-dimensional projection of the worldline resulting from the physical braiding of Majoranas in (2+1)-dimensional spacetime, analogous to how a knot diagram represents a two-dimensional projection of a knot in three-dimensional space.} and the \textit{scattering elements} (see TABLE~\ref{tab:Majorana-diagram-expansion} for various examples of how diagrammatic elements expressed as linear combinations of other diagrammatic elements). In fact, braids are special kinds of scatterings where the scattering angle is $\theta = \pm \frac{\pi}{2}$; however, we use separate diagrammatic notations because braids play important roles in representing Clifford circuits. We refer to the scattering element that does not reduce to either a braid (with scattering angle $\theta = \pm \frac{\pi}{2}$) or a pair of dots ($\theta = \pm \pi$, see TABLE~\ref{tab:majorana-rewriting-rules}) as a \textit{generic} scattering element (i.e., one for which $\theta$ is \textit{not} an integer multiple of $\frac{\pi}{2}$). Finally, the $\dagger$-operation in a Majorana diagram corresponds to the vertical reflection, where each scattering angle is mapped to its negation ($\theta \mapsto -\theta$) and the scalar coefficient is mapped to its complex conjugation ($c \mapsto c^*$). 

\begin{table*}[ht!]
\centering
\begin{tabular}{|c|c|c|c|}
\cline{2-4}
\multicolumn{1}{c|}{} & \multicolumn{2}{c|}{state/operator} & \multirow{2}{*}{Majorana diagram} \\ \cline{2-3}
\multicolumn{1}{c|}{} & fermionic Fock space & qubit Hilbert space & \\ \hline \hline 
identity operator & $\openone : \mathcal{H}_n \to \mathcal{H}_n$ & $\openone : H_n \to H_n$ & \raisebox{-0.6cm}{
\tikz{
\draw (-1.2, 0)--(-1.2, 0.8);
\draw (-0.8, 0)--(-0.8, 0.8);

\draw (0.8, 0)--(0.8, 0.8);
\draw (1.2, 0)--(1.2, 0.8);

\node at (0.05, 0.4) {$\cdots$};

\node at (-1.2, -0.2) {\scriptsize $1$};
\node at (-0.8, -0.2) {\scriptsize $2$};
\node at (0.05, -0.2) {\scriptsize $\cdots$};
\node at (1.2, -0.2) {\scriptsize $2n$};
}} \rule[-1ex]{0pt}{6ex} \\ \hline 
Majorana operator & $\gamma_{j}$ & 
$\begin{cases}
\big[ \prod_{a=1}^{k-1} Z_a \big] X_k, & \textrm{if } j=2k-1 \\
i \big[ \prod_{a=1}^{k} Z_a \big] X_k, & \textrm{if } j=2k 
\end{cases}$ & \raisebox{-0.6cm}{
\tikz{
\draw (-1.2, 0)--(-1.2, 0.8);

\draw (-0.4, 0)--(-0.4, 0.8);
\draw (0, 0)--(0, 0.8);
\draw (0.4, 0)--(0.4, 0.8);

\draw (1.2, 0)--(1.2, 0.8);

\node at (-0.75, 0.4) {$\cdots$};
\node at (0.85, 0.4) {$\cdots$};

\node at (0, 0.4) [circle, fill, inner sep=1pt] {};

\node at (-1.2, -0.2) {\scriptsize $1$};
\node at (-0.6, -0.2) {\scriptsize $\cdots$};
\node at (0, -0.2) {\scriptsize $j$};
\node at (0.6, -0.2) {\scriptsize $\cdots$};
\node at (1.2, -0.2) {\scriptsize $2n$};
}} \rule[-1.5ex]{0pt}{7ex} \\ \hline 
parity operator & $i \gamma_{j} \gamma_{j+1}$ & 
$\begin{cases}
Z_{k}, & \textrm{if } j=2k-1 \\ 
X_{k} X_{k+1}, & \textrm{if } j=2k \\ 
\end{cases}$
& \raisebox{-0.6cm}{
\tikz{
\draw (-1.2, 0)--(-1.2, 0.8);
\draw (-0.4, 0)--(-0.4, 0.8);

\draw (0, 0)--(0, 0.8);
\draw (0.6, 0)--(0.6, 0.8);

\draw (1, 0)--(1, 0.8);
\draw (1.8, 0)--(1.8, 0.8);

\node at (-0.75, 0.4) {$\cdots$};
\node at (1.45, 0.4) {$\cdots$};

\node at (0, 0.4) [circle, fill, inner sep=1pt] {};
\node at (0.6, 0.4) [circle, fill, inner sep=1pt] {};

\node at (-1.2, -0.2) {\scriptsize $1$};
\node at (-0.6, -0.2) {\scriptsize $\cdots$};
\node at (0, -0.2) {\scriptsize $j$};
\node at (0.6, -0.2) {\scriptsize $j+1$};
\node at (1.3, -0.2) {\scriptsize $\cdots$};
\node at (1.8, -0.2) {\scriptsize $2n$};
}} \rule[-1ex]{0pt}{7ex} \\ \hline 
global parity operator & $\displaystyle \hat{P}_{n} = \prod_{k=1}^n (i \gamma_{2k-1} \gamma_{2k} )$ & $\displaystyle \prod_{k=1}^n Z_k$ & \raisebox{-0.6cm}{
\tikz{
\draw (-1.2, 0)--(-1.2, 0.8);
\draw (-0.8, 0)--(-0.8, 0.8);
\draw (-0.4, 0)--(-0.4, 0.8);
\draw (0, 0)--(0, 0.8);

\draw (1.4, 0)--(1.4, 0.8);
\draw (1.8, 0)--(1.8, 0.8);

\node at (0.75, 0.4) {$\cdots$};

\node at (-1.2, 0.4) [circle, fill, inner sep=1pt] {};
\node at (-0.8, 0.4) [circle, fill, inner sep=1pt] {};
\node at (-0.4, 0.4) [circle, fill, inner sep=1pt] {};
\node at (0, 0.4) [circle, fill, inner sep=1pt] {};
\node at (1.4, 0.4) [circle, fill, inner sep=1pt] {};
\node at (1.8, 0.4) [circle, fill, inner sep=1pt] {};

\node at (-1.2, -0.2) {\scriptsize $1$};
\node at (-0.8, -0.2) {\scriptsize $2$};
\node at (-0.4, -0.2) {\scriptsize $3$};
\node at (0, -0.2) {\scriptsize $4$};
\node at (0.75, -0.2) {\scriptsize $\cdots$};
\node at (1.8, -0.2) {\scriptsize $2n$};
}} \rule[-1ex]{0pt}{6ex} \\ \hline \hline 
\raisebox{-0.2cm}{\multirow{2}{*}{ket state}} & $\vert \textrm{vac} \rangle$ & $ \vert 0 \rangle^{\otimes n}$ & $2^{-n/4}$ \raisebox{-0.1cm}{\tikz{
\draw (0, 0) arc (180:0:0.3);
\draw (0.8, 0) arc (180:0:0.3);
\draw (2.05, 0) arc (180:0:0.3);

\node at (1.75, 0.1) {$\cdots$};
}} \rule[-1.4ex]{0pt}{4ex} \\ \cline{2-4} 
& $ (c_1^\dagger)^{p_1} (c_2^\dagger)^{p_2} \ldots (c_n^\dagger)^{p_n} \vert \textrm{vac} \rangle$ & $\vert p_1, p_2, \ldots, p_n \rangle$ & $2^{-n/4}$ \raisebox{-0.3cm}{\tikz{
\draw (0, 0) arc (180:0:0.3);
\draw (0.8, 0) arc (180:0:0.3);
\draw (2.05, 0) arc (180:0:0.3);

\draw (0, 0)--(0, -0.4);
\draw (0.6, 0)--(0.6, -0.4);
\draw (0.8, 0)--(0.8, -0.4);
\draw (1.4, 0)--(1.4, -0.4);
\draw (2.05, 0)--(2.05, -0.4);
\draw (2.65, 0)--(2.65, -0.4);

\node at (0, -0.25) [circle, fill, inner sep=1pt] {};
\node at (0.8, -0.15) [circle, fill, inner sep=1pt] {};
\node at (2.05, -0.05) [circle, fill, inner sep=1pt] {};

\node at (-0.2, -0.25) {\scriptsize $p_1$};
\node at (0.6, -0.15) {\scriptsize $p_2$};
\node at (1.85, -0.05) {\scriptsize $p_n$};

\node at (1.75, -0.28) {$\cdots$};
}} \rule[-2ex]{0pt}{6ex} \\ \hline 
\raisebox{-0.25cm}{\multirow{2}{*}{bra state}} & $\langle \textrm{vac} \vert$ & $ \langle 0 \vert^{\otimes n}$ & $2^{-n/4}$ \raisebox{-0.1cm}{\tikz{
\draw (0, 0) arc (-180:0:0.3);
\draw (0.8, 0) arc (-180:0:0.3);
\draw (2.05, 0) arc (-180:0:0.3);

\node at (1.75, -0.1) {$\cdots$};
}} \\ \cline{2-4} 
& $\langle \textrm{vac} \vert (c_n)^{p_n} \ldots (c_2)^{p_2} (c_1)^{p_1}$ & $\langle p_1, p_2, \ldots, p_n \vert$ & $2^{-n/4}$ \raisebox{-0.25cm}{\tikz{
\draw (0, -0.4) arc (-180:0:0.3);
\draw (0.8, -0.4) arc (-180:0:0.3);
\draw (2.05, -0.4) arc (-180:0:0.3);

\draw (0, 0)--(0, -0.4);
\draw (0.6, 0)--(0.6, -0.4);
\draw (0.8, 0)--(0.8, -0.4);
\draw (1.4, 0)--(1.4, -0.4);
\draw (2.05, 0)--(2.05, -0.4);
\draw (2.65, 0)--(2.65, -0.4);

\node at (0, -0.15) [circle, fill, inner sep=1pt] {};
\node at (0.8, -0.25) [circle, fill, inner sep=1pt] {};
\node at (2.05, -0.35) [circle, fill, inner sep=1pt] {};

\node at (-0.2, -0.15) {\scriptsize $p_1$};
\node at (0.6, -0.25) {\scriptsize $p_2$};
\node at (1.85, -0.35) {\scriptsize $p_n$};

\node at (1.75, -0.08) {$\cdots$};
}} \rule[-2.5ex]{0pt}{6ex} \\ \hline \hline 
pair-creation & $\mathcal{F}_{n, j} : \mathcal{H}_{n-1} \to \mathcal{H}_n$ & $F_{n, j} : H_{n-1} \to H_n$ & \raisebox{-0.6cm}{
\tikz{
\draw (-1.2, 0)--(-1.2, 0.8);
\draw (-0.9, 0)--(-0.9, 0.8);

\draw (-0.1, 0)--(-0.1, 0.8);
\draw (0.2, 0) arc (180:0:0.3);

\draw (1.1, 0)--(1.1, 0.8);
\draw (1.9, 0)--(1.9, 0.8);

\node at (-0.475, 0.4) {$\cdots$};
\node at (1.525, 0.4) {$\cdots$};

\node at (-1.2, -0.2) {\scriptsize $1$};
\node at (-0.9, -0.2) {\scriptsize $2$};
\node at (0.2, -0.2) {\scriptsize $j$};
\node at (0.8, -0.2) {\scriptsize $j+1$};
\node at (1.9, -0.2) {\scriptsize $2n$};
}} \rule[-0.5ex]{0pt}{6ex} \\ \hline 
pair-annihilation & $\mathcal{F}_{n, j}^\dagger : \mathcal{H}_n \to \mathcal{H}_{n-1}$ & $F_{n, j}^\dagger : H_n \to H_{n-1}$ & \raisebox{-0.5cm}{
\tikz{
\draw (-1.2, 0)--(-1.2, 0.8);
\draw (-0.9, 0)--(-0.9, 0.8);

\draw (-0.1, 0)--(-0.1, 0.8);
\draw (0.2, 0.8) arc (-180:0:0.3);

\draw (1.1, 0)--(1.1, 0.8);
\draw (1.9, 0)--(1.9, 0.8);

\node at (-0.475, 0.4) {$\cdots$};
\node at (1.525, 0.4) {$\cdots$};

\node at (-1.2, 1) {\scriptsize $1$};
\node at (-0.9, 1) {\scriptsize $2$};
\node at (0.2, 1) {\scriptsize $j$};
\node at (0.8, 1) {\scriptsize $j+1$};
\node at (1.9, 1) {\scriptsize $2n$};
}} \rule[-4.5ex]{0pt}{10ex} \\ \hline \hline 
positive braid & $\frac{e^{-i \frac{\pi}{8}}}{\sqrt{2}} ( 1 + i \gamma_j \gamma_{j + 1})$ & $\begin{cases}
e^{-i \frac{\pi}{8}} e^{i \frac{\pi}{4} Z_k}, & \textrm{if } j=2k-1 \\ 
e^{-i \frac{\pi}{8}} e^{i \frac{\pi}{4} X_k X_{k+1}}, & \textrm{if } j=2k
\end{cases}$ & \raisebox{-0.6cm}{\tikz{
\draw (-1, 0)--(-1, 0.8);
\draw (-0.3, 0)--(-0.3, 0.8);

\draw (0, 0)--(0.8, 0.8);
\draw (0, 0.8)--(0.3, 0.5);
\draw (0.8, 0)--(0.5, 0.3);

\draw (1.1, 0)--(1.1, 0.8);
\draw (1.8, 0)--(1.8, 0.8);

\node at (-0.65, 0.4) {$\cdots$};
\node at (1.5, 0.4) {$\cdots$};

\node at (-1, -0.2) {\scriptsize $1$};
\node at (-0.5, -0.2) {\scriptsize $\cdots$};
\node at (0, -0.2) {\scriptsize $j$};
\node at (0.8, -0.2) {\scriptsize $j+1$};
\node at (1.4, -0.2) {\scriptsize $\cdots$};
\node at (1.8, -0.2) {\scriptsize $2n$};
}} \rule[-0.5ex]{0pt}{6ex} \\ \hline 
negative braid & $\frac{e^{i \frac{\pi}{8}}}{\sqrt{2}} ( 1 - i \gamma_j \gamma_{j + 1})$ & $\begin{cases}
e^{i \frac{\pi}{8}} e^{-i \frac{\pi}{4} Z_k}, & \textrm{if } j=2k-1 \\ 
e^{i \frac{\pi}{8}} e^{-i \frac{\pi}{4} X_k X_{k+1}}, & \textrm{if } j=2k
\end{cases}$ & \raisebox{-0.6cm}{\tikz{
\draw (-1, 0)--(-1, 0.8);
\draw (-0.3, 0)--(-0.3, 0.8);

\draw (0, 0)--(0.3, 0.3);
\draw (0.5, 0.5)--(0.8, 0.8);
\draw (0.8, 0)--(0, 0.8);

\draw (1.1, 0)--(1.1, 0.8);
\draw (1.8, 0)--(1.8, 0.8);

\node at (-0.65, 0.4) {$\cdots$};
\node at (1.5, 0.4) {$\cdots$};

\node at (-1, -0.2) {\scriptsize $1$};
\node at (-0.5, -0.2) {\scriptsize $\cdots$};
\node at (0, -0.2) {\scriptsize $j$};
\node at (0.8, -0.2) {\scriptsize $j+1$};
\node at (1.4, -0.2) {\scriptsize $\cdots$};
\node at (1.8, -0.2) {\scriptsize $2n$};
}} \rule[-0.5ex]{0pt}{6ex} \\ \hline 
scattering & $\frac{1 + e^{i \theta}}{2} + \frac{1 - e^{i \theta}}{2} i \gamma_j \gamma_{j+1}$ & $\begin{cases}
e^{i \frac{\theta}{2}} e^{-i \frac{\theta}{2} Z_k}, & \textrm{if } j=2k-1 \\ 
e^{i \frac{\theta}{2}} e^{-i \frac{\theta}{2} X_k X_{k+1}}, & \textrm{if } j=2k 
\end{cases}$ & \raisebox{-0.6cm}{\tikz{
\draw (-1, 0)--(-1, 0.8);
\draw (-0.3, 0)--(-0.3, 0.8);

\draw (0, 0)--(0.2, 0.2);
\draw (0.6, 0.6)--(0.8, 0.8);
\draw (0.8, 0)--(0.6, 0.2);
\draw (0.2, 0.6)--(0, 0.8);

\draw ({0.4+0.2*sqrt(2)}, 0.4) arc (0:360:{0.2*sqrt(2)});

\draw (1.1, 0)--(1.1, 0.8);
\draw (1.8, 0)--(1.8, 0.8);

\node at (-0.65, 0.4) {$\cdots$};
\node at (1.5, 0.4) {$\cdots$};

\node at (0.4, 0.4) {$\theta_\updownarrow$};

\node at (-1, -0.2) {\scriptsize $1$};
\node at (-0.5, -0.2) {\scriptsize $\cdots$};
\node at (0, -0.2) {\scriptsize $j$};
\node at (0.8, -0.2) {\scriptsize $j+1$};
\node at (1.4, -0.2) {\scriptsize $\cdots$};
\node at (1.8, -0.2) {\scriptsize $2n$};
}} \rule[-0.5ex]{0pt}{6ex} \\ \hline 
\end{tabular}
\caption{We present a dictionary between quantum state and operators in the (fermionic) Fock space, in the qubit Hilbert space (obtained via the Jordan-Wigner transformation), and in the Majorana diagram. The identity evolution of Majoranas is represented by parallel lines and a Majorana operator is depicted as a dot in the Majorana diagrams. The parity operator between two Majoranas is represented by two dots on the equal time slice, and the global parity operators is represented by dots appearing simultaneously on all lines. When a horizontal cut (i.e., equal time slice) intersects $2n$ Majoranas, we associate it with the Fock space of $n$ complex fermions, denoted by $\mathcal{H}_n$, and equivalently, the $n$-qubit Hilbert space $H_n$. The vacuum ket and bra states are represented by $n$-caps and $n$-cups, respectively. Computational basis states are obtained by adding dots in an appropriate order to the vacuum state diagram. Cap and cup are also used to represent a pair-creation and pair-annihilation process, which, as operators, are denoted as $\mathcal{F}_{n, j}$ ($F_{n, j}$) and $\mathcal{F}_{n, j}^\dagger$ ($F_{n, j}^\dagger$) in the table, with $F_{n, j} \vert p_1, \ldots, p_{j-1}, p_j, \ldots, p_{n-1} \rangle = 2^{1/4} \vert p_1, \ldots, p_{j-1}, 0, p_j, \ldots, p_{n-1} \rangle$ if $j=2k-1$ and $F_{n, j} \vert p_1, \ldots, p_{j-1}, p_j, \ldots, p_{n-1} \rangle = 2^{-1/4} \sum_{p = 0, 1} \vert p_1, \ldots, p_{j-1}, p, p + p_j, \ldots, p_{n-1} \rangle$ if $j = 2k$. We use separate diagrammatic notations to denote special linear combinations of diagrams, notably, the braids and the scattering elements. Braids are special cases of the scattering elements, where the scattering angle $\theta$ equals $-\frac{\pi}{2}$ for a positive braid and $+\frac{\pi}{2}$ for a negative braid\footnote{A crossing with an orientation (i.e., a direction along each strand, typically indicated by arrows) is defined as positive (negative) if the minimal rotation from the overpassing strand to the underpassing strand, which aligns the strands' orientations, is counterclockwise (clockwise). A braid is considered positive (negative) when its corresponding crossing is positive (negative) under the following top-to-bottom orientation convention: for two strands near the crossing in a braid, assign each an orientation following the direction of the time flow, i.e., from top to bottom.}.}
\label{tab:majorana-diagrams}
\end{table*}

Using these diagrammatic elements, we now formally define the \textit{Majorana diagrams}. A Majorana diagram is a parity-even diagram, i.e., one that contains an even number of dots, consisting of strings, caps, cups, dots, braids, and scatterings. We often encapsulate a Majorana diagram within a box. This box is connected to Majorana strings, referred to as \textit{open Majorana strings}, with those located in the upper part designated as \textit{input strings} and those in the lower part as \textit{output strings}, in accordance with our top-to-bottom time convention. Note that the numbers of both input and output legs must be even. An example of a Majorana diagram with $4$ input strings and $6$ output strings, denoted as $f$ below, is: 
\begin{equation}
\raisebox{-0.6cm}{\tikz{
\draw (0.05, 1)--(0.05, 0.7);
\draw (0.35, 1)--(0.35, 0.7);
\draw (0.65, 1)--(0.65, 0.7);
\draw (0.95, 1)--(0.95, 0.7);

\draw (-0.2, 0.7) rectangle (1.2, -0.1);
\node at (0.5, 0.3) {$f$};

\draw (0, -0.1)--(0, -0.4);
\draw (0.2, -0.1)--(0.2, -0.4);
\draw (0.4, -0.1)--(0.4, -0.4);
\draw (0.6, -0.1)--(0.6, -0.4);
\draw (0.8, -0.1)--(0.8, -0.4);
\draw (1, -0.1)--(1, -0.4);
}} \,\, = \,\, \raisebox{-0.6cm}{\tikz{
\draw (0.05, 1) to[out=-90, in=110] (0.3, 0.13);
\draw (0.33, 0.05) to[out=-70, in=90] (0.4, -0.4);

\draw (0.35, 1) to[out=-90, in=75] (0.23, 0.43);
\draw (0.18, 0.3) to[out=-105, in=90] (0, -0.4);

\draw (0.2, -0.1) arc (180:0:0.2);

\draw (1.1, 0.6) arc (0:360:0.25);
\node at (0.85, 0.6) {$\theta_\updownarrow$};

\draw (0.65, 1) to[out=-90, in=135] (+{0.85-0.25*cos(45)}, +{0.6+0.25*sin(45)});
\draw (0.95, 1) to[out=-90, in=45] (+{0.85+0.25*cos(45)}, +{0.6+0.25*sin(45)});

\draw (0.8, -0.4) to[out=90, in=-135] (+{0.85-0.25*cos(45)}, +{0.6-0.25*sin(45)});
\draw (1, -0.4) to[out=90, in=-45] (+{0.85+0.25*cos(45)}, +{0.6-0.25*sin(45)});

\draw (0.2, -0.1)--(0.2, -0.4);
\draw (0.6, -0.1)--(0.6, -0.4);

\node at (0.2, -0.2) [circle, fill, inner sep=1pt] {};
\node at (0.6, -0.2) [circle, fill, inner sep=1pt] {};
}} \,\, . 
\end{equation}
We denote a Majorana diagram that has open Majorana strings as an \textit{open Majorana diagram} and one does not have any open Majorana strings as a \textit{closed Majorana diagram}, with the latter evaluating to a number. Since Majorana diagrams are always assumed to be parity-even, any two Majorana diagrams commute in the following sense: 
\begin{equation}
\label{eq:f-g-equal-time}
\raisebox{-0.55cm}{\tikz{
\draw (0.2, 1)--(0.2, 0.65);
\draw (0.8, 1)--(0.8, 0.65);

\draw (0, 0.65) rectangle (1, 0.35);
\node at (0.5, 0.5) {\scriptsize $f$};

\draw (0.1, 0)--(0.1, 0.35);
\draw (0.9, 0)--(0.9, 0.35);

\node at (0.55, 0.95) {$\cdots$};
\node at (0.55, 0.05) {$\cdots$};

\draw (1.4, 0)--(1.4, 0.35);
\draw (2, 0)--(2, 0.35);

\draw (1.2, 0.65) rectangle (2.2, 0.35);
\node at (1.7, 0.5) {\scriptsize $g$};

\draw (1.3, 1)--(1.3, 0.65);
\draw (2.1, 1)--(2.1, 0.65);

\node at (1.75, 0.95) {$\cdots$};
\node at (1.75, 0.05) {$\cdots$};
}} \,\, := \,\, \raisebox{-0.55cm}{\tikz{
\draw (0.2, 1)--(0.2, 0.85);
\draw (0.8, 1)--(0.8, 0.85);

\draw (0, 0.85) rectangle (1, 0.55);
\node at (0.5, 0.7) {\scriptsize $f$};

\draw (0.1, 0)--(0.1, 0.55);
\draw (0.9, 0)--(0.9, 0.55);

\node at (0.55, 0.95) {$\cdots$};
\node at (0.55, 0.05) {$\cdots$};

\draw (1.4, 0)--(1.4, 0.15);
\draw (2, 0)--(2, 0.15);

\draw (1.2, 0.45) rectangle (2.2, 0.15);
\node at (1.7, 0.3) {\scriptsize $g$};

\draw (1.3, 1)--(1.3, 0.45);
\draw (2.1, 1)--(2.1, 0.45);

\node at (1.75, 0.95) {$\cdots$};
\node at (1.75, 0.05) {$\cdots$};
}} \,\, = \,\, \raisebox{-0.55cm}{\tikz{
\draw (0.2, 1)--(0.2, 0.45);
\draw (0.8, 1)--(0.8, 0.45);

\draw (0, 0.45) rectangle (1, 0.15);
\node at (0.5, 0.3) {\scriptsize $f$};

\draw (0.1, 0)--(0.1, 0.15);
\draw (0.9, 0)--(0.9, 0.15);

\node at (0.55, 0.95) {$\cdots$};
\node at (0.55, 0.05) {$\cdots$};

\draw (1.4, 0)--(1.4, 0.55);
\draw (2, 0)--(2, 0.55);

\draw (1.2, 0.85) rectangle (2.2, 0.55);
\node at (1.7, 0.7) {\scriptsize $g$};

\draw (1.3, 1)--(1.3, 0.85);
\draw (2.1, 1)--(2.1, 0.85);

\node at (1.75, 0.95) {$\cdots$};
\node at (1.75, 0.05) {$\cdots$};
}} \,\, , 
\end{equation}
where $f$ and $g$ denote two Majorana diagrams and, if there is no ambiguity, we draw two boxes on the same time slice as shown on the LHS.

\begin{table}[ht!]
\centering
\begin{tabular}{|c|l|}
\hline 
\raisebox{-0.4cm}{\multirow{2}{*}{cup and cap}} & \raisebox{-0.33cm}{
\begin{tikzpicture}[scale=0.8]
\coordinate[label = left:$l$] (A) at (0, 0.5);
\coordinate[label = left:$l$] (B) at (0.7, 0.5);

\draw (0,0) --(0,1);
\draw (0.7,0) --(0.7,1);

\node at (A) [circle, fill, inner sep=1pt] {};
\node at (B) [circle, fill, inner sep=1pt] {};
\end{tikzpicture}} $\displaystyle = \frac{1}{\sqrt{2}} \sum_{l=0, 1} (-1)^{k l}$ \raisebox{-0.55cm}{ 
\begin{tikzpicture}[scale=0.8]
\coordinate[label = left:$k$] (A) at (0.7, -0.05);
\coordinate[label = left:$k$] (B) at (0.7, 0.85);

\draw (0,0.8) --(0,1);
\draw (0.7,0.8) --(0.7,1);
\draw (0,-0.2) --(0,0);
\draw (0.7,-0.2) --(0.7,0);

\draw (0.7,0) arc (0:180:0.35);
\draw (0.7,0.8) arc (0:-180:0.35);

\node at (A) [circle, fill, inner sep=1pt] {};
\node at (B) [circle, fill, inner sep=1pt] {};
\end{tikzpicture}} \\ \cline{2-2} 
& \raisebox{-0.52cm}{
\begin{tikzpicture}[scale=0.8]
\coordinate[label = left:$k$] (A) at (0.7, -0.05);
\coordinate[label = left:$k$] (B) at (0.7, 0.85);

\draw (0,0.8)--(0,1);
\draw (0.7,0.8)--(0.7,1);
\draw (0,-0.2)--(0,0);
\draw (0.7,-0.2)--(0.7,0);

\draw (0.7,0) arc (0:180:0.35);
\draw (0.7,0.8) arc (0:-180:0.35);

\node at (A) [circle, fill, inner sep=1pt] {};
\node at (B) [circle, fill, inner sep=1pt] {};
\end{tikzpicture}} $\displaystyle = \frac{1}{\sqrt{2}} \sum_{l=0, 1} (-1)^{k l}$ \raisebox{-0.33cm}{
\begin{tikzpicture}[scale=0.8]
\coordinate[label = left:$l$] (A) at (0, 0.5);
\coordinate[label = left:$l$] (B) at (0.7, 0.5);

\draw (0,0)--(0,1);
\draw (0.7,0)--(0.7,1);

\node at (A) [circle, fill, inner sep=1pt] {};
\node at (B) [circle, fill, inner sep=1pt] {};
\end{tikzpicture}} \\ \hline \hline 
positive braid & $\begin{aligned}
\raisebox{-0.3cm}{
\begin{tikzpicture}[scale=0.9]
\draw (0, 0)--(0.8, 0.8);
\draw (0, 0.8)--(0.3, 0.5);
\draw (0.8, 0)--(0.5, 0.3);
\end{tikzpicture}} &= \frac{e^{-i \frac{\pi}{8}}}{\sqrt{2}} \bigg(
\raisebox{-0.3cm}{
\begin{tikzpicture}[scale=0.9]
\draw (0, 0)--(0, 0.8);
\draw (0.7, 0)--(0.7, 0.8);
\end{tikzpicture}} + i \raisebox{-0.3cm}{
\begin{tikzpicture}[scale=0.9]
\draw (0, 0)--(0, 0.8);
\draw (0.7, 0)--(0.7, 0.8);

\node at (0, 0.4) [circle, fill, inner sep=1pt] {};
\node at (0.7, 0.4) [circle, fill, inner sep=1pt] {};
\end{tikzpicture}} \bigg) \\ 
&= \frac{e^{i \frac{\pi}{8}}}{\sqrt{2}} \bigg(
\raisebox{-0.3cm}{
\begin{tikzpicture}[scale=0.9]
\draw (0, 0) arc (180:0:0.35);
\draw (0, 0.8) arc (-180:0:0.35);
\end{tikzpicture} } - i \raisebox{-0.3cm}{
\begin{tikzpicture}[scale=0.9]
\draw (0, 0) arc (180:0:0.35);
\draw (0, 0.8) arc (-180:0:0.35);

\node at ({0.35 + 0.35*cos(40)}, {0.35*sin(40)}) [circle, fill, inner sep=1pt] {};
\node at ({0.35 + 0.35*cos(40)}, {0.8 - 0.35*sin(40)}) [circle, fill, inner sep=1pt] {};
\end{tikzpicture} } \bigg)
\end{aligned}$ \rule{0pt}{7ex} \rule[-6ex]{0pt}{30pt} \\ \hline 
negative braid & $\begin{aligned}
\raisebox{-0.3cm}{
\begin{tikzpicture}[scale=0.9]
\draw (0, 0)--(0.3, 0.3);
\draw (0.5, 0.5)--(0.8, 0.8);
\draw (0.8, 0)--(0, 0.8);
\end{tikzpicture} }
&= \frac{e^{i \frac{\pi}{8}}}{\sqrt{2}} \bigg(
\raisebox{-0.3cm}{
\begin{tikzpicture}[scale=0.9]
\draw (0, 0)--(0, 0.8);
\draw (0.7, 0)--(0.7, 0.8);
\end{tikzpicture} } - i \raisebox{-0.3cm}{
\begin{tikzpicture}[scale=0.9]
\draw (0, 0)--(0, 0.8);
\draw (0.7, 0)--(0.7, 0.8);

\node at (0, 0.4) [circle, fill, inner sep=1pt] {};
\node at (0.7, 0.4) [circle, fill, inner sep=1pt] {};
\end{tikzpicture} } \bigg) \nonumber \\
&= \frac{e^{-i \frac{\pi}{8}}}{\sqrt{2}} \bigg(
\raisebox{-0.33cm}{
\begin{tikzpicture}[scale=0.9]
\draw (0, 0) arc (180:0:0.35);
\draw (0, 0.8) arc (-180:0:0.35);
\end{tikzpicture} } + i \raisebox{-0.3cm}{
\begin{tikzpicture}[scale=0.9]
\draw (0, 0) arc (180:0:0.35);
\draw (0, 0.8) arc (-180:0:0.35);

\node at ({0.35 + 0.35*cos(40)}, {0.35*sin(40)}) [circle, fill, inner sep=1pt] {};
\node at ({0.35 + 0.35*cos(40)}, {0.8 - 0.35*sin(40)}) [circle, fill, inner sep=1pt] {};
\end{tikzpicture} } \bigg)
\end{aligned} $ \rule{0pt}{7ex} \rule[-6ex]{0pt}{30pt} \\ \hline \hline 
\raisebox{-0.75cm}{\multirow{2}{*}{scattering}} & $\begin{aligned}
\raisebox{-0.3cm}{
\begin{tikzpicture}[scale=0.9]
\node at (0.4, 0.4) {$\theta_\updownarrow$};

\draw ({0.4+0.2*sqrt(2)}, 0.4) arc (0:360:{0.2*sqrt(2)});

\draw (0, 0)--(0.2, 0.2);
\draw (0.6, 0.6)--(0.8, 0.8);
\draw (0, 0.8)--(0.2, 0.6);
\draw (0.8, 0)--(0.6, 0.2);
\end{tikzpicture} } &= \frac{1 + e^{i \theta}}{2} 
\raisebox{-0.3cm}{
\begin{tikzpicture}[scale=0.9]
\draw (0, 0)--(0, 0.8);
\draw (0.7, 0)--(0.7, 0.8);
\end{tikzpicture} } + \frac{1 - e^{i \theta}}{2} 
\raisebox{-0.3cm}{
\begin{tikzpicture}[scale=0.9]
\draw (0, 0)--(0, 0.8);
\draw (0.7, 0)--(0.7, 0.8);

\node at (0, 0.4) [circle, fill, inner sep=1pt] {};
\node at (0.7, 0.4) [circle, fill, inner sep=1pt] {};
\end{tikzpicture} } \\ 
&= A_\theta \raisebox{-0.3cm}{
\begin{tikzpicture}[scale=0.9]
\draw (0, 0)--(0.8, 0.8);
\draw (0, 0.8)--(0.3, 0.5);
\draw (0.8, 0)--(0.5, 0.3);
\end{tikzpicture} } + B_\theta \raisebox{-0.3cm}{
\begin{tikzpicture}[scale=0.9]
\draw (0, 0)--(0.3, 0.3);
\draw (0.5, 0.5)--(0.8, 0.8);
\draw (0, 0.8)--(0.8, 0);
\end{tikzpicture} }
\end{aligned}$ \rule{0pt}{7ex} \rule[-6ex]{0pt}{30pt} \\ \cline{2-2} 
& $\begin{aligned}
\raisebox{-0.3cm}{
\begin{tikzpicture}[scale=0.9]
\node at (0.4, 0.4) {$\theta_\leftrightarrow$};

\draw ({0.4+0.2*sqrt(2)}, 0.4) arc (0:360:{0.2*sqrt(2)});

\draw (0, 0)--(0.2, 0.2);
\draw (0.6, 0.6)--(0.8, 0.8);
\draw (0, 0.8)--(0.2, 0.6);
\draw (0.8, 0)--(0.6, 0.2);
\end{tikzpicture} } &= \frac{1 + e^{i \theta}}{2} \raisebox{-0.3cm}{
\begin{tikzpicture}[scale=0.9]
\draw (0, 0) arc (180:0:0.35);
\draw (0, 0.8) arc (-180:0:0.35);
\end{tikzpicture} } + \frac{1 - e^{i \theta}}{2} 
\raisebox{-0.3cm}{
\begin{tikzpicture}[scale=0.9]
\draw (0, 0) arc (180:0:0.35);
\draw (0, 0.8) arc (-180:0:0.35);

\node at ({0.35 + 0.35*cos(40)}, {0.35*sin(40)}) [circle, fill, inner sep=1pt] {};
\node at ({0.35 + 0.35*cos(40)}, {0.8 - 0.35*sin(40)}) [circle, fill, inner sep=1pt] {};
\end{tikzpicture} } \\ 
&= B_\theta \raisebox{-0.3cm}{
\begin{tikzpicture}[scale=0.9]
\draw (0, 0)--(0.8, 0.8);
\draw (0, 0.8)--(0.3, 0.5);
\draw (0.8, 0)--(0.5, 0.3);
\end{tikzpicture} } + A_\theta\raisebox{-0.3cm}{
\begin{tikzpicture}[scale=0.9]
\draw (0, 0)--(0.3, 0.3);
\draw (0.5, 0.5)--(0.8, 0.8);
\draw (0, 0.8)--(0.8, 0);
\end{tikzpicture} }
\end{aligned}$ \rule{0pt}{7ex} \rule[-6ex]{0pt}{30pt} \\ \hline \hline 
\raisebox{-0.33cm}{\multirow{2}{*}{scattering*}} & $\begin{aligned}
\raisebox{-0.3cm}{
\begin{tikzpicture}[scale=0.9]
\node at (0.4, 0.4) {\scriptsize $\phi_\updownarrow$};

\draw ({0.4+0.16*sqrt(2)}, 0.4) arc (0:360:{0.16*sqrt(2)});
\draw ({0.4+0.2*sqrt(2)}, 0.4) arc (0:360:{0.2*sqrt(2)});

\draw (0, 0)--(0.2, 0.2);
\draw (0.6, 0.6)--(0.8, 0.8);
\draw (0, 0.8)--(0.2, 0.6);
\draw (0.8, 0)--(0.6, 0.2);
\end{tikzpicture} } = \frac{1 + e^{\phi}}{2} 
\raisebox{-0.3cm}{
\begin{tikzpicture}[scale=0.9]
\draw (0, 0)--(0, 0.8);
\draw (0.7, 0)--(0.7, 0.8);
\end{tikzpicture} } + \frac{1 - e^{\phi}}{2} 
\raisebox{-0.3cm}{
\begin{tikzpicture}[scale=0.9]
\draw (0, 0)--(0, 0.8);
\draw (0.7, 0)--(0.7, 0.8);

\node at (0, 0.4) [circle, fill, inner sep=1pt] {};
\node at (0.7, 0.4) [circle, fill, inner sep=1pt] {};
\end{tikzpicture} }
\end{aligned}$ \rule[-3ex]{0pt}{3pt} \\ \cline{2-2} 
& $\begin{aligned}
\raisebox{-0.3cm}{
\begin{tikzpicture}[scale=0.9]
\node at (0.4, 0.4) {\scriptsize $\phi_\leftrightarrow$};

\draw ({0.4+0.16*sqrt(2)}, 0.4) arc (0:360:{0.16*sqrt(2)});
\draw ({0.4+0.2*sqrt(2)}, 0.4) arc (0:360:{0.2*sqrt(2)});

\draw (0, 0)--(0.2, 0.2);
\draw (0.6, 0.6)--(0.8, 0.8);
\draw (0, 0.8)--(0.2, 0.6);
\draw (0.8, 0)--(0.6, 0.2);
\end{tikzpicture} } = \frac{1 + e^{\phi}}{2} 
\raisebox{-0.3cm}{
\begin{tikzpicture}[scale=0.9]
\draw (0, 0) arc (180:0:0.35);
\draw (0, 0.8) arc (-180:0:0.35);
\end{tikzpicture} } + \frac{1 - e^{\phi}}{2} 
\raisebox{-0.3cm}{
\begin{tikzpicture}[scale=0.9]
\draw (0, 0) arc (180:0:0.35);
\draw (0, 0.8) arc (-180:0:0.35);

\node at ({0.35 + 0.35*cos(40)}, {0.35*sin(40)}) [circle, fill, inner sep=1pt] {};
\node at ({0.35 + 0.35*cos(40)}, {0.8 - 0.35*sin(40)}) [circle, fill, inner sep=1pt] {};
\end{tikzpicture} } 
\end{aligned}$ \rule[-3ex]{0pt}{3pt} \\ \hline 
\end{tabular}
\caption{Examples of how diagrammatic elements can be expanded as superpositions of different diagrammatic elements. Here, we present two types of scattering elements that differ by an arrow appearing in the subscript of the \textit{scattering angle} $\theta \in \mathbb{C}$ (which is $2\pi$-periodicity). The arrow indicates the direction along which the string endpoints are ``connected.'' When the arrow is omitted, we assume that it aligns with the time direction. Note that when the scattering angle equals $\pm \pi/2$, the scattering reduces to a braid up to a $\mathrm{U}(1)$ factor. We also present expansions of the scattering elements as sums of braids, where $A_\theta := \frac{e^{-i \frac{\pi}{8}}}{2} \big( 1 + i e^{i \theta} \big)$ and $B_\theta := \frac{e^{i \frac{\pi}{8}}}{2} \big( 1 - i e^{i \theta} \big)$. In the last two columns, we introduce another scattering element, which reduces to the ordinary scattering if we set $\phi = i \theta$. While the ordinary scattering elements are useful when expressing unitary quantum circuits, these elements turn out to be useful when expressing non-unitary components as well as expressing the partition function of the Ising model.} 
\label{tab:Majorana-diagram-expansion}
\end{table}

\begin{table*}
\includegraphics[width=0.65\textwidth]{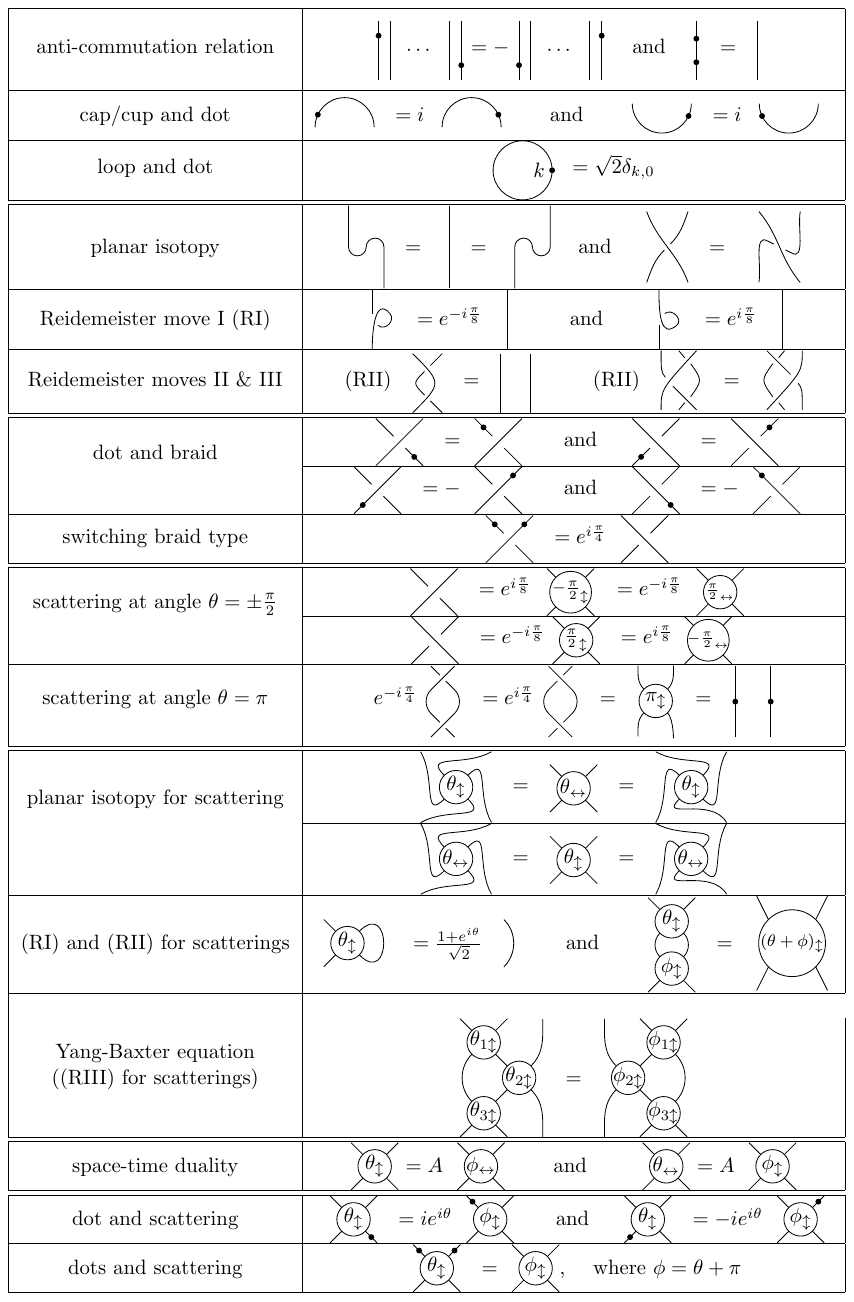}
\caption{Collection of \textit{diagrammatic rewriting rules} for Majorana diagrams. A dot can relocate its position in a cap or cup, introducing a $i$ factor, which exemplifies spin-statistics theorem. The Reidemeister move I applies a self-crossing based on the \textit{writhe}\footnote{The writhe of a self-crossing equals $+1$ ($-1$) if the crossing is positive (negative), with the orientation determined by the direction of traversing along the strand following the self-crossing rather than the top-to-bottom convention previously used for the braid. Therefore, the positivity of the braid at a self-crossing does not necessarily coincide with the positivity of the writhe at that self-crossing.}. A dot can pass under or over a braid, potentially up to $-1$ factor. One can even switch the braid type at the expense of using two dots. A scattering element with the scattering angle $\theta$ that is an integer multiple of $\frac{\pi}{2}$ can be reduced to another diagrammatic element; therefore, we call a scattering element \textit{generic} if the angle $\theta$ is \textit{not} a multiple of $\frac{\pi}{2}$. Planar isotopy and Reidemeister moves for braids can be appropriately generalized to scatterings. In particular, the generalized form of Reidemeister move III for scatterings corresponds to the famous \textit{Yang-Baxter equation}. In the space-time duality, $A$ and $\phi$ are complex numbers satisfying $A = \frac{1  + e^{i \theta}}{2}$ and $e^{i \phi} = \frac{1-e^{i \theta}}{1 + e^{i \theta}}$. A dot can also ``pass'' a scattering element, resulting in a change to the scattering angle (here, $\phi = \pi-\theta$). Similarly, two dots can be absorbed into a scattering element, also resulting in a change to the scattering angle.}
\label{tab:majorana-rewriting-rules}
\end{table*}

While it is possible to describe quantum processes solely in terms of states and Majorana operators, diagrammatic notation often offers an intuitive understanding and greater flexibility in manipulation. For example, the Reidemeister move III (RIII) presented in TABLE~\ref{tab:majorana-rewriting-rules} can be intuitively understood using braid notation, which is far less straightforward when expressed in terms of Majorana operators: 
\begin{equation}
e^{-\frac{\pi}{4} \gamma_1 \gamma_2} e^{-\frac{\pi}{4} \gamma_2 \gamma_3} e^{-\frac{\pi}{4} \gamma_1 \gamma_2} = e^{-\frac{\pi}{4} \gamma_2 \gamma_3} e^{-\frac{\pi}{4} \gamma_1 \gamma_2} e^{-\frac{\pi}{4} \gamma_2 \gamma_3} . 
\end{equation}
Furthermore, identities describing different quantum processes (so that the Majorana operators describing those processes look drastically different from each other) can be cast into the same diagrammatic identity: 
\begin{equation}
\raisebox{-0.35cm}{
\begin{tikzpicture}[scale=1.7]
\draw (0, 0.1)--(0, 0) to[out=-90, in=-90, looseness=1.5] (0.6, 0)--(0.6, 0.1);

\draw (0, -0.4)--(0, -0.3) to[out=90, in=-120] (0.03, -0.2);
\draw (0.08, -0.13) to[out=50, in=130] (0.52, -0.13);
\draw (0.6, -0.4)--(0.6, -0.3) to[out=90, in=-60] (0.57, -0.2);
\end{tikzpicture}
} \, = \raisebox{-0.35cm}{
\begin{tikzpicture}[scale=1.7]
\draw (0, -0.1) to[out=-90, in=-90, looseness=1.2] (0.6, -0.1);
\draw (0, -0.6) to[out=90, in=90, looseness=1.2] (0.6, -0.6);
\end{tikzpicture}
} \qquad \textrm{vs} \qquad \raisebox{-0.4cm}{
\begin{tikzpicture}[scale=1.7]
\draw (0, 0.6) to[out=-40, in=40, looseness=1.7] (0, 0);

\draw (0.3, 0.6) to[out=-140, in=42] (0.18, 0.49);
\draw (0.124, 0.43) to[out=-128, in=128] (0.124, 0.17);
\draw (0.3, 0) to[out=140, in=-42] (0.18, 0.11);
\end{tikzpicture}
} = \, \raisebox{-0.4cm}{
\begin{tikzpicture}[scale=1.7]
\draw (0, 0.6)--(0, 0);
\draw (0.3, 0.6)--(0.3, 0);
\end{tikzpicture}
} \,\, , 
\end{equation}
where the former involves pair-creation and annihilation process in addition to the braiding processes present in both cases. 

In TABLE~\ref{tab:majorana-rewriting-rules}, we summarize the \textit{diagrammatic rewriting rules} for Majorana diagrams, which are local rewritings of diagrams that preserves their semantic content. While most of the rewriting rules are self-explanatory, let us mention the Yang-Baxter equation (YBE) which involves three scatterings. The equality in the YBE holds in the following sense: given a triple $(\theta_1, \theta_2, \theta_3) \in \mathbb{C}^3$, except at a set of measure-zero points, there exists another triple $(\phi_1, \phi_2, \phi_3) \in \mathbb{C}^3$ that satisfies the equality. Conversely, given a triple $(\phi_1, \phi_2, \phi_3) \in \mathbb{C}^3$, except at another set of measure-zero points, there exists another triple $(\phi_1, \phi_2, \phi_3) \in \mathbb{C}^3$ that satisfies the equality. When two of a triple $(\theta_1, \theta_2, \theta_3)$ correspond to two positive (or negative) braids, the equality holds in the expected manner: e.g., 
\begin{equation}
\label{eq:YBE-two-braids}
\raisebox{-0.9cm}{\tikz{
\draw (0, 1)--(0.2, 0.8);
\draw (0.6, 0.4)--(1.2, -0.2) to[out=-45, in=90] (1.4, -1);

\draw (0.8, 1)--(0.6, 0.8);
\draw (0.2, 0.4) to[out=-135, in=135] (0.2, -0.4)--(0.8, -1);

\draw (1.4, 1) to[out=-90, in=45] (1.2, 0.2)--(1.1, 0.1);
\draw (0.5, -0.5)--(0.9, -0.1);
\draw (0, -1)--(0.3, -0.7);

\draw ({0.4+0.2*sqrt(2)}, 0.6) arc (0:360:{0.2*sqrt(2)});

\node at (0.4, 0.6) {$\theta_{\updownarrow}$};
}} \,\, = \,\, \raisebox{-0.9cm}{\tikz{
\draw (0, 1) to[out=-90, in=135] (0.2, 0.2)--(0.8, -0.4);
\draw (1.4, -1)--(1.2, -0.8);

\draw (0.6, 1)--(1.2, 0.4) to[out=-45, in=45] (1.2, -0.4);
\draw (0.6, -1)--(0.8, -0.8);

\draw (1.4, 1)--(1.1, 0.7);
\draw (0.9, 0.5)--(0.5, 0.1);
\draw (0, -1) to[out=90, in=-135] (0.2, -0.2)--(0.3, -0.1);

\draw ({1+0.2*sqrt(2)}, -0.6) arc (0:360:{0.2*sqrt(2)});

\node at (1, -0.6) {$\theta_{\updownarrow}$};
}} \,\, 
\end{equation}
and all the variants are also satisfied. 

As a side note, in principle, it is possible to substitute the scalar coefficient in a Majorana diagram with another Majorana diagram using the identity: 
\begin{equation}
\raisebox{-0.63cm}{
\tikz{
\node at (0, 0) {$\theta_\updownarrow$};

\draw ({0.2*sqrt(2)}, 0) arc (0:360:{0.2*sqrt(2)});

\draw [looseness=5] (-0.2, 0.2) to[out=135, in=-135] (-0.2, -0.2);
\draw [looseness=5] (0.2, 0.2) to[out=45, in=-45] (0.2, -0.2);
}} = 1 + e^{i \theta} , 
\end{equation}
where $\theta$ can be an arbitrary complex number, allowing one to choose $\theta$ such that the diagram evaluates to the desired scalar coefficient. However, to avoid unnecessary clutter, we typically refrain from replacing scalar coefficients with Majorana diagrams. 

Using Majorana diagrams, a restricted class of quantum computations can be represented. For instance, the following diagram corresponds to a $3$-qubit quantum computation:  
\begin{equation}
\raisebox{-1.9cm}{\tikz{
\draw (0.4, -0.47) to[out=145, in=-90] (0.1, 0)--(0.1, 1) arc (180:0:0.35) to[out=-90, in=145] (1.1, 0.53);

\draw (1.1+1.4, 0.53) to[out=145, in=-90] (2.2, 1) arc (0:180:0.35) to[out=-90, in=90] (0.8, 0) to[out=-90, in=90] (0.1, -1)--(0.1, -2) arc (-180:0:0.35) to[out=90, in=-90] (1.5, -1) to[out=90, in=-90] (0.8+1.4, 0) to[out=90, in=-90] (2.9, 1) arc (180:0:0.35)--(3.6, 0) to[out=-90, in=90] (2.9, -1) to[out=-90, in=90] (2.2, -2) arc (0:-180:0.35) to[out=90, in=-35] (1.2, 0.45-2);

\draw (2.6, -1.55) to[out=-35, in=90] (2.9, -2) arc (-180:0:0.35)--(3.6, -1) to[out=90, in=-35] (3.3, -0.55);

\draw (1.8, -0.47) to[out=145, in=-90] (1.5, 0) to[out=90, in=-35] (1.2, 0.45);

\draw (3.2, -0.47) to[out=145, in=-90] (2.9, 0) to[out=90, in=-35] (2.6, 0.45);

\draw (1.1, -1.47) to[out=145, in=-90] (0.8, -1) to[out=90, in=-35] (0.5, -0.55);

\draw (2.5, -1.47) to[out=145, in=-90] (2.2, -1) to[out=90, in=-35] (1.9, -0.55);

\node at (0.8, 0) [circle, fill, inner sep=1pt] {};
\node at (2.9, 0) [circle, fill, inner sep=1pt] {};

\node at (0.1, -1) [circle, fill, inner sep=1pt] {};
\node at (1.5, -1) [circle, fill, inner sep=1pt] {};

\draw [decorate, decoration = {brace}] (-0.05, 1.05)--(-0.05, 1.4);
\draw [decorate, decoration = {brace}] (-0.05, -1.95)--(-0.05, 0.95);
\draw [decorate, decoration = {brace}] (-0.05, -2.4)--(-0.05, -2.05);

\node at (-1.6, 1.25) {initial state: $\vert 0 \rangle^{\otimes 3}$};
\node at (-1.2, -0.5) {operator: $\hat{U}$};
\node at (-1.5, -2.25) {final state: $\langle 0 \vert^{\otimes 3}$};
}} \,\, , 
\end{equation}
where the complex number to which the diagram evaluates equals $2^{3/2} \langle 0,0,0 \vert \hat{U} \vert 0,0,0 \rangle$ with $\hat{U}$ being a $3$-qubit unitary. It turns out that Majorana diagrams with caps, cups, dots, and braids, but without generic scattering elements, can encode the intersection of Clifford and matchgate unitaries~\cite{PhysRevA.73.042313, PhysRevA.79.032311}. By incorporating generic scattering elements, the entire set of matchgate unitaries can be encoded. Therefore, one approach for efficiently evaluating a closed Majorana diagram is to translate it into a matchgate computation. However, an arbitrary Clifford unitary gate cannot be encoded using the dense encoding discussed here, since the SWAP gate, which is a Clifford gate, is \textit{not} encodable. One can at best realize the fermionic-SWAP ($f\textrm{SWAP}$) gate using the dense encoding. To achieve full capability, we introduce the \textit{background manifold}, which was originally used in the 3D Quon language~\cite{liu2017quon} and has now been adapted to the 2D Quon setting\footnote{Alternatively, universal quantum computation can be encoded by introducing a quartic Majorana operator~\cite{bravyi2002fermionic}. Additionally, Sec. 7 of the same paper discusses measurements of fermion parity operators, which are compatible with the background manifold.}. 

\subsection{Background Manifold and Diagrammatic Rewriting Rules}
\label{sec:background-manifold}
As a preliminary step toward introducing the \textit{background manifold}, we present the following diagrammatic notation for imposing the fermion parity-even projection on an even number of strings: 
\begin{align}
\label{eq:parity-even-projection}
\raisebox{-0.42cm}{\tikz{
\draw (-0.3, 0)--(-0.3, 1);
\draw (-0.1, 0)--(-0.1, 1);
\draw (0.1, 0)--(0.1, 1);
\draw (0.3, 0)--(0.3, 1);
\draw (1.1, 0)--(1.1, 1);
\draw (1.3, 0)--(1.3, 1);
\draw[thick, blue] (-0.5, 0.5)--(1.5, 0.5);
\node at (0.725, 0.7) {$\cdots$};
}} \,\, &:= \,\, \frac{1}{2} \Bigg( \,\, \raisebox{-0.42cm}{\tikz{
\draw (-0.3, 0)--(-0.3, 1);
\draw (-0.1, 0)--(-0.1, 1);
\draw (0.1, 0)--(0.1, 1);
\draw (0.3, 0)--(0.3, 1);
\draw (1.1, 0)--(1.1, 1);
\draw (1.3, 0)--(1.3, 1);
\node at (0.725, 0.5) {$\cdots$};
}} \,\, + \,\, \raisebox{-0.42cm}{\tikz{
\draw (-0.3, 0)--(-0.3, 1);
\draw (-0.1, 0)--(-0.1, 1);
\draw (0.1, 0)--(0.1, 1);
\draw (0.3, 0)--(0.3, 1);
\draw (1.1, 0)--(1.1, 1);
\draw (1.3, 0)--(1.3, 1);
\node at (-0.3, 0.5) [circle, fill, inner sep=1pt] {};
\node at (-0.1, 0.5) [circle, fill, inner sep=1pt] {};
\node at (0.1, 0.5) [circle, fill, inner sep=1pt] {};
\node at (0.3, 0.5) [circle, fill, inner sep=1pt] {};
\node at (1.1, 0.5) [circle, fill, inner sep=1pt] {};
\node at (1.3, 0.5) [circle, fill, inner sep=1pt] {};
\node at (0.725, 0.5) {$\cdots$};
}} \,\, \Bigg) \nonumber \\ 
& = \frac{\openone + \hat{P}_n}{2} , 
\end{align}
where the last term in the first line with dots represents the global parity operator Eq.~\eqref{eq:global-fermion-parity}. Importantly, the parity-even projection commutes with any parity-even diagram, i.e., for a parity-even Majorana diagram $f$: 
\begin{equation}
\label{eq:p-projections-majoroana-diagram}
\raisebox{-0.7cm}{\tikz{
\draw (-0.1, 1)--(-0.1, 0.5);
\draw (0.1, 1)--(0.1, 0.5);
\draw (0.9, 1)--(0.9, 0.5);
\draw (1.1, 1)--(1.1, 0.5);

\draw[thick, blue] (-0.3, 0.7)--(1.3, 0.7);

\draw (-0.4, 0.5) rectangle (1.4, 0.1);
\node at (0.5, 0.3) {$f$};

\draw (-0.2, 0.1)--(-0.2, -0.4);
\draw (0, 0.1)--(0, -0.4);
\draw (1, 0.1)--(1, -0.4);
\draw (1.2, 0.1)--(1.2, -0.4);

\node at (0.5, 0.9) {$\cdots$};
\node at (0.5, -0.3) {$\cdots$};
}} \,\, = \,\, \raisebox{-0.7cm}{\tikz{
\draw (-0.1, 1)--(-0.1, 0.5);
\draw (0.1, 1)--(0.1, 0.5);
\draw (0.9, 1)--(0.9, 0.5);
\draw (1.1, 1)--(1.1, 0.5);

\draw[thick, blue] (-0.4, -0.1)--(1.4, -0.1);

\draw (-0.4, 0.5) rectangle (1.4, 0.1);
\node at (0.5, 0.3) {$f$};

\draw (-0.2, 0.1)--(-0.2, -0.4);
\draw (0, 0.1)--(0, -0.4);
\draw (1, 0.1)--(1, -0.4);
\draw (1.2, 0.1)--(1.2, -0.4);

\node at (0.5, 0.9) {$\cdots$};
\node at (0.5, -0.3) {$\cdots$};
}}\,\, = \,\, \raisebox{-0.7cm}{\tikz{
\draw (-0.1, 1)--(-0.1, 0.5);
\draw (0.1, 1)--(0.1, 0.5);
\draw (0.9, 1)--(0.9, 0.5);
\draw (1.1, 1)--(1.1, 0.5);

\draw[thick, blue] (-0.3, 0.7)--(1.3, 0.7);
\draw[thick, blue] (-0.4, -0.1)--(1.4, -0.1);

\draw (-0.4, 0.5) rectangle (1.4, 0.1);
\node at (0.5, 0.3) {$f$};

\draw (-0.2, 0.1)--(-0.2, -0.4);
\draw (0, 0.1)--(0, -0.4);
\draw (1, 0.1)--(1, -0.4);
\draw (1.2, 0.1)--(1.2, -0.4);

\node at (0.5, 0.9) {$\cdots$};
\node at (0.5, -0.3) {$\cdots$};
}} \,\, , 
\end{equation}
where the number of input lines and the number of output lines may differ in general. The last equality implies that certain parity-even projections are inherently satisfied. For example, the parity-even projection becomes trivial in a closed Majorana diagram:
\begin{equation}
\raisebox{-1.4cm}{
\begin{tikzpicture}[scale=0.8]
\draw (0.4, -0.47) to[out=145, in=-90] (0.1, 0)--(0.1, 1) arc (180:0:0.35) to[out=-90, in=145] (1.1, 0.53);

\draw (1.1+1.4, 0.53) to[out=145, in=-90] (2.2, 1) arc (0:180:0.35) to[out=-90, in=90] (0.8, 0) to[out=-90, in=90] (0.1, -1)--(0.1, -2) arc (-180:0:0.35) to[out=90, in=-90] (1.5, -1) to[out=90, in=-90] (0.8+1.4, 0) to[out=90, in=-90] (2.9, 1) arc (180:0:0.35)--(3.6, 0) to[out=-90, in=90] (2.9, -1) to[out=-90, in=90] (2.2, -2) arc (0:-180:0.35) to[out=90, in=-35] (1.2, 0.45-2);

\draw (2.6, -1.55) to[out=-35, in=90] (2.9, -2) arc (-180:0:0.35)--(3.6, -1) to[out=90, in=-35] (3.3, -0.55);

\draw (1.8, -0.47) to[out=145, in=-90] (1.5, 0) to[out=90, in=-35] (1.2, 0.45);

\draw (3.2, -0.47) to[out=145, in=-90] (2.9, 0) to[out=90, in=-35] (2.6, 0.45);

\draw (1.1, -1.47) to[out=145, in=-90] (0.8, -1) to[out=90, in=-35] (0.5, -0.55);

\draw (2.5, -1.47) to[out=145, in=-90] (2.2, -1) to[out=90, in=-35] (1.9, -0.55);

\node at (0.8, 0) [circle, fill, inner sep=1pt] {};
\node at (2.9, 0) [circle, fill, inner sep=1pt] {};

\node at (0.1, -1) [circle, fill, inner sep=1pt] {};
\node at (1.5, -1) [circle, fill, inner sep=1pt] {};

\draw[thick, blue] (-0.1, 1)--(3.8, 1);
\end{tikzpicture}
} = \raisebox{-1.4cm}{
\begin{tikzpicture}[scale=0.8]
\draw (0.4, -0.47) to[out=145, in=-90] (0.1, 0)--(0.1, 1) arc (180:0:0.35) to[out=-90, in=145] (1.1, 0.53);

\draw (1.1+1.4, 0.53) to[out=145, in=-90] (2.2, 1) arc (0:180:0.35) to[out=-90, in=90] (0.8, 0) to[out=-90, in=90] (0.1, -1)--(0.1, -2) arc (-180:0:0.35) to[out=90, in=-90] (1.5, -1) to[out=90, in=-90] (0.8+1.4, 0) to[out=90, in=-90] (2.9, 1) arc (180:0:0.35)--(3.6, 0) to[out=-90, in=90] (2.9, -1) to[out=-90, in=90] (2.2, -2) arc (0:-180:0.35) to[out=90, in=-35] (1.2, 0.45-2);

\draw (2.6, -1.55) to[out=-35, in=90] (2.9, -2) arc (-180:0:0.35)--(3.6, -1) to[out=90, in=-35] (3.3, -0.55);

\draw (1.8, -0.47) to[out=145, in=-90] (1.5, 0) to[out=90, in=-35] (1.2, 0.45);

\draw (3.2, -0.47) to[out=145, in=-90] (2.9, 0) to[out=90, in=-35] (2.6, 0.45);

\draw (1.1, -1.47) to[out=145, in=-90] (0.8, -1) to[out=90, in=-35] (0.5, -0.55);

\draw (2.5, -1.47) to[out=145, in=-90] (2.2, -1) to[out=90, in=-35] (1.9, -0.55);

\node at (0.8, 0) [circle, fill, inner sep=1pt] {};
\node at (2.9, 0) [circle, fill, inner sep=1pt] {};

\node at (0.1, -1) [circle, fill, inner sep=1pt] {};
\node at (1.5, -1) [circle, fill, inner sep=1pt] {};
\end{tikzpicture}
}
\end{equation}
The strength of the parity-even projection lies in its ability to split into subsectors: 
\begin{align}
\raisebox{-0.4cm}{\tikz{
\draw (-0.2, 0)--(-0.2, 1);
\draw (0.4, 0)--(0.4, 1);
\draw (0.7, 0)--(0.7, 1);
\draw (1.3, 0)--(1.3, 1);
\draw[thick, blue] (-0.4, 0.7)--(1.5, 0.7);
\draw[thick, blue] (-0.4, 0.4)--(0.5, 0.4);
\draw[thick, blue] (0.6, 0.4)--(1.5, 0.4);
\node at (0.15, 0.2) {$\cdots$};
\node at (1.05, 0.2) {$\cdots$};
}} \,\, &= \raisebox{-0.4cm}{\tikz{
\draw (-0.2, 0)--(-0.2, 1);
\draw (0.4, 0)--(0.4, 1);
\draw (0.7, 0)--(0.7, 1);
\draw (1.3, 0)--(1.3, 1);
\draw[thick, blue] (-0.4, 0.7)--(1.5, 0.7);
\draw[thick, blue] (-0.4, 0.4)--(0.5, 0.4);
\node at (0.15, 0.2) {$\cdots$};
\node at (1.05, 0.2) {$\cdots$};
}} \,\, = \,\, \raisebox{-0.4cm}{\tikz{
\draw (-0.2, 0)--(-0.2, 1);
\draw (0.4, 0)--(0.4, 1);
\draw (0.7, 0)--(0.7, 1);
\draw (1.3, 0)--(1.3, 1);
\draw[thick, blue] (-0.4, 0.7)--(1.5, 0.7);
\draw[thick, blue] (0.6, 0.4)--(1.5, 0.4);
\node at (0.15, 0.2) {$\cdots$};
\node at (1.05, 0.2) {$\cdots$};
}} \nonumber \\ 
& \ne \raisebox{-0.4cm}{\tikz{
\draw (-0.2, 0)--(-0.2, 1);
\draw (0.4, 0)--(0.4, 1);
\draw (0.7, 0)--(0.7, 1);
\draw (1.3, 0)--(1.3, 1);
\draw[thick, blue] (-0.4, 0.7)--(1.5, 0.7);
\node at (0.15, 0.2) {$\cdots$};
\node at (1.05, 0.2) {$\cdots$};
}} \,\, , 
\end{align}
where the equalities in the first line demonstrate that some of the parity-even projections are automatically imposed. 

The above observations suggest that the parity-even projection has a \textit{topological} nature. Thus, we introduce the \textit{background manifold} as a way to keep track of the parity-even projections imposed in Majoranas: 
\begin{equation}
\raisebox{-0.45cm}{\tikz{
\fill [blue!10] (-0.2, 0) rectangle (1.4, 1);

\draw (0, 0)--(0, 1);
\draw (0.2, 0)--(0.2, 1);
\draw (1, 0)--(1, 1);
\draw (1.2, 0)--(1.2, 1);

\draw[thick,blue] (-0.2, 0)--(-0.2, 1);
\draw[thick,blue] (1.4, 0)--(1.4, 1);

\node at (0.65, 0.5) {$\cdots$};
}} \,\, := \raisebox{-0.45cm}{\tikz{
\draw (0, 0) --(0, 1);
\draw (0.2, 0) --(0.2, 1);
\draw (1, 0) --(1, 1);
\draw (1.2, 0) --(1.2, 1);

\draw[thick,blue] (-0.2,0.5) --(1.4,0.5);

\node at (0.65, 0.2) {$\cdots$};
}}
\end{equation}
and 
\begin{equation}
\raisebox{-0.45cm}{\tikz{
\fill[blue!10] (-0.4, 1)--(-0.4, 0)--(0.6, 0)--(0.6, 0.4) arc (180:0:0.1)--(0.8, 0)--(1.8, 0)--(1.8, 1);

\draw[thick, blue] (-0.4, 0)--(-0.4, 1);
\draw[thick, blue] (1.8, 0)--(1.8, 1);

\draw[thick, blue] (0.8, 0.4) arc (0:180:0.1);
\draw[thick, blue] (0.8, 0.4)--(0.8, 0);
\draw[thick, blue] (0.6, 0.4)--(0.6, 0);

\draw (-0.2, 0)--(-0.2, 1);
\draw (0.4, 0)--(0.4, 1);
\draw (1, 0)--(1, 1);
\draw (1.6, 0)--(1.6, 1);

\node at (0.15, 0.5) {$\cdots$};
\node at (1.35, 0.5) {$\cdots$};
}} \,\, := \,\,
\raisebox{-0.45cm}{\tikz{
\draw (-0.2, 0)--(-0.2, 1);
\draw (0.4, 0)--(0.4, 1);
\draw (0.7, 0)--(0.7, 1);
\draw (1.3, 0)--(1.3, 1);

\draw[thick, blue] (-0.4, 0.7)--(1.5, 0.7);
\draw[thick, blue] (-0.4, 0.4)--(0.5, 0.4);
\draw[thick, blue] (0.6, 0.4)--(1.5, 0.4);

\node at (0.15, 0.2) {$\cdots$};
\node at (1.05, 0.2) {$\cdots$};
}} \,\, , 
\end{equation}
where the background manifold is depicted as a shaded region in each diagram. As described in FIG~\ref{fig:quon-example} (a), a background manifold is a $(1+1)$-dimensional space-time manifold equipped with a local time flow at each point in the bulk, dictating the time ordering of the embdeed Majorana diagram. In principle, one should follow the local time flow provided by the background manifold; however, for simplicity, we always use the global time flow, which is top-to-bottom unless explicitly stated otherwise. The \textit{boundary} of a background manifold can have multiple connected components. Each component, which is topologically equivalent to a circle, consists of \textit{open} and \textit{closed} intervals, with the closed intervals highlighted by thick blue lines in the diagram. An open interval aligns with a constant time slice and is the only location where open Majorana strings can end transversally. We refer to a Quon diagram containing one or more open intervals as an \textit{open} Quon diagram, and one without any open intervals as a \textit{closed} Quon diagram. 

Using the background manifold, we can realize the SWAP gate, which is a magic gate for matchgate (see Appendix~\ref{app:matchgate}), using the following diagram: 
\begin{equation}
\label{eq:SWAP-quon-diagram}
\raisebox{-0.75cm}{\tikz{
\fill[blue!10] (0, 0.85)--(0.8, 0.85) arc (-180:0:0.1)--(1.8, 0.85)--(1.8, 0)--(1, 0) arc (0:180:0.1)--(0, 0);

\draw[thick, blue] (1, 0.85) arc (0:-180:0.1);
\draw[thick, blue] (1, 0) arc (0:180:0.1);

\draw (1.2, 0) to[out=110, in=-70] (0.1, 0.85);
\draw (1.7, 0) to[out=110, in=-70] (0.6, 0.85);

\draw (0.1, 0) to[out=70, in=-160] (0.6, 0.4);
\draw (0.7, 0.44)--(0.85, 0.5);
\draw (0.95, 0.55) to[out=25, in=-110] (1.2, 0.85);

\draw (0.6, 0) to[out=70, in=-145] (0.86, 0.3);
\draw (0.95, 0.34)--(1.1, 0.4);
\draw (1.2, 0.43) to[out=20, in=-110] (1.7, 0.85);

\draw[thick, blue] (0, 0)--(0, 0.85);
\draw[thick, blue] (1.8, 0)--(1.8, 0.85);

\node at (0.45, 0.1) {\scriptsize $\cdots$};
\node at (0.45, 0.75) {\scriptsize $\cdots$};
\node at (1.45, 0.1) {\scriptsize $\cdots$};
\node at (1.45, 0.75) {\scriptsize $\cdots$};

\draw [decorate, decoration = {brace, mirror}] (0.7, 0.9)--(0.05, 0.9);
\node at (0.375, 1.1) {\scriptsize $n$};

\draw [decorate, decoration = {brace, mirror}] (1.75, 0.9)--(1.15, 0.9);
\node at (1.45, 1.1) {\scriptsize $m$};

\draw [decorate, decoration = {brace, mirror}] (0.05, -0.05)--(0.7, -0.05);
\node at (0.375, -0.25) {\scriptsize $m$};

\draw [decorate, decoration = {brace, mirror}] (1.15, -0.05)--(1.75, -0.05);
\node at (1.45, -0.25) {\scriptsize $n$};
}} \,\, , 
\end{equation}
where $n$ and $m$ are even integers. Here, we use negative braids at every crossing, but one could instead use positive braids. 

In the remainder of the section, we discuss \textit{diagrammatic rewriting rules} for the 2D Quon diagrams. If the background manifold contains no Majorana diagrams, it evaluates to $1$: 
\begin{equation}
\label{eq:eval-empty-background}
\raisebox{-0.35cm}{
\tikz{
\fill [blue!10] (0.4, 0) arc (0:360:0.4);
\draw [thick, blue] (0.4, 0) arc (0:360:0.4);
}} \,\, = 1 \,\, , 
\end{equation}
where we note that this diagram is different from the Majorana loop presented in TABLE.~\ref{tab:majorana-rewriting-rules}. One can freely move a Majorana diagram as follows: 
\begin{equation}
\label{eq:push-neutral-diagram}
\raisebox{-0.45cm}{\tikz{
\fill[blue!10] (-0.4, 1)--(-0.4, 0)--(0.6, 0)--(0.6, 0.4) arc (180:0:0.1)--(0.8, 0)--(1.8, 0)--(1.8, 1);

\draw[thick, blue] (-0.4, 0)--(-0.4, 1);
\draw[thick, blue] (1.8, 0)--(1.8, 1);
\draw[thick, blue] (0.8,0)--(0.8, 0.4) arc (0:180:0.1)--(0.6, 0);

\draw (-0.2, 0)--(-0.2, 0.6);
\draw (-0.1, 0.9)--(-0.1, 1);
\draw (0.4, 0)--(0.4, 0.6);
\draw (0.3, 0.9)--(0.3, 1);
\draw (1, 0)--(1, 1);
\draw (1.6, 0)--(1.6, 1);

\draw (-0.3, 0.6) rectangle (0.5, 0.9);
\node at (0.1, 0.75) {\scriptsize $f$};

\node at (0.125, 0.2) {\scriptsize $\cdots$};
\node at (1.325, 0.5) {\scriptsize $\cdots$};
}} \,\, = \,\, \raisebox{-0.45cm}{\tikz{
\fill[blue!10] (-0.4, 1)--(-0.4, 0)--(0.6, 0)--(0.6, 0.4) arc (180:0:0.1)--(0.8, 0)--(1.8, 0)--(1.8, 1);

\draw[thick, blue] (-0.4, 0)--(-0.4, 1);
\draw[thick, blue] (1.8, 0)--(1.8, 1);
\draw[thick, blue] (0.8,0)--(0.8, 0.4) arc (0:180:0.1)--(0.6, 0);

\draw (-0.2, 0)--(-0.2, 0.1);
\draw (-0.1, 0.4)--(-0.1, 1);
\draw (0.4, 0)--(0.4, 0.1);
\draw (0.3, 0.4)--(0.3, 1);
\draw (1, 0)--(1, 1);
\draw (1.6, 0)--(1.6, 1);

\draw (-0.3, 0.1) rectangle (0.5, 0.4);
\node at (0.1, 0.25) {\scriptsize $f$};

\node at (0.125, 0.7) {\scriptsize $\cdots$};
\node at (1.325, 0.5) {\scriptsize $\cdots$};
}} \,\, , 
\end{equation}
where $f$ is a parity-even Majorana diagram generally having different numbers of input and output legs, which can be shown by using Eq.~\eqref{eq:p-projections-majoroana-diagram}. Furthermore,  
\begin{equation}
\label{eq:charges-out-of-nowhere}
\raisebox{-0.4cm}{\tikz{
\fill [blue!10] (-0.2, 0) rectangle (1.4, 1);

\draw[thick, blue] (-0.2, 0)--(-0.2, 1);
\draw[thick, blue] (1.4, 0)--(1.4, 1);

\draw (0, 0)--(0, 1);
\draw (0.2, 0)--(0.2, 1);
\draw (1, 0)--(1, 1);
\draw (1.2, 0)--(1.2, 1);

\node at (0.625, 0.5) {$\cdots$};
}} \,\, = \raisebox{-0.4cm}{\tikz{
\fill [blue!10] (-0.2, 0) rectangle (1.4, 1);

\draw[thick, blue] (-0.2, 0)--(-0.2, 1);
\draw[thick, blue] (1.4, 0)--(1.4, 1);

\draw (0, 0)--(0, 1);
\draw (0.2, 0)--(0.2, 1);
\draw (1, 0)--(1, 1);
\draw (1.2, 0)--(1.2, 1);

\node at (0, 0.5) [circle, fill, inner sep=1pt] {};
\node at (0.2, 0.5) [circle, fill, inner sep=1pt] {};
\node at (0.625, 0.5) {$\cdots$};
\node at (1, 0.5) [circle, fill, inner sep=1pt] {};
\node at (1.2, 0.5) [circle, fill, inner sep=1pt] {};
}} \,\, , 
\end{equation}
where the Majorana diagram in the RHS equals the fermion parity operator Eq.~\eqref{eq:global-fermion-parity}. When there are four strings in the manifold, one can move the location of a braid or scattering: 
\begin{equation}
\raisebox{-0.4cm}{\tikz{
\draw[fill, blue!10] (-0.1, -0.2) rectangle (1.2, 0.8);

\draw[thick, blue] (-0.1, -0.2)--(-0.1, 0.8);
\draw[thick, blue] (1.2, -0.2)--(1.2, 0.8);

\draw (0.15, -0.2) to[out=90, in=-120] (0.2, 0.13);
\draw (0.2, 0.47) to[out=120, in=-90] (0.15, 0.8);

\draw (0.45, -0.2) to[out=90, in=-60] (0.4, 0.13);
\draw (0.4, 0.47) to[out=60, in=-90] (0.45, 0.8);

\draw (0.7, -0.2)--(0.7, 0.8);
\draw (0.95, -0.2)--(0.95, 0.8);

\draw (0.5, 0.3) arc (0:360:0.2);
\node at (0.3, 0.3) {$\theta$};
}} \,\, = \,\, \raisebox{-0.4cm}{\tikz{
\draw[fill, blue!10] (-0.1, -0.2) rectangle (1.2, 0.8);

\draw[thick, blue] (-0.1, -0.2)--(-0.1, 0.8);
\draw[thick, blue] (1.2, -0.2)--(1.2, 0.8);

\draw (0.65, -0.2) to[out=90, in=-120] (0.7, 0.13);
\draw (0.7, 0.47) to[out=120, in=-90] (0.65, 0.8);

\draw (0.95, -0.2) to[out=90, in=-60] (0.9, 0.13);
\draw (0.9, 0.47) to[out=60, in=-90] (0.95, 0.8);

\draw (0.15, -0.2)--(0.15, 0.8);
\draw (0.4, -0.2)--(0.4, 0.8);

\draw (1, 0.3) arc (0:360:0.2);
\node at (0.8, 0.3) {$\theta$};
}} \,\, , 
\end{equation}
which can be derived by first expanding the scattering element as a sum of two diagrams using the expansion in TABLE~\ref{tab:Majorana-diagram-expansion}, and then applying Eq.~\eqref{eq:charges-out-of-nowhere}. 

The topology of the background manifold can be modified. Whenever the vertical or horizontal slice of the manifold contains no Majorana lines, we are free to reconnect the boundaries: 
\begin{equation}
\label{eq:shrink-bdy-1}
\raisebox{-0.45cm}{
\tikz{
\draw[fill, blue!10] (0, 0) rectangle (0.6, 1);

\draw[thick, blue] (0, 0)--(0, 1);
\draw[thick, blue] (0.6, 0)--(0.6, 1);
}} \,\, = \raisebox{-0.45cm}{
\tikz{
\draw[thick, blue, fill=blue!10] (0, 0)--(0, 0.1) arc (180:0:0.3)--(0.6, 0);
\draw[thick, blue, fill=blue!10] (0, 1)--(0, 0.9) arc (-180:0:0.3)--(0.6, 1);
}} 
\end{equation}
and
\begin{equation}
\label{eq:shrink-bdy-2}
\raisebox{-0.45cm}{
\tikz{
\draw[fill, blue!10] (-0.2, 1)--(-0.2, 0)--(1.1, 0)--(1.1, 0.2) arc (180:0:0.2)--(1.5, 0)--(2.6, 0)--(2.6, 1)--(1.5, 1)--(1.5, 0.8) arc (0:-180:0.2)--(1.1, 1);

\draw[thick,blue] (-0.2, 0)--(-0.2, 1);
\draw[thick,blue] (2.6, 0)--(2.6, 1);

\draw[thick,blue] (1.1, 0)--(1.1, 0.2) arc (180:0:0.2)--(1.5, 0);
\draw[thick,blue] (1.1, 1)--(1.1, 0.8) arc (-180:0:0.2)--(1.5, 1);

\draw (0, 0)--(0, 1);
\draw (0.9, 0)--(0.9, 1);
\draw (1.7, 0)--(1.7, 1);
\draw (2.4, 0)--(2.4, 1);

\node at (0.475, 0.5) {$\cdots$};
\node at (2.075, 0.5) {$\cdots$};
}} \,\, = \raisebox{-0.45cm}{
\tikz{
\draw[fill, blue!10] (-0.2, 0) rectangle (1.1, 1);
\draw[fill, blue!10] (1.5, 0) rectangle (2.6, 1);

\draw[thick, blue] (-0.2, 0)--(-0.2, 1);
\draw[thick, blue] (2.6, 0)--(2.6, 1);

\draw[thick,blue] (1.1, 0)--(1.1, 1);
\draw[thick,blue] (1.5, 0)--(1.5, 1);

\draw (0, 0)--(0, 1);
\draw (0.9, 0)--(0.9, 1);
\draw (1.7, 0)--(1.7, 1);
\draw (2.4, 0)--(2.4, 1);

\node at (0.475, 0.5) {$\cdots$};
\node at (2.075, 0.5) {$\cdots$};
}} \,\, . 
\end{equation}
Eq.~\eqref{eq:shrink-bdy-1} trivially holds as no Majorana lines intersect every time slice, making the parity projection Eq.~\eqref{eq:parity-even-projection} trivial. Eq.~\eqref{eq:shrink-bdy-2} also holds since the parity-even projection imposed in the ``middle'' part (along the vertical axis) of the diagram on the LHS is trivial as the Majoranas on the ``left'' and ''right'' parts (along the horizontal axis) are already in the parity-even sectors. 

When instead two Majorana lines are present between vertical or horizontal slices, we can still reconnect the boundaries of the manifold as follows: 
\begin{equation}
\label{eq:shrink-bdy-3}
\raisebox{-0.5cm}{
\tikz{
\fill[blue!10] (0, 0) rectangle (1, 1.1);

\draw[thick, blue] (0, 0)--(0, 1.1);
\draw[thick, blue] (1, 0)--(1, 1.1);

\draw (0.3, 0) --(0.3, 1.1);
\draw (0.7, 0) --(0.7, 1.1);
}} \,\, = \frac{1}{\sqrt{2}} \sum_{k=0, 1} \raisebox{-0.5cm}{
\tikz{
\fill[blue!10] (0, 0) rectangle (1, 1.1);

\draw[thick, blue] (0, 0)--(0, 1.1);
\draw[thick, blue] (1, 0)--(1, 1.1);

\draw (0.3, 0) --(0.3, 0.2);
\draw (0.3, 0.9) --(0.3, 1.1);
\draw (0.7, 0) --(0.7, 0.2);
\draw (0.7, 0.9) --(0.7, 1.1);

\draw (0.7, 0.9) arc (0:-180:0.2);
\draw (0.7, 0.2) arc (0:180:0.2);

\coordinate[label = left:$k$] (A) at (0.7, 0.2);
\coordinate[label = left:$k$] (B) at (0.7, 0.9);

\node at (A) [circle, fill, inner sep=1pt] {};
\node at (B) [circle, fill, inner sep=1pt] {};
}} \,\, = \frac{1}{\sqrt{2}} \raisebox{-0.5cm}{
\tikz{
\draw[thick, blue, fill=blue!10] (1, 1.1) arc (0:-180:0.5);
\draw[thick, blue, fill=blue!10] (1, 0) arc (0:180:0.5);

\draw (0.3, 0) --(0.3, 0.2);
\draw (0.3, 0.9) --(0.3, 1.1);
\draw (0.7, 0) --(0.7, 0.2);
\draw (0.7, 0.9) --(0.7, 1.1);

\draw (0.7, 0.9) arc (0:-180:0.2);
\draw (0.7, 0.2) arc (0:180:0.2);
}} \,\, , 
\end{equation}
where we used an expansion appearing in TABLE~\ref{tab:Majorana-diagram-expansion} in the first equality and applied Eq.~\eqref{eq:shrink-bdy-1}, along with the fact that the $k=1$ term vanishes due to the parity-even projection in the second equality. Similarly, 
\begin{align}
\label{eq:shrink-bdy-4}
\raisebox{-0.4cm}{
\tikz{
\fill[blue!10] (-0.2, 1)--(1.1, 1) arc (-180:0:0.2)--(2.8, 1)--(2.8, 0)--(1.5, 0) arc (0:180:0.2)--(-0.2, 0);
\draw[thick, blue] (-0.2, 0)--(-0.2, 1);
\draw[thick, blue] (2.8, 0)--(2.8, 1);
\draw[thick, blue] (1.5, 0.0) arc (0:180:0.2);
\draw[thick, blue] (1.5, 1) arc (0:-180:0.2);
\draw (0, 0)--(0, 1);
\draw (0.7, 0)--(0.7, 1);
\draw (1.9, 0)--(1.9, 1);
\draw (2.6, 0)--(2.6, 1);
\draw[red] (1.7, 1) arc (0:-180:0.4);
\draw[red] (1.7, 0) arc (0:180:0.4);
\node at (0.375, 0.5) {$\cdots$};
\node at (2.275, 0.5) {$\cdots$};
}} \,\, &= \frac{1}{\sqrt{2}} \sum_{\textcolor{red}{k}=0,1} \,\, \raisebox{-0.4cm}{
\tikz{
\fill[blue!10] (-0.2, 1)--(1.1, 1) arc (-180:0:0.2)--(2.8, 1)--(2.8, 0)--(1.5, 0) arc (0:180:0.2)--(-0.2, 0);
\draw[thick, blue] (-0.2, 0)--(-0.2, 1);
\draw[thick, blue] (2.8, 0)--(2.8, 1);
\draw[thick, blue] (1.5, 0.0) arc (0:180:0.2);
\draw[thick, blue] (1.5, 1) arc (0:-180:0.2);
\draw (0, 0)--(0, 1);
\draw (0.7, 0)--(0.7, 1);
\draw (1.9, 0) --(1.9, 1);
\draw (2.6, 0) --(2.6, 1);
\draw[red] (0.9,0)--(0.9,1);
\draw[red] (1.7,0)--(1.7,1);
\node at (0.375, 0.5) {$\cdots$};
\node at (2.275, 0.5) {$\cdots$};
\coordinate[label = left:{\scriptsize \textcolor{red}{$k$}}] (A) at (0.9, 0.5);
\coordinate[label = left:{\scriptsize \textcolor{red}{$k$}}] (B) at (1.7, 0.5);
\node at (A) [circle, fill, inner sep=1pt] {};
\node at (B) [circle, fill, inner sep=1pt] {};
}} \nonumber \\ 
&= \frac{1}{\sqrt{2}} \sum_{\textcolor{red}{k}=0, 1} \,\, \raisebox{-0.4cm}{
\tikz{
\fill[blue!10] (-0.2, 0) rectangle (1.1, 1);
\fill[blue!10] (1.5, 0) rectangle (2.8, 1);
\draw[thick, blue] (-0.2, 0)--(-0.2, 1);
\draw[thick, blue] (2.8, 0)--(2.8, 1);
\draw[thick, blue] (1.1, 0)--(1.1, 1);
\draw[thick, blue] (1.5, 0)--(1.5, 1);
\draw (0, 0)--(0, 1);
\draw (0.7, 0)--(0.7, 1);
\draw (1.9, 0)--(1.9, 1);
\draw (2.6, 0)--(2.6, 1);
\draw[red] (0.9, 0)--(0.9, 1);
\draw[red] (1.7, 0)--(1.7, 1);
\node at (0.375, 0.5) {$\cdots$};
\node at (2.275, 0.5) {$\cdots$};
\coordinate[label = left:{\scriptsize \textcolor{red}{$k$}}] (A) at (0.9, 0.5);
\coordinate[label = left:{\scriptsize \textcolor{red}{$k$}}] (B) at (1.7, 0.5);
\node at (A) [circle, fill, inner sep=1pt] {};
\node at (B) [circle, fill, inner sep=1pt] {};
}} \nonumber \\ 
&= \frac{1}{\sqrt{2}} \,\, \raisebox{-0.4cm}{
\tikz{
\fill[blue!10] (-0.2, 0) rectangle (1.1, 1);
\fill[blue!10] (1.5, 0) rectangle (2.8, 1);
\draw[thick, blue] (-0.2, 0)--(-0.2, 1);
\draw[thick, blue] (2.8, 0)--(2.8, 1);
\draw[thick, blue] (1.1, 0)--(1.1, 1);
\draw[thick, blue] (1.5, 0)--(1.5, 1);
\draw (0, 0)--(0, 1);
\draw (0.7, 0)--(0.7, 1);
\draw (1.9, 0) --(1.9, 1);
\draw (2.6, 0) --(2.6, 1);
\draw[red] (0.9,0)--(0.9,1);
\draw[red] (1.7,0)--(1.7,1);
\node at (0.375, 0.5) {$\cdots$};
\node at (2.275, 0.5) {$\cdots$};
}} \,\, , 
\end{align}
where we used an expansion appearing in TABLE~\ref{tab:Majorana-diagram-expansion} and Eq.~\eqref{eq:shrink-bdy-2}, and the strings being manipulated are highlighted in red for clarity. 

More non-trivial relations also follow. The first one is the \textit{string-genus relation}, which is the two-dimensional version of a more general string-genus relation appearing in the 3D Quon language~\cite{liu2017quon}: 
\begin{equation}
\label{eq:string-genus}
\raisebox{-1cm}{\tikz{
\begin{scope}[even odd rule]
\clip (-0.2, 0) rectangle (2.6, 1)  (1.5, 0.5) arc (0:360:0.2);

\fill [blue!10] (-0.2, 0) rectangle (2.6, 1);
\end{scope}

\draw[thick, blue] (-0.2, 0) --(-0.2, 1);
\draw[thick, blue] (2.6, 0) --(2.6, 1);

\draw[thick, blue] (1.5, 0.5) arc (0:360:0.2);

\draw (0, 0)--(0, 1);
\draw (0.9, 0)--(0.9, 1);
\draw (1.7, 0)--(1.7, 1);
\draw (2.4, 0)--(2.4, 1);

\node at (0.475, 0.5) {$\cdots$};
\node at (2.075, 0.5) {$\cdots$};

\draw[red] (1.6, 0.5) arc (0:360:0.3);

\draw [decorate, decoration = {brace, mirror}] (-0.1, -0.05)--(1.0, -0.05);
\draw [decorate, decoration = {brace, mirror}] (1.6, -0.05)--(2.5, -0.05);

\node at (0.5, -0.35) {\scriptsize $2m+1$};
\node at (2.1, -0.35) {\scriptsize $2n+1$};
}} \,\, = \frac{1}{\sqrt{2}} \raisebox{-1cm}{
\tikz{
\fill [blue!10] (-0.2, 0) rectangle (1.1, 1);

\draw[thick, blue] (-0.2, 0)--(-0.2, 1);
\draw[thick, blue] (1.1, 0)--(1.1, 1);

\draw (0, 0)--(0, 1);
\draw (0.9, 0)--(0.9, 1);

\node at (0.475, 0.5) {$\cdots$};

\draw [decorate, decoration = {brace, mirror}] (-0.1, -0.05)--(1.0, -0.05);
\node at (0.5, -0.35) {\scriptsize $2m+2n+2$};
}} \,\, , 
\end{equation}
where the hole in the background manifold can be removed along with its adjacent enclosing Majorana loop. As the genus in a 3D Quon diagram corresponds to a hole in its 2D counterpart, the relation could be referred to as the \textit{string-hole} relation. However, to maintain the consistency in terminology, we retain the nomenclature ``string-genus relation.'' 

\begin{proof}
The LHS is equal to 
\begin{align}
&\raisebox{-1cm}{\tikz{
\draw (0, 0)--(0, 1);
\draw (0.1, 0)--(0.1, 1);
\draw (0.9, 0)--(0.9, 1);
\draw (1.7, 0)--(1.7, 1);
\draw (2.3, 0)--(2.3, 1);
\draw (2.4, 0)--(2.4, 1);
\draw[red] (1.6, 0.5) arc (0:360:0.3);
\draw[thick, blue] (-0.2, 0.9) --(2.6, 0.9);
\node at (0.525, 0.2) {$\cdots$};
\node at (2.05, 0.2) {$\cdots$};
\draw[thick, blue] (1.15, 0.35)--(-0.2, 0.35);
\draw [decorate, decoration = {brace, mirror}] (-0.1, -0.05)--(1.0, -0.05);
\draw [decorate, decoration = {brace, mirror}] (1.6, -0.05)--(2.5, -0.05);
\node at (0.5, -0.35) {\scriptsize $2m+1$};
\node at (2.1, -0.35) {\scriptsize $2n+1$};
}} = \frac{1}{\sqrt{2}} \sum_{k=0,1} \,\, \raisebox{-0.7cm}{\tikz{
\draw (0, 0)--(0, 1.2);
\draw (0.1, 0)--(0.1, 1.2);
\draw (1.7, 0)--(1.7, 1.2);
\draw (2.3, 0)--(2.3, 1.2);
\draw (2.4, 0)--(2.4, 1.2);
\draw (0.9, 0)--(0.9, 0.45) to[out=90, in=110, looseness=1.5] (1.2, 0.45) to[out=-70, in=-90, looseness=2] (1.6, 0.65) to[out=90, in=70, looseness=2] (1.2, 0.85) to[out=-110, in=-90, looseness=1.5] (0.9, 0.85)--(0.9, 1.2);
\draw[thick, blue] (-0.2, 1.1) --(2.6, 1.1);
\node at (0.525, 0.2) {$\cdots$};
\node at (2.05, 0.2) {$\cdots$};
\draw[thick, blue] (-0.2, 0.33)--(1.38, 0.33);
\node[label={[xshift=-0.15cm, yshift=-0.23cm] \scriptsize $k$}] at (0.9, 0.45) [circle, fill, inner sep=1pt] {};
\node[label={[xshift=-0.15cm, yshift=-0.23cm] \scriptsize $k$}] at (0.9, 0.85) [circle, fill, inner sep=1pt] {};
}} \nonumber \\ 
& = \frac{1}{2^{3/2}} \sum_{k=0,1} \Bigg( \raisebox{-0.7cm}{\tikz{
\draw (0, 0)--(0, 1.2);
\draw (0.1, 0)--(0.1, 1.2);
\draw (1.6, 0)--(1.6, 1.2);
\draw (2.2, 0)--(2.2, 1.2);
\draw (2.3, 0)--(2.3, 1.2);
\draw (0.9, 0)--(0.9, 0.45) to[out=90, in=90, looseness=1.8] (1.2, 0.45)--(1.2, 0.3) to[out=-90, in=-90, looseness=2] (1.5, 0.65) to[out=90, in=90, looseness=2] (1.2, 0.85) to[out=-90, in=-90, looseness=1.8] (0.9, 0.85)--(0.9, 1.2);
\draw[thick, blue] (-0.2, 1.1)--(2.5, 1.1);
\node at (0.525, 0.2) {$\cdots$};
\node at (1.925, 0.2) {$\cdots$};
\node[label={[xshift=-0.15cm, yshift=-0.23cm] \scriptsize $k$}] at (0.9, 0.45) [circle, fill, inner sep=1pt] {};
\node[label={[xshift=-0.15cm, yshift=-0.23cm] \scriptsize $k$}] at (0.9, 0.85) [circle, fill, inner sep=1pt] {};
}} + \raisebox{-0.7cm}{\tikz{
\draw (0, 0)--(0, 1.2);
\draw (0.1, 0)--(0.1, 1.2);
\draw (1.6, 0)--(1.6, 1.2);
\draw (2.2, 0)--(2.2, 1.2);
\draw (2.3, 0)--(2.3, 1.2);
\draw (0.9, 0)--(0.9, 0.45) to[out=90, in=90, looseness=1.8] (1.2, 0.45)--(1.2, 0.3) to[out=-90, in=-90, looseness=2] (1.5, 0.65) to[out=90, in=90, looseness=2] (1.2, 0.85) to[out=-90, in=-90, looseness=1.8] (0.9, 0.85)--(0.9, 1.2);
\draw[thick, blue] (-0.2, 1.1)--(2.5, 1.1);
\node at (0.525, 0.2) {$\cdots$};
\node at (1.925, 0.2) {$\cdots$};
\node[label={[xshift=-0.15cm, yshift=-0.23cm] \scriptsize $k$}] at (0.9, 0.45) [circle, fill, inner sep=1pt] {};
\node[label={[xshift=-0.15cm, yshift=-0.23cm] \scriptsize $k$}] at (0.9, 0.85) [circle, fill, inner sep=1pt] {};
\node at (0, 0.1) [circle, fill, inner sep=1pt] {};
\node at (0.1, 0.1) [circle, fill, inner sep=1pt] {};
\node at (0.9, 0.3) [circle, fill, inner sep=1pt] {};
\node at (1.2, 0.3) [circle, fill, inner sep=1pt] {};
}} \Bigg) \nonumber \\ 
& = \frac{1}{2^{3/2}} \sum_{k=0,1} \Bigg( \raisebox{-1.1cm}{\tikz{
\draw (0, 0)--(0, 1.2);
\draw (0.1, 0)--(0.1, 1.2);
\draw (0.9, 0)--(0.9, 1.2);
\draw (1.2, 0)--(1.2, 1.2);
\draw (1.7, 0)--(1.7, 1.2);
\draw (1.8, 0)--(1.8, 1.2);
\draw[thick, blue] (-0.2, 1.1)--(2, 1.1);
\node at (0.525, 0.2) {$\cdots$};
\node at (1.475, 0.2) {$\cdots$};
\node[label={[xshift=-0.15cm, yshift=-0.23cm] \scriptsize $k$}] at (0.9, 0.45) [circle, fill, inner sep=1pt] {};
\node[label={[xshift=-0.15cm, yshift=-0.23cm] \scriptsize $k$}] at (0.9, 0.85) [circle, fill, inner sep=1pt] {};
\draw [decorate, decoration = {brace, mirror}] (-0.05, -0.05)--(0.95, -0.05);
\draw [decorate, decoration = {brace, mirror}] (1.15, -0.05)--(1.85, -0.05);
\node at (0.475, -0.35) {\scriptsize $2m+1$};
\node at (1.525, -0.35) {\scriptsize $2n+1$};
}} + (-1)^k \raisebox{-1.1cm}{\tikz{
\draw (0, 0)--(0, 1.2);
\draw (0.1, 0)--(0.1, 1.2);
\draw (0.9, 0)--(0.9, 1.2);
\draw (1.2, 0)--(1.2, 1.2);
\draw (1.7, 0)--(1.7, 1.2);
\draw (1.8, 0)--(1.8, 1.2);
\draw[thick, blue] (-0.2, 1.1)--(2, 1.1);
\node at (0.5, 0.2) {$\cdots$};
\node at (1.5, 0.2) {$\cdots$};
\node[label={[xshift=-0.15cm, yshift=-0.23cm] \scriptsize $k$}] at (0.9, 0.45) [circle, fill, inner sep=1pt] {};
\node[label={[xshift=-0.15cm, yshift=-0.23cm] \scriptsize $k$}] at (0.9, 0.85) [circle, fill, inner sep=1pt] {};
\node at (0, 0.1) [circle, fill, inner sep=1pt] {};
\node at (0.1, 0.1) [circle, fill, inner sep=1pt] {};
\draw [decorate, decoration = {brace, mirror}] (-0.05, -0.05)--(0.95, -0.05);
\draw [decorate, decoration = {brace, mirror}] (1.15, -0.05)--(1.85, -0.05);
\node at (0.475, -0.35) {\scriptsize $2m+1$};
\node at (1.525, -0.35) {\scriptsize $2n+1$};
}} \Bigg) \nonumber \\ 
& = \frac{1}{2^{3/2}} \sum_{k=0,1} \Bigg( \raisebox{-1.1cm}{\tikz{
\draw (0, 0)--(0, 1.2);
\draw (0.1, 0)--(0.1, 1.2);
\draw (0.9, 0)--(0.9, 1.2);
\draw (1.2, 0)--(1.2, 1.2);
\draw (1.7, 0)--(1.7, 1.2);
\draw (1.8, 0)--(1.8, 1.2);
\draw[thick, blue] (-0.2, 1.1)--(2, 1.1);
\node at (0.525, 0.2) {$\cdots$};
\node at (1.475, 0.2) {$\cdots$};
\draw [decorate, decoration = {brace, mirror}] (-0.05, -0.05)--(0.95, -0.05);
\draw [decorate, decoration = {brace, mirror}] (1.15, -0.05)--(1.85, -0.05);
\node at (0.475, -0.35) {\scriptsize $2m+1$};
\node at (1.525, -0.35) {\scriptsize $2n+1$};
}} + (-1)^k \raisebox{-1.1cm}{\tikz{
\draw (0, 0)--(0, 1.2);
\draw (0.1, 0)--(0.1, 1.2);
\draw (0.9, 0)--(0.9, 1.2);
\draw (1.2, 0)--(1.2, 1.2);
\draw (1.7, 0)--(1.7, 1.2);
\draw (1.8, 0)--(1.8, 1.2);
\draw[thick, blue] (-0.2, 1.1) --(2, 1.1);
\node at (0.525, 0.2) {$\cdots$};
\node at (1.475, 0.2) {$\cdots$};
\node at (0, 0.1) [circle, fill, inner sep=1pt] {};
\node at (0.1, 0.1) [circle, fill, inner sep=1pt] {};
\draw [decorate, decoration = {brace, mirror}] (-0.05, -0.05)--(0.95, -0.05);
\draw [decorate, decoration = {brace, mirror}] (1.15, -0.05)--(1.85, -0.05);
\node at (0.475, -0.35) {\scriptsize $2m+1$};
\node at (1.525, -0.35) {\scriptsize $2n+1$};
}} \Bigg) \nonumber \\ 
&= \frac{1}{\sqrt{2}} \raisebox{-0.8cm}{\tikz{
\draw (0, 0)--(0, 0.8);
\draw (0.1, 0)--(0.1, 0.8);
\draw (0.9, 0)--(0.9, 0.8);
\draw (1, 0)--(1, 0.8);
\draw[thick, blue] (-0.2, 0.4)--(1.2, 0.4);
\draw [decorate, decoration = {brace, mirror}] (-0.05, -0.05)--(1.05, -0.05);
\node at (0.525, 0.2) {$\cdots$};
\node at (0.5, -0.35) {\scriptsize $2m+2n+2$};
}} \,\, , 
\end{align}
where we used an expansion appearing in TABLE~\ref{tab:Majorana-diagram-expansion} in the first equality, the parity projection Eq.~\eqref{eq:parity-even-projection} in the second equality, 
\begin{equation}
\raisebox{-0.1cm}{\tikz{
\draw (0, -0.1)--(0, 0.2) arc (180:0:0.2)--(0.4, -0.1);

\node at (0, 0) [circle, fill, inner sep=1pt] {};
\node at (0.4, 0) [circle, fill, inner sep=1pt] {};

\node[label={[xshift=-0.15cm, yshift=-0.23cm] \scriptsize $k$}] at (0, 0.2) [circle, fill, inner sep=1pt] {};
}} \,\, = (-1)^k \raisebox{-0.1cm}{\tikz{
\draw (0, -0.1)--(0, 0.2) arc (180:0:0.2)--(0.4, -0.1);

\node[label={[xshift=-0.15cm, yshift=-0.23cm] \scriptsize $k$}] at (0, 0.2) [circle, fill, inner sep=1pt] {};
}}
\end{equation}
in the third equality, and the final equality follows from the fact that two diagrams on the fourth line are independent of $k$ and are, in fact, identical.
\end{proof}

The second one is the \textit{SWAP-hole relation}. This relation is one of the unique features of the 2d Quon language, for which there is no corresponding counterpart exists in the 3D Quon language. In a sense, it reflects that a 2D Quon diagram is ultimately a projection of a 3D Quon diagram. The SWAP-hole relation is given by 
\begin{equation}
\raisebox{-0.5cm}{\tikz{
\begin{scope}[even odd rule]
\clip (0, 0) rectangle (1.8, 1)  (1, 0.85) arc (0:360:0.1);
\clip (0, 0) rectangle (1.8, 1)  (1, 0) arc (0:180:0.1);

\fill [blue!10] (0, 0) rectangle (1.8, 1);
\end{scope}

\draw[thick, blue] (1, 0.85) arc (0:360:0.1);
\draw[thick, blue] (1, 0) arc (0:180:0.1);

\draw (1.2, 0) to[out=135, in=-80] (0.1, 1);
\draw (1.7, 0) to[out=135, in=-80] (0.6, 1);

\draw (0.1, 0) to[out=45, in=-150] (0.6, 0.35);
\draw (0.7, 0.4)--(0.85, 0.5);
\draw (0.95, 0.58) to[out=45, in=-100] (1.2, 1);

\draw (0.6, 0) to[out=45, in=-140] (0.85, 0.19);
\draw (0.94, 0.25)--(1.1, 0.35);
\draw (1.2, 0.4) to[out=30, in=-100] (1.7, 1);

\draw[thick, blue] (0, 0)--(0, 1);
\draw[thick, blue] (1.8, 0)--(1.8, 1);

\node at (0.5, 0.1) {\scriptsize $\cdots$};
\node at (0.4, 0.9) {\scriptsize $\cdots$};
\node at (1.35, 0.1) {\scriptsize $\cdots$};
\node at (1.45, 0.9) {\scriptsize $\cdots$};
}} \,\, = \,\, \raisebox{-0.5cm}{\tikz{
\begin{scope}[even odd rule]
\clip (0, 0) rectangle (1.8, 1)  (1, 0) arc (0:180:0.1);

\fill [blue!10] (0, 0) rectangle (1.8, 1);
\end{scope}

\draw[thick, blue] (1, 0) arc (0:180:0.1);

\draw (1.2, 0) to[out=135, in=-80] (0.1, 1);
\draw (1.7, 0) to[out=135, in=-80] (0.6, 1);

\draw (0.1, 0) to[out=45, in=-150] (0.6, 0.35);
\draw (0.7, 0.4)--(0.85, 0.5);
\draw (0.95, 0.58) to[out=45, in=-100] (1.2, 1);

\draw (0.6, 0) to[out=45, in=-140] (0.85, 0.19);
\draw (0.94, 0.25)--(1.1, 0.35);
\draw (1.2, 0.4) to[out=30, in=-100] (1.7, 1);

\draw[thick, blue] (0, 0)--(0, 1);
\draw[thick, blue] (1.8, 0)--(1.8, 1);

\node at (0.5, 0.1) {\scriptsize $\cdots$};
\node at (0.4, 0.9) {\scriptsize $\cdots$};
\node at (1.35, 0.1) {\scriptsize $\cdots$};
\node at (1.45, 0.9) {\scriptsize $\cdots$};
}} \,\, . 
\end{equation}
This is referred to as the SWAP-hole relation, as the ``inner'' hole can be removed if it is adjacent to the SWAP operator. The proof follows directly from the fact that the parity-even projections are preserved under the braidings within the SWAP operator.

Before closing the section, let us make a final remark regarding an algorithm for evaluating a closed 2D Quon diagram. Consider the following example: 
\begin{equation}
\label{eq:quon-diagram-to-sum-of-majorana-diagrams}
\raisebox{-1.6cm}{
\begin{tikzpicture}[scale=0.7]
\begin{scope}[even odd rule]
\clip (-0.2, 0.4) to[out=-130, in=130] (-0.2, -2.8) to[out=-50, in=-130, looseness=0.7] (3.6, -2.8) to[out=50, in=-50] (3.6, 0.4) to[out=130, in=50, looseness=0.7] (-0.2, 0.4)  (1.2, -0.5) arc (0:360:0.2);
\clip (-0.2, 0.4) to[out=-130, in=130] (-0.2, -2.8) to[out=-50, in=-130, looseness=0.7] (3.6, -2.8) to[out=50, in=-50] (3.6, 0.4) to[out=130, in=50, looseness=0.7] (-0.2, 0.4)  (2.6, -0.5) arc (0:360:0.2);
\clip (-0.2, 0.4) to[out=-130, in=130] (-0.2, -2.8) to[out=-50, in=-130, looseness=0.7] (3.6, -2.8) to[out=50, in=-50] (3.6, 0.4) to[out=130, in=50, looseness=0.7] (-0.2, 0.4)  (1.9, -1.8) arc (0:360:0.2);

\fill [blue!10] (-0.2, 0.4) to[out=-130, in=130] (-0.2, -2.8) to[out=-50, in=-130, looseness=0.7] (3.6, -2.8) to[out=50, in=-50] (3.6, 0.4) to[out=130, in=50, looseness=0.7] (-0.2, 0.4);
\end{scope}

\draw[thick, blue] (-0.2, 0.4) to[out=-130, in=130] (-0.2, -2.8) to[out=-50, in=-130, looseness=0.7] (3.6, -2.8) to[out=50, in=-50] (3.6, 0.4) to[out=130, in=50, looseness=0.7] (-0.2, 0.4);

\draw[thick, blue] (1.2, -0.5) arc (0:360:0.2);
\draw[thick, blue] (2.6, -0.5) arc (0:360:0.2);
\draw[thick, blue] (1.9, -1.8) arc (0:360:0.2);

\draw (0, 0.4) rectangle (3.4, 0);
\node at (1.7, 0.2) {\scriptsize $f_1$};

\draw (0.2, 0) to[out=-90, in=90] (-0.2, -1);
\draw (0.4, 0) to[out=-90, in=90] (0, -1);
\draw (0.6, 0) to[out=-90, in=90] (0.2, -1);
\draw (0.8, 0) to[out=-90, in=90] (0.4, -1);

\draw (1.4, 0) to[out=-90, in=90] (1.4, -1);
\draw (1.6, 0) to[out=-90, in=90] (1.6, -1);
\draw (1.8, 0) to[out=-90, in=90] (1.8, -1);
\draw (2, 0) to[out=-90, in=90] (2, -1);

\draw (2.6, 0) to[out=-90, in=90] (3, -1);
\draw (2.8, 0) to[out=-90, in=90] (3.2, -1);
\draw (3, 0) to[out=-90, in=90] (3.4, -1);
\draw (3.2, 0) to[out=-90, in=90] (3.6, -1);

\draw (-0.4, -1.4) rectangle (3.8, -1);
\node at (1.7, -1.2) {\scriptsize $f_2$};

\draw (0, -1.4) to[out=-90, in=90] (0.4, -2.4);
\draw (0.2, -1.4) to[out=-90, in=90] (0.6, -2.4);
\draw (0.4, -1.4) to[out=-90, in=90] (0.8, -2.4);
\draw (0.6, -1.4) to[out=-90, in=90] (1, -2.4);
\draw (0.8, -1.4) to[out=-90, in=90] (1.2, -2.4);
\draw (1, -1.4) to[out=-90, in=90] (1.4, -2.4);

\draw (2.4, -1.4) to[out=-90, in=90] (2, -2.4);
\draw (2.6, -1.4) to[out=-90, in=90] (2.2, -2.4);
\draw (2.8, -1.4) to[out=-90, in=90] (2.4, -2.4);
\draw (3, -1.4) to[out=-90, in=90] (2.6, -2.4);
\draw (3.2, -1.4) to[out=-90, in=90] (2.8, -2.4);
\draw (3.4, -1.4) to[out=-90, in=90] (3, -2.4);

\draw (0, -2.8) rectangle (3.4, -2.4);
\node at (1.7, -2.6) {\scriptsize $f_3$};
\end{tikzpicture}
} \, = \, \raisebox{-1.15cm}{
\begin{tikzpicture}[scale=0.7]
\draw[thick, blue] (-0.2, -0.3)--(1, -0.3);
\draw[thick, blue] (-0.3, -0.7)--(2.3, -0.7);
\draw[thick, blue] (-0.1, -1.9)--(1.5, -1.9);

\draw (0, 0.4) rectangle (3.4, 0);
\node at (1.7, 0.2) {\scriptsize $f_1$};

\draw (0.2, 0) to[out=-90, in=90] (-0.2, -1);
\draw (0.4, 0) to[out=-90, in=90] (0, -1);
\draw (0.6, 0) to[out=-90, in=90] (0.2, -1);
\draw (0.8, 0) to[out=-90, in=90] (0.4, -1);

\draw (1.4, 0) to[out=-90, in=90] (1.4, -1);
\draw (1.6, 0) to[out=-90, in=90] (1.6, -1);
\draw (1.8, 0) to[out=-90, in=90] (1.8, -1);
\draw (2, 0) to[out=-90, in=90] (2, -1);

\draw (2.6, 0) to[out=-90, in=90] (3, -1);
\draw (2.8, 0) to[out=-90, in=90] (3.2, -1);
\draw (3, 0) to[out=-90, in=90] (3.4, -1);
\draw (3.2, 0) to[out=-90, in=90] (3.6, -1);

\draw (-0.4, -1.4) rectangle (3.8, -1);
\node at (1.7, -1.2) {\scriptsize $f_2$};

\draw (0, -1.4) to[out=-90, in=90] (0.4, -2.4);
\draw (0.2, -1.4) to[out=-90, in=90] (0.6, -2.4);
\draw (0.4, -1.4) to[out=-90, in=90] (0.8, -2.4);
\draw (0.6, -1.4) to[out=-90, in=90] (1, -2.4);
\draw (0.8, -1.4) to[out=-90, in=90] (1.2, -2.4);
\draw (1, -1.4) to[out=-90, in=90] (1.4, -2.4);

\draw (2.4, -1.4) to[out=-90, in=90] (2, -2.4);
\draw (2.6, -1.4) to[out=-90, in=90] (2.2, -2.4);
\draw (2.8, -1.4) to[out=-90, in=90] (2.4, -2.4);
\draw (3, -1.4) to[out=-90, in=90] (2.6, -2.4);
\draw (3.2, -1.4) to[out=-90, in=90] (2.8, -2.4);
\draw (3.4, -1.4) to[out=-90, in=90] (3, -2.4);

\draw (0, -2.8) rectangle (3.4, -2.4);
\node at (1.7, -2.6) {\scriptsize $f_3$};
\end{tikzpicture}
} \,\, , 
\end{equation}
where $f_1$, $f_2$, and $f_3$ are parity-even Majorana diagrams. One approach to evaluate the Quon diagram is to apply Eq.~\eqref{eq:parity-even-projection} repeatedly starting from the RHS, thereby rewriting the diagram as a sum over closed Majorana diagrams. Since each closed Majorana diagram can be efficiently evaluated using matchgate computation as outlined in Sec.~\ref{sec:Majorana-diagrams}, a closed \textit{Quon diagram} containing $n_h$ holes can be rewritten as a summation of $2^{n_h}$ \textit{Majorana diagrams}, where each individual term can be efficiently computed using matchgate computations.

\section{Quon Representations and Universality}
\label{sec:quon-universal}
In this section, we employ 2D Quon diagrams to represent quantum states, gates, and, more generally, tensor networks. One of the main goal of this section is to show the universality of the Quon diagram. By \textit{universality}, we mean that any tensor network can be represented by a \textit{single} Quon diagram, not a superposition of Quon diagrams. To establish this, we first illustrate how Quon diagrams can represent arbitrary quantum states and unitary quantum gates, using Clifford and matchgate quantum computations as examples. 

We then take a slight detour by introducing \textit{planar tensor networks} and tensor network operations. While most of the tensor operations discussed in the section are fairly standard, we introduce the \textit{planar region} as a new tool to manifest the network's planarity. Additionally, we interpret matchgate tensor networks through the lens of the planar region and introduce a new class of tensor networks called the \textit{punctured matchgate tensor network}. This punctured matchgate tensor network is subsequently used to prove our decomposition theorem: any tensor network can be rewritten as a contraction between a single Clifford tensor and a single matchgate tensor. 

We conclude the section by demonstrating that planar tensor network operations are naturally compatible with the framework of 2D Quon diagrams and that a generating set of arbitrary tensor networks can be represented by 2D Quon diagrams. By leveraging this compatibility and the representability, we establish the universality of the 2D Quon diagrammatic language. Notably, constructing a 2D Quon diagram from a given tensor network is efficient in the following sense: if the tensor network is specified by tensors each with an efficient representation (e.g., (potentially large-ranked) Clifford and matchgate tensors and small-ranked tensors whose components are provided), then a representing Quon diagram can always be efficiently found on a classical computer, although it is not unique due to diagrammatic rewriting rules.

\subsection{Arbitrary Quantum States and Gates}
\label{sec:quon-for-q-states-gates}
Here, we illustrate how 2D Quon diagrams can encode arbitrary quantum states and gates. In Sec.~\ref{sec:Majorana-diagrams}, we demonstrated that a Majorana diagram can encode matchgate quantum circuits by employing the dense encoding, i.e., encoding $1$ qubit using $2$ Majoranas. We also explained that the SWAP gate is not encodable using the dense encoding; however, as shown in Sec.~\ref{sec:background-manifold}, especially in Eq.~\eqref{eq:SWAP-quon-diagram}, the SWAP gate can be realized using the background manifold. Here, we utilize the parity-even projections imposed by the background manifold Eq.~\eqref{eq:parity-even-projection} to encode $1$ qubit using $4$ Majoranas following the approach originally adopted using the 3D Quon language~\cite{liu2017quon, liu2019quon}, now using the 2D Quon language. Such an encoding is also referred to as the \textit{sparse encoding} in the literature~\cite{sarma2015majorana}, but now enriched by the additional flexibility afforded by the diagrammatic rewriting rules of Quon. 

In the remainder of this section, we use the term ``logical'' to denote sparse encoding in order to distinguish it from the dense encoding discussed in Sec.~\ref{sec:Majorana-diagrams}. However, whenever the meaning is clear from the context, we omit the term ``logical.'' As a proof-of-principle demonstration, we construct 2D Quon diagrams for examples drawn from Clifford and matchgate computations. 

As in Sec.~\ref{sec:Quon-diagrams}, time flows vertically from top to bottom. The $\dagger$-operation functions similarly, taking horizontal reflections of the diagrams, applying complex conjugation to scalar coefficients, negating the scattering angles. The tensor product corresponds to stacking in the horizontal direction, and the composition corresponds to gluing in the vertical direction. 

\subsubsection{States}
We first represent the logical computational basis states using 2D Quon diagrams\footnote{We adopt a different convention from Ref.~\onlinecite{liu2017quon}. Our convention differs by a Hadamard basis change and is particularly useful when representing matchgate computations.}: 
\begin{equation}
\label{eq:logical-qubit-4-majoranas}
\vert 0 \rangle_\textrm{L} := \frac{1}{\sqrt{2}} \,\, 
\raisebox{-0.2cm}{
\tikz{
\draw[thick, blue, fill=blue!10] (0, 0)--(0, 0.3) arc (180:0:0.5)--(1, 0);

\draw (0.2, 0)--(0.2, 0.3) arc (180:0:0.3)--(0.8, 0);
\draw (0.4, 0)--(0.4, 0.3) arc (180:0:0.1)--(0.6, 0);
}} \,\, \quad \textrm{and} \quad 
\vert 1 \rangle_\textrm{L} := \frac{1}{\sqrt{2}} \,\, 
\raisebox{-0.2cm}{
\tikz{
\draw[thick, blue, fill=blue!10] (0, 0)--(0, 0.3) arc (180:0:0.5)--(1, 0);

\draw (0.2, 0)--(0.2, 0.3) arc (180:0:0.3)--(0.8, 0);
\draw (0.4, 0)--(0.4, 0.3) arc (180:0:0.1)--(0.6, 0);

\node at (0.6, 0.2) [circle, fill, inner sep=1pt] {};
\node at (0.8, 0.2) [circle, fill, inner sep=1pt] {};
}} \,\, , 
\end{equation}
where $1/\sqrt{2}$ is a normalization factor due to Majorana loop amplitude appearing in TABLE.~\eqref{tab:majorana-rewriting-rules} and~\eqref{eq:eval-empty-background}. Under the $\dagger$-operation, we have the corresponding ket states: 
\begin{equation}
\langle 0 \vert_\textrm{L} = \frac{1}{\sqrt{2}} \,\, 
\raisebox{-0.35cm}{\tikz{
\draw[thick, blue, fill=blue!10] (0, 0.8)--(0, 0.5) arc (-180:0:0.5)--(1, 0.8);

\draw (0.2, 0.8)--(0.2, 0.5) arc (-180:0:0.3)--(0.8, 0.8);
\draw (0.4, 0.8)--(0.4, 0.5) arc (-180:0:0.1)--(0.6, 0.8);
}} \,\, \quad \textrm{and} \quad 
\langle 1 \vert_\textrm{L} = \frac{1}{\sqrt{2}} \,\, 
\raisebox{-0.35cm}{\tikz{
\draw[thick, blue, fill=blue!10] (0, 0.8)--(0, 0.5) arc (-180:0:0.5)--(1, 0.8);

\draw (0.2, 0.8)--(0.2, 0.5) arc (-180:0:0.3)--(0.8, 0.8);
\draw (0.4, 0.8)--(0.4, 0.5) arc (-180:0:0.1)--(0.6, 0.8);

\node at (0.6, 0.6) [circle, fill, inner sep=1pt] {};
\node at (0.8, 0.6) [circle, fill, inner sep=1pt] {};
}} \,\, . 
\end{equation}
The logical $\vert + \rangle_\textrm{L}$ and $\vert - \rangle_\textrm{L}$ states are represented as 
\begin{equation}
\label{eq:logical-qubit-pm-4-majoranas}
\vert + \rangle_\textrm{L} = \frac{1}{\sqrt{2}} \,\, 
\raisebox{-0.2cm}{
\tikz{
\draw[thick, blue, fill=blue!10] (0, 0)--(0, 0.3) arc (180:0:0.5)--(1, 0);

\draw (0.2, 0)--(0.2, 0.3) arc (180:0:0.1)--(0.4, 0);
\draw (0.6, 0)--(0.6, 0.3) arc (180:0:0.1)--(0.8, 0);
}} \,\, \quad \textrm{and} \quad 
\vert - \rangle_\textrm{L} = \frac{1}{\sqrt{2}} \,\, 
\raisebox{-0.2cm}{
\tikz{
\draw[thick, blue, fill=blue!10] (0, 0)--(0, 0.3) arc (180:0:0.5)--(1, 0);

\draw (0.2, 0)--(0.2, 0.3) arc (180:0:0.1)--(0.4, 0);
\draw (0.6, 0)--(0.6, 0.3) arc (180:0:0.1)--(0.8, 0);

\node at (0.4, 0.2) [circle, fill, inner sep=1pt] {};
\node at (0.6, 0.2) [circle, fill, inner sep=1pt] {};
}} \,\, , 
\end{equation}
which can be demonstrated, e.g., by taking the inner product with the states in Eq.~\eqref{eq:logical-qubit-4-majoranas} and applying the diagrammatic rules in TABLE.~\eqref{tab:majorana-rewriting-rules} to evaluate the diagrams. 

The tensor product is represented by the horizontal stacking of diagrams; for example: 
\begin{equation}
\vert 0, 0 \rangle_\textrm{L} = \frac{1}{2} \,\,\,\, 
\raisebox{-0.1cm}{\tikz{
\draw[thick, blue, fill=blue!10] (1, 0.5) arc (0:180:0.5);
\draw[thick, blue, fill=blue!10] (2.4, 0.5) arc (0:180:0.5);

\draw (0.6, 0.5) arc (0:180:0.1);
\draw (0.8, 0.5) arc (0:180:0.3);

\draw (2, 0.5) arc (0:180:0.1);
\draw (2.2, 0.5) arc (0:180:0.3);
}} \,\, , 
\end{equation}
where we also multiply the corresponding scalar coefficients. 

\subsubsection{Gates}
Next, we represent a set of $1$- and $2$-qubit unitary gates using 2D Quon diagrams, where each unitary is either Clifford or matchgate. Moreover, this set forms a universal generating set for arbitrary unitaries. Combined with the Quon diagrammatic representations of basis states, these results imply that the 2D Quon language is universal for arbitrary quantum computation. 

First of all, the \textit{identity} operator is represented by \textit{$4$ parallel Majorana strings}: 
\begin{equation}
I_\textrm{L} = 
\raisebox{-0.2cm}{\tikz{
\draw[fill, blue!10] (0, 0) rectangle (1, 0.6);

\draw[thick, blue] (0, 0)--(0, 0.6);
\draw[thick, blue] (1, 0)--(1, 0.6);

\draw (0.2, 0)--(0.2, 0.6);
\draw (0.4, 0)--(0.4, 0.6);
\draw (0.6, 0)--(0.6, 0.6);
\draw (0.8, 0)--(0.8, 0.6);
}} \,\, , 
\end{equation}
and the \textit{resolution of the identity}, a familiar identity $I_\textrm{L} = \vert 0 \rangle \langle 0 \vert_\textrm{L} + \vert 1 \rangle \langle 1 \vert_\textrm{L}$ expressed in diagrams, can be derived as follows: 
\begin{align}
\label{eq:resolution-of-id-1}
\raisebox{-0.55cm}{\tikz{
\draw[fill, blue!10] (0, 0) rectangle (1, 1.2);
\draw[thick, blue] (0, 0)--(0, 1.2);
\draw[thick, blue] (1, 0)--(1, 1.2);
\draw (0.2, 0)--(0.2, 1.2);
\draw (0.4, 0)--(0.4, 1.2);
\draw (0.6, 0)--(0.6, 1.2);
\draw (0.8, 0)--(0.8, 1.2);
}} \,\, &= \frac{1}{2} \sum_{k, l = 0, 1} \raisebox{-0.55cm}{\tikz{
\draw[thick, blue] (0.1, 0.1)--(0.9, 0.1);
\draw[thick, blue] (0.1, 1.1)--(0.9, 1.1);
\draw (0.2, 0)--(0.2, 0.2) arc (180:0:0.3)--(0.8, 0);
\draw (0.4, 0)--(0.4, 0.2) arc (180:0:0.1)--(0.6, 0);
\draw (0.2, 1.2)--(0.2, 1) arc (-180:0:0.3)--(0.8, 1.2);
\draw (0.4, 1.2)--(0.4, 1) arc (-180:0:0.1)--(0.6, 1.2);
\node[label={[xshift=-0.11cm, yshift=-0.23cm] \tiny $k$}] at (0.6, 0.2) [circle, fill, inner sep=1pt] {};
\node[label={[xshift=-0.11cm, yshift=-0.23cm] \tiny $k$}] at (0.6, 1) [circle, fill, inner sep=1pt] {};
\node[label={[xshift=-0.11cm, yshift=-0.23cm] \tiny $l$}] at ({0.5 + 0.3*cos(35)}, {0.2 + 0.3*sin(35)}) [circle, fill, inner sep=1pt] {};
\node[label={[xshift=-0.11cm, yshift=-0.23cm] \tiny $l$}] at ({0.5 + 0.3*cos(35)}, {1 - 0.3*sin(35)}) [circle, fill, inner sep=1pt] {};
}} \,\, = \frac{1}{2} \sum_{k = 0, 1} \raisebox{-0.55cm}{\tikz{
\draw[thick, blue] (0.1, 0.1)--(0.9, 0.1);
\draw[thick, blue] (0.1, 1.1)--(0.9, 1.1);
\draw (0.2, 0)--(0.2, 0.2) arc (180:0:0.3)--(0.8, 0);
\draw (0.4, 0)--(0.4, 0.2) arc (180:0:0.1)--(0.6, 0);
\draw (0.2, 1.2)--(0.2, 1) arc (-180:0:0.3)--(0.8, 1.2);
\draw (0.4, 1.2)--(0.4, 1) arc (-180:0:0.1)--(0.6, 1.2);
\node[label={[xshift=-0.11cm, yshift=-0.23cm] \tiny $k$}] at (0.6, 0.2) [circle, fill, inner sep=1pt] {};
\node[label={[xshift=-0.11cm, yshift=-0.23cm] \tiny $k$}] at (0.6, 1) [circle, fill, inner sep=1pt] {};
\node[label={[xshift=-0.11cm, yshift=-0.23cm] \tiny $k$}] at ({0.5 + 0.3*cos(35)}, {0.2 + 0.3*sin(35)}) [circle, fill, inner sep=1pt] {};
\node[label={[xshift=-0.11cm, yshift=-0.23cm] \tiny $k$}] at ({0.5 + 0.3*cos(35)}, {1 - 0.3*sin(35)}) [circle, fill, inner sep=1pt] {};
}} \nonumber \\ 
&= \frac{1}{2} \Bigg[ \,\, \raisebox{-0.55cm}{\tikz{
\draw[thick, blue, fill=blue!10] (0, 1.2) arc (-180:0:0.5);
\draw[thick, blue, fill=blue!10] (0, 0) arc (180:0:0.5);
\draw (0.2, 0)--(0.2, 0.1) arc (180:0:0.3)--(0.8, 0);
\draw (0.4, 0)--(0.4, 0.1) arc (180:0:0.1)--(0.6, 0);
\draw (0.2, 1.2)--(0.2, 1.1) arc (-180:0:0.3)--(0.8, 1.2);
\draw (0.4, 1.2)--(0.4, 1.1) arc (-180:0:0.1)--(0.6, 1.2);
}} \,\, + \,\, \raisebox{-0.55cm}{\tikz{
\draw[thick, blue, fill=blue!10] (0, 1.2) arc (-180:0:0.5);
\draw[thick, blue, fill=blue!10] (0, 0) arc (180:0:0.5);
\draw (0.2, 0)--(0.2, 0.1) arc (180:0:0.3)--(0.8, 0);
\draw (0.4, 0)--(0.4, 0.1) arc (180:0:0.1)--(0.6, 0);
\draw (0.2, 1.2)--(0.2, 1.1) arc (-180:0:0.3)--(0.8, 1.2);
\draw (0.4, 1.2)--(0.4, 1.1) arc (-180:0:0.1)--(0.6, 1.2);
\node at (0.6, 0.1) [circle, fill, inner sep=1pt] {};
\node at (0.8, 0.1) [circle, fill, inner sep=1pt] {};
\node at (0.6, 1.1) [circle, fill, inner sep=1pt] {};
\node at (0.8, 1.1) [circle, fill, inner sep=1pt] {};
}} \,\, \Bigg] , 
\end{align}
where we used a diagrammatic expansion in TABLE~\ref{tab:Majorana-diagram-expansion} in the first equality and imposed the parity-even projection in the second equality, ensuring that the term vanishes when $k \ne l$. 

The Pauli-$X$, -$Y$, and -$Z$ operators can be represented using \textit{dots}: 
\begin{equation}
X_\textrm{L} = 
\raisebox{-0.2cm}{\tikz{
\draw[fill, blue!10] (0, 0) rectangle (1, 0.6);
\draw[thick, blue] (0, 0)--(0, 0.6);
\draw[thick, blue] (1, 0)--(1, 0.6);

\draw (0.2, 0)--(0.2, 0.6);
\draw (0.4, 0)--(0.4, 0.6);
\draw (0.6, 0)--(0.6, 0.6);
\draw (0.8, 0)--(0.8, 0.6);

\node at (0.6, 0.3) [circle, fill, inner sep=1pt] {};
\node at (0.8, 0.3) [circle, fill, inner sep=1pt] {};
}} \,\, \textrm{,} \quad 
Y_\textrm{L} = 
\raisebox{-0.2cm}{\tikz{
\draw[fill, blue!10] (0, 0) rectangle (1, 0.6);
\draw[thick, blue] (0, 0)--(0, 0.6);
\draw[thick, blue] (1, 0)--(1, 0.6);

\draw (0.2, 0)--(0.2, 0.6);
\draw (0.4, 0)--(0.4, 0.6);
\draw (0.6, 0)--(0.6, 0.6);
\draw (0.8, 0)--(0.8, 0.6);

\node at (0.4, 0.3) [circle, fill, inner sep=1pt] {};
\node at (0.8, 0.3) [circle, fill, inner sep=1pt] {};
}} \,\, , \quad 
Z_\textrm{L} = 
\raisebox{-0.2cm}{\tikz{
\draw[fill, blue!10] (0, 0) rectangle (1, 0.6);
\draw[thick, blue] (0, 0)--(0, 0.6);
\draw[thick, blue] (1, 0)--(1, 0.6);

\draw (0.2, 0)--(0.2, 0.6);
\draw (0.4, 0)--(0.4, 0.6);
\draw (0.6, 0)--(0.6, 0.6);
\draw (0.8, 0)--(0.8, 0.6);

\node at (0.4, 0.3) [circle, fill, inner sep=1pt] {};
\node at (0.6, 0.3) [circle, fill, inner sep=1pt] {};
}} 
\end{equation}

The \textit{phase} ($S$) and the \textit{Hadamard gate} ($H$) can be represented using \textit{braids}: 
\begin{equation}
\label{eq:Phase-gate-quon}
\textrm{S}_\textrm{L} = e^{-i \frac{\pi}{8}} \, 
\raisebox{-0.2cm}{\tikz{ 
\draw[fill, blue!10] (0, 0) rectangle (1, 0.6);
\draw[thick, blue] (0, 0)--(0, 0.6);
\draw[thick, blue] (1, 0)--(1, 0.6);

\draw (0.2, 0)--(0.2, 0.6);

\draw (0.4, 0) to[out=90, in=-110] (0.47, 0.25);
\draw (0.53, 0.35) to[out=70, in=-90] (0.6, 0.6);

\draw (0.6, 0) to[out=90, in=-90] (0.4, 0.6);
\draw (0.8, 0)--(0.8, 0.6);
}} 
\end{equation}
and
\begin{equation}
\label{eq:Hadamard-gate-quon}
H_\textrm{L} = e^{i \frac{\pi}{8}} \, 
\raisebox{-0.2cm}{\tikz{
\draw[fill, blue!10] (0, 0) rectangle (1, 0.6);
\draw[thick, blue] (0, 0)--(0, 0.6);
\draw[thick, blue] (1, 0)--(1, 0.6);

\draw (0.2, 0) to[out=90, in=-90] (0.8, 0.6);

\draw (0.4, 0) to[out=90, in=-70] (0.36, 0.18);
\draw (0.31, 0.27) to[out=110, in=-90] (0.2, 0.6);

\draw (0.6, 0) to[out=90, in=-70] (0.52, 0.26);
\draw (0.47, 0.35) to[out=110, in=-90] (0.4, 0.6);

\draw (0.8, 0) to[out=90, in=-60] (0.68, 0.33);
\draw (0.63, 0.43) to[out=120, in=-90] (0.6, 0.6);
}} \,\, , 
\end{equation}
which together generate all $1$-qubit Clifford gates. Since the Quon diagram for the Hadamard gate Eq.~\eqref{eq:Hadamard-gate-quon} involves three braids, it is often convenient to consider the following Quon diagram for the following $X$-rotation gate, which involves a single braid: 
\begin{equation}
e^{-i \frac{\pi}{4} X_\textrm{L}} = e^{i \frac{\pi}{8}} \,\, \raisebox{-0.2cm}{\tikz{ 
\draw[fill, blue!10] (0, 0) rectangle (1, 0.6);
\draw[thick, blue] (0, 0)--(0, 0.6);
\draw[thick, blue] (1, 0)--(1, 0.6);

\draw (0.2, 0) to[out=90, in=-90] (0.4, 0.6);

\draw (0.4, 0) to[out=90, in=-70] (0.33, 0.25);
\draw (0.27, 0.35) to[out=110, in=-90] (0.2, 0.6);

\draw (0.6, 0)--(0.6, 0.6);
\draw (0.8, 0)--(0.8, 0.6);
}} = e^{i \frac{\pi}{8}} \,\, \raisebox{-0.2cm}{\tikz{ 
\draw[fill, blue!10] (0, 0) rectangle (1, 0.6);
\draw[thick, blue] (0, 0)--(0, 0.6);
\draw[thick, blue] (1, 0)--(1, 0.6);

\draw (0.2, 0)--(0.2, 0.6);
\draw (0.4, 0)--(0.4, 0.6);

\draw (0.6, 0) to[out=90, in=-90] (0.8, 0.6);

\draw (0.8, 0) to[out=90, in=-70] (0.73, 0.25);
\draw (0.67, 0.35) to[out=110, in=-90] (0.6, 0.6);
}} \,\, , 
\end{equation}
where we also note that $S$ gate is a $Z$-rotation gate with the same rotation angle, i.e., $S = e^{i \frac{\pi}{4}} e^{-i \frac{\pi}{4} Z}$. By parsing the Hadamard Quon diagram into three parts, one can immediately show the following identity: $H = e^{i \frac{\pi}{4}} (e^{i \frac{\pi}{4} X}) (S^{-1}) (e^{i \frac{\pi}{4} X})$. 

The \textit{Pauli-$Z$ rotation}, a matchgate $1$-qubit unitary gate, can be represented using a \textit{scattering} element: 
\begin{equation}
\label{eq:Pauli-Z-rotation-quon}
e^{i \frac{\theta}{2}} e^{-i \frac{\theta}{2} Z_\textrm{L}} = 
\raisebox{-0.2cm}{\tikz{
\draw[fill, blue!10] (0, 0) rectangle (1, 0.6);
\draw[thick, blue] (0, 0)--(0, 0.6);
\draw[thick, blue] (1, 0)--(1, 0.6);

\draw (0.2, 0)--(0.2, 0.6);

\draw (0.35, 0) to[out=90, in=-120] (0.4, 0.13);
\draw (0.4, 0.47) to[out=120, in=-90] (0.35, 0.6);

\draw (0.65, 0) to[out=90, in=-60] (0.6, 0.13);
\draw (0.6, 0.47) to[out=60, in=-90] (0.65, 0.6);

\draw (0.8, 0)--(0.8, 0.6);

\draw (0.7, 0.3) arc (0:360:0.2);

\node at (0.5, 0.3) {\scriptsize $\theta_{\updownarrow}$};
}} \,\, , 
\end{equation}
where $\theta \in [0, 2\pi)$ to ensure unitarity. At $\theta = \pm \pi/2$, the gate reduces to $S$ and $S^{-1}$, respectively, and at $\theta = \pm \pi$, the gate reduces to $Z$ and $Z^{-1}$, respectively.  

Using the above diagrammatic representations, it is clear that the Quon is \textit{universal} for $1$-qubit gates.

New, we present the Quon diagrams for two-qubit gates. The $\textrm{CNOT}$ gate has the following Quon representation: 
\begin{equation}
\textrm{CNOT}_{1,2} = 
\raisebox{-0.85cm}{\tikz{
\draw[fill, blue!10] (0, -1.2)--(0, 0.6)--(1, 0.6)--(1, 0.15) to[out=-90, in=-90] (1.8, 0.15)--(1.8, 0.6)--(2.8, 0.6)--(2.8, -1.2)--(1.8, -1.2)--(1.8, -0.75) to[out=90, in=90] (1, -0.75)--(1, -1.2);

\draw[thick, blue] (0, -1.2)--(0, 0.6);
\draw[thick, blue] (1, 0.6)--(1, 0.15) to[out=-90, in=-90] (1.8, 0.15)--(1.8, 0.6);
\draw[thick, blue] (1, -1.2)--(1, -0.75) to[out=90, in=90] (1.8, -0.75)--(1.8, -1.2);
\draw[thick, blue] (2.8, -1.2)--(2.8, 0.6);

\draw (0.8, 0.6) to[out=-90, in=90] (0.2, 0)--(0.2, -0.6) to[out=-90, in=110] (0.31, -1.2+0.27);
\draw (0.31, 0.27) to[out=110, in=-90] (0.2, 0.6);
\draw (0.47, 0.35) to[out=110, in=-90] (0.4, 0.6);
\draw (0.63, 0.43) to[out=120, in=-90] (0.6, 0.6);

\draw (0.36, 0.18) to[out=-70, in=90] (0.4, 0)--(0.4, -0.6) to[out=-90, in=110] (0.47, -1.2+0.35);
\draw (0.52, 0.26) to[out=-70, in=170, looseness=1.7] (1.32, -0.28);
\draw (0.68, 0.33) to[out=-70, in=170, looseness=1.7] (1.43, -0.21);
\draw (0.63, -1.2+0.43) to[out=120, in=170, looseness=1.5] (1.2, -0.37);

\draw (0.2, -1.2) to[out=90, in=-135] (0.8, -0.7) to[out=45, in=-100] (2, 0.4) to[out=80, in=-90] (2, 0.6);
\draw (0.4, -1.2) to[out=90, in=-70] (0.36, -1.2+0.18);
\draw (0.6, -1.2) to[out=90, in=-70] (0.52, -1.2+0.26);
\draw (0.8, -1.2) to[out=90, in=-60] (0.68, -1.2+0.33);

\draw (1.35, -0.39) to[out=-10, in=90, looseness=1.5] (2, -1.2);
\draw (1.45, -0.31) to[out=-10, in=90, looseness=1.5] (2.2, -1.2);
\draw (1.57, -0.23) to[out=-10, in=-90, looseness=1.5] (2.2, 0.6);
\draw (2.4, -1.2)--(2.4, 0.6);
\draw (2.6, -1.2)--(2.6, 0.6);
}} \,\, , 
\end{equation}
where the qubit associated with the left (right) is the control (target) qubit. The equation can be shown by computing its components with respect to the computational basis states Eq.~\eqref{eq:logical-qubit-4-majoranas}. Since the Quon diagram for the $\textrm{CNOT}$-gate involves multiple braids, it is useful to consider the \textit{two-qubit gate} $e^{i \frac{\pi}{4} X \otimes X}$ involving a single \textit{braid}: 
\begin{equation}
\label{eq:matchgate-XX-rotation}
e^{i \frac{\pi}{4} X_\textrm{L} \otimes X_\textrm{L}} = e^{-i \frac{\pi}{8}} \, 
\raisebox{-0.6cm}{\tikz{
\fill[blue!10] (-1.3, 0.6)--(-1.3, -0.6)--(-0.25, -0.6) arc (180:0:0.25)--(1.3, -0.6)--(1.3, 0.6)--(0.25, 0.6) arc (0:-180:0.25);

\draw[thick,blue] (-1.3, -0.6)--(-1.3, 0.6);
\draw[thick,blue] (1.3, -0.6)--(1.3, 0.6);
\draw[thick,blue] (0.25, 0.6) arc (0:-180:0.25);
\draw[thick,blue] (0.25, -0.6) arc (0:180:0.25);

\draw (-1.1, 0.6)--(-1.1, -0.6);
\draw (-0.9, 0.6)--(-0.9, -0.6);
\draw (0.9, 0.6)--(0.9, -0.6);
\draw (1.1, 0.6)--(1.1, -0.6);

\draw (-0.7, 0.6) to[out=-90, in=160] (-0.1, 0.04);
\draw (0.1, -0.04) to[out=-20, in=90] (0.7, -0.6);

\draw (-0.7, -0.6) to[out=90, in=-90] (0.7, 0.6);

\draw[looseness=1.3] (-0.5, -0.6) to[out=90, in=90] (0.5, -0.6);
\draw[looseness=1.3] (-0.5, 0.6) to[out=-90, in=-90] (0.5, 0.6);
}} \,\, , 
\end{equation}
which will be proven for a more general case below. Note that, together with one-qubit Clifford gates, $e^{i \frac{\pi}{4} X_1 \otimes X_2}$ generates the $\textrm{CNOT}$-gate: 
\begin{equation}
\textrm{CNOT}_{1 \to 2} = 
\raisebox{-0.85cm}{\tikz{
\draw[fill, blue!10] (0, -1.2)--(0, 0.6)--(1, 0.6)--(1, 0.15) to[out=-90, in=-90] (1.8, 0.15)--(1.8, 0.6)--(2.8, 0.6)--(2.8, -1.2)--(1.8, -1.2)--(1.8, -0.75) to[out=90, in=90] (1, -0.75)--(1, -1.2);

\draw[thick,blue] (0, -1.2)--(0, 0.6);
\draw[thick,blue] (1, 0.6)--(1, 0.15) to[out=-90, in=-90] (1.8, 0.15)--(1.8, 0.6);
\draw[thick,blue] (1, -1.2)--(1, -0.75) to[out=90, in=90] (1.8, -0.75)--(1.8, -1.2);
\draw[thick, blue] (2.8, -1.2)--(2.8, 0.6);

\draw (0.8, 0.6) to[out=-90, in=90] (0.2, 0)--(0.2, -0.6) to[out=-90, in=110] (0.33, -1.0);
\draw (0.31, 0.27) to[out=110, in=-90] (0.2, 0.6);
\draw (0.47, 0.35) to[out=110, in=-90] (0.4, 0.6);
\draw (0.63, 0.43) to[out=120, in=-90] (0.6, 0.6);

\draw (0.36, 0.18) to[out=-70, in=110] (0.5, -1.2+0.24);
\draw (0.52, 0.26) to[out=-70, in=160, looseness=1.7] (1.2, -0.25);
\draw (0.68, 0.33) to[out=-60, in=180, looseness=1.2] (1.4, -0.18) to[out=0, in=-130] (2.07, 0.15);
\draw (0.75, -0.68) to[out=60, in=180, looseness=1.2] (1.4, -0.45) to[out=0, in=90, looseness=1.5] (2, -1.2);

\draw (0.2, -1.2) to[out=90, in=-90] (0.8, -0.8) to[out=90, in=-120] (0.6, -0.6) to[out=60, in=-50, looseness=0.8] (2.2, 0.1) to[out=130, in=-90] (2, 0.6);

\draw (0.4, -1.2) to[out=90, in=-70] (0.36, -1.2+0.13);
\draw (0.6, -1.2) to[out=90, in=-70] (0.54, -1.2+0.17);
\draw (0.8, -1.2) to[out=90, in=-60] (0.72, -1.2+0.22);

\draw (0.68, -1.2+0.3) to[out=120, in=-120] (0.7, -0.75);

\draw (1.5, -0.33) to[out=-20, in=90, looseness=1.5] (2.2, -1.2);

\draw (2.2, 0.6) to[out=-90, in=50] (2.13, 0.25);

\draw (2.4, -1.2)--(2.4, 0.6);
\draw (2.6, -1.2)--(2.6, 0.6);
}} \,\, , 
\end{equation}
which implies that 
\begin{equation}
e^{i \frac{\pi}{4} X_1 \otimes X_2} = e^{- i \frac{\pi}{4}} \big((e^{i \frac{\pi}{4} X_1} H_1) \otimes I_2 \big) \textrm{CNOT}_{1 \to 2} (H_1 \otimes e^{i \frac{\pi}{4} X_2}) . 
\end{equation}

The logical SWAP gate between two logical qubits can be represented by a Quon diagram in Eq.~\eqref{eq:SWAP-quon-diagram} with $n=m=4$. Since the SWAP gate is equivalent to three CNOT gates applied in an appropriate order, it is instructive to start from the Quon diagram for three CNOT gates and derive the Quon diagram for the SWAP gate, with the help of the string-genus relation Eq.~\eqref{eq:string-genus}. 

We show that $2$-qubit \textit{matchgate unitary} $e^{i \frac{\theta}{2}} e^{- i \frac{\theta}{2} X \otimes X}$, where $\theta \in [0, 2 \pi)$, has the following Quon representation involving a single \textit{scattering element}\footnote{Note that at $\theta = -\pi/2$, Eq.~\eqref{eq:e-to-the-i-theta-XX-v1} reduces to Eq.~\eqref{eq:matchgate-XX-rotation}. Similar to Eq.~\eqref{eq:Pauli-Z-rotation-quon}, the gate becomes a Clifford gate when $\theta$ is equal to an integer multiple of $\frac{\pi}{2}$.}: 
\begin{equation}
\label{eq:e-to-the-i-theta-XX-v1}
e^{i \frac{\theta}{2}} e^{-i \frac{\theta}{2} X_\textrm{L} \otimes X_\textrm{L}} = \sqrt{2} \,\, 
\raisebox{-0.6cm}{\tikz{
\fill[blue!10] (-1.3, -0.7)--(-1.3, 0.7)--(-0.25, 0.7) arc (-180:0:0.25)--(1.3, 0.7)--(1.3, -0.7)--(0.25, -0.7) arc (0:180:0.25);

\draw[thick, blue] (-1.3, -0.7)--(-1.3, 0.7);
\draw[thick, blue] (1.3, -0.7)--(1.3, 0.7);
\draw[thick, blue] (0.25, 0.7) arc (0:-180:0.25);
\draw[thick, blue] (0.25, -0.7) arc (0:180:0.25);

\draw (-1.1, 0.7)--(-1.1, -0.7);
\draw (-0.9, 0.7)--(-0.9, -0.7);
\draw (0.9, 0.7)--(0.9, -0.7);
\draw (1.1, 0.7)--(1.1, -0.7);

\draw[looseness=1.3] (-0.5, 0.7) to[out=-90, in=-90] (0.5, 0.7);
\draw[looseness=1.3] (-0.5, -0.7) to[out=90, in=90] (0.5, -0.7);

\draw (-0.7, 0.7) to[out=-90, in=150] (-{0.1*sqrt(2)}, {0.1*sqrt(2)});
\draw (0.7, 0.7) to[out=-90, in=30] (+{0.1*sqrt(2)}, {0.1*sqrt(2)});

\draw (-0.7, -0.7) to[out=90, in=-150] (-{0.1*sqrt(2)}, -{0.1*sqrt(2)});
\draw (0.7, -0.7) to[out=90, in=-30] (+{0.1*sqrt(2)}, -{0.1*sqrt(2)});

\draw (0.2, 0) arc (0:360:0.2);
\node at (0, 0) {\scriptsize $\theta_{\updownarrow}$};
}} 
\end{equation}
\begin{proof} 
If we expand the RHS using a diagrammatic expansion in TABLE~\ref{tab:Majorana-diagram-expansion}, 
\begin{align}
& \frac{1 + e^{i \theta}}{\sqrt{2}} \,\, \raisebox{-0.6cm}{\tikz{
\fill[blue!10] (-1.3, -0.7)--(-1.3, 0.7)--(-0.25, 0.7) arc (-180:0:0.25)--(1.3, 0.7)--(1.3, -0.7)--(0.25, -0.7) arc (0:180:0.25);
\draw[thick, blue] (-1.3, -0.7)--(-1.3, 0.7);
\draw[thick, blue] (1.3, -0.7)--(1.3, 0.7);
\draw[thick, blue] (0.25, 0.7) arc (0:-180:0.25);
\draw[thick, blue] (0.25, -0.7) arc (0:180:0.25);
\draw (-1.1, 0.7)--(-1.1, -0.7);
\draw (-0.9, 0.7)--(-0.9, -0.7);
\draw (0.9, 0.7)--(0.9, -0.7);
\draw (1.1, 0.7)--(1.1, -0.7);
\draw[looseness=1.3] (-0.5, 0.7) to[out=-90, in=-90] (0.5, 0.7);
\draw[looseness=1.3] (-0.5, -0.7) to[out=90, in=90] (0.5, -0.7);
\draw (-0.7, 0.7) to[out=-90, in=90] (-0.2, 0) to[out=-90, in=90] (-0.7, -0.7);
\draw (0.7, 0.7) to[out=-90, in=90] (0.2, 0) to[out=-90, in=90] (0.7, -0.7);
}} \,\, + \frac{1 - e^{i \theta}}{\sqrt{2}} \,\, \raisebox{-0.6cm}{\tikz{
\fill[blue!10] (-1.3, -0.7)--(-1.3, 0.7)--(-0.25, 0.7) arc (-180:0:0.25)--(1.3, 0.7)--(1.3, -0.7)--(0.25, -0.7) arc (0:180:0.25);
\draw[thick, blue] (-1.3, -0.7)--(-1.3, 0.7);
\draw[thick, blue] (1.3, -0.7)--(1.3, 0.7);
\draw[thick, blue] (0.25, 0.7) arc (0:-180:0.25);
\draw[thick, blue] (0.25, -0.7) arc (0:180:0.25);
\draw (-1.1, 0.7)--(-1.1, -0.7);
\draw (-0.9, 0.7)--(-0.9, -0.7);
\draw (0.9, 0.7)--(0.9, -0.7);
\draw (1.1, 0.7)--(1.1, -0.7);
\draw[looseness=1.3] (-0.5, 0.7) to[out=-90, in=-90] (0.5, 0.7);
\draw[looseness=1.3] (-0.5, -0.7) to[out=90, in=90] (0.5, -0.7);
\draw (-0.7, 0.7) to[out=-90, in=90] (-0.2, 0) to[out=-90, in=90] (-0.7, -0.7);
\draw (0.7, 0.7) to[out=-90, in=90] (0.2, 0) to[out=-90, in=90] (0.7, -0.7);
\node at (-0.2, 0) [circle, fill, inner sep=1pt] {};
\node at (0.2, 0) [circle, fill, inner sep=1pt] {};
}} \nonumber \\ 
&= \frac{1 + e^{i \theta}}{\sqrt{2}} \,\, \raisebox{-0.6cm}{\tikz{
\fill[blue!10] (-1.3, -0.7)--(-1.3, 0.7)--(-0.25, 0.7) arc (-180:0:0.25)--(1.3, 0.7)--(1.3, -0.7)--(0.25, -0.7) arc (0:180:0.25);
\draw[thick, blue] (-1.3, -0.7)--(-1.3, 0.7);
\draw[thick, blue] (1.3, -0.7)--(1.3, 0.7);
\draw[thick, blue] (0.25, 0.7) arc (0:-180:0.25);
\draw[thick, blue] (0.25, -0.7) arc (0:180:0.25);
\draw (-1.1, 0.7)--(-1.1, -0.7);
\draw (-0.9, 0.7)--(-0.9, -0.7);
\draw (0.9, 0.7)--(0.9, -0.7);
\draw (1.1, 0.7)--(1.1, -0.7);
\draw[looseness=1.3] (-0.5, 0.7) to[out=-90, in=-90] (0.5, 0.7);
\draw[looseness=1.3] (-0.5, -0.7) to[out=90, in=90] (0.5, -0.7);
\draw (-0.7, 0.7)--(-0.7, -0.7);
\draw (0.7, 0.7)--(0.7, -0.7);
}} \,\, + \frac{1 - e^{i \theta}}{\sqrt{2}} \,\, \raisebox{-0.6cm}{\tikz{
\fill[blue!10] (-1.3, -0.7)--(-1.3, 0.7)--(-0.25, 0.7) arc (-180:0:0.25)--(1.3, 0.7)--(1.3, -0.7)--(0.25, -0.7) arc (0:180:0.25);
\draw[thick, blue] (-1.3, -0.7)--(-1.3, 0.7);
\draw[thick, blue] (1.3, -0.7)--(1.3, 0.7);
\draw[thick, blue] (0.25, 0.7) arc (0:-180:0.25);
\draw[thick, blue] (0.25, -0.7) arc (0:180:0.25);
\draw (-1.1, 0.7)--(-1.1, -0.7);
\draw (-0.9, 0.7)--(-0.9, -0.7);
\draw (0.9, 0.7)--(0.9, -0.7);
\draw (1.1, 0.7)--(1.1, -0.7);
\draw[looseness=1.3] (-0.5, 0.7) to[out=-90, in=-90] (0.5, 0.7);
\draw[looseness=1.3] (-0.5, -0.7) to[out=90, in=90] (0.5, -0.7);
\draw (-0.7, 0.7)--(-0.7, -0.7);
\draw (0.7, 0.7)--(0.7, -0.7);
\node at (-1.1, -0.1) [circle, fill, inner sep=1pt] {};
\node at (-0.9, -0.1) [circle, fill, inner sep=1pt] {};
\node at (1.1, 0.1) [circle, fill, inner sep=1pt] {};
\node at (0.9, 0.1) [circle, fill, inner sep=1pt] {};
}} \nonumber \\ 
&= \frac{1 + e^{i \theta}}{2} \,\, \raisebox{-0.6cm}{\tikz{
\fill [blue!10] (-1.3, -0.7) rectangle (-0.25, 0.7);
\fill [blue!10] (1.3, -0.7) rectangle (0.25, 0.7);
\draw[thick, blue] (-1.3, -0.7)--(-1.3, 0.7);
\draw[thick, blue] (1.3, -0.7)--(1.3, 0.7);
\draw[thick, blue] (-0.25, 0.7)--(-0.25, -0.7);
\draw[thick, blue] (0.25, -0.7)--(0.25, 0.7);
\draw (-1.1, 0.7)--(-1.1, -0.7);
\draw (-0.9, 0.7)--(-0.9, -0.7);
\draw (0.9, 0.7)--(0.9, -0.7);
\draw (1.1, 0.7)--(1.1, -0.7);
\draw (-0.5, 0.7)--(-0.5, -0.7);
\draw (0.5, -0.7)--(0.5, 0.7);
\draw (-0.7, 0.7)--(-0.7, -0.7);
\draw (0.7, 0.7)--(0.7, -0.7);
}} \,\, + \frac{1 - e^{i \theta}}{\sqrt{2}} \,\, \raisebox{-0.6cm}{\tikz{
\fill[blue!10] (-1.3, -0.7)--(-1.3, 0.7)--(-0.25, 0.7) arc (-180:0:0.25)--(1.3, 0.7)--(1.3, -0.7)--(0.25, -0.7) arc (0:180:0.25);
\draw[thick, blue] (-1.3, -0.7)--(-1.3, 0.7);
\draw[thick, blue] (1.3, -0.7)--(1.3, 0.7);
\draw[thick, blue] (0.25, 0.7) arc (0:-180:0.25);
\draw[thick, blue] (0.25, -0.7) arc (0:180:0.25);
\draw (-1.1, 0.7)--(-1.1, -0.7);
\draw (-0.9, 0.7)--(-0.9, -0.7);
\draw (0.9, 0.7)--(0.9, -0.7);
\draw (1.1, 0.7)--(1.1, -0.7);
\draw[looseness=1.3] (-0.5, 0.7) to[out=-90, in=-90] (0.5, 0.7);
\draw[looseness=1.3] (-0.5, -0.7) to[out=90, in=90] (0.5, -0.7);
\draw (-0.7, 0.7)--(-0.7, -0.7);
\draw (0.7, 0.7)--(0.7, -0.7);
\node at (-1.1, -0.6) [circle, fill, inner sep=1pt] {};
\node at (-0.9, -0.6) [circle, fill, inner sep=1pt] {};
\node at (1.1, -0.6) [circle, fill, inner sep=1pt] {};
\node at (0.9, -0.6) [circle, fill, inner sep=1pt] {};
}} \nonumber \\ 
&= \frac{1 + e^{i \theta}}{2} I_\textrm{L} \otimes I_\textrm{L} + \frac{1 - e^{i \theta}}{2} \,\, \raisebox{-0.25cm}{\tikz{
\fill [blue!10] (-1.3, -0.3) rectangle (-0.25, 0.3);
\fill [blue!10] (1.3, -0.3) rectangle (0.25, 0.3);
\draw[thick, blue] (-1.3, -0.3)--(-1.3, 0.3);
\draw[thick, blue] (1.3, -0.3)--(1.3, 0.3);
\draw[thick, blue] (0.25, 0.3)--(0.25, -0.3);
\draw[thick, blue] (-0.25, 0.3)--(-0.25, -0.3);
\draw (-1.1, 0.3)--(-1.1, -0.3);
\draw (-0.9, 0.3)--(-0.9, -0.3);
\draw (0.9, 0.3)--(0.9, -0.3);
\draw (1.1, 0.3)--(1.1, -0.3);
\draw (0.5, 0.3)--(0.5, -0.3);
\draw (-0.5, -0.3)--(-0.5, 0.3);
\draw (-0.7, 0.3)--(-0.7, -0.3);
\draw (0.7, 0.3)--(0.7, -0.3);
\node at (-1.1, 0) [circle, fill, inner sep=1pt] {};
\node at (-0.9, 0) [circle, fill, inner sep=1pt] {};
\node at (1.1, 0) [circle, fill, inner sep=1pt] {};
\node at (0.9, 0) [circle, fill, inner sep=1pt] {};
}} \nonumber \\ 
&= \frac{1 + e^{i \theta}}{2} I_\textrm{L} \otimes I_\textrm{L} + \frac{1 - e^{i \theta}}{2} X_\textrm{L} \otimes X_\textrm{L} , 
\end{align}
where we used Eq.~\eqref{eq:charges-out-of-nowhere} in the first equality, Eqs.~\eqref{eq:shrink-bdy-4} and~\eqref{eq:push-neutral-diagram} in the second equality, and the last line is equal to the LHS. 
\end{proof}

In sum, using the above 2D Quon diagrammatic representations, it is immediate to show that the 2D Quon language is universal for arbitrary quantum states and gates, and is even universal for quantum circuits with postselections. In particular, we note that among the gates we discussed, $\{e^{\pm i \frac{\pi}{4} X}, e^{\pm i \frac{\pi}{4} Z}, e^{\pm i \frac{\pi}{4} X \otimes X} \}$ forms a universal generating set for Clifford unitaries and $\{X\} \cup \{ e^{i \frac{\theta}{2}} e^{-i \frac{\theta}{2} Z}, e^{i \frac{\theta}{2}} e^{-i \frac{\theta}{2} X \otimes X} \}_{\theta \in [0, 2\pi)}$ forms a universal generating set for matchgate unitaries\footnote{In the literature, the matchgate unitaries usually refer to parity-preserving ones, thereby excluding the $X$-gate. Here, we include the $X$ gate as well, since it satisfies the matchgate identity as a $2$-leg tensor, thereby ensuring that the circuit remains classically simulable even with the inclusion of $X$ gates (see Appendix~\ref{app:matchgate} for more details).}.

\subsubsection{Examples}
As an application of the Quon representations from the previous section, we explore examples of Quon diagrammatic representations for Clifford and matchgate quantum computations. While relegating the pictorial characterizations of Cliffords and matchgates in terms of Quon diagrams to Sec.~\ref{sec:clifford-matchgate-quon-diagrams}, we instead highlight their key diagrammatic properties through illustrative examples. 

The first example is the following Clifford computation in which every gate is a Clifford gate: 
\begin{align}
\langle 0, 0, 0, 0 \vert &\big( e^{i \frac{\pi}{4} X_1 X_2} e^{i \frac{\pi}{4} X_2 X_3} e^{i \frac{\pi}{4} X_3 X_4} \big) \nonumber \\
& \big( e^{-i \frac{\pi}{4} X_1} S_2 e^{-i \frac{\pi}{4} X_3} S_4 \big) \nonumber \\ 
& \big( e^{i \frac{\pi}{4} X_1 X_2} e^{i \frac{\pi}{4} X_2 X_3} e^{i \frac{\pi}{4} X_3 X_4} \big) \vert 1, 0, 0, 0 \rangle , 
\end{align}
which can be represented by the following 2D Quon diagram: 
\begin{equation}
\raisebox{-0.88cm}{
\begin{tikzpicture}[scale=0.7]
\begin{scope}[even odd rule]
\clip (0, -1.5) arc (-180:0:0.5) arc (180:0:0.15) arc (-180:0:0.5) arc (180:0:0.15) arc (-180:0:0.5) arc (180:0:0.15) arc (-180:0:0.5)--(4.9, 0.2) arc (0:180:0.5) arc (0:-180:0.15) arc (0:180:0.5) arc (0:-180:0.15) arc (0:180:0.5) arc (0:-180:0.15) arc (0:180:0.5)  (1.3, -0.5) arc (0:180:0.15)--(1, -0.8) arc (-180:0:0.15)--(1.3, -0.5);
\clip (0, -1.5) arc (-180:0:0.5) arc (180:0:0.15) arc (-180:0:0.5) arc (180:0:0.15) arc (-180:0:0.5) arc (180:0:0.15) arc (-180:0:0.5)--(4.9, 0.2) arc (0:180:0.5) arc (0:-180:0.15) arc (0:180:0.5) arc (0:-180:0.15) arc (0:180:0.5) arc (0:-180:0.15) arc (0:180:0.5)  (2.6, -0.5) arc (0:180:0.15)--(2.3, -0.8) arc (-180:0:0.15)--(2.6, -0.5);
\clip (0, -1.5) arc (-180:0:0.5) arc (180:0:0.15) arc (-180:0:0.5) arc (180:0:0.15) arc (-180:0:0.5) arc (180:0:0.15) arc (-180:0:0.5)--(4.9, 0.2) arc (0:180:0.5) arc (0:-180:0.15) arc (0:180:0.5) arc (0:-180:0.15) arc (0:180:0.5) arc (0:-180:0.15) arc (0:180:0.5)  (3.9, -0.5) arc (0:180:0.15)--(3.6, -0.8) arc (-180:0:0.15)--(3.9, -0.5);

\fill[blue!10] (0, -1.5) arc (-180:0:0.5) arc (180:0:0.15) arc (-180:0:0.5) arc (180:0:0.15) arc (-180:0:0.5) arc (180:0:0.15) arc (-180:0:0.5)--(4.9, 0.2) arc (0:180:0.5) arc (0:-180:0.15) arc (0:180:0.5) arc (0:-180:0.15) arc (0:180:0.5) arc (0:-180:0.15) arc (0:180:0.5);
\end{scope}

\draw[thick, blue] (0, 0.2)--(0, -1.5) arc (-180:0:0.5) arc (180:0:0.15) arc (-180:0:0.5) arc (180:0:0.15) arc (-180:0:0.5) arc (180:0:0.15) arc (-180:0:0.5)--(4.9, 0.2) arc (0:180:0.5) arc (0:-180:0.15) arc (0:180:0.5) arc (0:-180:0.15) arc (0:180:0.5) arc (0:-180:0.15) arc (0:180:0.5);

\draw[thick, blue] (1.3, -0.5) arc (0:180:0.15)--(1, -0.8) arc (-180:0:0.15)--(1.3, -0.5);
\draw[thick, blue] (2.6, -0.5) arc (0:180:0.15)--(2.3, -0.8) arc (-180:0:0.15)--(2.6, -0.5);
\draw[thick, blue] (3.9, -0.5) arc (0:180:0.15)--(3.6, -0.8) arc (-180:0:0.15)--(3.9, -0.5);

\draw (0.25, -0.6) to[out=140, in=-90] (0.2, -0.5)--(0.2, 0.2) arc (180:0:0.3) to[out=-90, in=-90, looseness=1.2] (1.5, 0.2) arc (180:0:0.3) to[out=-90, in=-90, looseness=1.2] (2.8, 0.2) arc (180:0:0.3) to[out=-90, in=-90, looseness=1.2] (4.1, 0.2) arc (180:0:0.3)--(4.7, -1.5) arc (0:-180:0.3) to[out=90, in=90, looseness=1.2] (3.4, -1.5) arc (0:-180:0.3) to[out=90, in=90, looseness=1.2] (2.1, -1.5) arc (0:-180:0.3) to[out=90, in=90, looseness=1.2] (0.8, -1.5) arc (0:-180:0.3)--(0.2, -0.8) to[out=90, in=-90] (0.4, -0.5)--(0.4, 0.2) arc (180:0:0.1) to[out=-90, in=170] (1.05, -0.13);

\draw (1.25, -0.17) to[out=-10, in=90] (1.7, -0.5) to[out=-90, in=90] (1.9, -0.8) to[out=-90, in=170] (2.35, -1.13);
\draw (2.55, -1.17) to[out=-10, in=90] (3, -1.5) arc (-180:0:0.1) to[out=90, in=-90, looseness=0.7] (4.3, -0.8) to[out=90, in=-140] (4.35, -0.7);

\draw (3.85, -1.17) to[out=-10, in=90] (4.3, -1.5) arc (-180:0:0.1)--(4.5, -0.8) to[out=90, in=-90] (4.3, -0.5) to[out=90, in=-10] (3.85, -0.17);
\draw (3.65, -0.13) to[out=170, in=-90] (3.2, 0.2) arc (0:180:0.1) to[out=-90, in=90] (1.9, -0.5) to[out=-90, in=40] (1.85, -0.6);

\draw (1.75, -0.7) to[out=-140, in=90] (1.7, -0.8) to[out=-90, in=90, looseness=0.7] (0.6, -1.5) arc (0:-180:0.1)--(0.4, -0.8) to[out=90, in=-40] (0.35, -0.7);
\draw (1.05, -1.13) to[out=170, in=-90] (0.6, -0.8)--(0.6, -0.5) to[out=90, in=-90, looseness=0.7] (1.7, 0.2) arc (180:0:0.1) (1.9, 0.2) to[out=-90, in=170] (2.35, -0.13);

\draw (2.55, -0.17) to[out=-10, in=90] (3, -0.5) to[out=-90, in=90] (2.8, -0.8) to[out=-90, in=-90, looseness=1.2] (2.1, -0.8)--(2.1, -0.5) to[out=90, in=90, looseness=1.2] (2.8, -0.5) to[out=-90, in=140] (2.85, -0.6);
\draw (2.95, -0.7) to[out=-40, in=90] (3, -0.8) to[out=-90, in=90] (1.9, -1.5) arc (0:-180:0.1) to[out=90, in=-10] (1.25, -1.17);

\draw (3.65, -1.13) to[out=170, in=-90] (3.2, -0.8)--(3.2, -0.5) to[out=90, in=-90] (4.3, 0.2) arc (180:0:0.1)--(4.5, -0.5) to[out=-90, in=40] (4.45, -0.6);

\draw[red] (0.8, -0.5) to[out=90, in=90, looseness=1.2] (1.5, -0.5)--(1.5, -0.8) to[out=-90, in=-90, looseness=1.2] (0.8, -0.8)--(0.8, -0.5);
\draw[red] (3.4, -0.5) to[out=90, in=90, looseness=1.2] (4.1, -0.5)--(4.1, -0.8) to[out=-90, in=-90, looseness=1.2] (3.4, -0.8)--(3.4, -0.5);

\node at (0.6, 0.2) [circle, fill, inner sep=1pt] {};
\node at (0.8, 0.2) [circle, fill, inner sep=1pt] {};
\end{tikzpicture}
} = \frac{1}{2} \raisebox{-0.88cm}{
\begin{tikzpicture}[scale=0.7]
\begin{scope}[even odd rule]
\clip (0, -1.5) arc (-180:0:0.5) arc (180:0:0.15) arc (-180:0:0.5) arc (180:0:0.15) arc (-180:0:0.5) arc (180:0:0.15) arc (-180:0:0.5)--(4.9, 0.2) arc (0:180:0.5) arc (0:-180:0.15) arc (0:180:0.5) arc (0:-180:0.15) arc (0:180:0.5) arc (0:-180:0.15) arc (0:180:0.5)  (2.6, -0.5) arc (0:180:0.15)--(2.3, -0.8) arc (-180:0:0.15)--(2.6, -0.5);

\fill[blue!10] (0, -1.5) arc (-180:0:0.5) arc (180:0:0.15) arc (-180:0:0.5) arc (180:0:0.15) arc (-180:0:0.5) arc (180:0:0.15) arc (-180:0:0.5)--(4.9, 0.2) arc (0:180:0.5) arc (0:-180:0.15) arc (0:180:0.5) arc (0:-180:0.15) arc (0:180:0.5) arc (0:-180:0.15) arc (0:180:0.5);
\end{scope}

\draw[thick, blue] (0, 0.2)--(0, -1.5) arc (-180:0:0.5) arc (180:0:0.15) arc (-180:0:0.5) arc (180:0:0.15) arc (-180:0:0.5) arc (180:0:0.15) arc (-180:0:0.5)--(4.9, 0.2) arc (0:180:0.5) arc (0:-180:0.15) arc (0:180:0.5) arc (0:-180:0.15) arc (0:180:0.5) arc (0:-180:0.15) arc (0:180:0.5);

\draw[thick, blue] (2.6, -0.5) arc (0:180:0.15)--(2.3, -0.8) arc (-180:0:0.15)--(2.6, -0.5);

\draw (0.25, -0.6) to[out=140, in=-90] (0.2, -0.5)--(0.2, 0.2) arc (180:0:0.3) to[out=-90, in=-90, looseness=1.2] (1.5, 0.2) arc (180:0:0.3) to[out=-90, in=-90, looseness=1.2] (2.8, 0.2) arc (180:0:0.3) to[out=-90, in=-90, looseness=1.2] (4.1, 0.2) arc (180:0:0.3)--(4.7, -1.5) arc (0:-180:0.3) to[out=90, in=90, looseness=1.2] (3.4, -1.5) arc (0:-180:0.3) to[out=90, in=90, looseness=1.2] (2.1, -1.5) arc (0:-180:0.3) to[out=90, in=90, looseness=1.2] (0.8, -1.5) arc (0:-180:0.3)--(0.2, -0.8) to[out=90, in=-90] (0.4, -0.5)--(0.4, 0.2) arc (180:0:0.1) to[out=-90, in=170] (1.05, -0.13);

\draw (1.25, -0.17) to[out=-10, in=90] (1.7, -0.5) to[out=-90, in=90] (1.9, -0.8) to[out=-90, in=170] (2.35, -1.13);
\draw (2.55, -1.17) to[out=-10, in=90] (3, -1.5) arc (-180:0:0.1) to[out=90, in=-90, looseness=0.7] (4.3, -0.8) to[out=90, in=-140] (4.35, -0.7);

\draw (3.85, -1.17) to[out=-10, in=90] (4.3, -1.5) arc (-180:0:0.1)--(4.5, -0.8) to[out=90, in=-90] (4.3, -0.5) to[out=90, in=-10] (3.85, -0.17);
\draw (3.65, -0.13) to[out=170, in=-90] (3.2, 0.2) arc (0:180:0.1) to[out=-90, in=90] (1.9, -0.5) to[out=-90, in=40] (1.85, -0.6);

\draw (1.75, -0.7) to[out=-140, in=90] (1.7, -0.8) to[out=-90, in=90, looseness=0.7] (0.6, -1.5) arc (0:-180:0.1)--(0.4, -0.8) to[out=90, in=-40] (0.35, -0.7);
\draw (1.05, -1.13) to[out=170, in=-90] (0.6, -0.8)--(0.6, -0.5) to[out=90, in=-90, looseness=0.7] (1.7, 0.2) arc (180:0:0.1) (1.9, 0.2) to[out=-90, in=170] (2.35, -0.13);

\draw (2.55, -0.17) to[out=-10, in=90] (3, -0.5) to[out=-90, in=90] (2.8, -0.8) to[out=-90, in=-90, looseness=1.2] (2.1, -0.8)--(2.1, -0.5) to[out=90, in=90, looseness=1.2] (2.8, -0.5) to[out=-90, in=140] (2.85, -0.6);
\draw (2.95, -0.7) to[out=-40, in=90] (3, -0.8) to[out=-90, in=90] (1.9, -1.5) arc (0:-180:0.1) to[out=90, in=-10] (1.25, -1.17);

\draw (3.65, -1.13) to[out=170, in=-90] (3.2, -0.8)--(3.2, -0.5) to[out=90, in=-90] (4.3, 0.2) arc (180:0:0.1)--(4.5, -0.5) to[out=-90, in=40] (4.45, -0.6);

\node at (0.6, 0.2) [circle, fill, inner sep=1pt] {};
\node at (0.8, 0.2) [circle, fill, inner sep=1pt] {};
\end{tikzpicture}
} \,\, , 
\end{equation}
where we simplified the diagram using the string-genus relation Eq.~\eqref{eq:string-genus} for the Majorana strings colored in red and their enclosing holes. From the final diagram, we observe key diagrammatic characterization for Clifford: there exists a Quon diagrammatic representation where the Majorana diagram contains no scattering elements other than braids. The background manifold may, as in this case, contain holes. 

The second example is the following matchgate computation in which every gate is matchgate: 
\begin{align}
\langle 0, 0, 1, 1 \vert & (e^{i \phi_4 X_1 X_2} e^{i \phi_5 X_2 X_3} e^{i \phi_6 X_3 X_4}) \nonumber \\
& \times (e^{i \theta_1 Z_1} e^{i \theta_2 Z_2} e^{i \theta_3 Z_3} e^{i \theta_4 Z_4}) \nonumber \\
& \times (e^{i \phi_1 X_1 X_2} e^{i \phi_2 X_2 X_3} e^{i \phi_3 X_3 X_4}) \vert 0, 0, 0, 0 \rangle , 
\end{align}
which can be represented by the following 2D Quon diagram: 
\begin{equation}
\label{eq:matchgate-circuit-example}
\raisebox{-1.15cm}{
\begin{tikzpicture}[scale=0.7]
\begin{scope}[even odd rule]
\clip (0, -2.4) arc (-180:0:0.5) arc (180:0:0.15) arc (-180:0:0.5) arc (180:0:0.15) arc (-180:0:0.5) arc (180:0:0.15) arc (-180:0:0.5)--(4.9, 0.2) arc (0:180:0.5) arc (0:-180:0.15) arc (0:180:0.5) arc (0:-180:0.15) arc (0:180:0.5) arc (0:-180:0.15) arc (0:180:0.5)  (1.3, -0.95) arc (0:180:0.15)--(1, -1.25) arc (-180:0:0.15)--(1.3, -0.95);
\clip (0, -2.4) arc (-180:0:0.5) arc (180:0:0.15) arc (-180:0:0.5) arc (180:0:0.15) arc (-180:0:0.5) arc (180:0:0.15) arc (-180:0:0.5)--(4.9, 0.2) arc (0:180:0.5) arc (0:-180:0.15) arc (0:180:0.5) arc (0:-180:0.15) arc (0:180:0.5) arc (0:-180:0.15) arc (0:180:0.5)  (2.6, -0.95) arc (0:180:0.15)--(2.3, -1.25) arc (-180:0:0.15)--(2.6, -0.95);
\clip (0, -2.4) arc (-180:0:0.5) arc (180:0:0.15) arc (-180:0:0.5) arc (180:0:0.15) arc (-180:0:0.5) arc (180:0:0.15) arc (-180:0:0.5)--(4.9, 0.2) arc (0:180:0.5) arc (0:-180:0.15) arc (0:180:0.5) arc (0:-180:0.15) arc (0:180:0.5) arc (0:-180:0.15) arc (0:180:0.5)  (3.9, -0.95) arc (0:180:0.15)--(3.6, -1.25) arc (-180:0:0.15)--(3.9, -0.95);

\fill [blue!10] (0, -2.4) arc (-180:0:0.5) arc (180:0:0.15) arc (-180:0:0.5) arc (180:0:0.15) arc (-180:0:0.5) arc (180:0:0.15) arc (-180:0:0.5)--(4.9, 0.2) arc (0:180:0.5) arc (0:-180:0.15) arc (0:180:0.5) arc (0:-180:0.15) arc (0:180:0.5) arc (0:-180:0.15) arc (0:180:0.5);
\end{scope}

\draw [thick, blue] (0, 0.2)--(0, -2.4) arc (-180:0:0.5) arc (180:0:0.15) arc (-180:0:0.5) arc (180:0:0.15) arc (-180:0:0.5) arc (180:0:0.15) arc (-180:0:0.5)--(4.9, 0.2) arc (0:180:0.5) arc (0:-180:0.15) arc (0:180:0.5) arc (0:-180:0.15) arc (0:180:0.5) arc (0:-180:0.15) arc (0:180:0.5);

\draw[thick, blue] (1.3, -0.95) arc (0:180:0.15)--(1, -1.25) arc (-180:0:0.15)--(1.3, -0.95);
\draw[thick, blue] (2.6, -0.95) arc (0:180:0.15)--(2.3, -1.25) arc (-180:0:0.15)--(2.6, -0.95);
\draw[thick, blue] (3.9, -0.95) arc (0:180:0.15)--(3.6, -1.25) arc (-180:0:0.15)--(3.9, -0.95);

\draw[red] (0.2, 0.2) arc (180:0:0.3) (0.8, 0.2) to[out=-90, in=-90, looseness=1.2] (1.5, 0.2) arc (180:0:0.3) to[out=-90, in=-90, looseness=1.2] (2.8, 0.2) arc (180:0:0.3) to[out=-90, in=-90, looseness=1.2] (4.1, 0.2) arc (180:0:0.3)--(4.7, -2.4) arc (0:-180:0.3) to[out=90, in=90, looseness=1.2] (3.4, -2.4) arc (0:-180:0.3) to[out=90, in=90, looseness=1.2] (2.1, -2.4) arc (0:-180:0.3) to[out=90, in=90, looseness=1.2] (0.8, -2.4) arc (0:-180:0.3)--(0.2, 0.2);

\draw (0.4, -0.87)--(0.4, 0.2) arc (180:0:0.1) to[out=-90, in=135] (+{1.15-0.25*cos(45)}, {-0.35+0.25*sin(45)});

\draw (+{1.15+0.25*cos(45)}, {-0.35+0.25*sin(45)}) to[out=45, in=-90] (1.7, 0.2) arc (180:0:0.1) to[out=-90, in=135] (+{1.15+1.3-0.25*cos(45)}, {-0.35+0.25*sin(45)});
\draw (+{1.15+1.3+0.25*cos(45)}, {-0.35+0.25*sin(45)}) to[out=45, in=-90] (3, 0.2) arc (180:0:0.1) to[out=-90, in=135] (+{1.15+2.6-0.25*cos(45)}, {-0.35+0.25*sin(45)});

\draw (+{1.15+2.6+0.25*cos(45)}, {-0.35+0.25*sin(45)}) to[out=45, in=-90] (4.3, 0.2) arc (180:0:0.1)--(4.5, -0.87);

\draw (0.6, -0.87) to[out=60, in=-135] (+{1.15-0.25*cos(45)}, {-0.35-0.25*sin(45)});
\draw (1.7, -0.87) to[out=120, in=-45] (+{1.15+0.25*cos(45)}, {-0.35-0.25*sin(45)});
\draw (0.6+1.3, -0.87) to[out=60, in=-135] (+{1.15+1.3-0.25*cos(45)}, {-0.35-0.25*sin(45)});
\draw (1.7+1.3, -0.87) to[out=120, in=-45] (+{1.15+1.3+0.25*cos(45)}, {-0.35-0.25*sin(45)});
\draw (0.6+2.6, -0.87) to[out=60, in=-135] (+{1.15+2.6-0.25*cos(45)}, {-0.35-0.25*sin(45)});
\draw (1.7+2.6, -0.87) to[out=120, in=-45] (+{1.15+2.6+0.25*cos(45)}, {-0.35-0.25*sin(45)});

\draw (0.6, -1.33) to[out=-60, in=135] (+{1.15-0.25*cos(45)}, {-1.5-0.35+0.25*sin(45)});
\draw (1.7, -1.33) to[out=-120, in=45] (+{1.15+0.25*cos(45)}, {-1.5-0.35+0.25*sin(45)});
\draw (0.6+1.3, -1.33) to[out=-60, in=135] (+{1.3+1.15-0.25*cos(45)}, {-1.5-0.35+0.25*sin(45)});
\draw (1.7+1.3, -1.33) to[out=-120, in=45] (+{1.3+1.15+0.25*cos(45)}, {-1.5-0.35+0.25*sin(45)});
\draw (0.6+2.6, -1.33) to[out=-60, in=135] (+{2.6+1.15-0.25*cos(45)}, {-1.5-0.35+0.25*sin(45)});
\draw (1.7+2.6, -1.33) to[out=-120, in=45] (+{2.6+1.15+0.25*cos(45)}, {-1.5-0.35+0.25*sin(45)});

\draw (0.4, -1.33)--(0.4, -2.4) arc (-180:0:0.1) to[out=90, in=-135] (+{1.15-0.25*cos(45)}, {-1.5-0.35-0.25*sin(45)});
\draw (+{1.15+0.25*cos(45)}, {-1.5-0.35-0.25*sin(45)}) to[out=-45, in=90] (1.7, -2.4) arc (-180:0:0.1) to[out=90, in=-135] (+{1.3+1.15-0.25*cos(45)}, {-1.5-0.35-0.25*sin(45)});
\draw (+{1.3+1.15+0.25*cos(45)}, {-1.5-0.35-0.25*sin(45)}) to[out=-45, in=90] (3, -2.4) arc (-180:0:0.1) to[out=90, in=-135] (+{2.6+1.15-0.25*cos(45)}, {-1.5-0.35-0.25*sin(45)});
\draw (+{2.6+1.15+0.25*cos(45)}, {-1.5-0.35-0.25*sin(45)}) to[out=-45, in=90] (4.3, -2.4) arc (-180:0:0.1)--(4.5, -1.33);

\draw [red, looseness=1.7] (0.9, -0.95) to[out=90, in=90] (1.4, -0.95)--(1.4, -1.25) to[out=-90, in=-90] (0.9, -1.25)--cycle;
\draw [red, looseness=1.7] (2.2, -0.95) to[out=90, in=90] (2.7, -0.95)--(2.7, -1.25) to[out=-90, in=-90] (2.2, -1.25)--cycle;
\draw [red, looseness=1.7] (3.5, -0.95) to[out=90, in=90] (4, -0.95)--(4, -1.25) to[out=-90, in=-90] (3.5, -1.25)--cycle;

\draw (1.4, -0.35) arc (0:360:0.25);
\node at (1.15, -0.35) {\scriptsize $\phi_1$};

\draw (1.4+1.3, -0.35) arc (0:360:0.25);
\node at (1.15+1.3, -0.35) {\scriptsize $\phi_2$};

\draw (1.4+2.6, -0.35) arc (0:360:0.25);
\node at (1.15+2.6, -0.35) {\scriptsize $\phi_3$};

\draw (0.75, -1.1) arc (0:360:0.25);
\node at (0.5, -1.1) {\scriptsize $\theta_1$};

\draw (0.75+1.3, -1.1) arc (0:360:0.25);
\node at (0.5+1.3, -1.1) {\scriptsize $\theta_2$};

\draw (0.75+2.6, -1.1) arc (0:360:0.25);
\node at (0.5+2.6, -1.1) {\scriptsize $\theta_3$};

\draw (0.75+3.9, -1.1) arc (0:360:0.25);
\node at (0.5+3.9, -1.1) {\scriptsize $\theta_4$};

\draw (1.4, -0.3-1.55) arc (0:360:0.25);
\node at (1.15, -0.3-1.55) {\scriptsize $\phi_4$};

\draw (1.4+1.3, -0.3-1.55) arc (0:360:0.25);
\node at (1.15+1.3, -0.3-1.55) {\scriptsize $\phi_5$};

\draw (1.4+2.6, -0.3-1.55) arc (0:360:0.25);
\node at (1.15+2.6, -0.3-1.55) {\scriptsize $\phi_6$};

\node at (0.6+2.6, -2.4) [circle, fill, inner sep=1pt] {};
\node at (0.8+2.6, -2.4) [circle, fill, inner sep=1pt] {};

\node at (0.6+3.9, -2.4) [circle, fill, inner sep=1pt] {};
\node at (0.8+3.9, -2.4) [circle, fill, inner sep=1pt] {};
\end{tikzpicture}
} = \frac{1}{2 \sqrt{2}} 
\raisebox{-1.3cm}{
\begin{tikzpicture}[scale=0.7]
\draw[thick, blue, fill=blue!10] (0, -2) to[out=-90, in=-90, looseness=0.6] (4.9, -2)--(4.9, -0.2) to[out=90, in=90, looseness=0.6] (0, -0.2)--cycle;

\draw[red] (0.15, -1.8) to[out=-90, in=-90, looseness=0.6] (4.75, -1.8)--(4.75, -0.4) to[out=90, in=90, looseness=0.6] (0.15, -0.4)--cycle;

\draw (0.4, -0.87) to[out=120, in=135, looseness=1.7] (+{1.15-0.25*cos(45)}, {-0.35+0.25*sin(45)});
\draw (+{1.15+0.25*cos(45)}, {-0.35+0.25*sin(45)}) to[out=45, in=135] (+{1.15+1.3-0.25*cos(45)}, {-0.35+0.25*sin(45)});
\draw (+{1.15+1.3+0.25*cos(45)}, {-0.35+0.25*sin(45)}) to[out=45, in=135] (+{1.15+2.6-0.25*cos(45)}, {-0.35+0.25*sin(45)});
\draw (+{1.15+2.6+0.25*cos(45)}, {-0.35+0.25*sin(45)}) to[out=45, in=60, looseness=1.7] (4.5, -0.87);

\draw (0.6, -0.87) to[out=60, in=-135] (+{1.15-0.25*cos(45)}, {-0.35-0.25*sin(45)});
\draw (1.7, -0.87) to[out=120, in=-45] (+{1.15+0.25*cos(45)}, {-0.35-0.25*sin(45)});
\draw (0.6+1.3, -0.87) to[out=60, in=-135] (+{1.15+1.3-0.25*cos(45)}, {-0.35-0.25*sin(45)});
\draw (1.7+1.3, -0.87) to[out=120, in=-45] (+{1.15+1.3+0.25*cos(45)}, {-0.35-0.25*sin(45)});
\draw (0.6+2.6, -0.87) to[out=60, in=-135] (+{1.15+2.6-0.25*cos(45)}, {-0.35-0.25*sin(45)});
\draw (1.7+2.6, -0.87) to[out=120, in=-45] (+{1.15+2.6+0.25*cos(45)}, {-0.35-0.25*sin(45)});

\draw (0.6, -1.33) to[out=-60, in=135] (+{1.15-0.25*cos(45)}, {-1.5-0.35+0.25*sin(45)});
\draw (1.7, -1.33) to[out=-120, in=45] (+{1.15+0.25*cos(45)}, {-1.5-0.35+0.25*sin(45)});
\draw (0.6+1.3, -1.33) to[out=-60, in=135] (+{1.3+1.15-0.25*cos(45)}, {-1.5-0.35+0.25*sin(45)});
\draw (1.7+1.3, -1.33) to[out=-120, in=45] (+{1.3+1.15+0.25*cos(45)}, {-1.5-0.35+0.25*sin(45)});
\draw (0.6+2.6, -1.33) to[out=-60, in=135] (+{2.6+1.15-0.25*cos(45)}, {-1.5-0.35+0.25*sin(45)});
\draw (1.7+2.6, -1.33) to[out=-120, in=45] (+{2.6+1.15+0.25*cos(45)}, {-1.5-0.35+0.25*sin(45)});

\draw (0.4, -1.33) to[out=-120, in=-135, looseness=1.7] (+{1.15-0.25*cos(45)}, {-1.5-0.35-0.25*sin(45)});
\draw (+{1.15+0.25*cos(45)}, {-1.5-0.35-0.25*sin(45)}) to[out=-45, in=-135] (+{1.3+1.15-0.25*cos(45)}, {-1.5-0.35-0.25*sin(45)});
\draw (+{1.3+1.15+0.25*cos(45)}, {-1.5-0.35-0.25*sin(45)}) to[out=-45, in=-135] (+{2.6+1.15-0.25*cos(45)}, {-1.5-0.35-0.25*sin(45)});
\draw (+{2.6+1.15+0.25*cos(45)}, {-1.5-0.35-0.25*sin(45)}) to[out=-45, in=-60, looseness=1.7] (4.5, -1.33);

\draw (1.4, -0.35) arc (0:360:0.25);
\node at (1.15, -0.35) {\scriptsize $\phi_1$};

\draw (1.4+1.3, -0.35) arc (0:360:0.25);
\node at (1.15+1.3, -0.35) {\scriptsize $\phi_2$};

\draw (1.4+2.6, -0.35) arc (0:360:0.25);
\node at (1.15+2.6, -0.35) {\scriptsize $\phi_3$};

\draw (0.75, -1.1) arc (0:360:0.25);
\node at (0.5, -1.1) {\scriptsize $\theta_1$};

\draw (0.75+1.3, -1.1) arc (0:360:0.25);
\node at (0.5+1.3, -1.1) {\scriptsize $\theta_2$};

\draw (0.75+2.6, -1.1) arc (0:360:0.25);
\node at (0.5+2.6, -1.1) {\scriptsize $\theta_3$};

\draw (0.75+3.9, -1.1) arc (0:360:0.25);
\node at (0.5+3.9, -1.1) {\scriptsize $\theta_4$};

\draw (1.4, -0.3-1.55) arc (0:360:0.25);
\node at (1.15, -0.3-1.55) {\scriptsize $\phi_4$};

\draw (1.4+1.3, -0.3-1.55) arc (0:360:0.25);
\node at (1.15+1.3, -0.3-1.55) {\scriptsize $\phi_5$};

\draw (1.4+2.6, -0.3-1.55) arc (0:360:0.25);
\node at (1.15+2.6, -0.3-1.55) {\scriptsize $\phi_6$};

\node at (3.45, -2.12) [circle, fill, inner sep=1pt] {};
\node at (4.05, -2.12) [circle, fill, inner sep=1pt] {};
\end{tikzpicture}
} , 
\end{equation}
where we removed the holes using the string-genus relation Eq.~\eqref{eq:string-genus}. Note that one could, in principle, simplify the outermost red Majorana line by first shrinking it into a small circle using Eq.~\eqref{eq:YBE-two-braids} and then removing it from the diagram by applying the Majorana loop amplitude in TABLE~\ref{tab:majorana-rewriting-rules}. However, we have not performed this simplification to elaborate key diagrammatic characterizations of matchgate: along every boundary component of the background manifold, there exists a ``boundary-tracking'' Majorana lines (the highlighted outermost Majorana line in this case), and the background manifold contains no hole. See Sec.~\ref{sec:matchgate-TN-using-quon} for a complete discussion on the importance of boundary-tracking Majorana lines.

\subsection{Arbitrary Tensor Networks}
In Sec.~\ref{sec:quon-for-q-states-gates}, we demonstrated how Quon diagrams can represent arbitrary quantum states and unitary gates. In this section, we extend those results by establishing the \textit{universality} of the 2D Quon language for \textit{planar tensor networks}. Therefore, we lift the unitarity constraints that were imposed in the previous subsections. As expected, the planar nature of 2D Quon diagrams is naturally compatible with planar tensor networks. 

We begin by reviewing planar tensor networks, with a focus on the concept called the \textit{planar region}, and provide detailed explanations of tensor operations. Next, we present how matchgate tensor networks can be naturally understood as planar tensor networks, and introduce a new tensor network class called the \textit{punctured matchgate tensor networks} as a generalization of matchgate tensor networks. By using the concept of the punctured matchgate networks, we prove the \textit{decomposition theorem}, which states that any tensor network can be rewritten as the contraction of a single Clifford tensor with a single matchgate tensor. Furthermore, inspired by this decomposition theorem, we introduce a new tensor network ansatz called \textit{hybrid Clifford-matchgate-MPS}. Subsequently, we show how planar tensor operations can be directly translated into operations on 2D Quon diagrams. Finally, we establish the universality of the 2D Quon language by providing diagrammatic representations for an elementary generating set of tensors in planar tensor networks.

\subsubsection{Planar tensor networks and tensor operations}
\label{sec:planar-TN}
A tensor network can be viewed as a graph where nodes correspond to tensors and edges correspond to tensor legs. A \textit{planar tensor network} is a tensor network whose corresponding graph is a planar graph. As demonstrated in FIG.~\ref{fig:planar-TN-via-SWAP}, a generic tensor network can always be expressed as a planar tensor network at the expense of introducing additional SWAP gates. Thus, it is possible to consider only planar tensor networks without compromising the full capability of tensor networks. Furthermore, given a planar network, we consider not only the abstract connectivity between the tensors, but also their specific layout on the plane (or on the surface of a two-dimensional sphere). Throughout the rest of the paper, we assume that every tensor leg index can take only the values $0$ or $1$, i.e., each tensor leg corresponds to a qubit. 

\begin{figure*}[t]
\centering
\includegraphics[width=\textwidth]{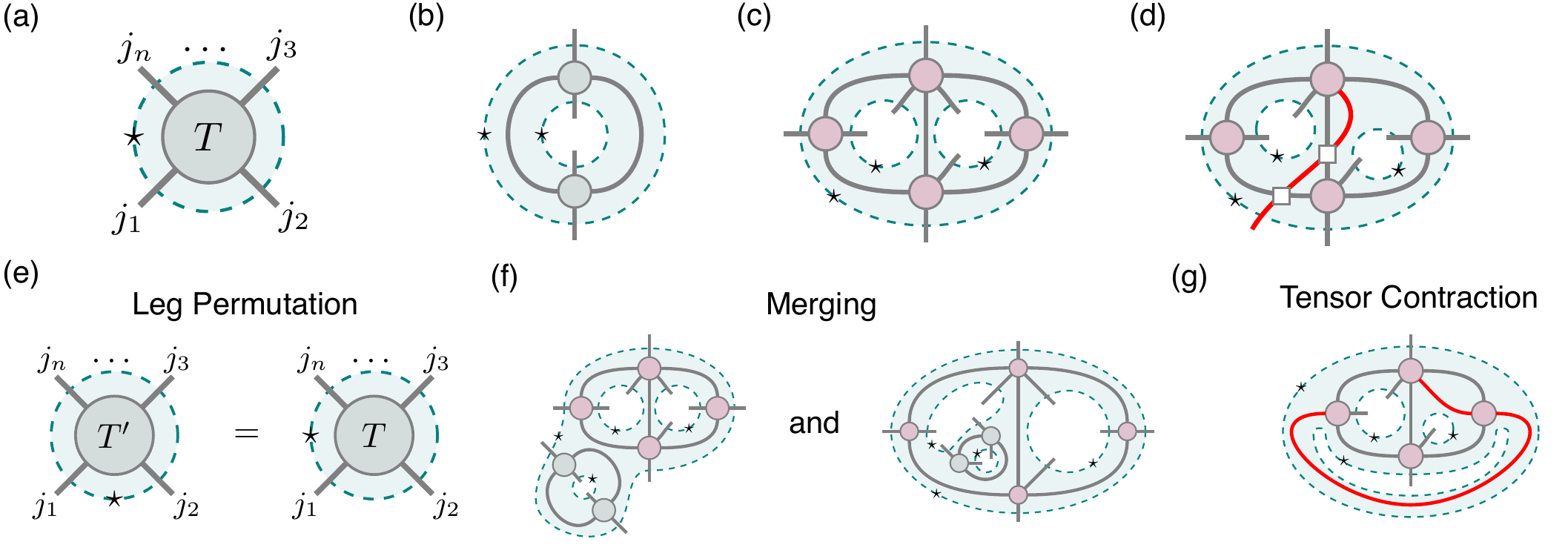}
\caption{(a) A single $n$-leg tensor $T$ in a planar tensor network, where $j_1, \ldots, j_n = 0, 1$ are tensor indices, and the planar region is indicated as a shaded region with its boundary indicated by a dashed line. Here, we use thick gray lines to represent tensor legs, in order to distinguish them from Majorana lines, which are drawn with thin black or red lines. Tensor legs are ordered counterclockwise, starting from $\star$ in the boundary of the planar region. (b) Example of a planar tensor network with one puncture in the planar region. (c) Example of a planar tensor network with two punctures in the planar region. We place $\star$ to each punctures, in order to provide the counterclockwise ordering among the tensor legs lying in the puncture. (d) Starting from a planar tensor network in (c), one can pull out a leg, highlighted in red, at the expense of using the SWAP gates, denoted as white boxes. \textit{Various planar tensor network operations}: (e) (\textit{Leg Permutation}) In a planar tensor network, one can cyclically permute the tensor legs lying in the same boundary component of the planar region, as indicated by shifting the position of $\star$. For example, the $n$-leg tensor $T'$, defined in terms of its components as $T'_{j_2, \ldots, j_n, j_1} = T_{j_1, j_2, \ldots, j_n}$, is obtained from $T$ in (a) by cyclically permuting its legs. (f) (\textit{Merging}) Merging of a two planar tensor networks, enabled by merging their planar regions. Here, we merge the tensor networks in (b) and (c). We merge two planar regions either by joining segments of their disk boundaries (as illustrated on the left), or by joining a segment of a puncture boundary in one region to a segment of the disk boundary of another region (as illustrated on the right). As a technical remark, one may have to adjust the locations of $\star$s on the two boundaries to bring them together before merging, at the expense of using the leg cyclic permutations presented in (e). (g) (\textit{Tensor Contractions}) Examples of neighboring and non-neighboring self-contractions of tensors, with the contracted legs highlighted in red.}
\label{fig:planar-tensor}
\end{figure*}

When visually displaying a planar tensor network, we always indicate a \textit{planar region}---shown as shaded areas in Fig.~\ref{fig:planar-tensor}---defined as follows: a planar region is a connected region\footnote{To avoid clutter, we assume that the planar region always has a single connected component} on a plane (or on the surface of a sphere), onto which the network is embedded, where the tensor nodes and all closed (internal) legs, those that are contracted over, are located within its interior, while all open (external) legs extend transversally across its boundary. While this may initially seem unnecessary, we will demonstrate later that indicating the planar region is, in fact, conceptually useful for understanding specific types of tensor networks, notably matchgate tensor networks. We note that the planar regions and the planar tensor network operations introduced here can be essentially understood in terms of the \textit{planar algebra} from the mathmatics literature~\cite{jones2021planar}. The planar region can have one or multiple boundary components, as shown in FIG.~\ref{fig:planar-tensor} (a)--(c). In the former case, the planar region is topologically equivalent to a disk, and in the latter, it is equivalent to a punctured disk\footnote{We will use the term ``puncture'' when referring to a hole in the planar region, to distinguish it from a hole in the background manifold of a 2D Quon diagram.}. As a remark, one can pull out a leg in one boundary component to another at the expense of using the SWAP gates, as demonstrated in FIG.~\ref{fig:planar-tensor} (d). Therefore, the planarity and the presence of punctures in the planar region become particularly relevant when the SWAP gate is not freely available, as in matchgate tensor networks~\cite{valiant2001quantum, cai2009theory, bravyi2009contraction}. We note that the planar region is often related to the background manifold of Quon, but serves a different role. At each boundary component of the planar region, we place a $\star$ to indicate relative counterclockwise ordering of the open legs associated with that boundary. 

We now explain three planar tensor network operations, namely (1) \textit{leg permutation}, (2) \textit{merging}, and (3) \textit{tensor contraction}, which are summarized in FIG.~\ref{fig:planar-tensor} (e)--(g). 

In a \textit{leg permutation}, owing to the planar embedding, we only consider a cyclic permutation of legs belonging to the same boundary component of the planar region, which corresponds to shifting the position of $\star$. 

A \textit{merging} amount to combining two planar tensor networks into a single planar tensor network. As part of this process, the corresponding planar regions are also merged\footnote{Note that the planar region is assumed to have a single connected component.}. As an operation in the planar tensor networks, we allow merging only when the resulting tensor network and the planar region remain planar, as demonstrated in FIG.~\ref{fig:planar-tensor} (f). After a merging operation, the number of boundary components, as well as the number of $\star$s, decreases by one. More general merging operation would result in a \textit{surface tensor network}, where the tensor network is embedded on a surface, potentially with genus. We defer the exploration of surface tensor networks to future work. 

The third tensor operation, \textit{tensor contraction}, is defined as a \textit{self-contraction} between legs that lie on the same boundary component of the planar region. Therefore, if necessary, we need to merge two planar regions before performing the tensor contraction to enable a self-contraction. This definition also implies that the planar tensor contraction between arbitrary pairs of legs is not generally possible, as such an operation could result in a surface tensor network. There are two types of self-contraction: one between \textit{neighboring legs} and the other between \textit{non-neighboring legs}, both of which are demonstrated in FIG.~\ref{fig:planar-tensor} (g). Note that the non-neighboring contraction changes the topology of the planar region by increasing the number of punctures by one. 

\begin{figure*}[t]
\centering
\includegraphics[width=\textwidth]{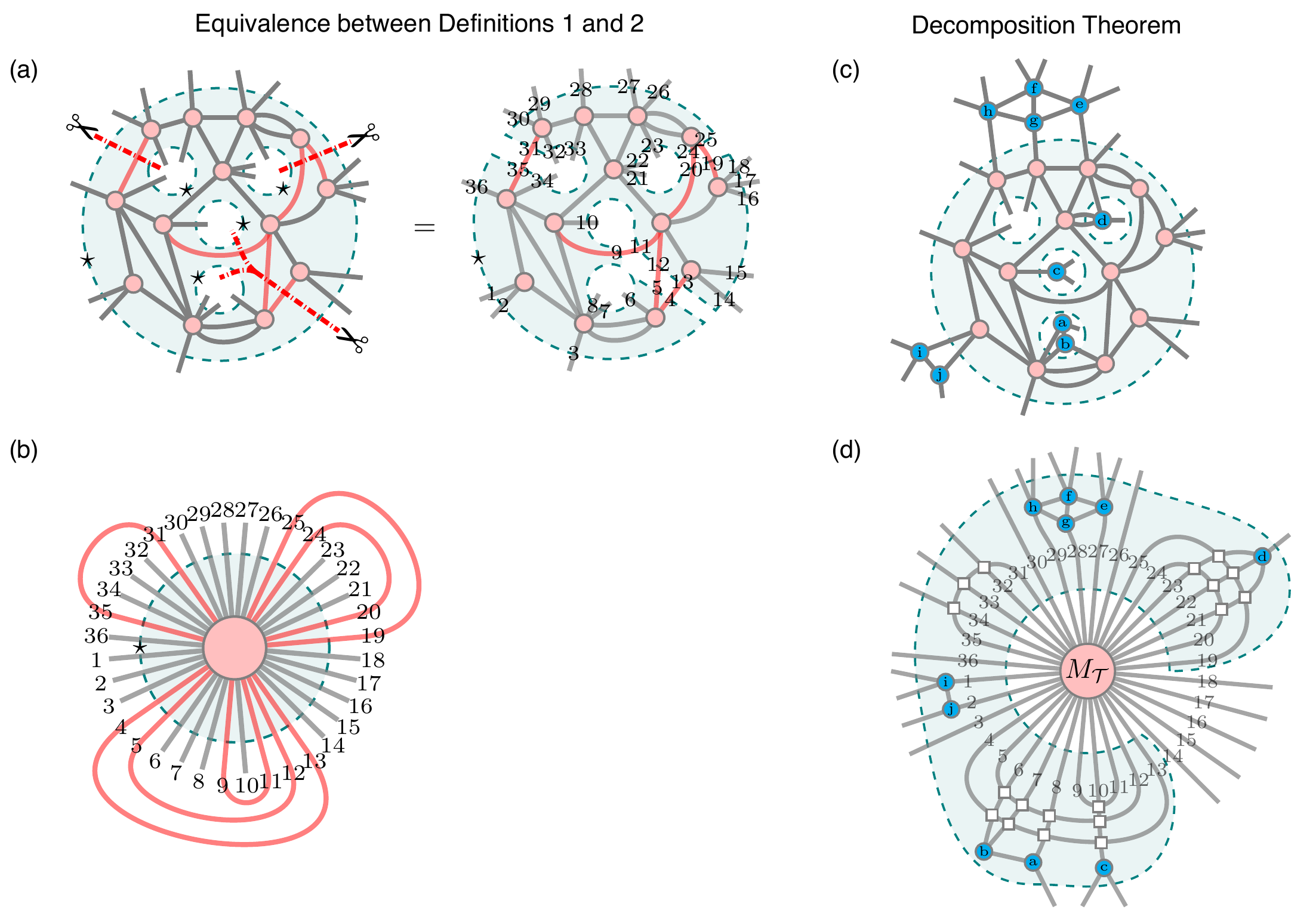}
\caption{(a)--(b) Tensor network diagrams used to prove the equivalence between Definitions~\ref{def:punctured-matchgate-tensor-network} and~\ref{def:punctured-matchgate-tensor}. (c)--(d) Tensor network diagrams used to prove Theorem~\ref{thm:Clifford-matchgate-decomposition}, the decomposition theorem. See the main text for detailed explanations of each diagram.}
\label{fig:punctured-decomposition}
\end{figure*}

Using the concept of the planar region, we now extend matchgate tensor networks by introducing a new class called \textit{punctured matchgate tensor networks}, which contain matchgates as a strict subset. Recall that a tensor network is called \textit{matchgate} if (1) it is a planar tensor network, (2) each constituent tensor is a matchgate tensor (see Appendix~\ref{app:matchgate} for the definition of a matchgate tensor), and (3) the planar region of the network is topologically equivalent to a (non-punctured) disk. In a matchgate tensor network, the resulting contracted tensor (involving only neighboring contractions) remains a matchgate tensor. We refer to the conventional matchgate tensor networks as the \textit{(non-punctured) matchgate tensor networks} and define a new, broader class of tensor networks, the \textit{punctured matchgate tensor networks}, by relaxing condition (3). This leads to the following definition of a \textit{punctured matchgate} tensor: 
\begin{definition}
\label{def:punctured-matchgate-tensor-network}
A tensor is called a \textbf{punctured matchgate tensor} if it can be obtained from (planar) contracting a punctured matchgate tensor network. 
\end{definition}
Alternatively, there exists an equivalent definition of a punctured matchgate tensor: 
\begin{definition}
\label{def:punctured-matchgate-tensor}
A tensor is called a \textbf{punctured matchgate tensor} if it can be obtained from (possibly non-neighboring but planar) self-contractions of a single matchgate tensor. 
\end{definition}
The equivalence of Definition~\ref{def:punctured-matchgate-tensor-network} and Definition~\ref{def:punctured-matchgate-tensor} is, in fact, far from obvious. Naively, one might merge all matchgate tensors in a punctured matchgate tensor network into a single tensor, and then perform the tensor contractions in the network. However, this merged tensor generally fails to be a matchgate tensor because the ordering of tensor legs---crucial for the matchgate structure---can become permuted. Therefore, we present a proof of the equivalence between two definitions below. 

A punctured matchgate tensor from Definition~\ref{def:punctured-matchgate-tensor} trivially aligns with that from Definition~\ref{def:punctured-matchgate-tensor-network}. To show the converse, consider a punctured tensor network. If we contract all the closed legs, the resulting tensor becomes a punctured matchgate tensor according to Definition~\ref{def:punctured-matchgate-tensor-network}. Now, in the punctured tensor network, consider cuts from each puncture in the planar region to the boundary of the disk, in a way that they intersect only closed tensor legs, as depicted in FIG.~\ref{fig:punctured-decomposition} (a), where the closed tensor legs that intersect the cuts are highlighted in red. In the RHS, we cut the planar region along these cuts, producing a puncture-free planar region; accordingly, we cyclically label tensor legs that cross the boundary of the planar region transversally. Furthermore, contracting the closed tensor legs lying within this puncture-free planar region yields in a matchgate tensor. Then, contracting the remaining tensor legs (those highlighted in red) produces FIG.~\ref{fig:punctured-decomposition} (b), which is a tensor obtained from non-neighboring planar contractions of a matchgate tensor and thus qualifies as a punctured matchgate per Definition~\ref{def:punctured-matchgate-tensor}. 

It turns out that the newly introduced punctured matchgates have a simple diagrammatic characterization in terms of the Quon language; see Sec.~\ref{sec:matchgate-TN-using-quon} for further details. 

We remark that given a punctured matchgate tensor, if we ``fill in'' the punctures with matchgates of appropriate size, the resulting tensor becomes a matchgate. This observation motivates the design a matchgate shadow-type protocol~\cite{wan2023matchgate} as a method of handling punctured matchgates. We emphasize that punctured matchgates include matchgates as a strict subset. Similar to matchgate tensors, computing the components of a punctured matchgate tensor remains efficient (since computational basis states are matchgates). However, when viewing a punctured matchgate tensor as a quantum state, evaluating the expectation value of an arbitrary (local) observable is generally intractable, with brute-force computation scaling \textit{exponentially} with the number of punctures---in contrast to matchgates\footnote{Any local observable can be expressed as a linear combination of Pauli string operators, i.e., as a sum of a small number of matchgates. Each expectation value on a matchgate state can be expressed as a closed matchgate tensor network: represent the ket state as a matchgate tensor network on the northern hemisphere and the bra state on the southern hemisphere. By gluing two hemispheres along their equators and inserting dots to represent the Pauli operators at the junction, one obtains a closed matchgate tensor network on the sphere. Therefore, the expectation value of a local observable on a matchgate state can be evaluated efficiently. In contrast, even computing the norm of a punctured matchgate state is generally \textit{intractable}, since, using the same construction as above, the norm is expressed as a \textit{surface} matchgate tensor network whose number of genus equals the number of punctures.}. 

We note that computing a component of a punctured matchgate can be expressed as a Pfaffian of a matrix, whose entries now encode multi-particle correlations, unlike the single-particle structure in matchgates. This feature bears some similarities with recent neural-network-based variational wavefunctions~\cite{PhysRevLett.122.226401, acevedo2020vandermonde, von2022self, chen2023exact, geier2025attention}, which generalize the Hartree-Fock wavefunction (or matchgate). In those neural network wavefunctions, each component is given by a single determinant of a matrix whose entries depends on multi-particle coordinates, rather than on a single-particle orbitals as in traditional Hartree-Fock. However, punctured matchgates offer a richer geometric picture inherited from matchgates, enabling more structured correlations than those in neural-network variational states. We leave the investigation of these encoded correlations for future work.

Using the concepts from punctured matchgates, we can prove the following decomposition theorem: 
\begin{theorem}
\label{thm:Clifford-matchgate-decomposition}
Any tensor network can be expressed as the tensor contractions between a single Clifford tensor and a single (non-punctured) matchgate tensor.
\end{theorem}
\begin{proof}
Suppose we are given a tensor network $\mathcal{T}$ where each tensor is either Clifford or matchgate. If a tensor in the network is neither Clifford nor matchgate, we replace it with an equivalent ``mini'' tensor network constructed solely from Clifford and matchgate tensors. Furthermore, we assume that the tensor network is planar by incorporating the SWAP gates following the procedure outlined in FIG.~\ref{fig:planar-TN-via-SWAP}, which are Cliffords. Under these assumptions, an example of $\mathcal{T}$ is given by FIG.~\ref{fig:punctured-decomposition} (c), where we label Clifford tensors by letters, facilitating the tracking of each tensor in later stages of the proof, and use a ``planar region'' to separate the matchgate tensors from Clifford tensors. One can immediately notice that this planar region and the enclosed matchgate tensors form a punctured matchgate tensor network, denoted by $\mathcal{T}_M$. 

Using the procedures outlined in FIG.~\ref{fig:punctured-decomposition} (a) and (b), we first contract all the matchgate tensors in $\mathcal{T}_M$ to form a single matchgate tensor, denoted by $M_\mathcal{T}$. We then ``pull out'' tensor legs from each puncture, and then reinsert the Clifford tensors from $\mathcal{T}$ into $\mathcal{T}_M$, leading to FIG.~\ref{fig:punctured-decomposition} (d), where we use white boxes to denote the SWAP gates,  which are generated by ``pulling out'' tensor legs from the punctures in $\mathcal{T}_M$. We also indicate another ``planar region'' to enclose all the \textit{Clifford} tensors in the network. Finally, we contract the Clifford tensors lying within this planar region to yield a \textit{Clifford} tensor $C_\mathcal{T}$. This produces a decomposition of tensor network $\mathcal{T}$ in terms of a single Clifford tensor $C_\mathcal{T}$ a single matchgate tensor $M_\mathcal{T}$: 
\begin{equation}
\mathcal{T} = \raisebox{-0.9cm}{
\begin{tikzpicture}
\draw[ultra thick, gray] (-1, 0)--(1, 0);
\draw[ultra thick, gray] (-1, 0)--(-2, -1.2);
\draw[ultra thick, gray] (1, 0)--(2, -1.2);

\draw[thick, gray, fill=pink] (1.5, 0) arc (0:360:0.5);
\draw[thick, gray, fill=cyan] (-1.5, 0) arc (180:-180:0.5);

\node at (1, 0) {$M_\mathcal{T}$};
\node at (-1, 0) {$C_\mathcal{T}$};

\node at (0, 0.2) {\tiny $\chi = 2^{31}$};
\node at (2, -0.5) {\tiny $\chi = 2^5$};
\node at (-1.9, -0.5) {\tiny $\chi = 2^{20}$};
\end{tikzpicture}
} \,\, , 
\end{equation}
where we denote the bond dimension $\chi$ to each tensor leg. Since all the tensor contractions involved are tractable, the decomposition can be constructed efficiently. 
\end{proof}

\begin{figure}[t]
\centering
\includegraphics[width=0.45\textwidth]{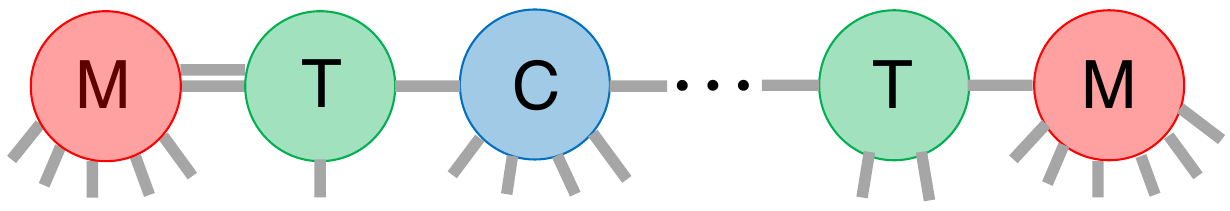}
\caption{Newly introduced hybrid Clifford-matchgate-MPS features tensors arranged in an architecture identical to that of MPS. Here, C, M, and T denote a Clifford tensor, a matchgate tensor, and a low-rank tensor, respectively. More general tree-tensor network replacing the MPS architecture is also possible. Each tensor in the network can be a Clifford, matchgate, or low-rank tensor. The number of legs connecting two tensor is assumed to be small, while the number of open legs attached to a Clifford or matchgate can be large. By design, our tensor network spatially separates the regions corresponding to Clifford and matchgate, enabling classical simulability when the bond dimensions are small.}
\label{fig:hybrid-cmt}
\end{figure}

Several remarks are in order. First, given an arbitrary tensor network, we provide an efficient algorithm to construct a Clifford and matchgate tensor as stated in the theorem. This may appear counterintuitive given that computing tensor network contractions of, e.g., a closed tensor network (i.e., one with no open legs) is $\textsc{\#P}$-hard~\cite{PhysRevLett.98.140506, PhysRevResearch.2.013010}. Our approach focuses on performing as many classically tractable tensor leg contractions as possible, ultimately reducing the network to a single Clifford tensor connected to a single matchgate tensor via the remaining closed tensor legs. The remaining closed tensor legs can be contracted through brute-force summations; however, this process scales exponentially with the number of legs connecting the Clifford tensor and the matchgate tensor. 

Second, Theorem~\ref{thm:Clifford-matchgate-decomposition} naturally motivates us to consider a new class of ansatz called the \textit{hybrid-Clifford-matchgate-MPS}, that remains classically simulable, despite typically exhibiting high non-Cliffordness, high non-matchgateness, and large bipartite entanglement entropy. Our \textit{hybrid Clifford-matchgate-MPS}, illustrated in FIG.~\ref{fig:hybrid-cmt}, retains the same tensor network architecture as a standard MPS, while allowing each tensor in the network to be a Clifford or a matchgate, along with a low-rank tensor. We note that more general tree-tensor networks can be used instead of a linear MPS arrangement, but we focus on the MPS architecture for concreteness. Surprisingly, we note that the classically simulability of non-planar dimer models defined on single-crossing minor-free graphs studied in the computer-science literature~\cite{vazirani1989nc, straub2016counting, curticapean2014counting, likhosherstov2020tractable} is precisely captured by our ansatz with a tree-tensor network architecture. To ensure classical simulability, the bond dimension, or the number of legs connecting two tensors\footnote{As a reminder, we assumed that each tensor leg has a bond dimension $2$.}, is kept small. Notably, a Clifford or matchgate tensor may have a large number of open legs, potentially scaling as $\mathrm{O}(n)$, where $n$ denotes the total number of open legs in the network. Using the bipartite cut through a Clifford or matchgate, we have a high bipartite entanglement entropy, a feature that is absent in the MPS with small bond dimensions. Consequently, a state representable by a hybrid Clifford-matchgate-MPS exhibits high non-Cliffordness and high non-matchgateness due to the presence of Clifford and matchgate tensors, yet remains efficiently contractable thanks to the geometrical separation of these tensors. Our hybrid Clifford-matchgate-MPS shares some similarities with other hybrid ansatze in the literature, including Clifford augmented MPS structure~\cite{PhysRevLett.133.190402}, Clifford transformed matchgate (Hartree-Fock) state~\cite{mishmash2023hierarchical, projansky2024extending}, Clifford conjugated matchgates~\cite{projansky2024extending}, in that geometrical separation is strategically employed to maintain tractability. We leave the exploration of the hybrid Clifford-matchgate-MPS as variational states for quantum many-body and quantum chemistry simulations to future work. 

Finally, we remark a potential of representing the third level of the matchgate hierarchy~\cite{matchgatehierarchiy}, which consists of unitaries that transform a Pauli string into a matchgate under conjugation, in terms of singly-punctured matchgates.

\subsubsection{Quon representations and tensor operations}
We now describe how planar tensor networks and their operations can be translated to the 2D Quon language. To achieve this, we introduce a diagrammatic element called the \textit{basis encoder}, which conveniently translates a tensor leg index into a 2D Quon diagram, defined as follows: 
\begin{equation}
\label{eq:basis-encoder}
\raisebox{-0.4cm}{\tikz{
\fill [blue!10] (0, 0.8) rectangle (1.5, 1);

\fill [blue!10] (1.5, 1) arc (0:180:0.75);

\draw[thick,blue] (0, 0.8)--(0, 1);
\draw[thick, blue] (1.5, 1) arc (0:180:0.75);
\draw[thick,blue] (1.5, 0.8)--(1.5, 1);

\draw (0.2, 0.8)--(0.2,1);
\draw (0.4, 0.8)--(0.4,1);
\draw (1.1, 0.8)--(1.1,1);
\draw (1.3, 0.8)--(1.3,1);

\draw (1.3, 1) arc (0:180:0.55);
\draw (1.1, 1) arc (0:180:0.35);

\draw[ultra thick, gray] (0.75, 1)--(0.75, 2);
\draw (0.75, 1)--(1.3+0.07, 1);

\node at (0.75, 1) [circle, fill, minimum size=4pt, inner sep=1pt] {};
\draw (1.3+0.07, 1) arc (0:360:0.07);
\draw (1.1+0.07, 1) arc (0:360:0.07);

\node at (0.75, 2.2) {$j$};
}} \,\, := \begin{cases}
\raisebox{-0.2cm}{\tikz{
\fill [blue!10] (0, 0.8) rectangle (1.5, 1);

\fill [blue!10] (1.5, 1) arc (0:180:0.75);

\draw[thick,blue] (0, 0.8)--(0, 1);
\draw[thick, blue] (1.5, 1) arc (0:180:0.75);
\draw[thick,blue] (1.5, 0.8)--(1.5, 1);

\draw (0.2, 0.8)--(0.2,1);
\draw (0.4, 0.8)--(0.4,1);
\draw (1.1, 0.8)--(1.1,1);
\draw (1.3, 0.8)--(1.3,1);

\draw (1.3, 1) arc (0:180:0.55);
\draw (1.1, 1) arc (0:180:0.35);
}} \,\, , \qquad \textrm{if } j = 0 \\ \\ 
\raisebox{-0.2cm}{\tikz{
\fill [blue!10] (0, 0.8) rectangle (1.5, 1);

\fill [blue!10] (1.5, 1) arc (0:180:0.75);

\draw[thick, blue] (0, 0.8)--(0, 1);
\draw[thick, blue] (1.5, 1) arc (0:180:0.75);
\draw[thick, blue] (1.5, 0.8)--(1.5, 1);

\draw (0.2, 0.8)--(0.2,1);
\draw (0.4, 0.8)--(0.4,1);
\draw (1.1, 0.8)--(1.1,1);
\draw (1.3, 0.8)--(1.3,1);

\draw (1.3, 1) arc (0:180:0.55);
\draw (1.1, 1) arc (0:180:0.35);

\node at (1.1, 1) [circle, fill, inner sep=1pt] {};
\node at (1.3, 1) [circle, fill, inner sep=1pt] {};
}} \,\, , \qquad \textrm{if } j = 1
\end{cases} , 
\end{equation}
where $j$ denotes the tensor leg index, and the tensor leg is depicted as a thick gray line. We often use the following simplified diagram: 
\begin{equation}
\label{eq:basis-encoder-simplified}
\raisebox{-0.4cm}{\tikz{
\draw (0.2, 0.8)--(0.2,1);
\draw (0.4, 0.8)--(0.4,1);
\draw (1.1, 0.8)--(1.1,1);
\draw (1.3, 0.8)--(1.3,1);

\draw (1.3, 1) arc (0:180:0.55);
\draw (1.1, 1) arc (0:180:0.35);

\fill [gray, fill opacity=0.6] (0, 0.8) rectangle (1.5, 1);

\fill[gray, fill opacity=0.6, draw=blue, thick] (1.5, 1) arc (0:180:0.75);

\draw[thick, blue] (0, 0.8)--(0, 1);
\draw[thick, blue] (1.5, 0.8)--(1.5, 1);

\draw[ultra thick, gray] (0.75, 1.75)--(0.75, 2);
}} \,\, := 
\raisebox{-0.4cm}{\tikz{
\fill [blue!10] (0, 0.8) rectangle (1.5, 1);

\fill [blue!10] (1.5, 1) arc (0:180:0.75);

\draw[thick, blue] (0, 0.8)--(0, 1);
\draw[thick, blue] (1.5, 1) arc (0:180:0.75);
\draw[thick, blue] (1.5, 0.8)--(1.5, 1);

\draw (0.2, 0.8)--(0.2,1);
\draw (0.4, 0.8)--(0.4,1);
\draw (1.1, 0.8)--(1.1,1);
\draw (1.3, 0.8)--(1.3,1);

\draw (1.3, 1) arc (0:180:0.55);
\draw (1.1, 1) arc (0:180:0.35);

\draw[ultra thick, gray] (0.75, 1)--(0.75, 2);
\draw (0.75, 1)--(1.3+0.07, 1);

\node at (0.75, 1) [circle, fill, minimum size=4pt, inner sep=1pt] {};
\draw (1.3+0.07, 1) arc (0:360:0.07);
\draw (1.1+0.07, 1) arc (0:360:0.07);
}} \,\, = \raisebox{-0.4cm}{\tikz{
\fill [blue!10] (0, 0.8) rectangle (1.5, 1);

\fill [blue!10] (1.5, 1) arc (0:180:0.75);

\draw[thick,blue] (0, 0.8)--(0, 1);
\draw[thick, blue] (1.5, 1) arc (0:180:0.75);
\draw[thick,blue] (1.5, 0.8)--(1.5, 1);

\draw (0.2, 0.8)--(0.2,1);
\draw (0.4, 0.8)--(0.4,1);
\draw (1.1, 0.8)--(1.1,1);
\draw (1.3, 0.8)--(1.3,1);

\draw (1.3, 1) arc (0:180:0.55);
\draw (1.1, 1) arc (0:180:0.35);

\draw[ultra thick, gray] (0.75, 1)--(0.75, 2);
\draw (0.75, 1)--(0.2-0.07, 1);

\node at (0.75, 1) [circle, fill, minimum size=4pt, inner sep=1pt] {};
\draw (0.2+0.07, 1) arc (0:360:0.07);
\draw (0.4+0.07, 1) arc (0:360:0.07);
}} \,\, . 
\end{equation}
Using the basis encoder, the \textit{resolution of the identity} Eq.~\eqref{eq:resolution-of-id-1} becomes 
\begin{equation}
\label{eq:resolution-of-id-2}
\raisebox{-0.55cm}{\tikz{
\fill[blue!10] (0, -0.6) rectangle (1, 0.6);

\draw[thick, blue] (0, 0.6)--(0, -0.6);
\draw[thick, blue] (1, 0.6)--(1, -0.6);

\draw (0.2, -0.6)--(0.2, 0.6);
\draw (0.4, -0.6)--(0.4, 0.6);
\draw (0.6, -0.6)--(0.6, 0.6);
\draw (0.8, -0.6)--(0.8, 0.6);
}} \,\, = \frac{1}{2} \,\, \raisebox{-0.55cm}{\tikz{
\draw (0.4, 0.6) arc (-180:0:0.1);
\draw (0.2, 0.6) arc (-180:0:0.3);

\draw (0.4, -0.6) arc (180:0:0.1);
\draw (0.2, -0.6) arc (180:0:0.3);

\fill[gray, fill opacity=0.6, draw=blue, thick] (0, 0.6) arc (-180:0:0.5);
\fill[gray, fill opacity=0.6, draw=blue, thick] (0, -0.6) arc (180:0:0.5);

\draw [ultra thick, gray] (0.5, -0.1)--(0.5, 0.1);
}} \,\, . 
\end{equation}
This resolution of the identity allows one to relate tensor network contractions with the gluing of the corresponding 2D Quon diagrams. As the name suggests, one can glue the basis encoder to an open boundary of a 2D Quon diagram, where four Majoarnas terminate. Doing these gluing clearly illustrates how 2D Quon diagrammatic representations serve as inner structures of planar tensors and planar tensor networks (see Table~\ref{tab:generating-tensors-quon-diagrams} for various examples).

Throughout the remainder of this section, we consider only 2D Quon diagrams with the basis encoders are glued, ensuring that the tensor legs are explicitly shown. For example, we properly glue the basis encoders for 2D Quon diagrams presented in Sec.~\ref{sec:quon-for-q-states-gates}. When more than four Majorana lines terminate at an open interval---implying that the open interval encodes more than one qubit---a corresponding gluing of basis encoders can be performed, as detailed in Appendix~\ref{app:encoding-general}. 

\begin{figure*}[t]
\centering
\includegraphics[width=\textwidth]{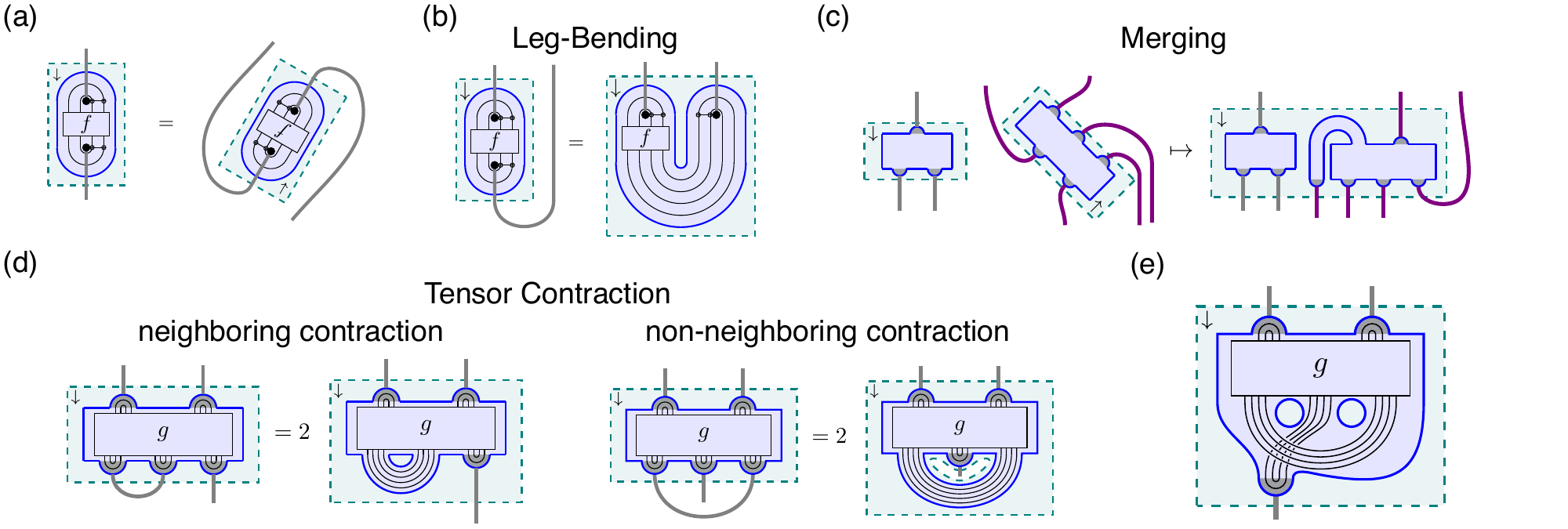}
\caption{(a) We can freely rotate a tensor within a planar tensor network while keeping the endpoints of the tensor legs fixed. When doing so, its 2D Quon diagrammatic representation, the planar region, and the time direction rotate together. Here, $f$ denotes a Majorana diagram and the arrow in the planar region, indicated as $\downarrow$, denotes the time direction. (b) The directionality of a tensor leg can be changed by using a \textit{leg-bending}. (c) As a result of a \textit{merging} operation between two planar tensor networks, the corresponding Quon diagrams are also merged as follows: First, we rotate one Quon diagram to align its time direction with that of the other diagram, then merge the two. If necessary, leg-bendings are performed prior to merging. (d) Effects of neighboring and non-neighboring tensor contractions on Quon diagrams, where leg-bendings and the resolution of the identity Eq.~\eqref{eq:resolution-of-id-2} are used. $g$ denotes a Majorana diagram in both examples. Note that a hole is introduced in the Quon diagram in both contractions; however, a puncture is introduced in the planar region only in the case of non-neighboring contraction. (e) A tensor leg can be moved from one puncture to outside the disk boundary with the help of the SWAP gate Eq.~\eqref{eq:SWAP-quon-diagram}.}
\label{fig:quon-tn}
\end{figure*}

In FIG.~\ref{fig:quon-tn}, we demonstrate how planar tensor network diagrams equipped with 2D Quon representations transform under various planar tensor network operations. As in Sec.~\ref{sec:planar-TN}, we explicitly indicate the \textit{planar region} in the tensor network diagrams. Moreover, we also indicate the \textit{time arrow}, denoted by an arrow inside the planar region, associated with the 2D Quon diagram in each tensor. For simplicity, we often omit placing $\star$s in the boundaries of the planar region. We also note that the time arrow can be removed when the 2D Quon diagram contains no dots (See Appendix~\ref{app:dots-and-time-direction} for further discussion). To avoid any ambiguity, we assume that tensor legs are not perpendicular to the time direction; thus, tensor legs are always classified as pointing either with or against the time flow.

\subsubsection{Quon representations of a generating set}
\label{sec:quon-for-tensor-networks}
Finally, we present an elementary generating set of tensors for arbitrary tensor networks, along with their 2D Quon diagrammatic representations in Table~\ref{tab:generating-tensors-quon-diagrams}. The results in TABLE~\ref{tab:generating-tensors-quon-diagrams}, together with planar tensor operations presented in FIG.~\ref{fig:quon-tn}, establish that the 2D Quon language is \textit{universal} for tensor network, meaning that any tensor network can be \textit{efficiently} represented by a \textit{single} 2D Quon diagram. Efficient construction of a 2D Quon diagrammatic representation proceeds as follows: Start with a tensor network in which each tensor is efficiently represented---for example, either as a Clifford tensor, a matchgate tensor, or a small-rank tensor with its component explicitly provided. First, decompose each tensor into a mini tensor network of elementary Clifford and matchgate tensors; this decomposition remains efficient even for a small-rank tensor, as their limited dimensionality allows a straightforward brute-force decomposition. Then, replace each tensor with its corresponding 2D Quon diagram from TABLE~\ref{tab:generating-tensors-quon-diagrams}, and perform the planar tensor operations following FIG.~\ref{fig:quon-tn} to assemble a single 2D Quon diagram representing the entire network. 

In TABLE~\ref{tab:generating-tensors-quon-diagrams}, we also incorporate the ZX-calculus diagrammatic notation~\cite{coecke2007graphical, coecke2008interacting, coecke2011interacting, coecke2010compositional, van2020zx} when presenting elementary tensors. One can view the representing Quon diagram as an internal structure for the ZX-diagram, offering finer-grained information and additional flexibility in manipulating diagrams. We note that all of the axioms of the ZX-calculus can be derived from Quon rewriting rules, notably the bialgebra rule~\cite{liu2017quon}. Most Quon representations in TABLE~\ref{tab:generating-tensors-quon-diagrams} are reformulations of the results from Sec.~\ref{sec:quon-for-q-states-gates} with the basis encoder Eq.~\eqref{eq:basis-encoder} merely glued. The main difference is the use of the $3$-leg \textit{parity tensor} $P$, instead of the $2$-qubit gates considered previously. In the following, we first discuss the properties of the parity tensor and then demonstrate that the $3$-leg parity tensor, in combination with $1$- and $2$-leg tensors in TABLE~\ref{tab:generating-tensors-quon-diagrams}, can generate the $2$-qubit gates considered in Sec.~\ref{sec:quon-for-q-states-gates}.

\begin{table*}[ht!]
\centering
\begin{tabular}{|c|c|c|c|c|c|c|c|} \hline
 & $1$-leg & \multicolumn{5}{c|}{$2$-leg} & $3$-leg \\ 
\hline
Tensor & $\vert 0 \rangle$ & $I$ & $X$ & $e^{i \frac{\pi}{4}} e^{-i \frac{\pi}{4} X}$ & $e^{-i \frac{\pi}{4}} e^{i \frac{\pi}{4} X}$ & $e^{i \frac{\theta}{2}} e^{-i \frac{\theta}{2} Z}$ & $P$ \\
\hline
ZX-diagram & $\frac{1}{\sqrt{2}}$ 
\raisebox{-0.55cm}{
\begin{tikzpicture}
\draw [dashed, thick, teal, fill=teal!20, fill opacity=0.4] (0.4, 0.7) arc (0:360:0.4);

\draw[ultra thick, gray] (0, 0.7)--(0, 0);

\draw [thick, gray, fill=zx_red] (0.2, 0.7) arc (0:360:0.2);
\node at (0, 0.7) {\scriptsize $0$};
\end{tikzpicture}
} & \raisebox{-0.65cm}{
\begin{tikzpicture}
\draw [dashed, thick, teal, fill=teal!20, fill opacity=0.4] (0.4, 0.7) arc (0:360:0.4);

\draw[ultra thick, gray] (0, 0)--(0, 1.4);

\end{tikzpicture}
} & \raisebox{-0.65cm}{
\begin{tikzpicture}
\draw [dashed, thick, teal, fill=teal!20, fill opacity=0.4] (0.4, 0.7) arc (0:360:0.4);

\draw[ultra thick, gray] (0, 0)--(0, 1.4);

\draw [thick, gray, fill=zx_red] (0.2, 0.7) arc (0:360:0.2);
\node at (0, 0.7) {\scriptsize $\pi$};
\end{tikzpicture}
} & \raisebox{-0.65cm}{
\begin{tikzpicture}
\draw [dashed, thick, teal, fill=teal!20, fill opacity=0.4] (0.4, 0.7) arc (0:360:0.4);

\draw[ultra thick, gray] (0, 0)--(0, 1.4);

\draw [thick, gray, fill=zx_red] (0.2, 0.7) arc (0:360:0.2);
\node at (0, 0.7) {\scriptsize $\frac{\pi}{2}$};
\end{tikzpicture}
} & \raisebox{-0.65cm}{
\begin{tikzpicture}
\draw [dashed, thick, teal, fill=teal!20, fill opacity=0.4] (0.4, 0.7) arc (0:360:0.4);

\draw[ultra thick, gray] (0, 0)--(0, 1.4);

\draw [thick, gray, fill=zx_red] (0.25, 0.7) arc (0:360:0.25);
\node at (0, 0.7) {\scriptsize $-\frac{\pi}{2}$};
\end{tikzpicture}
} & \raisebox{-0.65cm}{
\begin{tikzpicture}
\draw [dashed, thick, teal, fill=teal!20, fill opacity=0.4] (0.4, 0.7) arc (0:360:0.4);

\draw[ultra thick, gray] (0, 0)--(0, 1.4);

\draw [thick, gray, fill=zx_green] (0.2, 0.7) arc (0:360:0.2);
\node at (0, 0.7) {\scriptsize $\theta$};
\end{tikzpicture}
} & $\sqrt{2}$ \raisebox{-0.5cm}{
\begin{tikzpicture}
\draw [dashed, thick, teal, fill=teal!20, fill opacity=0.4] (0.4, 0.7) arc (0:360:0.4);

\draw[ultra thick, gray] (0, 0.7)--(0, 1.4);
\draw[ultra thick, gray] (0.5, 0.2) arc (0:180:0.5);

\draw [thick, gray, fill=zx_red] (0.2, 0.7) arc (0:360:0.2);

\node at (0, 0.7) {\scriptsize $0$};
\end{tikzpicture}
} \\
\hline
Quon Diagram & $\frac{1}{2}$ \raisebox{-0.75cm}{
\begin{tikzpicture}[scale=0.7]
\draw[dashed, thick, teal, fill=teal!20, fill opacity=0.4] (-0.2, -0.1) rectangle (1.7, 1.9);
\node at (-0.05, 1.7) {\scriptsize $\downarrow$};

\fill [blue!10] (1.5, 1) arc (0:180:0.75);
\fill [blue!10] (1.5, 0.8) arc (0:-180:0.75);

\fill [blue!10] (0, 0.8) rectangle (1.5, 1);

\draw[thick, blue] (0, 0.8)--(0, 1);
\draw[thick, blue] (1.5, 0.8)--(1.5, 1);
\draw[thick, blue] (1.5, 1) arc (0:180:0.75);
\draw[thick, blue] (1.5, 0.8) arc (0:-180:0.75);

\draw[red] (0.2, 0.8)--(0.2, 1);
\draw (0.4, 0.8)--(0.4, 1);
\draw (1.1, 0.8)--(1.1, 1);
\draw[red] (1.3, 0.8)--(1.3, 1);

\draw[red] (1.3, 1) arc (0:180:0.55);
\draw (1.1, 1) arc (0:180:0.35);

\draw[red] (1.3, 0.8) arc (0:-180:0.55);
\draw (1.1, 0.8) arc (0:-180:0.35);

\draw[ultra thick, gray] (0.75, 0.8)--(0.75, -0.3);
\draw (0.75, 0.8)--(1.3+0.07, 0.8);

\node at (0.75, 0.8) [circle, fill, minimum size=4pt, inner sep=1pt] {};
\draw (1.3+0.07, 0.8) arc (0:360:0.07);
\draw (1.1+0.07, 0.8) arc (0:360:0.07);
\end{tikzpicture}
} & $\frac{1}{2}$ \raisebox{-1cm}{
\begin{tikzpicture}[scale=0.8]
\draw[dashed, thick, teal, fill=teal!20, fill opacity=0.4] (-0.2, -0.1) rectangle (1.7, 2.1);
\node at (-0.05, 1.9) {\scriptsize $\downarrow$};

\fill [blue!10] (1.5, 1.2) arc (0:180:0.75);
\fill [blue!10] (1.5, 0.8) arc (0:-180:0.75);

\fill [blue!10] (0, 0.8) rectangle (1.5, 1.2);

\draw[thick, blue] (0, 0.8)--(0, 1.2);
\draw[thick, blue] (1.5, 0.8)--(1.5, 1.2);
\draw[thick, blue] (1.5, 1.2) arc (0:180:0.75);
\draw[thick, blue] (1.5, 0.8) arc (0:-180:0.75);

\draw[red] (0.2, 0.8)--(0.2, 1.2);
\draw (0.4, 0.8)--(0.4, 1.2);
\draw (1.1, 0.8)--(1.1, 1.2);
\draw[red] (1.3, 0.8)--(1.3, 1.2);

\draw[red] (1.3, 1.2) arc (0:180:0.55);
\draw (1.1, 1.2) arc (0:180:0.35);

\draw[red] (1.3, 0.8) arc (0:-180:0.55);
\draw (1.1, 0.8) arc (0:-180:0.35);

\draw[ultra thick, gray] (0.75, 1.2)--(0.75, 2.3);
\draw (0.75, 1.2)--(1.3+0.07, 1.2);

\node at (0.75, 1.2) [circle, fill, minimum size=4pt, inner sep=1pt] {};
\draw (1.3+0.07, 1.2) arc (0:360:0.07);
\draw (1.1+0.07, 1.2) arc (0:360:0.07);

\draw[ultra thick, gray] (0.75, 0.8)--(0.75, -0.3);
\draw (0.75, 0.8)--(1.3+0.07, 0.8);

\node at (0.75, 0.8) [circle, fill, minimum size=4pt, inner sep=1pt] {};
\draw (1.3+0.07, 0.8) arc (0:360:0.07);
\draw (1.1+0.07, 0.8) arc (0:360:0.07);
\end{tikzpicture}
} & $\frac{1}{2}$ \raisebox{-1cm}{
\begin{tikzpicture}[scale=0.8]
\draw[dashed, thick, teal, fill=teal!20, fill opacity=0.4] (-0.2, -0.1) rectangle (1.7, 2.1);
\node at (-0.05, 1.9) {\scriptsize $\downarrow$};

\fill [blue!10] (1.5, 1.2) arc (0:180:0.75);
\fill [blue!10] (1.5, 0.8) arc (0:-180:0.75);

\fill [blue!10] (0, 0.8) rectangle (1.5, 1.2);

\draw[thick, blue] (0, 0.8)--(0, 1.2);
\draw[thick, blue] (1.5, 0.8)--(1.5, 1.2);
\draw[thick, blue] (1.5, 1.2) arc (0:180:0.75);
\draw[thick, blue] (1.5, 0.8) arc (0:-180:0.75);

\draw[red] (0.2, 0.8)--(0.2, 1.2);
\draw (0.4, 0.8)--(0.4, 1.2);
\draw (1.1, 0.8)--(1.1, 1.2);
\draw[red] (1.3, 0.8)--(1.3, 1.2);

\draw[red] (1.3, 1.2) arc (0:180:0.55);
\draw (1.1, 1.2) arc (0:180:0.35);

\draw[red] (1.3, 0.8) arc (0:-180:0.55);
\draw (1.1, 0.8) arc (0:-180:0.35);

\draw[ultra thick, gray] (0.75, 1.2)--(0.75, 2.3);
\draw (0.75, 1.2)--(1.3+0.07, 1.2);

\node at (0.75, 1.2) [circle, fill, minimum size=4pt, inner sep=1pt] {};
\draw (1.3+0.07, 1.2) arc (0:360:0.07);
\draw (1.1+0.07, 1.2) arc (0:360:0.07);

\node at (1.3, 1) [circle, fill, inner sep=1pt] {};
\node at (1.1, 1) [circle, fill, inner sep=1pt] {};

\draw[ultra thick, gray] (0.75, 0.8)--(0.75, -0.3);
\draw (0.75, 0.8)--(1.3+0.07, 0.8);

\node at (0.75, 0.8) [circle, fill, minimum size=4pt, inner sep=1pt] {};
\draw (1.3+0.07, 0.8) arc (0:360:0.07);
\draw (1.1+0.07, 0.8) arc (0:360:0.07);
\end{tikzpicture}
} & $\frac{1}{2}$ \raisebox{-1cm}{
\begin{tikzpicture}[scale=0.8]
\draw[dashed, thick, teal, fill=teal!20, fill opacity=0.4] (-0.2, -0.1) rectangle (1.7, 2.1);
\node at (-0.05, 1.9) {\scriptsize $\downarrow$};

\draw[thick, blue, fill=blue!10] (0, 0.8)--(0, 1.2) arc (180:0:0.75)--(1.5, 0.8) arc(0:-180:0.75);

\draw (0.2, 0.8) to[out=90, in=-90] (0.4, 1.2) arc(180:0:0.35)--(1.1, 0.8) arc(0:-180:0.35) to[out=90, in=-50] (0.34, 0.95);
\draw (0.26, 1.05) to[out=130, in=-90] (0.2, 1.2) arc(180:0:0.55)--(1.3, 0.8) arc(0:-180:0.55); 

\draw[ultra thick, gray] (0.75, 1.2)--(0.75, 2.3);
\draw (0.75, 1.2)--(1.3+0.07, 1.2);

\node at (0.75, 1.2) [circle, fill, minimum size=4pt, inner sep=1pt] {};
\draw (1.3+0.07, 1.2) arc (0:360:0.07);
\draw (1.1+0.07, 1.2) arc (0:360:0.07);

\draw[ultra thick, gray] (0.75, 0.8)--(0.75, -0.3);
\draw (0.75, 0.8)--(1.3+0.07, 0.8);

\node at (0.75, 0.8) [circle, fill, minimum size=4pt, inner sep=1pt] {};
\draw (1.3+0.07, 0.8) arc (0:360:0.07);
\draw (1.1+0.07, 0.8) arc (0:360:0.07);
\end{tikzpicture}
} & $\frac{1}{2}$ \raisebox{-1cm}{
\begin{tikzpicture}[scale=0.8]
\draw[dashed, thick, teal, fill=teal!20, fill opacity=0.4] (-0.2, -0.1) rectangle (1.7, 2.1);
\node at (-0.05, 1.9) {\scriptsize $\downarrow$};

\draw[thick, blue, fill=blue!10] (0, 0.8)--(0, 1.2) arc (180:0:0.75)--(1.5, 0.8) arc(0:-180:0.75);

\draw (0.34, 1.05) to[out=50, in=-90] (0.4, 1.2) arc(180:0:0.35)--(1.1, 0.8) arc(0:-180:0.35) to[out=90, in=-90] (0.2, 1.2)arc(180:0:0.55)--(1.3, 0.8) arc(0:-180:0.55) to[out=90, in=-130] (0.26, 0.95);

\draw[ultra thick, gray] (0.75, 1.2)--(0.75, 2.3);
\draw (0.75, 1.2)--(1.3+0.07, 1.2);

\node at (0.75, 1.2) [circle, fill, minimum size=4pt, inner sep=1pt] {};
\draw (1.3+0.07, 1.2) arc (0:360:0.07);
\draw (1.1+0.07, 1.2) arc (0:360:0.07);

\draw[ultra thick, gray] (0.75, 0.8)--(0.75, -0.3);
\draw (0.75, 0.8)--(1.3+0.07, 0.8);

\node at (0.75, 0.8) [circle, fill, minimum size=4pt, inner sep=1pt] {};
\draw (1.3+0.07, 0.8) arc (0:360:0.07);
\draw (1.1+0.07, 0.8) arc (0:360:0.07);
\end{tikzpicture}
} & $\frac{1}{2}$ \raisebox{-1.25cm}{
\begin{tikzpicture}[scale=0.8]
\draw[dashed, thick, teal, fill=teal!20, fill opacity=0.4] (-0.2, -0.7) rectangle (1.7, 2.1);
\node at (-0.05, 1.9) {\scriptsize $\downarrow$};

\fill [blue!10] (1.5, 1.2) arc (0:180:0.75);
\fill [blue!10] (1.5, 0.2) arc (0:-180:0.75);

\fill [blue!10] (0, 0.2) rectangle (1.5, 1.2);

\draw[thick, blue] (0, 0.2)--(0, 1.2);
\draw[thick, blue] (1.5, 0.2)--(1.5, 1.2);
\draw[thick, blue] (1.5, 1.2) arc (0:180:0.75);
\draw[thick, blue] (1.5, 0.2) arc (0:-180:0.75);

\draw[red] (0.2, 0.2)--(0.2, 1.2);
\draw[red] (1.3, 0.2)--(1.3, 1.2);

\draw (1, 0.7) arc (0:360:0.25);
\node at (0.75, 0.7) {\scriptsize $\theta_\updownarrow$};

\draw (0.4, 1.2) to[out=-90, in=135] (+{0.75-0.25*cos(45)}, +{0.7+0.25*sin(45)});
\draw (1.1, 1.2) to[out=-90, in=45] (+{0.75+0.25*cos(45)}, +{0.7+0.25*sin(45)});

\draw (0.4, 0.2) to[out=90, in=-135] (+{0.75-0.25*cos(45)}, +{0.7-0.25*sin(45)});
\draw (1.1, 0.2) to[out=90, in=-45] (+{0.75+0.25*cos(45)}, +{0.7-0.25*sin(45)});

\draw[red] (1.3, 1.2) arc (0:180:0.55);
\draw (1.1, 1.2) arc (0:180:0.35);

\draw[red] (1.3, 0.2) arc (0:-180:0.55);
\draw (1.1, 0.2) arc (0:-180:0.35);

\draw[ultra thick, gray] (0.75, 1.2)--(0.75, 2.3);
\draw (0.75, 1.2)--(1.3+0.07, 1.2);

\node at (0.75, 1.2) [circle, fill, minimum size=4pt, inner sep=1pt] {};
\draw (1.3+0.07, 1.2) arc (0:360:0.07);
\draw (1.1+0.07, 1.2) arc (0:360:0.07);

\draw[ultra thick, gray] (0.75, 0.2)--(0.75, -0.9);
\draw (0.75, 0.2)--(1.3+0.07, 0.2);

\node at (0.75, 0.2) [circle, fill, minimum size=4pt, inner sep=1pt] {};
\draw (1.3+0.07, 0.2) arc (0:360:0.07);
\draw (1.1+0.07, 0.2) arc (0:360:0.07);
\end{tikzpicture}
} & $\frac{1}{2}$ \raisebox{-1cm}{
\begin{tikzpicture}[scale=0.5]
\draw[dashed, thick, teal, fill=teal!20, fill opacity=0.4] (-1.3, -0.9) rectangle (2.8, 2.4);
\node at (-1, 2.1) {\scriptsize $\downarrow$};

\begin{scope}[even odd rule]
\clip (0, 0) rectangle (1.5, 1.5)  (1.1, 0) arc (0:180:0.35);

\fill[blue!10] (0, 0) rectangle (1.5, 1.5);
\end{scope}

\fill[blue!10] (0, 1.5) to[out=-90, in=90] (-1.1, 0)--(0, 0)--(0, 1.5);
\fill[blue!10] (1.5, 1.5) to[out=-90, in=90] (2.6, 0)--(1.5, 0)--(1.5, 1.5);

\draw[thick, blue] (0, 1.5) to[out=-90, in=90] (-1.1, 0);
\draw[thick, blue] (1.5, 1.5) to[out=-90, in=90] (2.6, 0);
\draw[thick, blue] (1.1, 0) arc (0:180:0.35);

\draw[thick, blue, fill=blue!10] (1.5, 1.5) arc (0:180:0.75);
\draw[red] (1.3, 1.5) arc (0:180:0.55);
\draw (1.1, 1.5) arc (0:180:0.35);

\draw[thick, blue, fill=blue!10] (0.4, 0) arc (0:-180:0.75);
\draw[red] (0.2, 0) arc (0:-180:0.55);
\draw (0, 0) arc (0:-180:0.35);

\draw[thick, blue, fill=blue!10] (1.1, 0) arc (-180:0:0.75);
\draw[red] (2.4, 0) arc (0:-180:0.55);
\draw (2.2, 0) arc (0:-180:0.35);

\draw[red, looseness=0.8] (0.2, 1.5) to[out=-90, in=90] (-0.9, 0);
\draw[red, looseness=0.8] (1.3, 1.5) to[out=-90, in=90] (2.4, 0);
\draw[red, looseness=1.5] (1.3, 0) to[out=90, in=90] (0.2, 0);

\draw[looseness=0.8] (0.4, 1.5) to[out=-90, in=90] (-0.7, 0);
\draw[looseness=0.8] (1.1, 1.5) to[out=-90, in=90] (2.2, 0);
\draw[looseness=1.5] (1.5, 0) to[out=90, in=90] (0, 0);

\draw[ultra thick, gray] (0.75, 1.5)--(0.75, 2.8);
\draw (0.75, 1.5)--(1.3+0.07, 1.5);

\node at (0.75, 1.5) [circle, fill, minimum size=4pt, inner sep=1pt] {};
\draw (1.3+0.07, 1.5) arc (0:360:0.07);
\draw (1.1+0.07, 1.5) arc (0:360:0.07);

\draw[ultra thick, gray] (-0.35, 0)--(-0.35, -1.3);
\draw (-0.35, 0)--(0.2+0.07, 0);

\node at (-0.35, 0) [circle, fill, minimum size=4pt, inner sep=1pt] {};
\draw (0.2+0.07, 0) arc (0:360:0.07);
\draw (0.07, 0) arc (0:360:0.07);

\draw[ultra thick, gray] (1.85, 0)--(1.85, -1.3);
\draw (1.85, 0)--(2.4+0.07, 0);

\node at (1.85, 0) [circle, fill, minimum size=4pt, inner sep=1pt] {};
\draw (2.4+0.07, 0) arc (0:360:0.07);
\draw (2.2+0.07, 0) arc (0:360:0.07);
\end{tikzpicture}
} \\
\hline
Clifford & \checkmark & \checkmark & \checkmark & \checkmark & \checkmark & generally $\times$ & \checkmark \\
\hline
Matchgate & \checkmark & \checkmark & \checkmark & $\times$ & $\times$ & \checkmark & \checkmark \\ \hline 
\end{tabular}
\caption{A set of elementary generating tensors, along with their ZX-diagram notations and 2D Quon diagrams, with indications of whether each tensor is Clifford or matchgate. By slight abuse of notation, we use $\vert 0 \rangle$ to denote the $1$-leg tensor that equals $1$ when its leg index is $0$ and vanishes otherwise. Note that the $2$-leg tensor $e^{i \frac{\theta}{2}} e^{- i \frac{\theta}{2} Z}$ is generally not a Clifford tensor, except at special angles (integer multiples of $\pi/2$). Here, the $3$-leg tensor $P$ denotes the $3$-leg parity tensor.}
\label{tab:generating-tensors-quon-diagrams}
\end{table*}

The $n$-leg parity tensor (also called the XOR tensor~\cite{biamonte2011categorical}), defined for every integer $n$, takes the value $1$ if the sum of the leg indices is even and $0$ (so it vanishes) if the sum is odd. In particular, the cases $n=1$ and $n=2$ correspond to $\vert 0 \rangle$ and the identity tensor $I$, respectively, in TABLE~\ref{tab:generating-tensors-quon-diagrams}. Notably, the parity tensors satisfy the renormalization-group (RG) invariant property. For instance, the $4$-leg parity tensor can be decomposed into two 3-leg parity tensors: 
\begin{align}
&\raisebox{-0.65cm}{
\begin{tikzpicture}
\draw [dashed, thick, teal, fill=teal!20, fill opacity=0.4] (0.9, 0.7) arc (0:360:0.4);
\draw[ultra thick, gray] (0, 0)--(1, 1.4);
\draw[ultra thick, gray] (1, 0)--(0, 1.4);
\draw [thick, gray, fill=white] (0.7, 0.7) arc (0:360:0.2);
\node at (0.5, 0.7) {\scriptsize $P$};
\end{tikzpicture}
} = \raisebox{-0.65cm}{
\begin{tikzpicture}
\draw [dashed, thick, teal, fill=teal!20, fill opacity=0.4] (0.5, 0.75) ellipse (0.9cm and 0.4cm);
\draw[ultra thick, gray] (0, 0)--(0, 1.4);
\draw[ultra thick, gray] (1, 0)--(1, 1.4);
\draw[ultra thick, gray] (0, 0.8) to[out=-60, in=-120] (1, 0.8);
\draw [thick, gray, fill=white] (0.2, 0.8) arc (0:360:0.2);
\draw [thick, gray, fill=white] (1.2, 0.8) arc (0:360:0.2);
\node at (0, 0.8) {\scriptsize $P$};
\node at (1, 0.8) {\scriptsize $P$};
\end{tikzpicture}
} = 2 \raisebox{-0.65cm}{
\begin{tikzpicture}
\draw [dashed, thick, teal, fill=teal!20, fill opacity=0.4] (0.5, 0.75) ellipse (0.9cm and 0.4cm);
\draw[ultra thick, gray] (0, 0)--(0, 1.4);
\draw[ultra thick, gray] (1, 0)--(1, 1.4);
\draw[ultra thick, gray] (0, 0.8) to[out=-60, in=-120] (1, 0.8);
\draw [thick, gray, fill=zx_red] (0.15, 0.8) arc (0:360:0.15);
\draw [thick, gray, fill=zx_red] (1.15, 0.8) arc (0:360:0.15);
\node at (0, 0.8) {\scriptsize $0$};
\node at (1, 0.8) {\scriptsize $0$};
\end{tikzpicture}
} = 2 \raisebox{-0.65cm}{
\begin{tikzpicture}
\draw [dashed, thick, teal, fill=teal!20, fill opacity=0.4] (0.9, 0.7) arc (0:360:0.4);
\draw[ultra thick, gray] (0, 0)--(1, 1.4);
\draw[ultra thick, gray] (0, 1.4)--(1, 0);
\draw [thick, gray, fill=zx_red] (0.65, 0.7) arc (0:360:0.15);
\node at (0.5, 0.7) {\scriptsize $0$};
\end{tikzpicture}
} \nonumber \\ 
&= \frac{1}{4} \,\, \raisebox{-1.1cm}{
\begin{tikzpicture}[scale=0.5]
\draw[dashed, thick, teal, fill=teal!20, fill opacity=0.4] (-1.3, -0.9) rectangle (2.8, 2.4);
\node at (-1, 2.1) {\scriptsize $\downarrow$};
\begin{scope}[even odd rule]
\clip (0, 0) rectangle (1.5, 1.5)  (1.1, 0) arc (0:180:0.35);
\fill[blue!10] (0, 0) rectangle (1.5, 1.5);
\end{scope}
\fill[blue!10] (0, 1.5) to[out=-90, in=90] (-1.1, 0)--(0, 0)--(0, 1.5);
\fill[blue!10] (1.5, 1.5) to[out=-90, in=90] (2.6, 0)--(1.5, 0)--(1.5, 1.5);
\draw[thick, blue] (0, 1.5) to[out=-90, in=90] (-1.1, 0);
\draw[thick, blue] (1.5, 1.5) to[out=-90, in=90] (2.6, 0);
\draw[thick, blue] (1.1, 0) arc (0:180:0.35);
\draw[thick, blue, fill=blue!10] (1.5, 1.5) arc (0:180:0.75);
\draw[red] (1.3, 1.5) arc (0:180:0.55);
\draw (1.1, 1.5) arc (0:180:0.35);
\draw[thick, blue, fill=blue!10] (0.4, 0) arc (0:-180:0.75);
\draw[red] (0.2, 0) arc (0:-180:0.55);
\draw (0, 0) arc (0:-180:0.35);
\draw[thick, blue, fill=blue!10] (1.1, 0) arc (-180:0:0.75);
\draw[red] (2.4, 0) arc (0:-180:0.55);
\draw (2.2, 0) arc (0:-180:0.35);
\draw[red, looseness=0.8] (0.2, 1.5) to[out=-90, in=90] (-0.9, 0);
\draw[red, looseness=0.8] (1.3, 1.5) to[out=-90, in=90] (2.4, 0);
\draw[red, looseness=1.5] (1.3, 0) to[out=90, in=90] (0.2, 0);
\draw[looseness=0.8] (0.4, 1.5) to[out=-90, in=90] (-0.7, 0);
\draw[looseness=0.8] (1.1, 1.5) to[out=-90, in=90] (2.2, 0);
\draw[looseness=1.5] (1.5, 0) to[out=90, in=90] (0, 0);
%
\draw[ultra thick, gray] (0.75, 1.5)--(0.75, 2.8);
\draw (0.75, 1.5)--(1.3+0.07, 1.5);
\node at (0.75, 1.5) [circle, fill, minimum size=4pt, inner sep=1pt] {};
\draw (1.3+0.07, 1.5) arc (0:360:0.07);
\draw (1.1+0.07, 1.5) arc (0:360:0.07);
%
\draw[ultra thick, gray] (-0.35, 0)--(-0.35, -1.3);
\draw (-0.35, 0)--(0.2+0.07, 0);
\node at (-0.35, 0) [circle, fill, minimum size=4pt, inner sep=1pt] {};
\draw (0.2+0.07, 0) arc (0:360:0.07);
\draw (0.07, 0) arc (0:360:0.07);
%
\draw (1.85, 0)--(2.4+0.07, 0);
%
\draw (2.4+0.07, 0) arc (0:360:0.07);
\draw (2.2+0.07, 0) arc (0:360:0.07);
%
%
\draw[dashed, thick, teal, fill=teal!20, fill opacity=0.4] (4.5-1.3, -0.9) rectangle (4.5+2.8, 2.4);
\node at (-1+4.5, 2.1) {\scriptsize $\downarrow$};
\begin{scope}[even odd rule]
\clip (4.5+0, 0) rectangle (4.5+1.5, 1.5)  (4.5+1.1, 0) arc (0:180:0.35);
\fill[blue!10] (4.5+0, 0) rectangle (4.5+1.5, 1.5);
\end{scope}
\fill[blue!10] (4.5+0, 1.5) to[out=-90, in=90] (4.5-1.1, 0)--(4.5+0, 0)--(4.5+0, 1.5);
\fill[blue!10] (4.5+1.5, 1.5) to[out=-90, in=90] (4.5+2.6, 0)--(4.5+1.5, 0)--(4.5+1.5, 1.5);
\draw[thick, blue] (4.5+0, 1.5) to[out=-90, in=90] (4.5-1.1, 0);
\draw[thick, blue] (4.5+1.5, 1.5) to[out=-90, in=90] (4.5+2.6, 0);
\draw[thick, blue] (4.5+1.1, 0) arc (0:180:0.35);
\draw[thick, blue, fill=blue!10] (4.5+1.5, 1.5) arc (0:180:0.75);
\draw[red] (4.5+1.3, 1.5) arc (0:180:0.55);
\draw (4.5+1.1, 1.5) arc (0:180:0.35);
\draw[thick, blue, fill=blue!10] (4.5+0.4, 0) arc (0:-180:0.75);
\draw[red] (4.5+0.2, 0) arc (0:-180:0.55);
\draw (4.5+0, 0) arc (0:-180:0.35);
\draw[thick, blue, fill=blue!10] (4.5+1.1, 0) arc (-180:0:0.75);
\draw[red] (4.5+2.4, 0) arc (0:-180:0.55);
\draw (4.5+2.2, 0) arc (0:-180:0.35);
\draw[red, looseness=0.8] (4.5+0.2, 1.5) to[out=-90, in=90] (4.5-0.9, 0);
\draw[red, looseness=0.8] (4.5+1.3, 1.5) to[out=-90, in=90] (4.5+2.4, 0);
\draw[red, looseness=1.5] (4.5+1.3, 0) to[out=90, in=90] (4.5+0.2, 0);
\draw[looseness=0.8] (4.5+0.4, 1.5) to[out=-90, in=90] (4.5-0.7, 0);
\draw[looseness=0.8] (4.5+1.1, 1.5) to[out=-90, in=90] (4.5+2.2, 0);
\draw[looseness=1.5] (4.5+1.5, 0) to[out=90, in=90] (4.5+0, 0);
%
\draw[ultra thick, gray] (4.5+0.75, 1.5)--(4.5+0.75, 2.8);
\draw (4.5+0.75, 1.5)--(4.5+1.3+0.07, 1.5);
\node at (4.5+0.75, 1.5) [circle, fill, minimum size=4pt, inner sep=1pt] {};
\draw (4.5+1.3+0.07, 1.5) arc (0:360:0.07);
\draw (4.5+1.1+0.07, 1.5) arc (0:360:0.07);
%
\draw[ultra thick, gray, looseness=1.8] (4.5-0.35, 0) to[out=-90, in=-90] (1.85, 0);
\draw (4.5-0.35, 0)--(4.5+0.2+0.07, 0);
\node at (1.85, 0) [circle, fill, minimum size=4pt, inner sep=1pt] {};
\node at (4.5-0.35, 0) [circle, fill, minimum size=4pt, inner sep=1pt] {};
\draw (4.5+0.2+0.07, 0) arc (0:360:0.07);
\draw (4.5+0.07, 0) arc (0:360:0.07);
%
\draw[ultra thick, gray] (4.5+1.85, 0)--(4.5+1.85, -1.3);
\draw (4.5+1.85, 0)--(4.5+2.4+0.07, 0);
\node at (4.5+1.85, 0) [circle, fill, minimum size=4pt, inner sep=1pt] {};
\draw (4.5+2.4+0.07, 0) arc (0:360:0.07);
\draw (4.5+2.2+0.07, 0) arc (0:360:0.07);
\end{tikzpicture}
} = \frac{1}{2} \raisebox{-0.9cm}{
\begin{tikzpicture}[scale=0.5]
\draw[dashed, thick, teal, fill=teal!20, fill opacity=0.4] (-1.5, -0.9) rectangle (2.9, 2.4);
\node at (-1.2, 2.1) {\scriptsize $\downarrow$};
\begin{scope}[even odd rule]
\clip (-1.1, 0) rectangle (2.6, 1.5)  (1.1, 0) arc (0:180:0.35);
\clip (-1.1, 0) rectangle (2.6, 1.5)  (1.1, 1.5) arc (0:-180:0.35);
\fill[blue!10] (-1.1, 0) rectangle (2.6, 1.5);
\end{scope}
\draw [thick, blue] (-1.1, 0)--(-1.1, 1.5);
\draw [thick, blue] (2.6, 0)--(2.6, 1.5);
\draw [thick, blue] (1.1, 0) arc (0:180:0.35);
\draw [thick, blue] (1.1, 1.5) arc (0:-180:0.35);
\draw[thick, blue, fill=blue!10] (0.4, 1.5) arc (0:180:0.75);
\draw[red] (0.2, 1.5) arc (0:180:0.55);
\draw (0, 1.5) arc (0:180:0.35);
\draw[thick, blue, fill=blue!10] (1.1, 1.5) arc (180:0:0.75);
\draw[red] (2.4, 1.5) arc (0:180:0.55);
\draw (2.2, 1.5) arc (0:180:0.35);
\draw[thick, blue, fill=blue!10] (0.4, 0) arc (0:-180:0.75);
\draw[red] (0.2, 0) arc (0:-180:0.55);
\draw (0, 0) arc (0:-180:0.35);
\draw[thick, blue, fill=blue!10] (1.1, 0) arc (-180:0:0.75);
\draw[red] (2.4, 0) arc (0:-180:0.55);
\draw (2.2, 0) arc (0:-180:0.35);
\draw[red] (-0.9, 0)--(-0.9, 1.5);
\draw[red] (2.4, 1.5)--(2.4, 0);
\draw (-0.7, 1.5)--(-0.7, 0);
\draw (2.2, 1.5)--(2.2, 0);
\draw[red, looseness=1.5] (1.3, 0) to[out=90, in=90] (0.2, 0);
\draw[red, looseness=1.5] (1.3, 1.5) to[out=-90, in=-90] (0.2, 1.5);
\draw[looseness=1.5] (1.5, 0) to[out=90, in=90] (0, 0);
\draw[looseness=1.5] (1.5, 1.5) to[out=-90, in=-90] (0, 1.5);
\draw[ultra thick, gray] (-0.35, 1.5)--(-0.35, 2.8);
\draw (-0.35, 1.5)--(0.2+0.07, 1.5);
\node at (-0.35, 1.5) [circle, fill, minimum size=4pt, inner sep=1pt] {};
\draw (0.2+0.07, 1.5) arc (0:360:0.07);
\draw (0.07, 1.5) arc (0:360:0.07);
%
\draw[ultra thick, gray] (1.85, 1.5)--(1.85, 2.8);
\draw (1.85, 1.5)--(2.4+0.07, 1.5);
\node at (1.85, 1.5) [circle, fill, minimum size=4pt, inner sep=1pt] {};
\draw (2.4+0.07, 1.5) arc (0:360:0.07);
\draw (2.2+0.07, 1.5) arc (0:360:0.07);
%
\draw[ultra thick, gray] (-0.35, 0)--(-0.35, -1.3);
\draw (-0.35, 0)--(0.2+0.07, 0);
\node at (-0.35, 0) [circle, fill, minimum size=4pt, inner sep=1pt] {};
\draw (0.2+0.07, 0) arc (0:360:0.07);
\draw (0.07, 0) arc (0:360:0.07);
%
\draw[ultra thick, gray] (1.85, 0)--(1.85, -1.3);
\draw (1.85, 0)--(2.4+0.07, 0);
\node at (1.85, 0) [circle, fill, minimum size=4pt, inner sep=1pt] {};
\draw (2.4+0.07, 0) arc (0:360:0.07);
\draw (2.2+0.07, 0) arc (0:360:0.07);
\end{tikzpicture}
} \,\, , 
\end{align}
where we denote the parity tensor by P, presenting it alongside with its ZX-calculus notation and its Quon diagrammatic representation. Similarly the $n(>4)$-leg parity tensors can be generated by $3$-leg parity tensors. 

The two-qubit gate $e^{i \frac{\theta}{2}} e^{-i \frac{\theta}{2} X \otimes X}$, considered in Sec.~\ref{sec:quon-for-q-states-gates}, can be generated by $2$- and $3$-leg tensors from TABLE~\ref{tab:generating-tensors-quon-diagrams}: 
\begin{align}
&e^{i \frac{\theta}{2}} e^{-i \frac{\theta}{2} X \otimes X} = \frac{1}{2} \,\, 
\raisebox{-1.2cm}{
\begin{tikzpicture}[scale=0.8]
\draw[dashed, thick, teal, fill=teal!20, fill opacity=0.4] (-1.4, -1.4) rectangle (1.4, 1.4);
\node at (-1.2, 1.2) {\scriptsize $\downarrow$};
\begin{scope}[even odd rule]
\clip (-1.15, -0.8) rectangle (1.15, 0.8)  (0.25, 0.8) arc (0:-180:0.25);
\clip (-1.15, -0.8) rectangle (1.15, 0.8)  (0.25, -0.8) arc (0:180:0.25);
\fill [blue!10] (-1.15, -0.8) rectangle (1.15, 0.8);
\end{scope}
\draw[thick, blue] (-1.15, -0.8)--(-1.15, 0.8);
\draw[thick, blue] (1.15, -0.8)--(1.15, 0.8);
\draw[thick, blue] (0.25, 0.8) arc (0:-180:0.25);
\draw[thick, blue] (0.25, -0.8) arc (0:180:0.25);
%
\draw[red] (-1, -0.8) arc (-180:0:0.3);
\draw (-0.8, -0.8) arc (-180:0:0.1);
\fill[gray, fill opacity=0.6, draw=blue, thick] (-1.15, -0.8) arc (-180:0:0.45);
\draw[ultra thick, gray] (-0.7, -1.25)--(-0.7, -1.6);
%
\draw[red] (1, -0.8) arc (0:-180:0.3);
\draw (0.8, -0.8) arc (0:-180:0.1);
\fill[gray, fill opacity=0.6, draw=blue, thick] (1.15, -0.8) arc (0:-180:0.45);
\draw[ultra thick, gray] (0.7, -1.25)--(0.7, -1.6);
%
\draw[red] (1, 0.8) arc (0:180:0.3);
\draw (0.8, 0.8) arc (0:180:0.1);
\fill[gray, fill opacity=0.6, draw=blue, thick] (1.15, 0.8) arc (0:180:0.45);
\draw[ultra thick, gray] (0.7, 1.25)--(0.7, 1.6);
%
\draw[red] (-1, 0.8) arc (180:0:0.3);
\draw (-0.8, 0.8) arc (180:0:0.1);
\fill[gray, fill opacity=0.6, draw=blue, thick] (-1.15, 0.8) arc (180:0:0.45);
\draw[ultra thick, gray] (-0.7, 1.25)--(-0.7, 1.6);
\draw[red] (-1, 0.8)--(-1, -0.8);
\draw (-0.8, 0.8)--(-0.8, -0.8);
\draw (0.8, 0.8)--(0.8, -0.8);
\draw[red] (1, 0.8)--(1, -0.8);
\draw[red, looseness=1.5] (-0.4, 0.8) to[out=-90, in=-90] (0.4, 0.8);
\draw[red, looseness=1.5] (-0.4, -0.8) to[out=90, in=90] (0.4, -0.8);
\draw (-0.6, 0.8) to[out=-90, in=150] (-{0.15*sqrt(2)}, {0.15*sqrt(2)});
\draw (0.6, 0.8) to[out=-90, in=30] (+{0.15*sqrt(2)}, {0.15*sqrt(2)});
\draw (-0.6, -0.8) to[out=90, in=-150] (-{0.15*sqrt(2)}, -{0.15*sqrt(2)});
\draw (0.6, -0.8) to[out=90, in=-30] (+{0.15*sqrt(2)}, -{0.15*sqrt(2)});
\draw (0.3, 0) arc (0:360:0.3);
\node at (0, 0) {\scriptsize $\theta_{\updownarrow}$};
\end{tikzpicture}
} = \frac{1}{2} \,\, 
\raisebox{-1.2cm}{
\begin{tikzpicture}[scale=0.8]
\draw[dashed, thick, teal, fill=teal!20, fill opacity=0.4] (-1.4, -1.4) rectangle (1.4, 1.4);
\node at (-1.2, 1.2) {\scriptsize $\downarrow$};
\begin{scope}[even odd rule]
\clip (-1.15, -0.8) rectangle (1.15, 0.8)  (0.25, 0.8) arc (0:-180:0.25);
\clip (-1.15, -0.8) rectangle (1.15, 0.8)  (0.25, -0.8) arc (0:180:0.25);
\fill [blue!10] (-1.15, -0.8) rectangle (1.15, 0.8);
\end{scope}
\draw[thick, blue] (-1.15, -0.8)--(-1.15, 0.8);
\draw[thick, blue] (1.15, -0.8)--(1.15, 0.8);
\draw[thick, blue] (0.25, 0.8) arc (0:-180:0.25);
\draw[thick, blue] (0.25, -0.8) arc (0:180:0.25);
%
\draw[red] (-1, -0.8) arc (-180:0:0.3);
\draw (-0.8, -0.8) arc (-180:0:0.1);
\fill[gray, fill opacity=0.6, draw=blue, thick] (-1.15, -0.8) arc (-180:0:0.45);
\draw[ultra thick, gray] (-0.7, -1.25)--(-0.7, -1.6);
%
\draw[red] (1, -0.8) arc (0:-180:0.3);
\draw (0.8, -0.8) arc (0:-180:0.1);
\fill[gray, fill opacity=0.6, draw=blue, thick] (1.15, -0.8) arc (0:-180:0.45);
\draw[ultra thick, gray] (0.7, -1.25)--(0.7, -1.6);
%
\draw[red] (1, 0.8) arc (0:180:0.3);
\draw (0.8, 0.8) arc (0:180:0.1);
\fill[gray, fill opacity=0.6, draw=blue, thick] (1.15, 0.8) arc (0:180:0.45);
\draw[ultra thick, gray] (0.7, 1.25)--(0.7, 1.6);
%
\draw[red] (-1, 0.8) arc (180:0:0.3);
\draw (-0.8, 0.8) arc (180:0:0.1);
\fill[gray, fill opacity=0.6, draw=blue, thick] (-1.15, 0.8) arc (180:0:0.45);
\draw[ultra thick, gray] (-0.7, 1.25)--(-0.7, 1.6);
\draw[red] (-1, 0.8)--(-1, -0.8);
\draw (-0.8, 0.8)--(-0.8, -0.8);
\draw (0.8, 0.8)--(0.8, -0.8);
\draw[red] (1, 0.8)--(1, -0.8);
\draw[red, looseness=1.5] (-0.4, 0.8) to[out=-90, in=-90] (0.4, 0.8);
\draw[red, looseness=1.5] (-0.4, -0.8) to[out=90, in=90] (0.4, -0.8);
\draw[looseness=1.7] (-0.6, 0.8) to[out=-90, in=30] (+{0.15*sqrt(2)}, {0.15*sqrt(2)});
\draw (0.6, 0.8) to[out=-90, in=-30] (+{0.15*sqrt(2)}, -{0.15*sqrt(2)});
\draw (-0.6, -0.8) to[out=90, in=150] (-{0.15*sqrt(2)}, +{0.15*sqrt(2)});
\draw[looseness=1.7] (0.6, -0.8) to[out=90, in=-150] (-{0.15*sqrt(2)}, -{0.15*sqrt(2)});
\draw (0.3, 0) arc (0:360:0.3);
\node at (0, 0) {\scriptsize $\theta_{\leftrightarrow}$};
\end{tikzpicture}
} \nonumber \\ 
&= \frac{A}{2} \,\, 
\raisebox{-1.2cm}{
\begin{tikzpicture}[scale=0.8]
\draw[dashed, thick, teal, fill=teal!20, fill opacity=0.4] (-1.5, -1.4) rectangle (1.5, 1.4);
\node at (-1.3, 1.2) {\scriptsize $\downarrow$};
\begin{scope}[even odd rule]
\clip (-1.25, -0.8) rectangle (1.25, 0.8)  (0.25, 0.8) arc (0:-180:0.25);
\clip (-1.25, -0.8) rectangle (1.25, 0.8)  (0.25, -0.8) arc (0:180:0.25);
\fill [blue!10] (-1.25, -0.8) rectangle (1.25, 0.8);
\end{scope}
\draw[thick, blue] (-1.25, -0.8)--(-1.25, 0.8);
\draw[thick, blue] (1.25, -0.8)--(1.25, 0.8);
\draw[thick, blue] (0.25, 0.8) arc (0:-180:0.25);
\draw[thick, blue] (0.25, -0.8) arc (0:180:0.25);
%
\draw[red] (-1.1, -0.8) arc (-180:0:0.35);
\draw (-0.85, -0.8) arc (-180:0:0.1);
\fill[gray, fill opacity=0.6, draw=blue, thick] (-1.25, -0.8) arc (-180:0:0.5);
\draw[ultra thick, gray] (-0.75, -1.25)--(-0.75, -1.6);
%
\draw[red] (1.1, -0.8) arc (0:-180:0.35);
\draw (0.85, -0.8) arc (0:-180:0.1);
\fill[gray, fill opacity=0.6, draw=blue, thick] (1.25, -0.8) arc (0:-180:0.5);
\draw[ultra thick, gray] (0.75, -1.25)--(0.75, -1.6);
%
\draw[red] (1.1, 0.8) arc (0:180:0.35);
\draw (0.85, 0.8) arc (0:180:0.1);
\fill[gray, fill opacity=0.6, draw=blue, thick] (1.25, 0.8) arc (0:180:0.5);
\draw[ultra thick, gray] (0.75, 1.25)--(0.75, 1.6);
%
\draw[red] (-1.1, 0.8) arc (180:0:0.35);
\draw (-0.85, 0.8) arc (180:0:0.1);
\fill[gray, fill opacity=0.6, draw=blue, thick] (-1.25, 0.8) arc (180:0:0.5);
\draw[ultra thick, gray] (-0.75, 1.25)--(-0.75, 1.6);
\draw[red] (-1.1, 0.8)--(-1.1, -0.8);
\draw (-0.85, 0.8)--(-0.85, -0.8);
\draw (0.85, 0.8)--(0.85, -0.8);
\draw[red] (1.1, 0.8)--(1.1, -0.8);
\draw[red, looseness=1.5] (0.4, 0.8) to[out=-90, in=90] (0.5, -0.1) to[out=-90, in=-90] (0.4, -0.1) to[out=90, in=-20] (0, 0.48) to[out=160, in=-90] (-0.4, 0.8);
\draw[red, looseness=1.5] (-0.4, -0.8) to[out=90, in=-90] (-0.5, 0.1) to[out=90, in=90] (-0.4, 0.1) to[out=-90, in=160] (0, -0.48) to[out=-20, in=90] (0.4, -0.8);
\draw[looseness=0.8] (-0.65, 0.8) to[out=-90, in=90] (-0.65, 0.6) to[out=-90, in=70] (+{0.3*cos(45)}, {0.3*sin(45)});
\draw[looseness=1.5] (0.65, 0.8) to[out=-90, in=-45] (+{0.15*sqrt(2)}, -{0.15*sqrt(2)});
\draw[looseness=1.5] (-0.65, -0.8) to[out=90, in=135] (-{0.15*sqrt(2)}, +{0.15*sqrt(2)});
\draw[looseness=0.8] (0.65, -0.8) to[out=90, in=-90] (0.65, -0.6) to[out=90, in=-110] (-{0.15*sqrt(2)}, -{0.15*sqrt(2)});
\draw (0.3, 0) arc (0:360:0.3);
\node at (0, 0) {\scriptsize $\phi_{\updownarrow}$};
\end{tikzpicture}
} = \frac{A}{2} \,\, \raisebox{-1.7cm}{\tikz{
\draw[dashed, thick, teal, fill=teal!20, fill opacity=0.4] (-1.7, -1.6) rectangle (1.7, 1.6);
\node at (-1.55, 1.4) {\scriptsize $\downarrow$};
\fill[blue!10] (-1.5, 1)--(-1.5, -1)--(-0.5, -1)--(-0.5, -0.5) to[out=90, in=90] (0.5, -1)--(1.5, -1)--(1.5, 1)--(0.5, 1)--(0.5, 0.5) to[out=-90, in=-90] (-0.5, 1)--(-1.5, 1);
\draw[thick, blue] (-1.5, 1)--(-1.5, -1);
\draw[thick, blue] (-0.5, -1)--(-0.5, -0.5) to[out=90, in=90] (0.5, -1);
\draw[thick, blue] (0.5, 1)--(0.5, 0.5) to[out=-90, in=-90] (-0.5, 1);
\draw[thick, blue] (1.5, 1)--(1.5, -1);
\draw[red] (-1.3, -1)--(-1.3, 1);
\draw (-1.1, -1)--(-1.1, 1);
\draw (1.1, -1)--(1.1, 1);
\draw[red] (1.3, -1)--(1.3, 1);
\draw[red] (-0.7, -1)--(-0.7, -0.5) to[out=90, in=170] (-0.4, -0.2) to[out=-10, in=140] (0.4, -0.6) to[out=-40, in=90] (0.7, -1);
\draw[red] (0.7, 1)--(0.7, 0.5) to[out=-90, in=10] (0.4, 0.2) to[out=-170, in=-40] (-0.4, 0.6) to[out=140, in=-90] (-0.7, 1);
\draw (0.2, 0) arc (0:360:0.2);
\node at (0, 0) {\rotatebox{60}{\scriptsize $\phi_{\updownarrow}$}};
\draw (-0.9, -1)--(-0.9, -0.3) to[out=90, in=160] (-0.2, 0);
\draw (0.9, 1)--(0.9, 0.3) to[out=-90, in=-20] (0.2, 0);
\draw (-0.9, 1) to[out=-90, in=135] (-0.14, 0.15);
\draw (0.9, -1) to[out=90, in=-45] (0.14, -0.15);
%
%
\draw[red] (-1.3, -1) arc (-180:0:0.3);
\draw (-1.1, -1) arc (-180:0:0.1);
\fill[gray, fill opacity=0.6, draw=blue, thick] (-1.5, -1) arc (-180:0:0.5);
\draw[ultra thick, gray] (-1, -1.5)--(-1, -1.75);
%
\draw[red] (1.3, -1) arc (0:-180:0.3);
\draw (1.1, -1) arc (0:-180:0.1);
\fill[gray, fill opacity=0.6, draw=blue, thick] (1.5, -1) arc (0:-180:0.5);
\draw[ultra thick, gray] (1, -1.5)--(1, -1.75);
%
\draw[red] (1.3, 1) arc (0:180:0.3);
\draw (1.1, 1) arc (0:180:0.1);
\fill[gray, fill opacity=0.6, draw=blue, thick] (1.5, 1) arc (0:180:0.5);
\draw[ultra thick, gray] (1, 1.5)--(1, 1.75);
%
\draw[red] (-1.3, 1) arc (180:0:0.3);
\draw (-1.1, 1) arc (180:0:0.1);
\fill[gray, fill opacity=0.6, draw=blue, thick] (-1.5, 1) arc (180:0:0.5);
\draw[ultra thick, gray] (-1, 1.5)--(-1, 1.75);
}} \nonumber \\ 
&= \frac{A}{8} \,\, \raisebox{-2.7cm}{\tikz{
\draw[dashed, thick, teal, fill=teal!20, fill opacity=0.4] (-1, -0.6) rectangle (1.3, 1.4);
\node at (-0.8, 1.2) {\scriptsize $\downarrow$};
%
\fill[blue!10] (-0.25, 0.8) to[out=-90, in=90] (-0.8, 0)--(0.1, 0) arc (180:0:0.1) -- (1.2, 0) to[out=90, in=-90] (0.65, 0.8);
\draw[thick, blue] (-0.25, 0.8) to[out=-90, in=90] (-0.8, 0);
\draw[thick, blue] (0.1, 0) arc (180:0:0.1);
\draw[thick, blue] (1.2, 0) to[out=90, in=-90] (0.65, 0.8);
%
\draw[red] (-0.05, 0) arc (0:-180:0.3);
\draw (-0.25, 0) arc (0:-180:0.1);
\fill[gray, fill opacity=0.6, draw=blue, thick] (0.1, 0) arc (0:-180:0.45);
\draw[ultra thick, gray] (-0.35, -0.45)--(-0.35, -0.8);
%
\draw[red] (1.05, 0) arc (0:-180:0.3);
\draw (0.85, 0) arc (0:-180:0.1);
\fill[gray, fill opacity=0.6, draw=blue, thick] (1.2, 0) arc (0:-180:0.45);
\draw[ultra thick, gray] (1.1, -0.3)--(1.6, -0.72);
%
\draw[red] (0.5, 0.8) arc (0:180:0.3);
\draw (0.3, 0.8) arc (0:180:0.1);
\fill[gray, fill opacity=0.6, draw=blue, thick] (0.65, 0.8) arc (0:180:0.45);
\draw[ultra thick, gray] (0.2, 1.25)--(0.2, 1.6);
%
\draw[red] (-0.1, 0.8) to[out=-90, in=90] (-0.65, 0);
\draw (0.1, 0.8) to[out=-90, in=90] (-0.45, 0);
\draw (-0.25, 0) to[out=90, in=90] (0.65, 0);
\draw[red, looseness=1.3] (-0.05, 0) to[out=90, in=90] (0.45, 0);
\draw (0.3, 0.8) to[out=-90, in=90] (0.85, 0);
\draw[red] (0.5, 0.8) to[out=-90, in=90] (1.05, 0);
%
%
\draw[dashed, thick, teal, fill=teal!20, fill opacity=0.4] (1.1, -1)--(2.4, -2.3)--(3.25, -1.45)--(1.95, -0.15)--(1.1, -1);
\node[rotate=45] at (1.35, -1.05) {\scriptsize $\downarrow$};
%
\fill[blue!10] (2.2, -0.7)--(2.7, -1.2)--({2.7-0.4*sqrt(2)}, {-1.2-0.4*sqrt(2)})--({2.2-0.4*sqrt(2)}, {-0.7-0.4*sqrt(2)});
\draw[thick, blue] (2.2, -0.7)--(2.7, -1.2);
\draw[thick, blue] ({2.2-0.4*sqrt(2)}, {-0.7-0.4*sqrt(2)})--({2.7-0.4*sqrt(2)}, {-1.2-0.4*sqrt(2)});
%
\draw[red] ({2.2-0.05*sqrt(2)}, {-0.7-0.05*sqrt(2)}) arc (45:225:0.3);
\draw ({2.2-0.1*sqrt(2)}, {-0.7-0.1*sqrt(2)}) arc (45:225:0.2);
\fill[gray, fill opacity=0.6, draw=blue, thick] (2.2, -0.7) arc (45:225:0.4);
%
\draw[red] ({2.7-0.05*sqrt(2)}, {-1.2-0.05*sqrt(2)}) arc (45:-135:0.3);
\draw ({2.7-0.1*sqrt(2)}, {-1.2-0.1*sqrt(2)}) arc (45:-135:0.2);
\fill[gray, fill opacity=0.6, draw=blue, thick] (2.7, -1.2) arc (45:-135:0.4);
%
\draw[red] ({2.2-0.05*sqrt(2)}, {-0.7-0.05*sqrt(2)})--({2.7-0.05*sqrt(2)}, {-1.2-0.05*sqrt(2)});
\draw[red] ({2.2-0.35*sqrt(2)}, {-0.7-0.35*sqrt(2)})--({2.7-0.35*sqrt(2)}, {-1.2-0.35*sqrt(2)});
\draw ({2.2-0.1*sqrt(2)}, {-0.7-0.1*sqrt(2)}) to[out=-45, in=100] (2.13, -1.02);
\draw (2.38, -1.28) to[out=-20, in=135] (+{2.7-0.1*sqrt(2)}, {-1.2-0.1*sqrt(2)});
\draw ({2.2-0.3*sqrt(2)}, {-0.7-0.3*sqrt(2)}) to[out=-45, in=170] (2, -1.2);
\draw (2.2, -1.4) to[out=-60, in=135] (+{2.7-0.3*sqrt(2)}, {-1.2-0.3*sqrt(2)});
\draw (2.4, -1.2) arc (0:360:0.2);
\node at (2.2, -1.2) {\rotatebox{45}{\scriptsize $\phi_{\updownarrow}$}};
%
%
\draw[dashed, thick, teal, fill=teal!20, fill opacity=0.4] (3, -3.8) rectangle (5.4, -1.8);
\node at (3.2, -3.55) {\scriptsize $\downarrow$};
%
\fill[blue!10] (3.75, -3.2) to[out=90, in=-90] (3.2, -2.4)--(4.1, -2.4) arc (-180:0:0.1) -- (5.2, -2.4) to[out=-90, in=90] (4.65, -3.2);
\draw[thick, blue] (3.75, -3.2) to[out=90, in=-90] (3.2, -2.4);
\draw[thick, blue] (4.1, -2.4) arc (-180:0:0.1);
\draw[thick, blue] (5.2, -2.4) to[out=-90, in=90] (4.65, -3.2);
%
\draw[red] (4.5, -3.2) arc (0:-180:0.3);
\draw (4.3, -3.2) arc (0:-180:0.1);
\fill[gray, fill opacity=0.6, draw=blue, thick] (4.65, -3.2) arc (0:-180:0.45);
\draw[ultra thick, gray] (4.2, -3.65)--(4.2, -4);
%
\draw[red] (5.05, -2.4) arc (0:180:0.3);
\draw (4.85, -2.4) arc (0:180:0.1);
\fill[gray, fill opacity=0.6, draw=blue, thick] (5.2, -2.4) arc (0:180:0.45);
\draw[ultra thick, gray] (4.75, -1.95)--(4.75, -1.6);
%
\draw[red] (3.95, -2.4) arc (0:180:0.3);
\draw (3.75, -2.4) arc (0:180:0.1);
\fill[gray, fill opacity=0.6, draw=blue, thick] (4.1, -2.4) arc (0:180:0.45);
\draw[ultra thick, gray] (3.3, -2.1)--(2.7, -1.75);
%
\draw[red] (3.35, -2.4) to[out=-90, in=90] (3.9, -3.2);
\draw (3.55, -2.4) to[out=-90, in=90] (4.1, -3.2);
\draw (3.75, -2.4) to[out=-90, in=-90] (4.65, -2.4);
\draw[red, looseness=1.3] (3.95, -2.4) to[out=-90, in=-90] (4.45, -2.4);
\draw (4.3, -3.2) to[out=90, in=-90] (4.85, -2.4);
\draw[red] (4.5, -3.2) to[out=90, in=-90] (5.05, -2.4);
}} \,\, , 
\end{align}
where we apply the space-time duality in TABLE~\ref{tab:majorana-rewriting-rules} from the first line to the second line, define $A = \frac{1 + e^{i \theta}}{\sqrt{2}}$, and introduce $\phi$ satisfying $e^{i \phi} = \frac{1 - e^{i \theta}}{1 + e^{i \theta}}$.

\section{Efficient Diagrammatic Characterizations of Clifford, Matchgate, and Punctured Matchgate Tensor Networks}
\label{sec:clifford-matchgate-quon-diagrams}
In this section, we present efficient diagrammatic characterizations of Clifford, matchgate, and punctured matchgate tensor networks in terms of the 2D Quon language. We assume that a Clifford tensor network consists solely of Clifford tensors, and similarly, a matchgate or punctured matchgate tensor network consists solely of matchgate tensors. Under these assumptions, constructing a 2D Quon diagram for a given Clifford, matchgate, or punctured matchgate tensor network that satisfies our diagrammatic characterizations is also \textit{efficient}. 

Consequently, we do not directly address the question of testing whether a given tensor network is Clifford or (punctured) matchgate. Efficient testing of this kind is believed to be computationally intractable in general, as even the non-identity check on a given circuit is known to be hard~\cite{janzing2005non, ji2009non}. See recent advances on learning the T-count of T-doped Clifford circuits~\cite{PhysRevLett.133.010601, PhysRevLett.133.020601, leone2024learning} and in testing non-Gaussianity~\cite{lyu2024fermionic, coffman2025measuring}. In contrast, we highlight the potential of \textit{obfuscating} a given quantum circuit or tensor network~\cite{alagic2016quantum, alagic2021impossibility, broadbent2021constructions, bartusek2022indistinguishability, bartusek2023obfuscation, coladangelo2024use, bartusek2024quantum} via the Quon diagrammatic rewriting rules. See Sec.~\ref{subsec:factory} for further discussions along this direction. 

In the following, we show that any Clifford tensor network can be represented by a 2D Quon diagram containing no generic scattering elements. Furthermore, the converse holds: a 2D Quon diagram that does not contain any generic scattering elements represents a Clifford tensor network. This pictorial characterization of Clifford tensor networks was previously discussed using the 3D Quon language in Ref.~\onlinecite{liu2017quon}, and is now presented using the 2D Quon language. We also show that any matchgate tensor network can be represented by a 2D Quon diagram that satisfies the \textit{boundary-tracking} property and has a hole-free background manifold. Conversely, a 2D Quon diagram satisfying these two conditions always represents a matchgate tensor network. Finally, a punctured matchgate tensor network can be represented by a 2D Quon diagram that satisfies the boundary-tracking property, and vice versa. 

\subsection{Clifford Tensor Network using Quon Diagrams}
Suppose that we are given a \textit{Clifford} tensor network, consisting of the tensors from the following Clifford generating set $\{ \vert 0 \rangle, I, e^{\pm i \frac{\pi}{4} X}, e^{\pm i \frac{\pi}{4} Z}, P \}$, each of which is listed in TABLE~\ref{tab:generating-tensors-quon-diagrams}. We can further assume that the tensor network is \textit{planar}, since the SWAP gate---a resource gate for transforming an arbitrary tensor network into a planar one---is itself a Clifford gate and is therefore generated by the Clifford generating set. Using TABLE~\ref{tab:generating-tensors-quon-diagrams}, we observe that the Quon diagrams for the Clifford generating tensors do not contain any generic scattering elements. This observation is consistent with the previous diagrammatic characterizations of Clifford tensor networks in terms of the 3D Quon language~\cite{liu2017quon}. We therefore state the following diagrammatic characterization along with a sketch of its proof. 

\begin{theorem}
\label{thm:quon-clifford}
A planar tensor network is Clifford if and only if it can be represented by a 2D Quon diagram that does not contain any generic scattering elements. 
\end{theorem}

\begin{proof}
First, consider a Clifford tensor network composed of tensors from the generating set $\{ \vert 0 \rangle, I, e^{\pm i \frac{\pi}{4} X}, e^{\pm i \frac{\pi}{4} Z}, P \}$. The 2D Quon diagrams of these generators (see TABLE~\ref{tab:generating-tensors-quon-diagrams}), as well as the planar tensor operations [see FIG.~\ref{fig:quon-tn} (e)-(g)], do not introduce generic scattering elements. Consequently, performing planar tensor operations yields a 2D Quon diagram representation of the given Clifford tensor network with no generic scattering elements. 

To show the converse, consider a 2D Quon diagram that contains no generic scattering elements. By employing the string-genus relation Eq.~\eqref{eq:string-genus} when necessary, one can always group Majoranas into groups of four. The background manifold of the resulting Quon diagram can be interpreted as a thickened graph, where each edge supports $4$ Majoranas. Using TABLE~\ref{tab:generating-tensors-quon-diagrams}, we replace every edge with a $2$-leg Clifford tensor. Then the vertices in the thickened graph can be substituted by the parity tensors of appropriate degree, which are themselves Clifford tensors. Therefore, we convert the Quon diagram into a Clifford tensor network. 
\end{proof}

\subsection{Matchgate and Punctured Matchgate Tensor Networks using Quon Diagrams}
\label{sec:matchgate-TN-using-quon}
As explained in Sec.~\ref{sec:planar-TN}, a matchgate (punctured matchgate) tensor network can be expressed as a planar tensor network in which the constituent tensors are matchgates and the planar region of the network is topologically equivalent to a disk (punctured disk). In the following, we provide diagrammatic characterizations of matchgate and punctured matchgate tensor networks using the 2D Quon language. 

\subsubsection{Diagrammatic characterizations of non-punctured matchgate}
Here, we show that a matchgate tensor network admits a 2D Quon diagrammatic representation which (i) satisfies the \textit{boundary-tracking} property and (ii) has a hole-free background manifold. The \textit{boundary-tracking} property asserts that for every closed interval along the boundary of the background manifold, there exists an isolated Majorana string spatially adjacent to it, tracing the closed interval without scattering or braiding with other Majorana strings. See FIG.~\ref{fig:matchgate-example} (a) for an example and non-example of Quon diagrams exhibiting the boundary-tracking property. As shown in TABLE~\ref{tab:generating-tensors-quon-diagrams}, each tensor in the following matchgate generating set $\{ \vert 0 \rangle, I, X, P \} \cup \{ e^{i \frac{\theta}{2}} e^{-i \frac{\theta}{2} Z} \}_{\theta \in [0, 2\pi)}$ admits a 2D Quon diagrammatic representation that satisfies the boundary-tracking property---with the corresponding boundary-tracking Majorana strings highlighted in red---and the background manifold is hole-free. In contrast, the non-matchgate tensor $e^{-\frac{i \pi}{4} X}$ violates the boundary-tracking property, as is evident in TABLE~\ref{tab:generating-tensors-quon-diagrams}. 

Suppose our matchgate tensor network contains loops when the network is viewed as a graph. Then, for each loop, we introduce a hole in the background manifold after performing tensor contractions. This seems to violate the hole-free background manifold condition. However, due to the boundary-tracking property, any newly generated hole can be removed, along with its enclosing boundary-tracking Majorana loop, via the string-genus relation Eq.~\eqref{eq:string-genus}, as exemplified below: 
\begin{align}
\label{eq:matchgate-loop-example}
&\raisebox{-0.6cm}{\tikz{
\draw [dashed, thick, teal, fill=teal!20, fill opacity=0.4] (0.85, 0.2) arc (0:360:0.85);
\draw[ultra thick, gray] (0, 0.7)--(0, 1.3);
\draw[ultra thick, gray] (0, 0.7)--(0.5, 0);
\draw[ultra thick, gray] (0, 0.7)--(-0.5, 0);
\draw[ultra thick, gray] (0.5, 0)--(-0.5, 0);
\draw[ultra thick, gray] (-0.5, 0)--(-1, -0.3);
\draw[ultra thick, gray] (0.5, 0)--(1, -0.3);
\draw [thick, gray, fill=white] (0.2, 0.7) arc (0:360:0.2);
\draw [thick, gray, fill=white] (0.7, 0) arc (0:360:0.2);
\draw [thick, gray, fill=white] (-0.3, 0) arc (0:360:0.2);
\node at (0, 0.7) {\scriptsize $P$};
\node at (0.5, 0) {\scriptsize $P$};
\node at (-0.5, 0) {\scriptsize $P$};
}} \,\, = \frac{1}{8} \,\, \raisebox{-2.1cm}{
\begin{tikzpicture}[scale=0.5]
\coordinate (A) at (0, 0);
%
\draw[dashed, thick, teal, fill=teal!20, fill opacity=0.4] ($(-1.3, -0.9) + (A)$) rectangle ($(2.8, 2.4) + (A)$);
\node at ($(-1, 2)+(A)$) {\scriptsize $\downarrow$};
\begin{scope}[even odd rule]
\clip ($(0, 0) + (A)$) rectangle ($(1.5, 1.5) + (A)$)  ($(1.1, 0) + (A)$) arc (0:180:0.35);
\fill[blue!10] ($(0, 0) + (A)$) rectangle ($(1.5, 1.5) + (A)$);
\end{scope}
\fill[blue!10] ($(0, 1.5) + (A)$) to[out=-90, in=90] ($(-1.1, 0) + (A)$)--($(0, 0) + (A)$)--($(0, 1.5) + (A)$);
\fill[blue!10] ($(1.5, 1.5) + (A)$) to[out=-90, in=90] ($(2.6, 0) + (A)$)--($(1.5, 0) + (A)$)--($(1.5, 1.5) + (A)$);
\draw[thick, blue] ($(0, 1.5) + (A)$) to[out=-90, in=90] ($(-1.1, 0) + (A)$);
\draw[thick, blue] ($(1.5, 1.5) + (A)$) to[out=-90, in=90] ($(2.6, 0) + (A)$);
\draw[thick, blue] ($(1.1, 0) + (A)$) arc (0:180:0.35);
\draw[thick, blue, fill=blue!10] ($(1.5, 1.5) + (A)$) arc (0:180:0.75);
\draw[red] ($(1.3, 1.5) + (A)$) arc (0:180:0.55);
\draw ($(1.1, 1.5) + (A)$) arc (0:180:0.35);
\draw[thick, blue, fill=blue!10] ($(0.4, 0) + (A)$) arc (0:-180:0.75);
\draw[red] ($(0.2, 0) + (A)$) arc (0:-180:0.55);
\draw ($(0, 0) + (A)$) arc (0:-180:0.35);
\draw[thick, blue, fill=blue!10] ($(1.1, 0) + (A)$) arc (-180:0:0.75);
\draw[red] ($(2.4, 0) + (A)$) arc (0:-180:0.55);
\draw ($(2.2, 0) + (A)$) arc (0:-180:0.35);
\draw[red, looseness=0.8] ($(0.2, 1.5) + (A)$) to[out=-90, in=90] ($(-0.9, 0) + (A)$);
\draw[red, looseness=0.8] ($(1.3, 1.5) + (A)$) to[out=-90, in=90] ($(2.4, 0) + (A)$);
\draw[red, looseness=1.5] ($(1.3, 0) + (A)$) to[out=90, in=90] ($(0.2, 0) + (A)$);
\draw[looseness=0.8] ($(0.4, 1.5) + (A)$) to[out=-90, in=90] ($(-0.7, 0) + (A)$);
\draw[looseness=0.8] ($(1.1, 1.5) + (A)$) to[out=-90, in=90] ($(2.2, 0) + (A)$);
\draw[looseness=1.5] ($(1.5, 0) + (A)$) to[out=90, in=90] ($(0, 0) + (A)$);
%
\draw[ultra thick, gray] ($(0.75, 1.5) + (A)$)--($(0.75, 2.8) + (A)$);
\draw ($(0.75, 1.5) + (A)$)--($(1.3+0.07, 1.5) + (A)$);
\node at ($(0.75, 1.5) + (A)$) [circle, fill, minimum size=4pt, inner sep=1pt] {};
\draw ($(1.3+0.07, 1.5) + (A)$) arc (0:360:0.07);
\draw ($(1.1+0.07, 1.5) + (A)$) arc (0:360:0.07);
%
\draw ($(-0.35, 0) + (A)$)--($(0.2+0.07, 0) + (A)$);
%
\draw ($(0.2+0.07, 0) + (A)$) arc (0:360:0.07);
\draw ($(0.07, 0) + (A)$) arc (0:360:0.07);
%
\draw ($(1.85, 0) + (A)$)--($(2.4+0.07, 0) + (A)$);
%
\draw ($(2.4+0.07, 0) + (A)$) arc (0:360:0.07);
\draw ($(2.2+0.07, 0) + (A)$) arc (0:360:0.07);
%
%
\coordinate (A) at (-2.3, -3.6);
%
\draw[dashed, thick, teal, fill=teal!20, fill opacity=0.4] ($(-1.3, -0.9) + (A)$) rectangle ($(2.8, 2.4) + (A)$);
\node at ($(-1, 2)+(A)$) {\scriptsize $\downarrow$};
\begin{scope}[even odd rule]
\clip ($(0, 0) + (A)$) rectangle ($(1.5, 1.5) + (A)$)  ($(1.1, 0) + (A)$) arc (0:180:0.35);
\fill[blue!10] ($(0, 0) + (A)$) rectangle ($(1.5, 1.5) + (A)$);
\end{scope}
\fill[blue!10] ($(0, 1.5) + (A)$) to[out=-90, in=90] ($(-1.1, 0) + (A)$)--($(0, 0) + (A)$)--($(0, 1.5) + (A)$);
\fill[blue!10] ($(1.5, 1.5) + (A)$) to[out=-90, in=90] ($(2.6, 0) + (A)$)--($(1.5, 0) + (A)$)--($(1.5, 1.5) + (A)$);
\draw[thick, blue] ($(0, 1.5) + (A)$) to[out=-90, in=90] ($(-1.1, 0) + (A)$);
\draw[thick, blue] ($(1.5, 1.5) + (A)$) to[out=-90, in=90] ($(2.6, 0) + (A)$);
\draw[thick, blue] ($(1.1, 0) + (A)$) arc (0:180:0.35);
\draw[thick, blue, fill=blue!10] ($(1.5, 1.5) + (A)$) arc (0:180:0.75);
\draw[red] ($(1.3, 1.5) + (A)$) arc (0:180:0.55);
\draw ($(1.1, 1.5) + (A)$) arc (0:180:0.35);
\draw[thick, blue, fill=blue!10] ($(0.4, 0) + (A)$) arc (0:-180:0.75);
\draw[red] ($(0.2, 0) + (A)$) arc (0:-180:0.55);
\draw ($(0, 0) + (A)$) arc (0:-180:0.35);
\draw[thick, blue, fill=blue!10] ($(1.1, 0) + (A)$) arc (-180:0:0.75);
\draw[red] ($(2.4, 0) + (A)$) arc (0:-180:0.55);
\draw ($(2.2, 0) + (A)$) arc (0:-180:0.35);
\draw[red, looseness=0.8] ($(0.2, 1.5) + (A)$) to[out=-90, in=90] ($(-0.9, 0) + (A)$);
\draw[red, looseness=0.8] ($(1.3, 1.5) + (A)$) to[out=-90, in=90] ($(2.4, 0) + (A)$);
\draw[red, looseness=1.5] ($(1.3, 0) + (A)$) to[out=90, in=90] ($(0.2, 0) + (A)$);
\draw[looseness=0.8] ($(0.4, 1.5) + (A)$) to[out=-90, in=90] ($(-0.7, 0) + (A)$);
\draw[looseness=0.8] ($(1.1, 1.5) + (A)$) to[out=-90, in=90] ($(2.2, 0) + (A)$);
\draw[looseness=1.5] ($(1.5, 0) + (A)$) to[out=90, in=90] ($(0, 0) + (A)$);
%
\draw ($(0.75, 1.5) + (A)$)--($(1.3+0.07, 1.5) + (A)$);
%
\draw ($(1.3+0.07, 1.5) + (A)$) arc (0:360:0.07);
\draw ($(1.1+0.07, 1.5) + (A)$) arc (0:360:0.07);
%
\draw[ultra thick, gray] ($(-0.35, 0) + (A)$)--($(-0.35, -1.3) + (A)$);
\draw ($(-0.35, 0) + (A)$)--($(0.2+0.07, 0) + (A)$);
\node at ($(-0.35, 0) + (A)$) [circle, fill, minimum size=4pt, inner sep=1pt] {};
\draw ($(0.2+0.07, 0) + (A)$) arc (0:360:0.07);
\draw ($(0.07, 0) + (A)$) arc (0:360:0.07);
%
\draw ($(1.85, 0) + (A)$)--($(2.4+0.07, 0) + (A)$);
%
\draw ($(2.4+0.07, 0) + (A)$) arc (0:360:0.07);
\draw ($(2.2+0.07, 0) + (A)$) arc (0:360:0.07);
%
%
\coordinate (A) at (2.3, -3.6);
%
\draw[dashed, thick, teal, fill=teal!20, fill opacity=0.4] ($(-1.3, -0.9) + (A)$) rectangle ($(2.8, 2.4) + (A)$);
\node at ($(-1, 2)+(A)$) {\scriptsize $\downarrow$};
\begin{scope}[even odd rule]
\clip ($(0, 0) + (A)$) rectangle ($(1.5, 1.5) + (A)$)  ($(1.1, 0) + (A)$) arc (0:180:0.35);
\fill[blue!10] ($(0, 0) + (A)$) rectangle ($(1.5, 1.5) + (A)$);
\end{scope}
\fill[blue!10] ($(0, 1.5) + (A)$) to[out=-90, in=90] ($(-1.1, 0) + (A)$)--($(0, 0) + (A)$)--($(0, 1.5) + (A)$);
\fill[blue!10] ($(1.5, 1.5) + (A)$) to[out=-90, in=90] ($(2.6, 0) + (A)$)--($(1.5, 0) + (A)$)--($(1.5, 1.5) + (A)$);
\draw[thick, blue] ($(0, 1.5) + (A)$) to[out=-90, in=90] ($(-1.1, 0) + (A)$);
\draw[thick, blue] ($(1.5, 1.5) + (A)$) to[out=-90, in=90] ($(2.6, 0) + (A)$);
\draw[thick, blue] ($(1.1, 0) + (A)$) arc (0:180:0.35);
\draw[thick, blue, fill=blue!10] ($(1.5, 1.5) + (A)$) arc (0:180:0.75);
\draw[red] ($(1.3, 1.5) + (A)$) arc (0:180:0.55);
\draw ($(1.1, 1.5) + (A)$) arc (0:180:0.35);
\draw[thick, blue, fill=blue!10] ($(0.4, 0) + (A)$) arc (0:-180:0.75);
\draw[red] ($(0.2, 0) + (A)$) arc (0:-180:0.55);
\draw ($(0, 0) + (A)$) arc (0:-180:0.35);
\draw[thick, blue, fill=blue!10] ($(1.1, 0) + (A)$) arc (-180:0:0.75);
\draw[red] ($(2.4, 0) + (A)$) arc (0:-180:0.55);
\draw ($(2.2, 0) + (A)$) arc (0:-180:0.35);
\draw[red, looseness=0.8] ($(0.2, 1.5) + (A)$) to[out=-90, in=90] ($(-0.9, 0) + (A)$);
\draw[red, looseness=0.8] ($(1.3, 1.5) + (A)$) to[out=-90, in=90] ($(2.4, 0) + (A)$);
\draw[red, looseness=1.5] ($(1.3, 0) + (A)$) to[out=90, in=90] ($(0.2, 0) + (A)$);
\draw[looseness=0.8] ($(0.4, 1.5) + (A)$) to[out=-90, in=90] ($(-0.7, 0) + (A)$);
\draw[looseness=0.8] ($(1.1, 1.5) + (A)$) to[out=-90, in=90] ($(2.2, 0) + (A)$);
\draw[looseness=1.5] ($(1.5, 0) + (A)$) to[out=90, in=90] ($(0, 0) + (A)$);
%
\draw ($(0.75, 1.5) + (A)$)--($(1.3+0.07, 1.5) + (A)$);
%
\draw ($(1.3+0.07, 1.5) + (A)$) arc (0:360:0.07);
\draw ($(1.1+0.07, 1.5) + (A)$) arc (0:360:0.07);
%
\draw ($(-0.35, 0) + (A)$)--($(0.2+0.07, 0) + (A)$);
%
\draw ($(0.2+0.07, 0) + (A)$) arc (0:360:0.07);
\draw ($(0.07, 0) + (A)$) arc (0:360:0.07);
%
\draw[ultra thick, gray] ($(1.85, 0) + (A)$)--($(1.85, -1.3) + (A)$);
\draw ($(1.85, 0) + (A)$)--($(2.4+0.07, 0) + (A)$);
\node at ($(1.85, 0) + (A)$) [circle, fill, minimum size=4pt, inner sep=1pt] {};
\draw ($(2.4+0.07, 0) + (A)$) arc (0:360:0.07);
\draw ($(2.2+0.07, 0) + (A)$) arc (0:360:0.07);
%
%
\draw[ultra thick, gray] (-0.35, 0)--(0.75-2.3, 1.5-3.6);
\draw[ultra thick, gray] (1.85, 0)--(0.75+2.3, 1.5-3.6);
\draw[ultra thick, gray, looseness=1.8] (1.85-2.3, 0-3.6) to[out=-90, in=-90] (-0.35+2.3, 0-3.6);
%
\node at (-0.35, 0) [circle, fill, minimum size=4pt, inner sep=1pt] {};
\node at (1.85, 0) [circle, fill, minimum size=4pt, inner sep=1pt] {};
\node at (0.75-2.3, 1.5-3.6) [circle, fill, minimum size=4pt, inner sep=1pt] {};
\node at (0.75+2.3, 1.5-3.6) [circle, fill, minimum size=4pt, inner sep=1pt] {};
\node at (1.85-2.3, 0-3.6) [circle, fill, minimum size=4pt, inner sep=1pt] {};
\node at (-0.35+2.3, 0-3.6) [circle, fill, minimum size=4pt, inner sep=1pt] {};
\end{tikzpicture}
} \nonumber \\ 
& = \raisebox{-1.8cm}{
\begin{tikzpicture}[scale=0.8]
\coordinate (A) at (0, 0);
%
\draw[dashed, thick, teal, fill=teal!20, fill opacity=0.4] ($(-1.3, -0.9) + (A)$) rectangle ($(2.8, 2.9) + (A)$);
\node at ($(-1, 2.6)+(A)$) {\scriptsize $\downarrow$};
\begin{scope}[even odd rule]
\clip ($(0, 0) + (A)$) rectangle ($(1.5, 2) + (A)$)  ($(1.1, 0) + (A)$) arc (0:180:0.35);
\clip ($(0, 0) + (A)$) rectangle ($(1.5, 2) + (A)$)  (1, 1.2) arc (0:360:0.25);
\fill[blue!10] ($(0, 0) + (A)$) rectangle ($(1.5, 2) + (A)$);
\end{scope}
\fill[blue!10] ($(0, 2) + (A)$) to[out=-90, in=90] ($(-1.1, 0) + (A)$)--($(0, 0) + (A)$)--($(0, 2) + (A)$);
\fill[blue!10] ($(1.5, 2) + (A)$) to[out=-90, in=90] ($(2.6, 0) + (A)$)--($(1.5, 0) + (A)$)--($(1.5, 2) + (A)$);
\draw[thick, blue] ($(0, 2) + (A)$) to[out=-90, in=90] ($(-1.1, 0) + (A)$);
\draw[thick, blue] ($(1.5, 2) + (A)$) to[out=-90, in=90] ($(2.6, 0) + (A)$);
\draw[thick, blue] ($(1.1, 0) + (A)$) arc (0:180:0.35);
\draw[thick, blue] (1, 1.2) arc (0:360:0.25);
\draw[red] (1.1, 1.2) arc (0:360:0.35);
\draw (1.2, 1.2) arc (0:360:0.45);
\draw[thick, blue, fill=blue!10] ($(1.5, 2) + (A)$) arc (0:180:0.75);
\draw[red] ($(1.3, 2) + (A)$) arc (0:180:0.55);
\draw ($(1.1, 2) + (A)$) arc (0:180:0.35);
\draw[thick, blue, fill=blue!10] ($(0.4, 0) + (A)$) arc (0:-180:0.75);
\draw[red] ($(0.2, 0) + (A)$) arc (0:-180:0.55);
\draw ($(0, 0) + (A)$) arc (0:-180:0.35);
\draw[thick, blue, fill=blue!10] ($(1.1, 0) + (A)$) arc (-180:0:0.75);
\draw[red] ($(2.4, 0) + (A)$) arc (0:-180:0.55);
\draw ($(2.2, 0) + (A)$) arc (0:-180:0.35);
\draw[red, looseness=0.8] ($(0.2, 2) + (A)$) to[out=-90, in=90] ($(-0.9, 0) + (A)$);
\draw[red, looseness=0.8] ($(1.3, 2) + (A)$) to[out=-90, in=90] ($(2.4, 0) + (A)$);
\draw[red, looseness=1.5] ($(1.3, 0) + (A)$) to[out=90, in=90] ($(0.2, 0) + (A)$);
\draw[looseness=0.8] ($(0.4, 2) + (A)$) to[out=-90, in=90] ($(-0.7, 0) + (A)$);
\draw[looseness=0.8] ($(1.1, 2) + (A)$) to[out=-90, in=90] ($(2.2, 0) + (A)$);
\draw[looseness=1.5] ($(1.5, 0) + (A)$) to[out=90, in=90] ($(0, 0) + (A)$);
%
\draw[ultra thick, gray] ($(0.75, 2) + (A)$)--($(0.75, 3.3) + (A)$);
\draw ($(0.75, 2) + (A)$)--($(1.3+0.07, 2) + (A)$);
\node at ($(0.75, 2) + (A)$) [circle, fill, minimum size=4pt, inner sep=1pt] {};
\draw ($(1.3+0.07, 2) + (A)$) arc (0:360:0.07);
\draw ($(1.1+0.07, 2) + (A)$) arc (0:360:0.07);
%
\draw[ultra thick, gray] ($(-0.35, 0) + (A)$)--($(-0.35, -1.3) + (A)$);
\draw ($(-0.35, 0) + (A)$)--($(0.2+0.07, 0) + (A)$);
\node at ($(-0.35, 0) + (A)$) [circle, fill, minimum size=4pt, inner sep=1pt] {};
\draw ($(0.2+0.07, 0) + (A)$) arc (0:360:0.07);
\draw ($(0.07, 0) + (A)$) arc (0:360:0.07);
%
\draw[ultra thick, gray] ($(1.85, 0) + (A)$)--($(1.85, -1.3) + (A)$);
\draw ($(1.85, 0) + (A)$)--($(2.4+0.07, 0) + (A)$);
\node at ($(1.85, 0) + (A)$) [circle, fill, minimum size=4pt, inner sep=1pt] {};
\draw ($(2.4+0.07, 0) + (A)$) arc (0:360:0.07);
\draw ($(2.2+0.07, 0) + (A)$) arc (0:360:0.07);
\end{tikzpicture}
} = \raisebox{-1.8cm}{
\begin{tikzpicture}[scale=0.8]
\coordinate (A) at (0, 0);
%
\draw[dashed, thick, teal, fill=teal!20, fill opacity=0.4] ($(-1.3, -0.9) + (A)$) rectangle ($(2.8, 2.9) + (A)$);
\node at ($(-1, 2.6)+(A)$) {\scriptsize $\downarrow$};
\begin{scope}[even odd rule]
\clip ($(0, 0) + (A)$) rectangle ($(1.5, 2) + (A)$)  ($(1.1, 0) + (A)$) arc (0:180:0.35);
\fill[blue!10] ($(0, 0) + (A)$) rectangle ($(1.5, 2) + (A)$);
\end{scope}
\fill[blue!10] ($(0, 2) + (A)$) to[out=-90, in=90] ($(-1.1, 0) + (A)$)--($(0, 0) + (A)$)--($(0, 2) + (A)$);
\fill[blue!10] ($(1.5, 2) + (A)$) to[out=-90, in=90] ($(2.6, 0) + (A)$)--($(1.5, 0) + (A)$)--($(1.5, 2) + (A)$);
\draw[thick, blue] ($(0, 2) + (A)$) to[out=-90, in=90] ($(-1.1, 0) + (A)$);
\draw[thick, blue] ($(1.5, 2) + (A)$) to[out=-90, in=90] ($(2.6, 0) + (A)$);
\draw[thick, blue] ($(1.1, 0) + (A)$) arc (0:180:0.35);
%
%
\draw[thick, blue, fill=blue!10] ($(1.5, 2) + (A)$) arc (0:180:0.75);
\draw[red] ($(1.3, 2) + (A)$) arc (0:180:0.55);
\draw ($(1.1, 2) + (A)$) arc (0:180:0.35);
\draw[thick, blue, fill=blue!10] ($(0.4, 0) + (A)$) arc (0:-180:0.75);
\draw[red] ($(0.2, 0) + (A)$) arc (0:-180:0.55);
\draw ($(0, 0) + (A)$) arc (0:-180:0.35);
\draw[thick, blue, fill=blue!10] ($(1.1, 0) + (A)$) arc (-180:0:0.75);
\draw[red] ($(2.4, 0) + (A)$) arc (0:-180:0.55);
\draw ($(2.2, 0) + (A)$) arc (0:-180:0.35);
\draw[red, looseness=0.8] ($(0.2, 2) + (A)$) to[out=-90, in=90] ($(-0.9, 0) + (A)$);
\draw[red, looseness=0.8] ($(1.3, 2) + (A)$) to[out=-90, in=90] ($(2.4, 0) + (A)$);
\draw[red, looseness=1.5] ($(1.3, 0) + (A)$) to[out=90, in=90] ($(0.2, 0) + (A)$);
\draw[looseness=0.8] ($(0.4, 2) + (A)$) to[out=-90, in=90] ($(-0.7, 0) + (A)$);
\draw[looseness=0.8] ($(1.1, 2) + (A)$) to[out=-90, in=90] ($(2.2, 0) + (A)$);
\draw[looseness=1.5] ($(1.5, 0) + (A)$) to[out=90, in=90] ($(0, 0) + (A)$);
%
\draw[ultra thick, gray] ($(0.75, 2) + (A)$)--($(0.75, 3.3) + (A)$);
\draw ($(0.75, 2) + (A)$)--($(1.3+0.07, 2) + (A)$);
\node at ($(0.75, 2) + (A)$) [circle, fill, minimum size=4pt, inner sep=1pt] {};
\draw ($(1.3+0.07, 2) + (A)$) arc (0:360:0.07);
\draw ($(1.1+0.07, 2) + (A)$) arc (0:360:0.07);
%
\draw[ultra thick, gray] ($(-0.35, 0) + (A)$)--($(-0.35, -1.3) + (A)$);
\draw ($(-0.35, 0) + (A)$)--($(0.2+0.07, 0) + (A)$);
\node at ($(-0.35, 0) + (A)$) [circle, fill, minimum size=4pt, inner sep=1pt] {};
\draw ($(0.2+0.07, 0) + (A)$) arc (0:360:0.07);
\draw ($(0.07, 0) + (A)$) arc (0:360:0.07);
%
\draw[ultra thick, gray] ($(1.85, 0) + (A)$)--($(1.85, -1.3) + (A)$);
\draw ($(1.85, 0) + (A)$)--($(2.4+0.07, 0) + (A)$);
\node at ($(1.85, 0) + (A)$) [circle, fill, minimum size=4pt, inner sep=1pt] {};
\draw ($(2.4+0.07, 0) + (A)$) arc (0:360:0.07);
\draw ($(2.2+0.07, 0) + (A)$) arc (0:360:0.07);
\end{tikzpicture}
} \,\, . 
\end{align}
These observations naturally lead to the following diagrammatic characterizations of the matchgate tensor networks. 

\begin{theorem}
\label{thm:quon-matchgate}
A planar tensor network is \textbf{matchgate} if and only if it has a 2D Quon diagrammatic representation that (i) satisfies the boundary-tracking property and (ii) has a hole-free background manifold. 
\end{theorem}

\begin{proof} 
Suppose that we are given a matchgate tensor network, expressed as a planar tensor network consisting of tensors in an elementary matchgate generating set $\{ \vert 0 \rangle, I, X, P \} \cup \{ e^{i \frac{\theta}{2}} e^{-i \frac{\theta}{2} Z} \}_{\theta \in [0, 2\pi)}$. First, we note that each generating tensor can be represented by a 2D Quon diagram satisfying both conditions (i) and (ii), as explicitly presented in TABLE~\ref{tab:generating-tensors-quon-diagrams}. When viewed as a graph, the tensor network may contain a loop, which introduces a hole in the background manifold of the Quon diagram. While this would violate the condition (ii), the boundary-tracking property (condition (i)) allows any newly introduced hole to be removed via the string-genus relation Eq.~\eqref{eq:string-genus}, thereby restoring the condition (ii), as demonstrated in Eq.~\eqref{eq:matchgate-loop-example}. 

To prove the converse, let us consider a 2D Quon diagram satisfying conditions (i) and (ii). The proof is essentially consist of parsing the diagram and identifying each component as a tensor from the matchgate generating set. We further assume that the Majorana diagram and the background manifold each has a single connected component, so that the latter is topologically equivalent to a circle. If there are multiple connected components, we can address each one individually. The boundary-tracking property implies a unique boundary-tracking Majorana line that follows the sole boundary of the background manifold. Thus, the Quon diagram takes the following form: 
\begin{equation}
\raisebox{-1cm}{
\begin{tikzpicture}
\draw[dashed, thick, teal, fill=teal!20, fill opacity=0.4] (-0.3, -0.4) rectangle (3.3, 1.4);
\node at (-0.1, 1.2) {\scriptsize $\downarrow$};
%
\fill[blue!10] (0, 0) rectangle (3, 1);
\draw[thick, blue] (0, 0)--(0, 1);
\draw[thick, blue] (3, 0)--(3, 1);
\draw[thick, blue] (0, 0)--(0.3, 0);
\draw[thick, blue] (0.8, 0)--(1.25, 0);
\draw[thick, blue] (1.75, 0)--(2.2, 0);
\draw[thick, blue] (2.7, 0)--(3, 0);
\draw[thick, blue] (0, 1)--(0.5, 1);
\draw[thick, blue] (1, 1)--(2, 1);
\draw[thick, blue] (2.5, 1)--(3, 1);
%
\fill[gray, fill opacity=0.7, draw=blue, thick] (0.3, 0) arc (180:360:0.25);
\fill[gray, fill opacity=0.7, draw=blue, thick] (1.25, 0) arc (180:360:0.25);
\fill[gray, fill opacity=0.7, draw=blue, thick] (2.2, 0) arc (180:360:0.25);
\fill[gray, fill opacity=0.7, draw=blue, thick] (0.5, 1) arc (180:0:0.25);
\fill[gray, fill opacity=0.7, draw=blue, thick] (2, 1) arc (180:0:0.25);
%
\draw[red] (0.4, 0) to[out=90, in=-90] (0.1, 0.25)--(0.1, 0.7) to[out=90, in=-90] (0.6, 1);
\draw[red] (0.7, 0) to[out=90, in=90] (1.35, 0);
\draw[red] (0.4, 0) arc (-180:0:0.15);
\draw (0.5, 0)--(0.5, 0.25);
\draw (0.6, 0)--(0.6, 0.25);
\draw (0.5, 0) arc (-180:0:0.05);
%
\draw[red] (1.65, 0) to[out=90, in=90] (2.3, 0);
\draw[red] (1.35, 0) arc (-180:0:0.15);
\draw (1.45, 0)--(1.45, 0.25);
\draw (1.55, 0)--(1.55, 0.25);
\draw (1.45, 0) arc (-180:0:0.05);
%
\draw[red] (2.6, 0) to[out=90, in=-90] (2.9, 0.25)--(2.9, 0.7) to[out=90, in=-90] (2.4, 1);
\draw[red] (2.3, 0) arc (-180:0:0.15);
\draw (2.4, 0)--(2.4, 0.25);
\draw (2.5, 0)--(2.5, 0.25);
\draw (2.4, 0) arc (-180:0:0.05);
%
\draw[red, looseness=0.5] (2.1, 1) to[out=-90, in=-90] (0.9, 1);
\draw[red] (2.1, 1) arc (180:0:0.15);
\draw (2.2, 0.75)--(2.2, 1);
\draw (2.3, 0.75)--(2.3, 1);
\draw (2.2, 1) arc (180:0:0.05);
%
\draw[red] (0.6, 1) arc (180:0:0.15);
\draw (0.7, 0.75)--(0.7, 1);
\draw (0.8, 0.75)--(0.8, 1);
\draw (0.7, 1) arc (180:0:0.05);
%
\node at (1.5, 0.5) {$f$};
\draw (0.2, 0.25) rectangle (2.8, 0.75);
%
\draw[ultra thick, gray] (0.55, -0.25)--(0.55, -0.8);
\draw[ultra thick, gray] (1.5, -0.25)--(1.5, -0.8);
\draw[ultra thick, gray] (2.45, -0.25)--(2.45, -0.8);
\draw[ultra thick, gray] (2.25, 1.25)--(2.25, 1.8);
\draw[ultra thick, gray] (0.75, 1.25)--(0.75, 1.8);
\end{tikzpicture}
} \,\, , 
\end{equation}
where $f$ denotes a Majorana diagram. Focusing on the Majorana diagram, an example of $f$, including the outermost boundary-tracking Majorana line, is depicted as: 
\begin{equation}
\raisebox{-2.2cm}{
\begin{tikzpicture}[scale=2.3]

\draw[fill=gray!10, even odd rule] (-0.2, 0)--(-0.2, 0.3) to[out=90, in=90] (0, 0.3)--(0, 0.1) to[out=-90, in=170] (0.4, -0.1) to[out=-10, in=90] (1.4, -0.7) to[out=-90, in=90] (1.2, -1.2)--(1.2, -1.5) to[out=-90, in=-90] (1, -1.5)--(1, -1.2) to[out=90, in=0] (-0.6, -0.2) to[out=180, in=90] (-0.8, -0.8) to[out=-90, in=90] (-0.5, -1.2)--(-0.5, -1.5) to[out=-90, in=-90] (-0.3, -1.5)--(-0.3, -1.2) to[out=90, in=-90] (0.7, 0)--(0.7, 0.3) to[out=90, in=90] (0.9, 0.3)--(0.9, 0) to[out=-90, in=0] (0.4, -1) to[out=180, in=90] (-0.6, -0.8) to[out=-90, in=90] (0.2, -1.2)--(0.2, -1.5) to[out=-90, in=-90] (0.4, -1.5)--(0.4, -1.2) to[out=90, in=-90] (-0.2, 0);

\draw[thick, dotted] (-1, -1.3) rectangle (1.6, 0.1);

\draw[red] (-0.25, 0.3) to[out=-90, in=120] (-0.25, -0.1) to[out=-60, in=60] (-0.8, -0.2) to[out=-120, in=90] (-0.85, -0.9) to[out=-90, in=90] (-0.6, -1.2) to[out=-90, in=90] (-0.55, -1.5) to[out=-90, in=-90] (-0.25, -1.5) to[out=90, in=180] (-0.1, -1.05) to[out=0, in=90] (0.15, -1.5) to[out=-90, in=-90] (0.45, -1.5) to[out=90, in=200] (0.6, -1) to[out=20, in=120] (0.85, -0.8) to[out=-60, in=90] (0.95, -1.5) to[out=-90, in=-90] (1.25, -1.5) to[out=90, in=-120] (1.3, -1) to[out=60, in=-90] (1.45, -0.7) to[out=90, in=-10] (1.1, -0.23) to[out=170, in=-90] (0.95, 0.3) to[out=90, in=90] (0.65, 0.3) to[out=-90, in=-20] (0.55, -0.08) to[out=160, in=-90] (0.05, 0.15) to[out=90, in=-90] (0.05, 0.3) to[out=90, in=90] (-0.25, 0.3);

\draw[very thick, teal, fill=white] (0.72, -0.15) arc (0:360:0.05);
\draw[very thick, teal, fill=white] (0.95, -0.2) arc (0:360:0.05);
\draw[very thick, teal, fill=white] (-0.11, -0.23) arc (0:360:0.05);
\draw[very thick, teal, fill=white] (0.52, -0.42) arc (0:360:0.05);
\draw[very thick, teal, fill=white] (0.18, -0.64) arc (0:360:0.05);
\draw[very thick, teal, fill=white] (-0.07, -0.82) arc (0:360:0.05);
\draw[very thick, teal, fill=white] (-0.2, -1) arc (0:360:0.05);
\draw[very thick, teal, fill=white] (0.42, -1) arc (0:360:0.05);
\draw[very thick, teal, fill=white] (0.85, -0.68) arc (0:360:0.05);

\node at (-0.283, 0.05) [circle, fill, inner sep=1pt] {};
\node at (-0.2, 0.05) [circle, fill, inner sep=1pt] {};

\node at (-0.79, -0.4) [circle, fill, inner sep=1pt] {};
\node at (-0.07, -0.4) [circle, fill, inner sep=1pt] {};

\node at (0.27, -0.8) [circle, fill, inner sep=1pt] {};
\node at (0.9, -0.8) [circle, fill, inner sep=1pt] {};
\end{tikzpicture}
} \,\, , 
\end{equation}
where the scattering elements are shown as ``crossings'' with circles while omitting the scattering angles for simplicity. Additionally, we introduce alternating shading to the diagram, with the convention that the outermost, unbounded region remains unshaded. Such alternating shading always uniquely exists. We first ``pair-annihilates'' as many as dots possible using diagrammatic rewriting rules involving a dot passing through scattering in TABLE~\ref{tab:majorana-rewriting-rules}. Then, for each remaining pair of dots, one of the dot must lie on the boundary-tracking (red) Majorana line. Next, we push each pair to the vicinity of one of the basis encoders. By appropriately applying Pauli gates, which consists of a pair of dots and are matchgates, to tensor legs, we can eliminate all the dots in the Majorana diagram. Returning to the 2D Quon diagram, we apply the string-genus relation Eq.~\eqref{eq:string-genus} to insert the string-hole pair into each unshaded region (the tensor legs and the planar region are suppressed for simplicity): 
\begin{equation}
\raisebox{-2.5cm}{
\begin{tikzpicture}[scale=2.3]
\begin{scope}[even odd rule]
\clip (-1.2, -1.45) rectangle (1.8, 0.25)  (0.18, -0.43) arc (0:360:0.05);
\clip (-1.2, -1.45) rectangle (1.8, 0.25)  (0.79, -0.4) arc (0:360:0.06);
\clip (-1.2, -1.45) rectangle (1.8, 0.25)  (0.19, -0.84) arc (0:360:0.04);
\clip (-1.2, -1.45) rectangle (1.8, 0.25)  (-0.34, -0.84) arc (0:360:0.06);

\fill[blue!10] (-1.2, -1.45) rectangle (1.8, 0.25);
\end{scope}

\draw[thick, blue] (0.18, -0.43) arc (0:360:0.05);
\draw[thick, blue] (0.79, -0.4) arc (0:360:0.06);
\draw[thick, blue] (0.19, -0.84) arc (0:360:0.04);
\draw[thick, blue] (-0.34, -0.84) arc (0:360:0.06);

\draw[thick, blue] (-0.35, 0.25)--(-1.2, 0.25);
\draw[thick, blue] (-1.2, 0.25)--(-1.2, -1.45);
\draw[thick, blue] (-1.2, -1.45)--(-0.65, -1.45);

\draw[thick, blue] (-0.15, -1.45)--(0.05, -1.45);
\draw[thick, blue] (0.55, -1.45)--(0.85, -1.45);

\draw[thick, blue] (1.35, -1.45)--(1.8, -1.45);
\draw[thick, blue] (1.8, -1.45)--(1.8, 0.25);
\draw[thick, blue] (1.8, 0.25)--(1.05, 0.25);

\draw[thick, blue] (0.15, 0.25)--(0.55, 0.25);

\draw[fill=gray!10, even odd rule] (-0.2, 0)--(-0.2, 0.3) to[out=90, in=90] (0, 0.3)--(0, 0.1) to[out=-90, in=170] (0.4, -0.1) to[out=-10, in=90] (1.4, -0.7) to[out=-90, in=90] (1.2, -1.2)--(1.2, -1.5) to[out=-90, in=-90] (1, -1.5)--(1, -1.2) to[out=90, in=0] (-0.6, -0.2) to[out=180, in=90] (-0.8, -0.8) to[out=-90, in=90] (-0.5, -1.2)--(-0.5, -1.5) to[out=-90, in=-90] (-0.3, -1.5)--(-0.3, -1.2) to[out=90, in=-90] (0.7, 0)--(0.7, 0.3) to[out=90, in=90] (0.9, 0.3)--(0.9, 0) to[out=-90, in=0] (0.4, -1) to[out=180, in=90] (-0.6, -0.8) to[out=-90, in=90] (0.2, -1.2)--(0.2, -1.5) to[out=-90, in=-90] (0.4, -1.5)--(0.4, -1.2) to[out=90, in=-90] (-0.2, 0);

\draw[thick, dotted] (-1, -1.3) rectangle (1.6, 0.1);

\draw[very thick, teal, fill=white] (0.72, -0.15) arc (0:360:0.05);
\draw[very thick, teal, fill=white] (0.95, -0.2) arc (0:360:0.05);
\draw[very thick, teal, fill=white] (-0.11, -0.23) arc (0:360:0.05);
\draw[very thick, teal, fill=white] (0.52, -0.42) arc (0:360:0.05);
\draw[very thick, teal, fill=white] (0.18, -0.64) arc (0:360:0.05);
\draw[very thick, teal, fill=white] (-0.07, -0.82) arc (0:360:0.05);
\draw[very thick, teal, fill=white] (-0.2, -1) arc (0:360:0.05);
\draw[very thick, teal, fill=white] (0.42, -1) arc (0:360:0.05);
\draw[very thick, teal, fill=white] (0.85, -0.68) arc (0:360:0.05);

\draw[red] (0.73, -0.38) ellipse (0.1cm and 0.15cm);
\draw[rotate=60, red] (-0.3, -0.32) ellipse (0.1cm and 0.15cm);
\draw[rotate=60, red] (-0.65, -0.54) ellipse (0.07cm and 0.13cm);
\draw[red] (-0.3, -0.84) arc (0:360:0.1);

\draw[red] (-0.25, 0.3) to[out=-90, in=120] (-0.25, -0.1) to[out=-60, in=60] (-0.8, -0.2) to[out=-120, in=90] (-0.85, -0.9) to[out=-90, in=90] (-0.6, -1.2) to[out=-90, in=90] (-0.55, -1.5) to[out=-90, in=-90] (-0.25, -1.5) to[out=90, in=180] (-0.1, -1.05) to[out=0, in=90] (0.15, -1.5) to[out=-90, in=-90] (0.45, -1.5) to[out=90, in=200] (0.6, -1) to[out=20, in=120] (0.85, -0.8) to[out=-60, in=90] (0.95, -1.5) to[out=-90, in=-90] (1.25, -1.5) to[out=90, in=-120] (1.3, -1) to[out=60, in=-90] (1.45, -0.7) to[out=90, in=-10] (1.1, -0.23) to[out=170, in=-90] (0.95, 0.3) to[out=90, in=90] (0.65, 0.3) to[out=-90, in=-20] (0.55, -0.08) to[out=160, in=-90] (0.05, 0.15) to[out=90, in=-90] (0.05, 0.3) to[out=90, in=90] (-0.25, 0.3);

\fill[gray, fill opacity=0.5, draw=blue, thick] (0.15, 0.25) arc (0:180:0.25);
\fill[gray, fill opacity=0.5, draw=blue, thick] (1.05, 0.25) arc (0:180:0.25);

\fill[gray, fill opacity=0.5, draw=blue, thick] (-0.15, -1.45) arc (0:-180:0.25);
\fill[gray, fill opacity=0.5, draw=blue, thick] (0.55, -1.45) arc (0:-180:0.25);
\fill[gray, fill opacity=0.5, draw=blue, thick] (1.35, -1.45) arc (0:-180:0.25);
\end{tikzpicture}
} 
\end{equation}
Now, for each scattering element, we apply the resolution of identity Eq.~\eqref{eq:resolution-of-id-2} and replace the scattering with a $Z$-rotation using TABLE~\ref{tab:generating-tensors-quon-diagrams}: 
\begin{equation}
\raisebox{-0.9cm}{
\begin{tikzpicture}
\fill[blue!10] (-0.8, 1) to[out=-70, in=70] (-0.8, -1)--(0.8, -1) to[out=110, in=-110] (0.8, 1);

\fill[gray!10] (-0.5, 1)--(0.5, 1)--(-0.5, -1)--(0.5, -1)--(-0.5, 1);
\draw (-0.5, 1)--(0.5, -1);
\draw (0.5, 1)--(-0.5, -1);

\draw[red] (-0.65, 1) to[out=-70, in=70] (-0.65, -1);
\draw[red] (0.65, 1) to[out=-110, in=110] (0.65, -1);

\draw[thick, blue] (-0.8, 1) to[out=-70, in=70] (-0.8, -1);
\draw[thick, blue] (0.8, 1) to[out=-110, in=110] (0.8, -1);

\draw[very thick, teal, fill=white] (0.3, 0) arc (0:360:0.3);

\node at (0, 0) {$\theta_\updownarrow$};
\end{tikzpicture}
} = \frac{1}{2} \,\, 
\raisebox{-1.2cm}{
\begin{tikzpicture}
\draw[ultra thick, gray, fill=zx_green] (0.4, 0) arc (0:360:0.4);
\node at (0, 0) {$\theta$};

\draw[ultra thick, gray] (0, 0.4)--(0, 0.8);
\draw[ultra thick, gray] (0, -0.4)--(0, -0.8);

\fill[gray, fill opacity=0.7, draw=blue, very thick] (0.5, 1.3) arc (0:-180:0.5);
\fill[gray, fill opacity=0.7, draw=blue, very thick] (0.5, -1.3) arc (0:180:0.5);

\draw[red] (0.35, 1.3) arc (0:-180:0.35);
\draw (0.2, 1.3) arc (0:-180:0.2);

\draw[red] (0.35, -1.3) arc (0:180:0.35);
\draw (0.2, -1.3) arc (0:180:0.2);
\end{tikzpicture}
} 
\end{equation}
and similarly
\begin{equation}
\raisebox{-0.8cm}{
\begin{tikzpicture}[rotate=90]
\fill[blue!10] (-0.8, 1) to[out=-70, in=70] (-0.8, -1)--(0.8, -1) to[out=110, in=-110] (0.8, 1);

\fill[gray!10] (-0.5, 1)--(0.5, 1)--(-0.5, -1)--(0.5, -1)--(-0.5, 1);
\draw (-0.5, 1)--(0.5, -1);
\draw (0.5, 1)--(-0.5, -1);

\draw[red] (-0.65, 1) to[out=-70, in=70] (-0.65, -1);
\draw[red] (0.65, 1) to[out=-110, in=110] (0.65, -1);

\draw[thick, blue] (-0.8, 1) to[out=-70, in=70] (-0.8, -1);
\draw[thick, blue] (0.8, 1) to[out=-110, in=110] (0.8, -1);

\draw[very thick, teal, fill=white] (0.3, 0) arc (0:360:0.3);

\node at (0, 0) {$\theta_\leftrightarrow$};
\end{tikzpicture}
} = \frac{1}{2} \,\, 
\raisebox{-0.4cm}{
\begin{tikzpicture}[rotate=90]
\draw[ultra thick, gray, fill=zx_green] (0.4, 0) arc (0:360:0.4);
\node at (0, 0) {$\theta$};

\draw[ultra thick, gray] (0, 0.4)--(0, 0.8);
\draw[ultra thick, gray] (0, -0.4)--(0, -0.8);

\fill[gray, fill opacity=0.7, draw=blue, very thick] (0.5, 1.3) arc (0:-180:0.5);
\fill[gray, fill opacity=0.7, draw=blue, very thick] (0.5, -1.3) arc (0:180:0.5);

\draw[red] (0.35, 1.3) arc (0:-180:0.35);
\draw (0.2, 1.3) arc (0:-180:0.2);

\draw[red] (0.35, -1.3) arc (0:180:0.35);
\draw (0.2, -1.3) arc (0:180:0.2);
\end{tikzpicture}
} \, \, , 
\end{equation}
where the arrow in each scattering element are oriented to follow the direction connecting adjacent shaded regions, which can always be ensured by applying the space-time duality presented in TABLE~\ref{tab:majorana-rewriting-rules}, if necessary. Note that we use the ZX-calculus notation to indicate the $Z$-rotation tensors in the RHS. After these replacements, each shaded region is replaced by a parity tensor, which is a matchgate: 
\begin{equation}
\raisebox{-1.1cm}{
\begin{tikzpicture}[scale=0.8]
\fill[blue!10] (-0.7, 1.35) to[out=-50, in=-130, looseness=0.3] (0.7, 1.35)--(1.3, 0.75) to[out=-130, in=130, looseness=0.5] (1.45, -0.3)--(0.98, -1.1) to[out=120, in=120] (0.4, -1.4)--(-0.4, -1.4) to[out=60, in=60] (-0.98, -1.1)--(-1.45, -0.3) to[out=50, in=-50, looseness=0.5] (-1.3, 0.75)--(-0.7, 1.35);

\draw[thick, blue, looseness=0.3] (-0.7, 1.35) to[out=-50, in=-130] (0.7, 1.35);
\draw[thick, blue, looseness=0.5] (1.3, 0.75) to[out=-130, in=130] (1.45, -0.3);
\draw[thick, blue] (0.98, -1.1) to[out=120, in=120] (0.4, -1.4);
\draw[thick, blue] (-0.98, -1.1) to[out=60, in=60] (-0.4, -1.4);
\draw[thick, blue, looseness=0.5] (-1.3, 0.75) to[out=-50, in=50] (-1.45, -0.3);

\draw[red] (-0.8, 1.25) to[out=-50, in=-150] (-0.5, 1) to [out=30, in=150] (0.5, 1) to[out=-30, in=-130] (0.8, 1.25);
\draw[red] (1.25, 0.8) to[out=-140, in=120] (0.95, 0.6) to[out=-60, in=70] (1.05, -0.25) to[out=-110, in=160] (1.37, -0.4);
\draw[red] (1.03, -1) to[out=120, in=60] (0.8, -0.77) to[out=-120, in=20] (0.25, -1.1) to[out=-160, in=120] (0.3, -1.4);
\draw[red] (-1.03, -1) to[out=60, in=120] (-0.8, -0.77) to[out=-60, in=160] (-0.25, -1.1) to[out=-20, in=60] (-0.3, -1.4);
\draw[red] (-1.25, 0.8) to[out=-40, in=60] (-0.95, 0.6) to[out=-120, in=110] (-1.05, -0.25) to[out=-70, in=20] (-1.37, -0.4);

\draw[fill=gray!10] (1, 0) arc (0:360:1);

\fill[gray!10] (0.7, 0.7)--(0.9, 1.15)--(1.15, 0.9);
\draw (0.7, 0.7)--(0.9, 1.15);
\draw (0.7, 0.7)--(1.15, 0.9);
\draw[very thick, teal, fill=white] (0.85, 0.7) arc (0:360:0.15);

\fill[gray!10] (0.85, -0.5)--(1.1, -0.9)--(1.3, -0.55);
\draw (0.85, -0.5)--(1.1, -0.9);
\draw (0.85, -0.5)--(1.3, -0.55);
\draw[very thick, teal, fill=white] (1, -0.5) arc (0:360:0.15);

\fill[gray!10] (0, -1)--(0.2, -1.4)--(-0.2, -1.4);
\draw (0, -1)--(0.2, -1.4);
\draw (0, -1)--(-0.2, -1.4);
\draw[very thick, teal, fill=white] (0.15, -1) arc (0:360:0.15);

\fill[gray!10] (-0.85, -0.5)--(-1.1, -0.9)--(-1.3, -0.55);
\draw (-0.85, -0.5)--(-1.1, -0.9);
\draw (-0.85, -0.5)--(-1.3, -0.55);
\draw[very thick, teal, fill=white] (-0.7, -0.5) arc (0:360:0.15);

\fill[gray!10] (-0.7, 0.7)--(-0.9, 1.15)--(-1.15, 0.9);
\draw (-0.7, 0.7)--(-0.9, 1.15);
\draw (-0.7, 0.7)--(-1.15, 0.9);
\draw[very thick, teal, fill=white] (-0.55, 0.7) arc (0:360:0.15);
\end{tikzpicture}
} \mapsto \raisebox{-1.1cm}{
\begin{tikzpicture}[scale=0.8]
\fill[gray, fill opacity=0.7, draw=blue, thick] (0.55, 0.3) arc (-30:150:0.2);
\fill[gray, fill opacity=0.7, draw=blue, thick] (-0.55, 0.3) arc (210:30:0.2);
\fill[gray, fill opacity=0.7, draw=blue, thick] (0.65, -0.1) arc (60:-120:0.2);
\fill[gray, fill opacity=0.7, draw=blue, thick] (-0.65, -0.1) arc (120:300:0.2);
\fill[gray, fill opacity=0.7, draw=blue, thick] (0.2, -0.6) arc (0:-180:0.2);

\fill[gray!10] ({+0.55-0.2*sqrt(3)}, 0.5) to[out=-120, in=-60] (+{-0.55+0.2*sqrt(3)}, 0.5)--(-0.55, 0.3) to[out=-60, in=30] (-0.65, -0.1)--(-0.65+0.2, +{-0.1-0.2*sqrt(3)}) to[out=30, in=90] (-0.2, -0.6)--(0.2, -0.6) to[out=90, in=150] (0.65-0.2, {-0.1-0.2*sqrt(3)})--(0.65, -0.1) to[out=150, in=-120] (0.55, 0.3);

\draw[thick, blue] ({+0.55-0.2*sqrt(3)}, 0.5) to[out=-120, in=-60] (+{-0.55+0.2*sqrt(3)}, 0.5);
\draw[thick, blue] (-0.55, 0.3) to[out=-60, in=30] (-0.65, -0.1);
\draw[thick, blue] (-0.65+0.2, +{-0.1-0.2*sqrt(3)}) to[out=30, in=90] (-0.2, -0.6);
\draw[thick, blue] (0.2, -0.6) to[out=90, in=150] (0.65-0.2, {-0.1-0.2*sqrt(3)});
\draw[thick, blue] (0.65, -0.1) to[out=150, in=-120] (0.55, 0.3);

\draw[red] ({0.55-0.04*sqrt(3)}, 0.3+0.04) arc (-30:150:0.12) to[out=-120, in=-60] (+{-0.55+0.16*sqrt(3)}, 0.3+0.16) arc (30:210:0.12) to[out=-60, in=30] (-0.65+0.04, +{-0.1-0.04*sqrt(3)}) arc (120:300:0.12) to[out=30, in=90] (-0.12, -0.6) arc (-180:0:0.12) to[out=90, in=150] (0.65-0.16, {-0.1-0.16*sqrt(3)}) arc (-120:60:0.12) to[out=150, in=-120] (+{0.55-0.04*sqrt(3)}, 0.3+0.04);

\draw ({0.55-0.08*sqrt(3)}, 0.3+0.08) arc (-30:150:0.04) to[out=-120, in=-60] (+{-0.55+0.12*sqrt(3)}, 0.3+0.12) arc (30:210:0.04) to[out=-60, in=30] (-0.65+0.08, +{-0.1-0.08*sqrt(3)}) arc (120:300:0.04) to[out=30, in=90] (-0.04, -0.6) arc (-180:0:0.04) to[out=90, in=150] (0.65-0.12, {-0.1-0.12*sqrt(3)}) arc (-120:60:0.04) to[out=150, in=-120] (+{0.55-0.08*sqrt(3)}, 0.3+0.08);

\fill[gray, fill opacity=0.7, draw=blue, thick] (0.95, {0.3+0.4*sqrt(3)}) arc (-30:-210:0.2);
\fill[gray, fill opacity=0.7, draw=blue, thick] (-0.95, {0.3+0.4*sqrt(3)}) arc (-150:30:0.2);
\fill[gray, fill opacity=0.7, draw=blue, thick] (+{0.65+0.4*sqrt(3)}, -0.5) arc (60:240:0.2);
\fill[gray, fill opacity=0.7, draw=blue, thick] (+{-0.65-0.4*sqrt(3)}, -0.5) arc (120:-60:0.2);
\fill[gray, fill opacity=0.7, draw=blue, thick] (0.2, -1.4) arc (0:180:0.2);

\draw[red] ({0.95-0.04*sqrt(3)}, {0.3+0.04+0.4*sqrt(3)}) arc (-30:-210:0.12);
\draw[red] (+{-0.95+0.04*sqrt(3)}, {0.3+0.04+0.4*sqrt(3)}) arc (-150:30:0.12);
\draw[red] (+{-0.65+0.04-0.4*sqrt(3)}, +{-0.5-0.04*sqrt(3)}) arc(120:-60:0.12);
\draw[red] (+{0.65-0.04+0.4*sqrt(3)}, +{-0.5-0.04*sqrt(3)}) arc(60:240:0.12);
\draw[red] (0.12, -1.4) arc (0:180:0.12);

\draw ({0.95-0.08*sqrt(3)}, {0.3+0.08+0.4*sqrt(3)}) arc (-30:-210:0.04);
\draw (+{-0.95+0.08*sqrt(3)}, {0.3+0.08+0.4*sqrt(3)}) arc (-150:30:0.04);
\draw (+{-0.65+0.08-0.4*sqrt(3)}, +{-0.5-0.08*sqrt(3)}) arc(120:-60:0.04);
\draw (+{0.65-0.08+0.4*sqrt(3)}, +{-0.5-0.08*sqrt(3)}) arc(60:240:0.04);
\draw (0.04, -1.4) arc (0:180:0.04);

\draw[very thick, gray] (+{0.65-0.1*sqrt(3)}, {0.4+0.1*sqrt(3)})--(+{0.85-0.1*sqrt(3)}, {0.4+0.3*sqrt(3)});
\draw[very thick, gray] (+{-0.65+0.1*sqrt(3)}, {0.4+0.1*sqrt(3)})--(+{-0.85+0.1*sqrt(3)}, {0.4+0.3*sqrt(3)});
\draw[very thick, gray] (+{-0.55-0.1*sqrt(3)}, +{-0.2-0.1*sqrt(3)})--(+{-0.55-0.3*sqrt(3)}, +{-0.4-0.1*sqrt(3)});
\draw[very thick, gray] (+{0.55+0.1*sqrt(3)}, +{-0.2-0.1*sqrt(3)})--(+{0.55+0.3*sqrt(3)}, +{-0.4-0.1*sqrt(3)});
\draw[very thick, gray] (0, -0.8)--(0, -1.2);

\draw[very thick, gray, fill=zx_green] (+{0.7-0.1*sqrt(3)}, {0.4+0.15*sqrt(3)}) arc (-120:240:0.1);
\draw[very thick, gray, fill=zx_green] (+{-0.7+0.1*sqrt(3)}, {0.4+0.15*sqrt(3)}) arc (-60:300:0.1);
\draw[very thick, gray, fill=zx_green] (+{0.55+0.15*sqrt(3)}, +{-0.25-0.1*sqrt(3)}) arc (150:150+360:0.1);
\draw[very thick, gray, fill=zx_green] (+{-0.55-0.15*sqrt(3)}, +{-0.25-0.1*sqrt(3)}) arc (30:30+360:0.1);
\draw[very thick, gray, fill=zx_green] (0, -0.9) arc (90:90+360:0.1);
\end{tikzpicture}
} = 2 \raisebox{-1.1cm}{
\begin{tikzpicture}[scale=0.8]
\draw[very thick, gray] (0, 0)--(+{0.85-0.1*sqrt(3)}, {0.4+0.3*sqrt(3)});
\draw[very thick, gray] (0, 0)--(+{-0.85+0.1*sqrt(3)}, {0.4+0.3*sqrt(3)});
\draw[very thick, gray] (0, 0)--(+{-0.55-0.3*sqrt(3)}, +{-0.4-0.1*sqrt(3)});
\draw[very thick, gray] (0, 0)--(+{0.55+0.3*sqrt(3)}, +{-0.4-0.1*sqrt(3)});
\draw[very thick, gray] (0, 0)--(0, -1.2);

\fill[gray, fill opacity=0.7, draw=blue, thick] (0.95, {0.3+0.4*sqrt(3)}) arc (-30:-210:0.2);
\fill[gray, fill opacity=0.7, draw=blue, thick] (-0.95, {0.3+0.4*sqrt(3)}) arc (-150:30:0.2);
\fill[gray, fill opacity=0.7, draw=blue, thick] (+{0.65+0.4*sqrt(3)}, -0.5) arc (60:240:0.2);
\fill[gray, fill opacity=0.7, draw=blue, thick] (+{-0.65-0.4*sqrt(3)}, -0.5) arc (120:-60:0.2);
\fill[gray, fill opacity=0.7, draw=blue, thick] (0.2, -1.4) arc (0:180:0.2);

\draw[red] ({0.95-0.04*sqrt(3)}, {0.3+0.04+0.4*sqrt(3)}) arc (-30:-210:0.12);
\draw[red] (+{-0.95+0.04*sqrt(3)}, {0.3+0.04+0.4*sqrt(3)}) arc (-150:30:0.12);
\draw[red] (+{-0.65+0.04-0.4*sqrt(3)}, +{-0.5-0.04*sqrt(3)}) arc(120:-60:0.12);
\draw[red] (+{0.65-0.04+0.4*sqrt(3)}, +{-0.5-0.04*sqrt(3)}) arc(60:240:0.12);
\draw[red] (0.12, -1.4) arc (0:180:0.12);

\draw ({0.95-0.08*sqrt(3)}, {0.3+0.08+0.4*sqrt(3)}) arc (-30:-210:0.04);
\draw (+{-0.95+0.08*sqrt(3)}, {0.3+0.08+0.4*sqrt(3)}) arc (-150:30:0.04);
\draw (+{-0.65+0.08-0.4*sqrt(3)}, +{-0.5-0.08*sqrt(3)}) arc(120:-60:0.04);
\draw (+{0.65-0.08+0.4*sqrt(3)}, +{-0.5-0.08*sqrt(3)}) arc(60:240:0.04);
\draw (0.04, -1.4) arc (0:180:0.04);

\draw [thick, gray, fill=white] (0.3, 0) arc (0:360:0.3);
\node at (0, 0) {$P$};

\draw[very thick, gray, fill=zx_green] (0.38, 0.5) arc (-120:240:0.1);
\draw[very thick, gray, fill=zx_green] (-0.38, 0.5) arc (-60:300:0.1);
\draw[very thick, gray, fill=zx_green] (0.55, -0.3) arc (150:150+360:0.1);
\draw[very thick, gray, fill=zx_green] (-0.55, -0.3) arc (30:30+360:0.1);
\draw[very thick, gray, fill=zx_green] (0, -0.65) arc (90:90+360:0.1);
\end{tikzpicture}
}
\end{equation}

As a series of these replacements, we obtain a planar tensor network where the constituent tensors are matchgates. Since the planar region of the network remains topologically equivalent to a disk, this qualifies the network as a matchgate tensor network. 
\end{proof}

Earlier in Sec.~\ref{sec:Majorana-diagrams}, we explained how matchgate unitaries can be represented by Majorana diagrams. This may appear puzzling when compared with the diagrammatic characterization of matchgates in terms of Quon diagrams in Theorem~\ref{thm:quon-matchgate}, since the boundary-tracking property was not imposed in Sec.~\ref{sec:Majorana-diagrams} to represent matchgates. The discrepancy arises from the different qubit encodings used: the dense encoding for Majorana diagrams versus the sparse encoding for Quon diagrams. We recall that the dense encoding utilizes the parity-even subspaces, enabling the representation of the entire Clifford group, a capability absent in the dense encoding. Notably, the Quon language enables switching from a sparse encoding to a dense encoding: in the trick outlined in Eq.~\eqref{eq:quon-diagram-to-sum-of-majorana-diagrams}, we rewrite a closed Quon diagram as a sum of closed Majorana diagrams, thereby translating a computation from the sparse encoding into the dense encoding. 

\subsubsection{Diagrammatic characterization of punctured matchgate}
In addition the diagrammatic characterizations of matchgate, we also provide a diagrammatic characterization of punctured matchgate tensor networks: 
\begin{theorem}
\label{thm:quon-punctured-matchgate}
A planar tensor network is a \textbf{punctured matchgate} if and only if it has a Quon diagrammatic representation that satisfies the boundary-tracking property. 
\end{theorem}

\begin{proof}
The proof of Theorem~\ref{thm:quon-matchgate} also applies here; however, in this case, the planar region of the network is topologically equivalent to a disk which might be punctured. 
\end{proof}

Example of a 2D Quon diagram for a punctured matchgate tensor network can be found in FIG.~\ref{fig:punctured-matchgate-example}. As mentioned earlier, the tensor contraction of a punctured matchgate becomes tractable if we fill in the punctures with matchgates of appropriate size, as the resulting tensor network then becomes a matchgate. This fact becomes readily apparent from the perspective of Quon. Since the holes in the background manifold are tied to the punctures in the planar region---a unique property of punctured matchgate tensor networks---all holes in the background manifold are removed via the string-genus relation when we fill in the punctures with matchgate tensors. We note that while a puncture in the planar region always implies a hole in the background manifold of its 2D Quon diagrammatic representation, the converse is not true in general. For example, the tensor network in Fig.~\ref{fig:quon-tn} (e) contains the SWAP gate, which introduces two holes that are not associated with the punctures in the planar region.

In the proof of Theorems~\ref{thm:quon-clifford} and~\ref{thm:quon-matchgate}, we present an efficient way to \textit{compile} a 2D Quon diagram into a tensor network. Using these compilations, it is also possible to efficiently compile an arbitrary 2D Quon diagram into a planar tensor network, and vice versa. 

\subsection{Computational Hardness Implications}
\label{sec:implication-on-computation}
Before concluding the section, let us discuss the implications of our diagrammatic characterizations for quantifying the computational hardness in evaluating 2D Quon diagrams. For simplicity, we consider closed 2D Quon diagrams that evaluate to complex numbers. Consider a closed 2D Quon diagram slightly deviates from the Clifford diagrammatic characterization by containing a small number of generic scattering elements. To evaluate the diagram, we treat Clifford computations as tractable computations, and use the \textit{number of generic scattering elements} in the diagram as a measure of \textit{non-Cliffordness}, due to the following observation: we can evaluate the Quon diagram by performing a brute-force summation over $2^\textrm{\# of scatterings}$ individually tractable Clifford computations, by applying diagrammatic expansions presented in TABLE~\ref{tab:Majorana-diagram-expansion} to each generic scattering element. Unless this summation is structured, e.g., when this decomposition is applied to a 2D Quon diagram satisfying the matchgate diagrammatic characterizations, the complexity of the summation grows exponentially with the number of generic scattering elements. This way of measuring non-Cliffordness is very similar to counting the $T$-gates in a $T$-gate doped Clifford circuit, namely $T$-count. Leveraging the flexibility of the diagrammatic rewriting rules of Quon, it would be compelling to investigate whether the Quon language can be used to reduce non-Clifford resource overheads, and to benchmark its performance against existing methods~\cite{amy2014polynomial, heyfron2018efficient, ruiz2025quantum, PhysRevA.102.022406, de2020fast, backens2021there}.

Consider a (closed) 2D Quon diagram that slightly deviates from the matchgate diagrammatic characterizations, in particular, a diagram with a few irremovable holes in the background manifold. To evaluate the diagram, we treat matchgate computation as tractable computations, and use the \textit{number of holes} in the background manifold as a measure of \textit{non-matchgateness}: we can evaluate the Quon diagram by performing a brute-force summation over $2^\textrm{\# of holes}$ individually matchgate tractable Majorana diagrams, by applying Eq.~\eqref{eq:parity-even-projection} to each hole, as outlined in Eq.~\eqref{eq:quon-diagram-to-sum-of-majorana-diagrams}. Unless the summation is structured, e.g., when this decomposition is applied to a 2D Quon diagram satisfying the Clifford diagrammatic characterization, the complexity of the summation grows exponentially with the number of holes. See earlier literature on the topological characterization on the non-matchgateness~\cite{cimasoni2007dimers, cimasoni2008dimers, cimasoni2009dimers, Dijkgraaf2009dimer, bravyi2009contraction}. It would be interesting to investigate the performance of this method and to benchmark against existing methods~\cite{cudby2023gaussian}.

In sum, each approach highlights a distinct classically simulable facet of the tensor network, with both being encompassed by the Quon language.

\section{Applications}
\label{sec:applications}
In this section, we present applications of the 2D Quon language by showcasing three examples: factories for tractable networks, the Kramers-Wannier duality of the two-dimensional Ising model, and the star-triangle relation of the two-dimensional Ising model. In all cases, we explicitly exploit the diagrammatic nature of the 2D Quon language. 

Factories for tractable networks offer a systematic method to generate a large class of networks with controlled tractability using three diagrammatic moves called \textit{stretching}, \textit{inserting}, and \textit{switching}. Each move transforms a 2D Quon diagram to produce new tensor networks often changing the total number of tensor legs. The first two moves preserve tractability, while the third move may increase the computational cost of the network. Specifically, if a braid is replaced by an arbitrary scattering element during a switching move, the computational cost increases by at most a factor of $2$; thus the number of such replacements, called the \textit{number of transformed scatterings}, serves as a measure of the computational cost. 

The Kramers-Wannier duality relates the classical Ising model defined on a planar graph to the model defined on its dual lattice, with the coupling constants appropriately transformed. It allows for a deeper understanding of the nature of the Ising model, including the identification of the self-dual point for the square lattice at which critical phenomena appear. Here, we reinterpret the famous Kramers-Wannier duality using the 2D Quon language. Starting from the observation that the partition function of the Ising model can be written as a matchgate tensor network, we derive the duality diagrammatically by successively applying the string-genus relation Eq.~\eqref{eq:string-genus} and the space-time duality of Majoarana scattering in TABLE~\ref{tab:majorana-rewriting-rules}. 

The star-triangle relation, also known as the Y-$\Delta$ transformation, establishes an equivalence between the partition functions of two systems that differ only in a local region where a Y-shaped configuration is substituted by a $\Delta$-shaped configuration. It is well known that the star-triangle relation is another manifestation of the Yang-Baxter equation, a cornerstone of integrability. Importantly, such local transformations integrate seamlessly with the pictorial nature of the Quon language. We show that the star-triangle relation for the two-dimensional classical Ising model is essentially caputred by the the Yang-Baxter equation of Majorana scatterings in TABLE~\ref{tab:majorana-rewriting-rules}. 

\subsection{Factories for Tractable Networks}
\label{subsec:factory}
In the following, we describe a method for transforming an input tractable network, represented as a 2D Quon diagram, into a large family of tensor networks with controlled tractability. Specifically, we introduce three diagrammatic moves---\textit{stretching}, \textit{inserting}, and \textit{switching}---each of which transforms a Quon diagram into a new one. 

Before introducing three moves in more detail, let us clarify our definition of ``tractability'' and explain how the ``number of transformed scatterings'' controls the tractability. In this subsection, a tensor network is called \textit{tractable} if the evaluation of every tensor network component, i.e., the amplitude corresponding to any given bit-string on the open tensor-leg indices, can be performed efficiently. When the number of transformed scattering elements, denoted $n_S$ is nonzero, the evaluation of the tensor network components reduces to a brute-force summation over $2^{n_S}$ terms, each of which can be computed efficiently. Therefore, by keeping $n_S$ small, the computational cost remains controlled, analogous to maintaining a low $T$-count in a $T$-gate-doped Clifford circuit or a small bond dimension in an MPS~\cite{PhysRevA.71.022316, schollwock2011density}. Our notion of tractability naturally encompasses Clifford, matchgate, and punctured matchgate tensor networks, for which $n_S = 0$ in all cases. Viewing a tensor network as a quantum state, our notion of tractability differs somewhat from the \textit{strong simulability}~\cite{terhal2004adaptive}, which is the ability to (approximately) compute any measurement outcome probability, including the marginals. While one can compute an \textit{unnormalized} probability by taking the absolute square of a component, computing the normalization is generally intractable; hence, such states are not, in general, strongly simulable. (see Sec.~\ref{sec:quon-universal} for a discussion of computing the normalization of a punctured matchgate state). Instead, by employing the Markov chain Monte Carlo (MCMC) via the Metropolis algorithm, which requires only probability ratios between two configurations, one can construct a sampler for the quantum state, achieving a \textit{weak simulation}, provided that the MCMC autocorrelation time is properly accounted for. Since we can compute not only the (unnormalized) probability but also the complex-valued amplitude for each component, our tractable networks can serve as variational wavefunctions; moreover, the aforementioned MCMC sampler can be employed in variational Monte Carlo simulations for quantum many-body problems. 

\begin{figure*}[t]
\centering
\includegraphics[width=\textwidth]{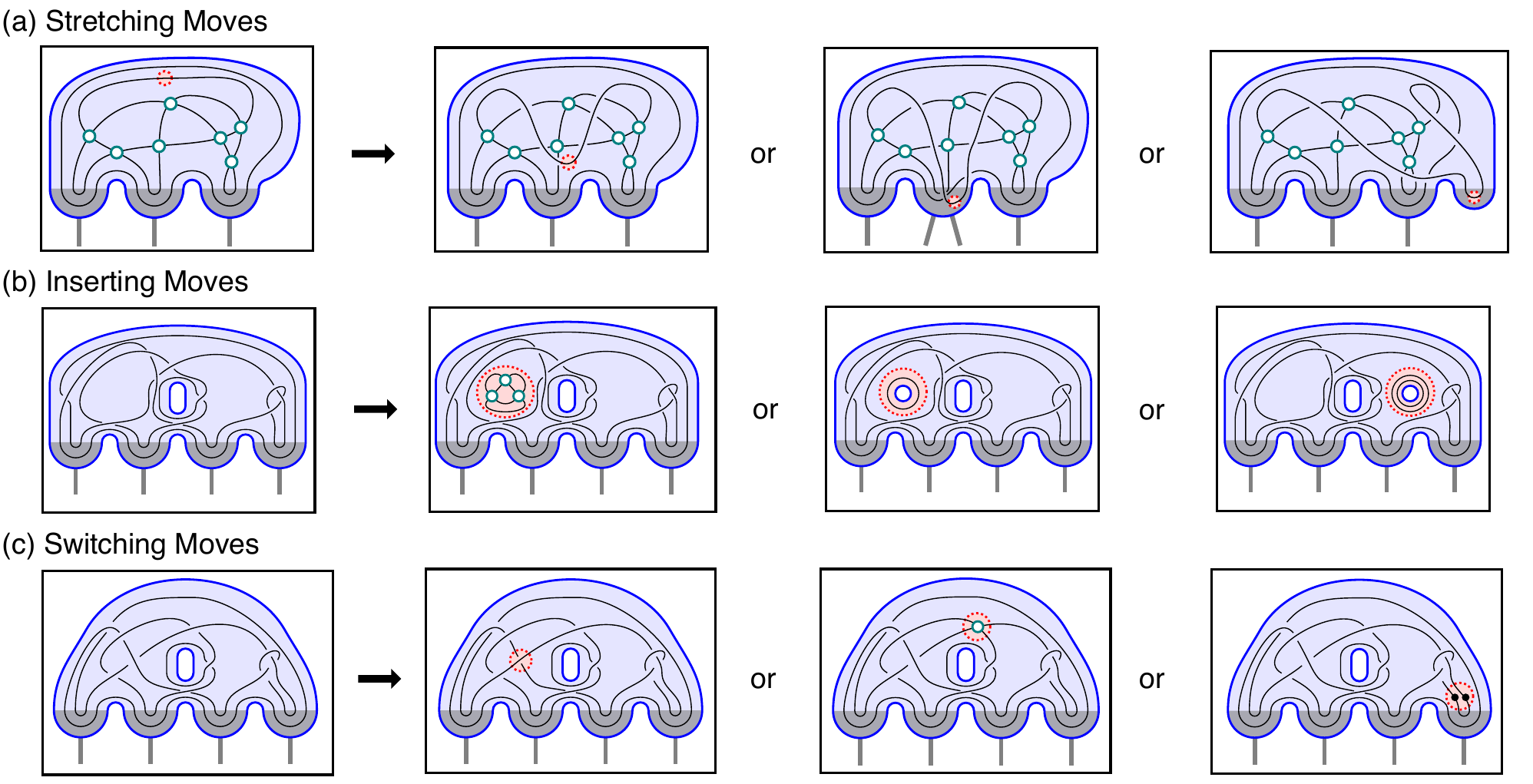}
\caption{Examples of three diagrammatic moves. We denote a scattering element as a circle with a crossing and omit the scattering angle for simplicity. (a) Examples of the \textit{stretching} moves. One stretch a Majorana string segment into a target location, which may lie in a bulk, on an existing basis encoder, or on a newly introduced basis encoder. In all cases, new braids are introduced during the move. When the target location is on an existing basis encoder, the stretching move introduces an additional tensor leg to that basis encoder. (b) Examples of the \textit{inserting} moves. An inserting move inserts an arbitrary closed Majorana diagram, a string-hole pair, or a double-string-hole pair into a small region of the Quon diagram. (c) Examples of \textit{switching} moves. A switching move either switches a braid type, a braid into a scattering element with an arbitrary scattering angle, add pair of dots, or change the scattering angle of a scattering element (the last one not shown in the figure).}
\label{fig:stretching-overlaying-switching}
\end{figure*}

We now explain three moves \textit{stretching}, \textit{inserting}, and \textit{switching} in detail. 

\subsubsection{Stretching move}
The first move is \textit{stretching}, where a segment of a Majorana string is stretched until it reaches a target location in a Quon diagram. The target location can be classified into three categories based on its position: the bulk, an existing basis encoder, or a newly created basis encoder, as illustrated in FIG.~\ref{fig:stretching-overlaying-switching} (a). During a stretching move, the Majorana string segment passes over other diagrammatic elements, introducing several braids into the Quon diagram. Crucially, stretching must occur entirely within the bulk of the Quon diagram; i.e., it cannot overpass any hole in the background manifold. When the target location lies on existing basis encoder, it introduces one additional open tensor leg to that basis encoder. For example, in the second case of FIG.~\ref{fig:stretching-overlaying-switching} (a), the number of Majorana lines contained in a basis encoder increases from $4$ to $6$, therefore the associated Hilbert space dimension increases from $2$ to $4$. In general, because each existing basis encoder contains at least $2$ Majoranas, adding two extra Majoranas introduces one additional tensor leg. As seen in the figure, the basis encoder also contains several braids, generalizing the simplest case shown in Eq.~\eqref{eq:basis-encoder-simplified}. In Eq.~\eqref{eq:basis-encoder-simplified} (or in Eq.~\eqref{eq:basis-encoder}), dots are added based on the open tensor-leg index; in Appendix~\ref{app:encoding-general}, we generalize this procedure to arbitrary basis encoders. When the marked region lies on a newly introduced basis encoder, that encoder contains only two Majorana lines, as shown in the third case of FIG.~\ref{fig:stretching-overlaying-switching} (a). In this case, no new tensor legs are introduced; however, if a subsequent stretching move stretches a Majorana string segment into this basis encoder, a new tensor leg is now introduced as the basis encoder contains $4$ Majoranas. 

Two remarks are in order. First, to ensure tractability, it is crucial to keep track of the sequence in which the braids are introduced during stretching moves. Second, if we restrict the stretching move solely to bulk regions, the tensor-network components remain unchanged, while the 2D Quon diagram potentially becomes more complex. 

\subsubsection{Inserting move}
The second move is \textit{inserting}, summarized in FIG.~\ref{fig:stretching-overlaying-switching} (b). The intuition behind the move is to obtain a new 2D Quon diagram by overlaying two 2D Quon diagrams representing tractable networks. However, simply overlaying two Quon diagrams would generally result in a non-tractable network, so a more systematic method is required. 

An inserting move is defined as follows: we insert a diagram---either an arbitrary closed Majorana diagram or one with a hole---into a small region in a 2D Quon diagram. When inserting a diagram with a hole, we insert either a string-hole pair or a double-string-hole pair by applying the string-genus relation Eq.~\eqref{eq:string-genus}. Only one of these two options leads to a valid Quon diagram, depending on the number of Majorana lines in the surrounding region. 

We remark that by multiplying an appropriate constant factor, the inserting move can be adjusted to preserve tensor components. Consequently, the inserting move preserves tensor content. Thus, it may seem that the inserting move plays a trivial role; however, when combined with the stretching move, these two moves are already powerful enough to generate a large class of tractable networks, including the punctured matchgates and beyond. The entire class of punctured matchgates\footnote{Technically speaking, this construction recovers all parity-even punctured matchgates; the parity-odd ones can simply be obtained by applying an $X$-gate to one tensor leg of the parity-even ones.} can in fact be generated from the simplest Quon diagram, namely, the one containing no Majorana diagrams [Eq.~\eqref{eq:eval-empty-background}]. One might wonder, how these two seemingly trivial looking moves are capable of generating complex tractable networks. A key insight lies in the non-trivial role of dots. Suppose that after a few rounds of the stretching and the inserting moves, several new tensor legs are introduced in the network. Now compare two components---one where all tensor-leg indices are $0$ and anoother with arbitrary indices. The corresponding Quon diagrams are identical except for the additional dots in the latter. One can imagine reducing the second diagram into the first one by pair-annihilating these extra dots. This can be achieved by moving dots from one location to another, passing several scattering elements, enabled by diagrammatic rewriting rules in TABLE~\ref{tab:majorana-rewriting-rules}. But crucially, the scattering angles change during this process, resulting in a diagram with modified scattering angles compared to that from the all-zero component. Given that punctured matchgates can be generated from the simplest Quon digram, a natural extension is to explore the family of tractable tensor networks generated by successive stretching and inserting moves, starting from a 2D Quon diagram representing a Clifford tensor network. 

\subsubsection{Switching move}
The third and final move is \textit{switching}, which locally modifies a 2D Quon diagram, as illustrated in FIG.~\ref{fig:stretching-overlaying-switching} (c). In a switching move, one considers a local region containing solely two parallel Majorana strings, a generic scattering element, or a braid. If the local region contains two parallel Majorana strings, we insert two parallel dots onto these strings. If the local region contains a generic scattering element, we change the scattering angle to an arbitrary value. Finally, if the local region contains a braid, we either switch the braid type, or transform the braid into a generic scattering element with an arbitrary scattering angle. It is only when a braid is transformed into a generic scattering element that increases the computational cost of the network, by at most a factor of $2$. Therefore, we track of such changes by the \textit{number of transformed scatterings}, often denoted $n_S$. 

Similar to the stretching moves, it is important to remember all the changes made during the switching moves. In particular, when computing a tensor component, we can expand each transformed scattering element into a weighted sum of two braids using the diagrammatic expansions in TABLE~\ref{tab:Majorana-diagram-expansion}. This yields a sum of $2^{n_S}$ Quon diagrams, each of which can be evaluated efficiently. As the number of transformed scatterings $n_S$ increases, a broader class of tensor networks is encompassed, ultimately covering the entire space as $n_S \to \infty$, albeit at increased computational cost. 

An immediate application of three moves is to employ the resulting tractable networks as variational ansatze for quantum many‑body systems. Since these networks generally exhibit high non-Cliffordness and non-matchgateness, with the Clifford and the matchgate components intertwined rather than geometrically separated, it would be interesting to evaluate the performance as variational ansatze. Additionally, factories can also be adapted for quantum circuit and tensor network \textit{obfuscation} by slightly modifying the three moves to preserve tensor components. We discuss this obfuscation application further in Sec.~\ref{sec:outlook}. 

\subsection{Kramers-Wannier Duality}
The second application of the Quon language is a simple pictorial derivation of the Kramers-Wannier (KW) duality of the two-dimensional classical Ising model. The KW duality relates the classical Ising model defined on a planar graph to the model defined on its dual lattice with modified Ising couplings. Here, we focus on the square lattice for clarity, though, the approach immediately generalizes to other two-dimensional lattices or planar graphs. We note that Jones studied the Ising model including the planar algebraic presentation of the partition function and Kramers-Wannier duality, which can be found, e.g., in Refs.~\onlinecite{jones1997introduction, evans1998quantum}. See also Ref.~\onlinecite{jiang2019block} for recent developments. 

First, we observe that the partition function can be written as a matchgate tensor network. Specifically, at each site, we assign an Ising variable $\sigma = \pm 1$, and at each link, we assign an Ising interaction $-J \sigma \sigma'$ with $J$ being the coupling constant. The partition function of the Ising model is equal to 
\begin{equation}
\sum_{ \{ \sigma_j = \pm 1 \} } e^{K \sum_{\langle \sigma, \sigma' \rangle} \sigma \sigma'} = \sum_{ \{ \sigma_j = \pm 1 \} } \prod_{\langle \sigma, \sigma' \rangle} e^{K \sigma \sigma'} , 
\end{equation}
where $\sigma_j$ denotes the Ising variable at site $j$, $K := \frac{J}{k_B T}$, and the product is taken over the links of the lattice. To re-express the partition function, consider the following tensor network: the $4$-leg Parity tensor $P$ is positioned at each site and the $2$-leg tensor $\frac{e^{K}}{\cosh \phi} e^{\phi Z}$, where $\phi$ is defined by the equation $\tanh \phi = e^{-2K}$, is positioned at each link. Since all of the tensors in the tensor network are matchgate tensors, the network forms a matchgate tensor network. In this re-expression, we identify the classical Ising variable $\pm$ with the Pauli-$X$ eigenstates $\vert \pm \rangle$, observe that the $4$-leg COPY tensor becomes the $4$-leg parity tensor $P$ after the Hadamard basis transformation on every leg, and note the following equation representing the Ising interaction: 
\begin{equation}
\langle \sigma \vert \big( e^{K} \openone + e^{-K} Z \big) \vert \sigma' \rangle = \frac{e^{K}}{\cosh \phi} \langle \sigma \vert e^{\phi Z} \vert \sigma' \rangle , 
\end{equation}
where $\vert \sigma \rangle$ and $\vert \sigma' \rangle$ denote two Ising variables associated with two nearest-neighbor sites and $\tanh \phi = e^{-2 K}$. 

Putting all these tensors together, we obtain the following 2D Quon diagram representing the partition function: 
\begin{equation}
\raisebox{-1.8cm}{
\begin{tikzpicture}[scale = 0.5]
\begin{scope}[even odd rule]
\coordinate (A) at (0, 0);

\clip ($(-4.2, 3) + (A)$) rectangle ($(-2.4, 4.8) + (A)$)  ($(-4, 4.8) + (A)$) arc (0:-180:0.2);
\clip ($(-4.2, 3) + (A)$) rectangle ($(-2.4, 4.8) + (A)$)  ($(-2.2, 4.8) + (A)$) arc (0:-180:0.2);

\clip ($(-4.2, 3) + (A)$) rectangle ($(-2.4, 4.8) + (A)$)  ($(-4, 3) + (A)$) arc (0:180:0.2);
\clip ($(-4.2, 3) + (A)$) rectangle ($(-2.4, 4.8) + (A)$)  ($(-2.2, 3) + (A)$) arc (0:180:0.2);

\fill [blue!10] ($(-4.2, 3) + (A)$) rectangle ($(-2.4, 4.8) + (A)$);
\end{scope}

\begin{scope}[even odd rule]
\coordinate (A) at (5.4, 0);

\clip ($(-4.2, 3) + (A)$) rectangle ($(-2.4, 4.8) + (A)$)  ($(-4, 4.8) + (A)$) arc (0:-180:0.2);
\clip ($(-4.2, 3) + (A)$) rectangle ($(-2.4, 4.8) + (A)$)  ($(-2.2, 4.8) + (A)$) arc (0:-180:0.2);

\clip ($(-4.2, 3) + (A)$) rectangle ($(-2.4, 4.8) + (A)$)  ($(-4, 3) + (A)$) arc (0:180:0.2);
\clip ($(-4.2, 3) + (A)$) rectangle ($(-2.4, 4.8) + (A)$)  ($(-2.2, 3) + (A)$) arc (0:180:0.2);

\fill [blue!10] ($(-4.2, 3) + (A)$) rectangle ($(-2.4, 4.8) + (A)$);
\end{scope}

\begin{scope}[even odd rule]
\coordinate (A) at (0, -4);

\clip ($(-4.2, 3) + (A)$) rectangle ($(-2.4, 4.8) + (A)$)  ($(-4, 4.8) + (A)$) arc (0:-180:0.2);
\clip ($(-4.2, 3) + (A)$) rectangle ($(-2.4, 4.8) + (A)$)  ($(-2.2, 4.8) + (A)$) arc (0:-180:0.2);

\clip ($(-4.2, 3) + (A)$) rectangle ($(-2.4, 4.8) + (A)$)  ($(-4, 3) + (A)$) arc (0:180:0.2);
\clip ($(-4.2, 3) + (A)$) rectangle ($(-2.4, 4.8) + (A)$)  ($(-2.2, 3) + (A)$) arc (0:180:0.2);

\fill [blue!10] ($(-4.2, 3) + (A)$) rectangle ($(-2.4, 4.8) + (A)$);
\end{scope}

\begin{scope}[even odd rule]
\coordinate (A) at (5.4, -4);

\clip ($(-4.2, 3) + (A)$) rectangle ($(-2.4, 4.8) + (A)$)  ($(-4, 4.8) + (A)$) arc (0:-180:0.2);
\clip ($(-4.2, 3) + (A)$) rectangle ($(-2.4, 4.8) + (A)$)  ($(-2.2, 4.8) + (A)$) arc (0:-180:0.2);

\clip ($(-4.2, 3) + (A)$) rectangle ($(-2.4, 4.8) + (A)$)  ($(-4, 3) + (A)$) arc (0:180:0.2);
\clip ($(-4.2, 3) + (A)$) rectangle ($(-2.4, 4.8) + (A)$)  ($(-2.2, 3) + (A)$) arc (0:180:0.2);

\fill [blue!10] ($(-4.2, 3) + (A)$) rectangle ($(-2.4, 4.8) + (A)$);
\end{scope}

\begin{scope}[even odd rule]
\coordinate (A) at (0, 0);

\clip ($(-0.8, 3) + (A)$) arc (0:180:1.6)  ($(-2.2, 3) + (A)$) arc (0:180:0.2);

\fill [blue!10] ($(-0.8, 3) + (A)$) arc (0:180:1.6);
\end{scope}

\begin{scope}[even odd rule]
\coordinate (A) at (1.8, 0);

\clip ($(-0.8, 3) + (A)$) arc (0:-180:1.6)  ($(-2.2, 3) + (A)$) arc (0:-180:0.2);

\fill [blue!10] ($(-0.8, 3) + (A)$) arc (0:-180:1.6);
\end{scope}

\begin{scope}[even odd rule]
\coordinate (A) at (3.6, 0);

\clip ($(-0.8, 3) + (A)$) arc (0:180:1.6)  ($(-2.2, 3) + (A)$) arc (0:180:0.2);

\fill [blue!10] ($(-0.8, 3) + (A)$) arc (0:180:1.6);
\end{scope}

\begin{scope}[even odd rule]
\coordinate (A) at (0, -4);

\clip ($(-0.8, 3) + (A)$) arc (0:180:1.6)  ($(-2.2, 3) + (A)$) arc (0:180:0.2);

\fill [blue!10] ($(-0.8, 3) + (A)$) arc (0:180:1.6);
\end{scope}

\begin{scope}[even odd rule]
\coordinate (A) at (1.8, -4);

\clip ($(-0.8, 3) + (A)$) arc (0:-180:1.6)  ($(-2.2, 3) + (A)$) arc (0:-180:0.2);

\fill [blue!10] ($(-0.8, 3) + (A)$) arc (0:-180:1.6);
\end{scope}

\begin{scope}[even odd rule]
\coordinate (A) at (3.6, -4);

\clip ($(-0.8, 3) + (A)$) arc (0:180:1.6)  ($(-2.2, 3) + (A)$) arc (0:180:0.2);

\fill [blue!10] ($(-0.8, 3) + (A)$) arc (0:180:1.6);
\end{scope}

\coordinate (A) at (3.6, -4);
\fill [blue!10] ($(-7.6, 7) + (A)$) rectangle ($(-6.2, 4.8) + (A)$);
\fill [blue!10] ($(-2.2, 7) + (A)$) rectangle ($(-0.8, 4.8) + (A)$);

\fill [blue!10] ($(-7.6, 1.2) + (A)$) rectangle ($(-6.2, 3) + (A)$);
\fill [blue!10] ($(-2.2, 1.2) + (A)$) rectangle ($(-0.8, 3) + (A)$);

\coordinate (A) at (0, 0);

\draw [thick, blue] ($(-4, 4.8) + (A)$) arc (0:-90:0.2);
\draw [thick, blue] ($(-2.6, 4.8) + (A)$) arc (180:270:0.2);
\draw [thick, blue] ($(-2.4, 4.6) + (A)$) arc (90:0:1.6);
\draw [thick, blue] ($(-0.8, 3) + (A)$) arc (180:360:0.2);
\draw [thick, blue] ($(-0.4, 3) + (A)$) arc (180:90:1.6);
\draw [thick, blue] ($(1.2, 4.6) + (A)$) arc (-90:0:0.2);
\draw [thick, blue] ($(2.8, 4.8) + (A)$) arc (180:270:0.2);

\draw [thick, blue] ($(-4, 3) + (A)$) arc (0:90:0.2);
\draw [thick, blue] ($(-2.6, 3) + (A)$) arc (180:0:0.2);
\draw [thick, blue] ($(-2.2, 3) + (A)$) arc (-180:0:1.6);
\draw [thick, blue] ($(1, 3) + (A)$) arc (180:0:0.2);
\draw [thick, blue] ($(2.8, 3) + (A)$) arc (180:90:0.2);

\draw [thick, blue] ($(-4, 3) + (A)$)--($(-4, 0.8) + (A)$);
\draw [thick, blue] ($(-2.6, 3) + (A)$)--($(-2.6, 0.8) + (A)$);

\draw [thick, blue] ($(1.4, 3) + (A)$)--($(1.4, 0.8) + (A)$);
\draw [thick, blue] ($(2.8, 3) + (A)$)--($(2.8, 0.8) + (A)$);

\draw [red] ($(-3.8, 4.8) + (A)$) arc (0:-90:0.4);
\draw ($(-3.5, 4.8) + (A)$) arc (0:-90:0.7);

\draw ($(-4.2, 3.7) + (A)$) arc (90:0:0.7);
\draw [red] ($(-4.2, 3.4) + (A)$) arc (90:0:0.4);

\draw [red] ($(-2, 3) + (A)$) arc (0:180:0.4);
\draw ($(-3.1, 3) + (A)$) arc (180:16:0.7);
\draw ($(-2.4, 4.1) + (A)$) arc (90:13:1.1);
\draw [red] ($(-1, 3) + (A)$) arc (0:90:1.4);

\draw [red] ($(-2.4, 4.4) + (A)$) arc (-90:-180:0.4);
\draw ($(-2.4, 4.1) + (A)$) arc (-90:-180:0.7);

\draw[red] ($(-3.8, 3) + (A)$)--($(-3.8, 0.8) + (A)$);

\draw ($(-3.5, 3) + (A)$)--($(-3.5, 2.23) + (A)$);
\draw ($(-3.5, 0.8) + (A)$)--($(-3.5, 1.77) + (A)$);

\draw ($(-3.1, 3) + (A)$)--($(-3.1, 2.23) + (A)$);
\draw ($(-3.1, 0.8) + (A)$)--($(-3.1, 1.77) + (A)$);

\draw[red] ($(-2.8, 3) + (A)$)--($(-2.8, 0.8) + (A)$);

\draw[red] ($(1.6, 3) + (A)$)--($(1.6, 0.8) + (A)$);

\draw ($(1.9, 3) + (A)$)--($(1.9, 2.23) + (A)$);
\draw ($(1.9, 0.8) + (A)$)--($(1.9, 1.77) + (A)$);

\draw ($(2.3, 3) + (A)$)--($(2.3, 2.23) + (A)$);
\draw ($(2.3, 0.8) + (A)$)--($(2.3, 1.77) + (A)$);

\draw[red] ($(2.6, 3) + (A)$)--($(2.6, 0.8) + (A)$);

\draw [red] ($(-1, 3) + (A)$) arc (180:360:0.4);
\draw ($(0.1, 3) + (A)$) arc (0:-164:0.7);
\draw ($(0.5, 3) + (A)$) arc (0:-167:1.1);
\draw [red] ($(-2, 3) + (A)$) arc (180:360:1.4);

\draw [red] ($(0.8, 3) + (A)$) arc (180:0:0.4);
\draw ($(0.5, 3) + (A)$) arc (180:0:0.7);
\draw ($(0.1, 3) + (A)$) arc (180:90:1.1);
\draw [red] ($(-0.2, 3) + (A)$) arc (180:90:1.4);

\draw [red] ($(1.2, 4.4) + (A)$) arc (-90:0:0.4);
\draw ($(1.2, 4.1) + (A)$) arc (-90:0:0.7);

\draw ($(2.3, 4.8) + (A)$) arc (180:270:0.7);
\draw [red] ($(2.6, 4.8) + (A)$) arc (180:270:0.4);

\draw ($(3, 3.7) + (A)$) arc (90:180:0.7);
\draw [red] ($(3, 3.4) + (A)$) arc (90:180:0.4);

\coordinate (A) at (0, -4);

\draw [thick, blue] ($(-4, 4.8) + (A)$) arc (0:-90:0.2);
\draw [thick, blue] ($(-2.6, 4.8) + (A)$) arc (180:270:0.2);
\draw [thick, blue] ($(-2.4, 4.6) + (A)$) arc (90:0:1.6);
\draw [thick, blue] ($(-0.8, 3) + (A)$) arc (180:360:0.2);
\draw [thick, blue] ($(-0.4, 3) + (A)$) arc (180:90:1.6);
\draw [thick, blue] ($(1.2, 4.6) + (A)$) arc (-90:0:0.2);
\draw [thick, blue] ($(2.8, 4.8) + (A)$) arc (180:270:0.2);

\draw [thick, blue] ($(-4, 3) + (A)$) arc (0:90:0.2);
\draw [thick, blue] ($(-2.6, 3) + (A)$) arc (180:0:0.2);
\draw [thick, blue] ($(-2.2, 3) + (A)$) arc (-180:0:1.6);
\draw [thick, blue] ($(1, 3) + (A)$) arc (180:0:0.2);
\draw [thick, blue] ($(2.8, 3) + (A)$) arc (180:90:0.2);

\draw [thick, blue] ($(-4, 3) + (A)$)--($(-4, 1.2) + (A)$);
\draw [thick, blue] ($(-2.6, 3) + (A)$)--($(-2.6, 1.2) + (A)$);

\draw [thick, blue] ($(1.4, 3) + (A)$)--($(1.4, 1.2) + (A)$);
\draw [thick, blue] ($(2.8, 3) + (A)$)--($(2.8, 1.2) + (A)$);

\draw [red] ($(-3.8, 4.8) + (A)$) arc (0:-90:0.4);
\draw ($(-3.5, 4.8) + (A)$) arc (0:-90:0.7);

\draw ($(-4.2, 3.7) + (A)$) arc (90:0:0.7);
\draw [red] ($(-4.2, 3.4) + (A)$) arc (90:0:0.4);

\draw [red] ($(-2, 3) + (A)$) arc (0:180:0.4);
\draw ($(-3.1, 3) + (A)$) arc (180:16:0.7);
\draw ($(-2.4, 4.1) + (A)$) arc (90:13:1.1);
\draw [red] ($(-1, 3) + (A)$) arc (0:90:1.4);

\draw [red] ($(-2.4, 4.4) + (A)$) arc (-90:-180:0.4);
\draw ($(-2.4, 4.1) + (A)$) arc (-90:-180:0.7);

\draw[red] ($(-3.8, 3) + (A)$)--($(-3.8, 1.2) + (A)$);

\draw ($(-3.5, 3) + (A)$)--($(-3.5, 2.23) + (A)$);
\draw ($(-3.5, 1.2) + (A)$)--($(-3.5, 1.77) + (A)$);

\draw ($(-3.1, 3) + (A)$)--($(-3.1, 2.23) + (A)$);
\draw ($(-3.1, 1.2) + (A)$)--($(-3.1, 1.77) + (A)$);

\draw[red] ($(-2.8, 3) + (A)$)--($(-2.8, 1.2) + (A)$);

\draw[red] ($(1.6, 3) + (A)$)--($(1.6, 1.2) + (A)$);

\draw ($(1.9, 3) + (A)$)--($(1.9, 2.23) + (A)$);
\draw ($(1.9, 1.2) + (A)$)--($(1.9, 1.77) + (A)$);

\draw ($(2.3, 3) + (A)$)--($(2.3, 2.23) + (A)$);
\draw ($(2.3, 1.2) + (A)$)--($(2.3, 1.77) + (A)$);

\draw[red] ($(2.6, 3) + (A)$)--($(2.6, 1.2) + (A)$);

\draw [red] ($(-1, 3) + (A)$) arc (180:360:0.4);
\draw ($(0.1, 3) + (A)$) arc (0:-164:0.7);
\draw ($(0.5, 3) + (A)$) arc (0:-167:1.1);
\draw [red] ($(-2, 3) + (A)$) arc (180:360:1.4);

\draw [red] ($(0.8, 3) + (A)$) arc (180:0:0.4);
\draw ($(0.5, 3) + (A)$) arc (180:0:0.7);
\draw ($(0.1, 3) + (A)$) arc (180:90:1.1);
\draw [red] ($(-0.2, 3) + (A)$) arc (180:90:1.4);

\draw [red] ($(1.2, 4.4) + (A)$) arc (-90:0:0.4);
\draw ($(1.2, 4.1) + (A)$) arc (-90:0:0.7);

\draw ($(2.3, 4.8) + (A)$) arc (180:270:0.7);
\draw [red] ($(2.6, 4.8) + (A)$) arc (180:270:0.4);

\draw ($(3, 3.7) + (A)$) arc (90:180:0.7);
\draw [red] ($(3, 3.4) + (A)$) arc (90:180:0.4);

\coordinate (A) at (-1.5, 3);

\draw ($(0.3, 0) + (A)$) arc (0:360:0.3);
\draw ($(0.2, 0) + (A)$) arc (0:360:0.2);

\node at ($(0, 0) + (A)$) {\tiny $\phi$};

\coordinate (A) at (-1.5, -1);

\draw ($(0.3, 0) + (A)$) arc (0:360:0.3);
\draw ($(0.2, 0) + (A)$) arc (0:360:0.2);

\node at ($(0, 0) + (A)$) {\tiny $\phi$};

\coordinate (A) at (-3.3, 2);

\draw ($(0.3, 0) + (A)$) arc (0:360:0.3);
\draw ($(0.2, 0) + (A)$) arc (0:360:0.2);

\node at ($(0, 0) + (A)$) {\tiny $\phi$};

\coordinate (A) at (-3.3, -2);

\draw ($(0.3, 0) + (A)$) arc (0:360:0.3);
\draw ($(0.2, 0) + (A)$) arc (0:360:0.2);

\node at ($(0, 0) + (A)$) {\tiny $\phi$};

\coordinate (A) at (2.1, 2);

\draw ($(0.3, 0) + (A)$) arc (0:360:0.3);
\draw ($(0.2, 0) + (A)$) arc (0:360:0.2);

\node at ($(0, 0) + (A)$) {\tiny $\phi$};

\coordinate (A) at (2.1, -2);

\draw ($(0.3, 0) + (A)$) arc (0:360:0.3);
\draw ($(0.2, 0) + (A)$) arc (0:360:0.2);

\node at ($(0, 0) + (A)$) {\tiny $\phi$};

\draw[ultra thick, teal, opacity=0.5] (-4.2, -0.1)--(3, -0.1);
\draw[ultra thick, teal, opacity=0.5] (-4.2, 3.9)--(3, 3.9);
\draw[ultra thick, teal, opacity=0.5] (-3.3, -2.8)--(-3.3, 4.8);
\draw[ultra thick, teal, opacity=0.5] (2.1, -2.8)--(2.1, 4.8);

\end{tikzpicture}
} \,\, , 
\end{equation}
where we overlay the square grid on top of the Quon diagram as a trace of the original square lattice and omit the arrows on the scattering elements, which are assumed to be oriented top-to-bottom. We now make a series of transformations: using the string-genus relation Eq.~\eqref{eq:string-genus}, remove the string-hole pair on every plaquette of the original lattice and insert the string-hole pair on every plaquette of the dual lattice: 
\begin{equation}
\raisebox{-1.6cm}{
\begin{tikzpicture}[scale=0.7]
\begin{scope}[even odd rule]
\coordinate (A) at (-4.8-0.85, 0.6-0.85);
\coordinate (B) at (0-0.95, 5.4-0.95);

\clip (A) rectangle (B)  (-2.9, 3.9) arc (0:360:0.4);
\clip (A) rectangle (B)  (-2.9+1.8, 3.9) arc (0:360:0.4);
\clip (A) rectangle (B)  (-2.9-1.8, 3.9) arc (0:360:0.4);

\clip (A) rectangle (B)  (-2.9, 3.9-1.8) arc (0:360:0.4);
\clip (A) rectangle (B)  (-2.9+1.8, 3.9-1.8) arc (0:360:0.4);
\clip (A) rectangle (B)  (-2.9-1.8, 3.9-1.8) arc (0:360:0.4);

\clip (A) rectangle (B)  (-2.9, 3.9-3.6) arc (0:360:0.4);
\clip (A) rectangle (B)  (-2.9+1.8, 3.9-3.6) arc (0:360:0.4);
\clip (A) rectangle (B)  (-2.9-1.8, 3.9-3.6) arc (0:360:0.4);

\fill [blue!10] (A) rectangle (B);
\end{scope}

\coordinate (A) at (-1.8, 0);

\draw [red] ($(-2.8, 3.9) + (A)$) arc (0:360:0.5);
\draw [thick, blue] ($(-2.9, 3.9) + (A)$) arc (0:360:0.4);

\draw ($({-3.3-0.7*cos(40)}, {3.9-0.7*sin(40)}) + (A)$) arc (220:254:0.7);
\draw ($({-3.3-0.7*cos(50)}, {3.9+0.7*sin(50)}) + (A)$) arc (130:140:0.7);

\draw ($({-3.3+0.7*cos(16)}, {3.9+0.7*sin(16)}) + (A)$) arc (16:50:0.7);
\draw ($({-3.3+0.7*cos(16)}, {3.9-0.7*sin(16)}) + (A)$) arc (-16:-74:0.7);

\coordinate (A) at (0, 0);

\draw [red] ($(-2.8, 3.9) + (A)$) arc (0:360:0.5);
\draw [thick, blue] ($(-2.9, 3.9) + (A)$) arc (0:360:0.4);

\draw ($({-3.3-0.7*cos(16)}, {3.9-0.7*sin(16)}) + (A)$) arc (196:254:0.7);
\draw ($({-3.3-0.7*cos(16)}, {3.9+0.7*sin(16)}) + (A)$) arc (164:130:0.7);

\draw ($({-3.3+0.7*cos(16)}, {3.9+0.7*sin(16)}) + (A)$) arc (16:50:0.7);
\draw ($({-3.3+0.7*cos(16)}, {3.9-0.7*sin(16)}) + (A)$) arc (-16:-74:0.7);

\coordinate (A) at (1.8, 0);

\draw [red] ($(-2.8, 3.9) + (A)$) arc (0:360:0.5);
\draw [thick, blue] ($(-2.9, 3.9) + (A)$) arc (0:360:0.4);

\draw ($({-3.3-0.7*cos(16)}, {3.9-0.7*sin(16)}) + (A)$) arc (196:254:0.7);
\draw ($({-3.3-0.7*cos(16)}, {3.9+0.7*sin(16)}) + (A)$) arc (164:130:0.7);

\draw ($({-3.3+0.7*cos(40)}, {3.9+0.7*sin(40)}) + (A)$) arc (40:50:0.7);
\draw ($({-3.3+0.7*cos(40)}, {3.9-0.7*sin(40)}) + (A)$) arc (-40:-74:0.7);

\coordinate (A) at (-1.8, -1.8);

\draw [red] ($(-2.8, 3.9) + (A)$) arc (0:360:0.5);
\draw [thick, blue] ($(-2.9, 3.9) + (A)$) arc (0:360:0.4);

\draw ($({-3.3-0.7*cos(40)}, {3.9-0.7*sin(40)}) + (A)$) arc (220:254:0.7);
\draw ($({-3.3-0.7*cos(40)}, {3.9+0.7*sin(40)}) + (A)$) arc (140:106:0.7);

\draw ($({-3.3+0.7*cos(16)}, {3.9+0.7*sin(16)}) + (A)$) arc (16:74:0.7);
\draw ($({-3.3+0.7*cos(16)}, {3.9-0.7*sin(16)}) + (A)$) arc (-16:-74:0.7);

\coordinate (A) at (0, -1.8);

\draw [red] ($(-2.8, 3.9) + (A)$) arc (0:360:0.5);
\draw [thick, blue] ($(-2.9, 3.9) + (A)$) arc (0:360:0.4);

\draw ($({-3.3-0.7*cos(16)}, {3.9-0.7*sin(16)}) + (A)$) arc (196:254:0.7);
\draw ($({-3.3-0.7*cos(16)}, {3.9+0.7*sin(16)}) + (A)$) arc (164:106:0.7);

\draw ($({-3.3+0.7*cos(16)}, {3.9+0.7*sin(16)}) + (A)$) arc (16:74:0.7);
\draw ($({-3.3+0.7*cos(16)}, {3.9-0.7*sin(16)}) + (A)$) arc (-16:-74:0.7);

\coordinate (A) at (1.8, -1.8);

\draw [red] ($(-2.8, 3.9) + (A)$) arc (0:360:0.5);
\draw [thick, blue] ($(-2.9, 3.9) + (A)$) arc (0:360:0.4);

\draw ($({-3.3-0.7*cos(16)}, {3.9-0.7*sin(16)}) + (A)$) arc (196:254:0.7);
\draw ($({-3.3-0.7*cos(16)}, {3.9+0.7*sin(16)}) + (A)$) arc (164:106:0.7);

\draw ($({-3.3+0.7*cos(40)}, {3.9+0.7*sin(40)}) + (A)$) arc (40:74:0.7);
\draw ($({-3.3+0.7*cos(74)}, {3.9-0.7*sin(74)}) + (A)$) arc (-74:-40:0.7);

\coordinate (A) at (-1.8, -3.6);

\draw [red] ($(-2.8, 3.9) + (A)$) arc (0:360:0.5);
\draw [thick, blue] ($(-2.9, 3.9) + (A)$) arc (0:360:0.4);

\draw ($({-3.3-0.7*cos(40)}, {3.9-0.7*sin(40)}) + (A)$) arc (220:230:0.7);
\draw ($({-3.3-0.7*cos(40)}, {3.9+0.7*sin(40)}) + (A)$) arc (140:106:0.7);

\draw ($({-3.3+0.7*cos(16)}, {3.9+0.7*sin(16)}) + (A)$) arc (16:74:0.7);
\draw ($({-3.3+0.7*cos(16)}, {3.9-0.7*sin(16)}) + (A)$) arc (-16:-50:0.7);

\coordinate (A) at (0, -3.6);

\draw [red] ($(-2.8, 3.9) + (A)$) arc (0:360:0.5);
\draw [thick, blue] ($(-2.9, 3.9) + (A)$) arc (0:360:0.4);

\draw ($({-3.3-0.7*cos(16)}, {3.9-0.7*sin(16)}) + (A)$) arc (196:230:0.7);
\draw ($({-3.3-0.7*cos(16)}, {3.9+0.7*sin(16)}) + (A)$) arc (164:106:0.7);

\draw ($({-3.3+0.7*cos(16)}, {3.9+0.7*sin(16)}) + (A)$) arc (16:74:0.7);
\draw ($({-3.3+0.7*cos(16)}, {3.9-0.7*sin(16)}) + (A)$) arc (-16:-50:0.7);

\coordinate (A) at (1.8, -3.6);

\draw [red] ($(-2.8, 3.9) + (A)$) arc (0:360:0.5);
\draw [thick, blue] ($(-2.9, 3.9) + (A)$) arc (0:360:0.4);

\draw ($({-3.3-0.7*cos(16)}, {3.9-0.7*sin(16)}) + (A)$) arc (196:230:0.7);
\draw ($({-3.3-0.7*cos(16)}, {3.9+0.7*sin(16)}) + (A)$) arc (164:106:0.7);

\draw ($({-3.3+0.7*cos(40)}, {3.9+0.7*sin(40)}) + (A)$) arc (40:74:0.7);
\draw ($({-3.3+0.7*cos(40)}, {3.9-0.7*sin(40)}) + (A)$) arc (-40:-50:0.7);

\coordinate (A) at (-2.4, 3.9);

\draw ($(0.3, 0) + (A)$) arc (0:360:0.3);
\draw ($(0.2, 0) + (A)$) arc (0:360:0.2);

\node at ($(0, 0) + (A)$) {\scriptsize $\phi$};

\coordinate (A) at (-4.2, 3.9);

\draw ($(0.3, 0) + (A)$) arc (0:360:0.3);
\draw ($(0.2, 0) + (A)$) arc (0:360:0.2);

\node at ($(0, 0) + (A)$) {\scriptsize $\phi$};

\coordinate (A) at (-2.4, 2.1);

\draw ($(0.3, 0) + (A)$) arc (0:360:0.3);
\draw ($(0.2, 0) + (A)$) arc (0:360:0.2);

\node at ($(0, 0) + (A)$) {\scriptsize $\phi$};

\coordinate (A) at (-4.2, 2.1);

\draw ($(0.3, 0) + (A)$) arc (0:360:0.3);
\draw ($(0.2, 0) + (A)$) arc (0:360:0.2);

\node at ($(0, 0) + (A)$) {\scriptsize $\phi$};

\coordinate (A) at (-2.4, 2.1-1.8);

\draw ($(0.3, 0) + (A)$) arc (0:360:0.3);
\draw ($(0.2, 0) + (A)$) arc (0:360:0.2);

\node at ($(0, 0) + (A)$) {\scriptsize $\phi$};

\coordinate (A) at (-4.2, 2.1-1.8);

\draw ($(0.3, 0) + (A)$) arc (0:360:0.3);
\draw ($(0.2, 0) + (A)$) arc (0:360:0.2);

\node at ($(0, 0) + (A)$) {\scriptsize $\phi$};

\coordinate (A) at (-3.3-1.8, 4.8-1.8);

\draw ($(0.3, 0) + (A)$) arc (0:360:0.3);
\draw ($(0.2, 0) + (A)$) arc (0:360:0.2);

\node at ($(0, 0) + (A)$) {\scriptsize $\theta$};

\coordinate (A) at (-3.3, 3);

\draw ($(0.3, 0) + (A)$) arc (0:360:0.3);
\draw ($(0.2, 0) + (A)$) arc (0:360:0.2);

\node at ($(0, 0) + (A)$) {\scriptsize $\theta$};

\coordinate (A) at (-1.5, 3);

\draw ($(0.3, 0) + (A)$) arc (0:360:0.3);
\draw ($(0.2, 0) + (A)$) arc (0:360:0.2);

\node at ($(0, 0) + (A)$) {\scriptsize $\theta$};

\coordinate (A) at (-3.3-1.8, 1.2);

\draw ($(0.3, 0) + (A)$) arc (0:360:0.3);
\draw ($(0.2, 0) + (A)$) arc (0:360:0.2);

\node at ($(0, 0) + (A)$) {\scriptsize $\theta$};

\coordinate (A) at (-3.3, 1.2);

\draw ($(0.3, 0) + (A)$) arc (0:360:0.3);
\draw ($(0.2, 0) + (A)$) arc (0:360:0.2);

\node at ($(0, 0) + (A)$) {\scriptsize $\theta$};

\coordinate (A) at (-1.5, 1.2);

\draw ($(0.3, 0) + (A)$) arc (0:360:0.3);
\draw ($(0.2, 0) + (A)$) arc (0:360:0.2);

\node at ($(0, 0) + (A)$) {\scriptsize $\theta$};

\draw[ultra thick, teal, opacity=0.5] (-5.65, 1.2)--(-0.95, 1.2);
\draw[ultra thick, teal, opacity=0.5] (-5.65, 3)--(-0.95, 3);

\draw[ultra thick, teal, opacity=0.5] (-4.2, -0.25)--(-4.2, 4.45);
\draw[ultra thick, teal, opacity=0.5] (-2.4, -0.25)--(-2.4, 4.45);

\draw[ultra thick, densely dashed, brown, opacity=0.5] (-5.65, 0.3)--(-0.95, 0.3);
\draw[ultra thick, densely dashed, brown, opacity=0.5] (-5.65, 2.1)--(-0.95, 2.1);
\draw[ultra thick, densely dashed, brown, opacity=0.5] (-5.65, 3.9)--(-0.95, 3.9);

\draw[ultra thick, densely dashed, brown, opacity=0.5] (-5.1, -0.25)--(-5.1, 4.45);
\draw[ultra thick, densely dashed, brown, opacity=0.5] (-3.3, -0.25)--(-3.3, 4.45);
\draw[ultra thick, densely dashed, brown, opacity=0.5] (-1.5, -0.25)--(-1.5, 4.45);

\end{tikzpicture}
} 
\,\, = \raisebox{-1.6cm}{
\begin{tikzpicture}[scale=0.7]
\begin{scope}[even odd rule]
\clip (-4.8, 0.6) rectangle (0, 5.4)  (-2.9, 3.9) arc (0:360:0.4);
\clip (-4.8, 0.6) rectangle (0, 5.4)  (-2.9+1.8, 3.9) arc (0:360:0.4);
\clip (-4.8, 0.6) rectangle (0, 5.4)  (-2.9, 3.9-1.8) arc (0:360:0.4);
\clip (-4.8, 0.6) rectangle (0, 5.4)  (-2.9+1.8, 3.9-1.8) arc (0:360:0.4);

\clip (-4.8, 0.6) rectangle (0, 5.4)  (-5.1, 4.3) arc (90:-90:0.4);
\clip (-4.8, 0.6) rectangle (0, 5.4)  (-5.1, 4.3-1.8) arc (90:-90:0.4);

\clip (-4.8, 0.6) rectangle (0, 5.4)  ({0.3-0.4*cos(42)}, {3.9+0.4*sin(42)}) arc (138:222:0.4);
\clip (-4.8, 0.6) rectangle (0, 5.4)  ({0.3-0.4*cos(42)}, {3.9-1.8+0.4*sin(42)}) arc (138:222:0.4);

\clip (-4.8, 0.6) rectangle (0, 5.4)  ({-2.9-0.4*(1-cos(48))}, {3.9+1.8-0.4*sin(48)}) arc (-48:-132:0.4);
\clip (-4.8, 0.6) rectangle (0, 5.4)  ({-1.1-0.4*(1-cos(48))}, {3.9+1.8-0.4*sin(48)}) arc (-48:-132:0.4);

\clip (-4.8, 0.6) rectangle (0, 5.4)  ({-2.9-0.4+0.4*cos(48)}, {0.3+0.4*sin(48)}) arc (48:132:0.4);
\clip (-4.8, 0.6) rectangle (0, 5.4)  ({-1.1-0.4+0.4*cos(48)}, {0.3+0.4*sin(48)}) arc (48:132:0.4);

\fill [blue!10] (-4.8, 0.6) rectangle (0, 5.4);
\end{scope}

\coordinate (A) at (0, 0);

\draw [red] ($({-4.6-0.5*(1-cos(53))}, {3.9+0.5*sin(53)}) + (A)$) arc (53:-53:0.5);
\draw [thick, blue] ($({-4.7-0.4*(1-cos(42))}, {3.9+0.4*sin(42)}) + (A)$) arc (42:-42:0.4);

\draw ($({-3.3-1.8+0.7*cos(16)}, {3.9+0.7*sin(16)}) + (A)$) arc (16:64:0.7);
\draw ($({-3.3-1.8+0.7*cos(16)}, {3.9-0.7*sin(16)}) + (A)$) arc (-16:-64:0.7);

\draw [red] ($({-4.6-0.5*(1-cos(53))}, {3.9-1.8+0.5*sin(53)}) + (A)$) arc (53:-53:0.5);
\draw [thick, blue] ($({-4.7-0.4*(1-cos(42))}, {3.9-1.8+0.4*sin(42)}) + (A)$) arc (42:-42:0.4);

\draw ($({-3.3-1.8+0.7*cos(16)}, {3.9-1.8+0.7*sin(16)}) + (A)$) arc (16:64:0.7);
\draw ($({-3.3-1.8+0.7*cos(16)}, {3.9-1.8-0.7*sin(16)}) + (A)$) arc (-16:-64:0.7);

\draw [red] ($({-2.8-0.5*(1-cos(37))}, {3.9+1.8-0.5*sin(37)}) + (A)$) arc (-37:-143:0.5);
\draw [thick, blue] ($({-2.9-0.4*(1-cos(48))}, {3.9+1.8-0.4*sin(48)}) + (A)$) arc (-48:-132:0.4);

\draw ($({-3.3+0.7*cos(26)}, {3.9+1.8-0.7*sin(26)}) + (A)$) arc (-26:-74:0.7);
\draw ($({-3.3-0.7*cos(26)}, {3.9+1.8-0.7*sin(26)}) + (A)$) arc (206:254:0.7);

\draw [red] ($({-2.8+1.8-0.5*(1-cos(37))}, {3.9+1.8-0.5*sin(37)}) + (A)$) arc (-37:-143:0.5);
\draw [thick, blue] ($({-2.9+1.8-0.4*(1-cos(48))}, {3.9+1.8-0.4*sin(48)}) + (A)$) arc (-48:-132:0.4);

\draw ($({-3.3+1.8+0.7*cos(26)}, {3.9+1.8-0.7*sin(26)}) + (A)$) arc (-26:-74:0.7);
\draw ($({-3.3+1.8-0.7*cos(26)}, {3.9+1.8-0.7*sin(26)}) + (A)$) arc (206:254:0.7);

\draw [red] ($(-2.8, 3.9) + (A)$) arc (0:360:0.5);
\draw [thick, blue] ($(-2.9, 3.9) + (A)$) arc (0:360:0.4);

\draw ($({-3.3-0.7*cos(16)}, {3.9-0.7*sin(16)}) + (A)$) arc (196:254:0.7);
\draw ($({-3.3-0.7*cos(16)}, {3.9+0.7*sin(16)}) + (A)$) arc (164:106:0.7);

\draw ($({-3.3+0.7*cos(16)}, {3.9+0.7*sin(16)}) + (A)$) arc (16:74:0.7);
\draw ($({-3.3+0.7*cos(16)}, {3.9-0.7*sin(16)}) + (A)$) arc (-16:-74:0.7);

\coordinate (A) at (1.8, 0);

\draw [red] ($(-2.8, 3.9) + (A)$) arc (0:360:0.5);
\draw [thick, blue] ($(-2.9, 3.9) + (A)$) arc (0:360:0.4);

\draw ($({-3.3-0.7*cos(16)}, {3.9-0.7*sin(16)}) + (A)$) arc (196:254:0.7);
\draw ($({-3.3-0.7*cos(16)}, {3.9+0.7*sin(16)}) + (A)$) arc (164:106:0.7);

\draw ($({-3.3+0.7*cos(16)}, {3.9+0.7*sin(16)}) + (A)$) arc (16:74:0.7);
\draw ($({-3.3+0.7*cos(16)}, {3.9-0.7*sin(16)}) + (A)$) arc (-16:-74:0.7);

\draw ($({1.8-3.3-0.7*cos(16)}, {3.9-0.7*sin(16)}) + (A)$) arc (196:244:0.7);
\draw ($({1.8-3.3-0.7*cos(16)}, {3.9+0.7*sin(16)}) + (A)$) arc (164:116:0.7);

\draw [red] ($({-1.5-0.5*cos(53)}, {3.9+0.5*sin(53)}) + (A)$) arc (127:233:0.5);
\draw [thick, blue] ($({-1.5-0.4*cos(42)}, {3.9+0.4*sin(42)}) + (A)$) arc (138:222:0.4);

\draw ($({1.8-3.3-0.7*cos(16)}, {3.9-1.8-0.7*sin(16)}) + (A)$) arc (196:244:0.7);
\draw ($({1.8-3.3-0.7*cos(16)}, {3.9-1.8+0.7*sin(16)}) + (A)$) arc (164:116:0.7);

\draw [red] ($({-1.5-0.5*cos(53)}, {3.9-1.8+0.5*sin(53)}) + (A)$) arc (127:233:0.5);
\draw [thick, blue] ($({-1.5-0.4*cos(42)}, {3.9-1.8+0.4*sin(42)}) + (A)$) arc (138:222:0.4);

\draw [red] ($({-2.8-3.6-0.5*(1-cos(37))}, {3.9+1.8-0.5*sin(37)}) + (A)$) arc (-37:-53:0.5);
\draw ($({-3.3-3.6+0.7*cos(26)}, {3.9+1.8-0.7*sin(26)}) + (A)$) arc (-26:-64:0.7);

\draw [red] ($({-2.8+1.8-0.5*(1+cos(37))}, {3.9+1.8-0.5*sin(37)}) + (A)$) arc (-143:-127:0.5);
\draw ($({-3.3+1.8-0.7*cos(26)}, {3.9+1.8-0.7*sin(26)}) + (A)$) arc (206:244:0.7);

\draw [red] ($({-2.8-3.6-0.5*(1-cos(37))}, {0.3+0.5*sin(37)}) + (A)$) arc (37:53:0.5);
\draw ($({-3.3-3.6+0.7*cos(26)}, {0.3+0.7*sin(26)}) + (A)$) arc (26:64:0.7);

\draw [red] ($({-2.8+1.8-0.5*(1+cos(37))}, {0.3+0.5*sin(37)}) + (A)$) arc (143:127:0.5);
\draw ($({-3.3+1.8-0.7*cos(26)}, {0.3+0.7*sin(26)}) + (A)$) arc (154:116:0.7);

\coordinate (A) at (0, -1.8);

\draw [red] ($(-2.8, 3.9) + (A)$) arc (0:360:0.5);
\draw [thick, blue] ($(-2.9, 3.9) + (A)$) arc (0:360:0.4);

\draw ($({-3.3-0.7*cos(16)}, {3.9-0.7*sin(16)}) + (A)$) arc (196:254:0.7);
\draw ($({-3.3-0.7*cos(16)}, {3.9+0.7*sin(16)}) + (A)$) arc (164:106:0.7);

\draw ($({-3.3+0.7*cos(16)}, {3.9+0.7*sin(16)}) + (A)$) arc (16:74:0.7);
\draw ($({-3.3+0.7*cos(16)}, {3.9-0.7*sin(16)}) + (A)$) arc (-16:-74:0.7);

\coordinate (A) at (1.8, -1.8);

\draw [red] ($(-2.8, 3.9) + (A)$) arc (0:360:0.5);
\draw [thick, blue] ($(-2.9, 3.9) + (A)$) arc (0:360:0.4);

\draw ($({-3.3-0.7*cos(16)}, {3.9-0.7*sin(16)}) + (A)$) arc (196:254:0.7);
\draw ($({-3.3-0.7*cos(16)}, {3.9+0.7*sin(16)}) + (A)$) arc (164:106:0.7);

\draw ($({-3.3+0.7*cos(16)}, {3.9+0.7*sin(16)}) + (A)$) arc (16:74:0.7);
\draw ($({-3.3+0.7*cos(16)}, {3.9-0.7*sin(16)}) + (A)$) arc (-16:-74:0.7);

\coordinate (A) at (0, -3.6);

\draw [red] ($({-2.8-0.5+0.5*cos(37)}, {3.9+0.5*sin(37)}) + (A)$) arc (37:143:0.5);
\draw [thick, blue] ($({-2.9-0.4+0.4*cos(48)}, {3.9+0.4*sin(48)}) + (A)$) arc (48:132:0.4);

\draw ($({-3.3-0.7*cos(74)}, {3.9+0.7*sin(74)}) + (A)$) arc (106:154:0.7);
\draw ($({-3.3+0.7*cos(74)}, {3.9+0.7*sin(74)}) + (A)$) arc (74:26:0.7);

\draw [red] ($({-1-0.5+0.5*cos(37)}, {3.9+0.5*sin(37)}) + (A)$) arc (37:143:0.5);
\draw [thick, blue] ($({-1.1-0.4+0.4*cos(48)}, {3.9+0.4*sin(48)}) + (A)$) arc (48:132:0.4);

\draw ($({-3.3+1.8-0.7*cos(74)}, {3.9+0.7*sin(74)}) + (A)$) arc (106:154:0.7);
\draw ($({-3.3+1.8+0.7*cos(74)}, {3.9+0.7*sin(74)}) + (A)$) arc (74:26:0.7);

\coordinate (A) at (-2.4, 3.9);

\draw ($(0.3, 0) + (A)$) arc (0:360:0.3);
\draw ($(0.2, 0) + (A)$) arc (0:360:0.2);

\node at ($(0, 0) + (A)$) {\scriptsize $\theta$};

\coordinate (A) at (-4.2, 3.9);

\draw ($(0.3, 0) + (A)$) arc (0:360:0.3);
\draw ($(0.2, 0) + (A)$) arc (0:360:0.2);

\node at ($(0, 0) + (A)$) {\scriptsize $\theta$};

\coordinate (A) at (-0.6, 3.9);

\draw ($(0.3, 0) + (A)$) arc (0:360:0.3);
\draw ($(0.2, 0) + (A)$) arc (0:360:0.2);

\node at ($(0, 0) + (A)$) {\scriptsize $\theta$};

\coordinate (A) at (-2.4, 2.1);

\draw ($(0.3, 0) + (A)$) arc (0:360:0.3);
\draw ($(0.2, 0) + (A)$) arc (0:360:0.2);

\node at ($(0, 0) + (A)$) {\scriptsize $\theta$};

\coordinate (A) at (-4.2, 2.1);

\draw ($(0.3, 0) + (A)$) arc (0:360:0.3);
\draw ($(0.2, 0) + (A)$) arc (0:360:0.2);

\node at ($(0, 0) + (A)$) {\scriptsize $\theta$};

\coordinate (A) at (-0.6, 2.1);

\draw ($(0.3, 0) + (A)$) arc (0:360:0.3);
\draw ($(0.2, 0) + (A)$) arc (0:360:0.2);

\node at ($(0, 0) + (A)$) {\scriptsize $\theta$};

\coordinate (A) at (-3.3, 4.8);

\draw ($(0.3, 0) + (A)$) arc (0:360:0.3);
\draw ($(0.2, 0) + (A)$) arc (0:360:0.2);

\node at ($(0, 0) + (A)$) {\scriptsize $\phi$};

\coordinate (A) at (-1.5, 4.8);

\draw ($(0.3, 0) + (A)$) arc (0:360:0.3);
\draw ($(0.2, 0) + (A)$) arc (0:360:0.2);

\node at ($(0, 0) + (A)$) {\scriptsize $\phi$};

\coordinate (A) at (-3.3, 3);

\draw ($(0.3, 0) + (A)$) arc (0:360:0.3);
\draw ($(0.2, 0) + (A)$) arc (0:360:0.2);

\node at ($(0, 0) + (A)$) {\scriptsize $\phi$};

\coordinate (A) at (-1.5, 3);

\draw ($(0.3, 0) + (A)$) arc (0:360:0.3);
\draw ($(0.2, 0) + (A)$) arc (0:360:0.2);

\node at ($(0, 0) + (A)$) {\scriptsize $\phi$};

\coordinate (A) at (-3.3, 1.2);

\draw ($(0.3, 0) + (A)$) arc (0:360:0.3);
\draw ($(0.2, 0) + (A)$) arc (0:360:0.2);

\node at ($(0, 0) + (A)$) {\scriptsize $\phi$};

\coordinate (A) at (-1.5, 1.2);

\draw ($(0.3, 0) + (A)$) arc (0:360:0.3);
\draw ($(0.2, 0) + (A)$) arc (0:360:0.2);

\node at ($(0, 0) + (A)$) {\scriptsize $\phi$};

\draw[ultra thick, teal, opacity=0.5] (-4.8, 2.1)--(0, 2.1);
\draw[ultra thick, teal, opacity=0.5] (-4.8, 3.9)--(0, 3.9);

\draw[ultra thick, teal, opacity=0.5] (-3.3, 0.6)--(-3.3, 5.4);
\draw[ultra thick, teal, opacity=0.5] (-1.5, 0.6)--(-1.5, 5.4);

\draw[ultra thick, densely dashed, brown, opacity=0.5] (-4.8, 1.2)--(0, 1.2);
\draw[ultra thick, densely dashed, brown, opacity=0.5] (-4.8, 3)--(0, 3);
\draw[ultra thick, densely dashed, brown, opacity=0.5] (-4.8, 4.8)--(0, 4.8);

\draw[ultra thick, densely dashed, brown, opacity=0.5] (-4.2, 0.6)--(-4.2, 5.4);
\draw[ultra thick, densely dashed, brown, opacity=0.5] (-2.4, 0.6)--(-2.4, 5.4);
\draw[ultra thick, densely dashed, brown, opacity=0.5] (-0.6, 0.6)--(-0.6, 5.4);

\end{tikzpicture}
} \,\, , 
\end{equation}
where we additionally overlay the dual lattice in dashed lines. Finally, we apply the space-time duality to the horizontal links of the dual lattice: 
\begin{equation}
\raisebox{-1.8cm}{
\begin{tikzpicture}[scale = 0.5]
\begin{scope}[even odd rule]
\coordinate (A) at (0, 0);

\clip ($(-4.2, 3) + (A)$) rectangle ($(-2.4, 4.8) + (A)$)  ($(-4, 4.8) + (A)$) arc (0:-180:0.2);
\clip ($(-4.2, 3) + (A)$) rectangle ($(-2.4, 4.8) + (A)$)  ($(-2.2, 4.8) + (A)$) arc (0:-180:0.2);

\clip ($(-4.2, 3) + (A)$) rectangle ($(-2.4, 4.8) + (A)$)  ($(-4, 3) + (A)$) arc (0:180:0.2);
\clip ($(-4.2, 3) + (A)$) rectangle ($(-2.4, 4.8) + (A)$)  ($(-2.2, 3) + (A)$) arc (0:180:0.2);

\fill [blue!10] ($(-4.2, 3) + (A)$) rectangle ($(-2.4, 4.8) + (A)$);
\end{scope}

\begin{scope}[even odd rule]
\coordinate (A) at (5.4, 0);

\clip ($(-4.2, 3) + (A)$) rectangle ($(-2.4, 4.8) + (A)$)  ($(-4, 4.8) + (A)$) arc (0:-180:0.2);
\clip ($(-4.2, 3) + (A)$) rectangle ($(-2.4, 4.8) + (A)$)  ($(-2.2, 4.8) + (A)$) arc (0:-180:0.2);

\clip ($(-4.2, 3) + (A)$) rectangle ($(-2.4, 4.8) + (A)$)  ($(-4, 3) + (A)$) arc (0:180:0.2);
\clip ($(-4.2, 3) + (A)$) rectangle ($(-2.4, 4.8) + (A)$)  ($(-2.2, 3) + (A)$) arc (0:180:0.2);

\fill [blue!10] ($(-4.2, 3) + (A)$) rectangle ($(-2.4, 4.8) + (A)$);
\end{scope}

\begin{scope}[even odd rule]
\coordinate (A) at (0, -4);

\clip ($(-4.2, 3) + (A)$) rectangle ($(-2.4, 4.8) + (A)$)  ($(-4, 4.8) + (A)$) arc (0:-180:0.2);
\clip ($(-4.2, 3) + (A)$) rectangle ($(-2.4, 4.8) + (A)$)  ($(-2.2, 4.8) + (A)$) arc (0:-180:0.2);

\clip ($(-4.2, 3) + (A)$) rectangle ($(-2.4, 4.8) + (A)$)  ($(-4, 3) + (A)$) arc (0:180:0.2);
\clip ($(-4.2, 3) + (A)$) rectangle ($(-2.4, 4.8) + (A)$)  ($(-2.2, 3) + (A)$) arc (0:180:0.2);

\fill [blue!10] ($(-4.2, 3) + (A)$) rectangle ($(-2.4, 4.8) + (A)$);
\end{scope}

\begin{scope}[even odd rule]
\coordinate (A) at (5.4, -4);

\clip ($(-4.2, 3) + (A)$) rectangle ($(-2.4, 4.8) + (A)$)  ($(-4, 4.8) + (A)$) arc (0:-180:0.2);
\clip ($(-4.2, 3) + (A)$) rectangle ($(-2.4, 4.8) + (A)$)  ($(-2.2, 4.8) + (A)$) arc (0:-180:0.2);

\clip ($(-4.2, 3) + (A)$) rectangle ($(-2.4, 4.8) + (A)$)  ($(-4, 3) + (A)$) arc (0:180:0.2);
\clip ($(-4.2, 3) + (A)$) rectangle ($(-2.4, 4.8) + (A)$)  ($(-2.2, 3) + (A)$) arc (0:180:0.2);

\fill [blue!10] ($(-4.2, 3) + (A)$) rectangle ($(-2.4, 4.8) + (A)$);
\end{scope}

\begin{scope}[even odd rule]
\coordinate (A) at (0, 0);

\clip ($(-0.8, 3) + (A)$) arc (0:180:1.6)  ($(-2.2, 3) + (A)$) arc (0:180:0.2);

\fill [blue!10] ($(-0.8, 3) + (A)$) arc (0:180:1.6);
\end{scope}

\begin{scope}[even odd rule]
\coordinate (A) at (1.8, 0);

\clip ($(-0.8, 3) + (A)$) arc (0:-180:1.6)  ($(-2.2, 3) + (A)$) arc (0:-180:0.2);

\fill [blue!10] ($(-0.8, 3) + (A)$) arc (0:-180:1.6);
\end{scope}

\begin{scope}[even odd rule]
\coordinate (A) at (3.6, 0);

\clip ($(-0.8, 3) + (A)$) arc (0:180:1.6)  ($(-2.2, 3) + (A)$) arc (0:180:0.2);

\fill [blue!10] ($(-0.8, 3) + (A)$) arc (0:180:1.6);
\end{scope}

\begin{scope}[even odd rule]
\coordinate (A) at (0, -4);

\clip ($(-0.8, 3) + (A)$) arc (0:180:1.6)  ($(-2.2, 3) + (A)$) arc (0:180:0.2);

\fill [blue!10] ($(-0.8, 3) + (A)$) arc (0:180:1.6);
\end{scope}

\begin{scope}[even odd rule]
\coordinate (A) at (1.8, -4);

\clip ($(-0.8, 3) + (A)$) arc (0:-180:1.6)  ($(-2.2, 3) + (A)$) arc (0:-180:0.2);

\fill [blue!10] ($(-0.8, 3) + (A)$) arc (0:-180:1.6);
\end{scope}

\begin{scope}[even odd rule]
\coordinate (A) at (3.6, -4);

\clip ($(-0.8, 3) + (A)$) arc (0:180:1.6)  ($(-2.2, 3) + (A)$) arc (0:180:0.2);

\fill [blue!10] ($(-0.8, 3) + (A)$) arc (0:180:1.6);
\end{scope}

\coordinate (A) at (3.6, -4);
\fill [blue!10] ($(-7.6, 7) + (A)$) rectangle ($(-6.2, 4.8) + (A)$);
\fill [blue!10] ($(-2.2, 7) + (A)$) rectangle ($(-0.8, 4.8) + (A)$);

\fill [blue!10] ($(-7.6, 1.2) + (A)$) rectangle ($(-6.2, 3) + (A)$);
\fill [blue!10] ($(-2.2, 1.2) + (A)$) rectangle ($(-0.8, 3) + (A)$);

\coordinate (A) at (0, 0);

\draw [thick, blue] ($(-4, 4.8) + (A)$) arc (0:-90:0.2);
\draw [thick, blue] ($(-2.6, 4.8) + (A)$) arc (180:270:0.2);
\draw [thick, blue] ($(-2.4, 4.6) + (A)$) arc (90:0:1.6);
\draw [thick, blue] ($(-0.8, 3) + (A)$) arc (180:360:0.2);
\draw [thick, blue] ($(-0.4, 3) + (A)$) arc (180:90:1.6);
\draw [thick, blue] ($(1.2, 4.6) + (A)$) arc (-90:0:0.2);
\draw [thick, blue] ($(2.8, 4.8) + (A)$) arc (180:270:0.2);

\draw [thick, blue] ($(-4, 3) + (A)$) arc (0:90:0.2);
\draw [thick, blue] ($(-2.6, 3) + (A)$) arc (180:0:0.2);
\draw [thick, blue] ($(-2.2, 3) + (A)$) arc (-180:0:1.6);
\draw [thick, blue] ($(1, 3) + (A)$) arc (180:0:0.2);
\draw [thick, blue] ($(2.8, 3) + (A)$) arc (180:90:0.2);

\draw [thick, blue] ($(-4, 3) + (A)$)--($(-4, 0.8) + (A)$);
\draw [thick, blue] ($(-2.6, 3) + (A)$)--($(-2.6, 0.8) + (A)$);

\draw [thick, blue] ($(1.4, 3) + (A)$)--($(1.4, 0.8) + (A)$);
\draw [thick, blue] ($(2.8, 3) + (A)$)--($(2.8, 0.8) + (A)$);

\draw [red] ($(-3.8, 4.8) + (A)$) arc (0:-90:0.4);
\draw ($(-3.5, 4.8) + (A)$) arc (0:-90:0.7);

\draw ($(-4.2, 3.7) + (A)$) arc (90:0:0.7);
\draw [red] ($(-4.2, 3.4) + (A)$) arc (90:0:0.4);

\draw [red] ($(-2, 3) + (A)$) arc (0:180:0.4);
\draw ($(-3.1, 3) + (A)$) arc (180:16:0.7);
\draw ($(-2.4, 4.1) + (A)$) arc (90:13:1.1);
\draw [red] ($(-1, 3) + (A)$) arc (0:90:1.4);

\draw [red] ($(-2.4, 4.4) + (A)$) arc (-90:-180:0.4);
\draw ($(-2.4, 4.1) + (A)$) arc (-90:-180:0.7);

\draw[red] ($(-3.8, 3) + (A)$)--($(-3.8, 0.8) + (A)$);

\draw ($(-3.5, 3) + (A)$)--($(-3.5, 2.23) + (A)$);
\draw ($(-3.5, 0.8) + (A)$)--($(-3.5, 1.77) + (A)$);

\draw ($(-3.1, 3) + (A)$)--($(-3.1, 2.23) + (A)$);
\draw ($(-3.1, 0.8) + (A)$)--($(-3.1, 1.77) + (A)$);

\draw[red] ($(-2.8, 3) + (A)$)--($(-2.8, 0.8) + (A)$);

\draw[red] ($(1.6, 3) + (A)$)--($(1.6, 0.8) + (A)$);

\draw ($(1.9, 3) + (A)$)--($(1.9, 2.23) + (A)$);
\draw ($(1.9, 0.8) + (A)$)--($(1.9, 1.77) + (A)$);

\draw ($(2.3, 3) + (A)$)--($(2.3, 2.23) + (A)$);
\draw ($(2.3, 0.8) + (A)$)--($(2.3, 1.77) + (A)$);

\draw[red] ($(2.6, 3) + (A)$)--($(2.6, 0.8) + (A)$);

\draw [red] ($(-1, 3) + (A)$) arc (180:360:0.4);
\draw ($(0.1, 3) + (A)$) arc (0:-164:0.7);
\draw ($(0.5, 3) + (A)$) arc (0:-167:1.1);
\draw [red] ($(-2, 3) + (A)$) arc (180:360:1.4);

\draw [red] ($(0.8, 3) + (A)$) arc (180:0:0.4);
\draw ($(0.5, 3) + (A)$) arc (180:0:0.7);
\draw ($(0.1, 3) + (A)$) arc (180:90:1.1);
\draw [red] ($(-0.2, 3) + (A)$) arc (180:90:1.4);

\draw [red] ($(1.2, 4.4) + (A)$) arc (-90:0:0.4);
\draw ($(1.2, 4.1) + (A)$) arc (-90:0:0.7);

\draw ($(2.3, 4.8) + (A)$) arc (180:270:0.7);
\draw [red] ($(2.6, 4.8) + (A)$) arc (180:270:0.4);

\draw ($(3, 3.7) + (A)$) arc (90:180:0.7);
\draw [red] ($(3, 3.4) + (A)$) arc (90:180:0.4);

\coordinate (A) at (0, -4);

\draw [thick, blue] ($(-4, 4.8) + (A)$) arc (0:-90:0.2);
\draw [thick, blue] ($(-2.6, 4.8) + (A)$) arc (180:270:0.2);
\draw [thick, blue] ($(-2.4, 4.6) + (A)$) arc (90:0:1.6);
\draw [thick, blue] ($(-0.8, 3) + (A)$) arc (180:360:0.2);
\draw [thick, blue] ($(-0.4, 3) + (A)$) arc (180:90:1.6);
\draw [thick, blue] ($(1.2, 4.6) + (A)$) arc (-90:0:0.2);
\draw [thick, blue] ($(2.8, 4.8) + (A)$) arc (180:270:0.2);

\draw [thick, blue] ($(-4, 3) + (A)$) arc (0:90:0.2);
\draw [thick, blue] ($(-2.6, 3) + (A)$) arc (180:0:0.2);
\draw [thick, blue] ($(-2.2, 3) + (A)$) arc (-180:0:1.6);
\draw [thick, blue] ($(1, 3) + (A)$) arc (180:0:0.2);
\draw [thick, blue] ($(2.8, 3) + (A)$) arc (180:90:0.2);

\draw [thick, blue] ($(-4, 3) + (A)$)--($(-4, 1.2) + (A)$);
\draw [thick, blue] ($(-2.6, 3) + (A)$)--($(-2.6, 1.2) + (A)$);

\draw [thick, blue] ($(1.4, 3) + (A)$)--($(1.4, 1.2) + (A)$);
\draw [thick, blue] ($(2.8, 3) + (A)$)--($(2.8, 1.2) + (A)$);

\draw [red] ($(-3.8, 4.8) + (A)$) arc (0:-90:0.4);
\draw ($(-3.5, 4.8) + (A)$) arc (0:-90:0.7);

\draw ($(-4.2, 3.7) + (A)$) arc (90:0:0.7);
\draw [red] ($(-4.2, 3.4) + (A)$) arc (90:0:0.4);

\draw [red] ($(-2, 3) + (A)$) arc (0:180:0.4);
\draw ($(-3.1, 3) + (A)$) arc (180:16:0.7);
\draw ($(-2.4, 4.1) + (A)$) arc (90:13:1.1);
\draw [red] ($(-1, 3) + (A)$) arc (0:90:1.4);

\draw [red] ($(-2.4, 4.4) + (A)$) arc (-90:-180:0.4);
\draw ($(-2.4, 4.1) + (A)$) arc (-90:-180:0.7);

\draw[red] ($(-3.8, 3) + (A)$)--($(-3.8, 1.2) + (A)$);

\draw ($(-3.5, 3) + (A)$)--($(-3.5, 2.23) + (A)$);
\draw ($(-3.5, 1.2) + (A)$)--($(-3.5, 1.77) + (A)$);

\draw ($(-3.1, 3) + (A)$)--($(-3.1, 2.23) + (A)$);
\draw ($(-3.1, 1.2) + (A)$)--($(-3.1, 1.77) + (A)$);

\draw[red] ($(-2.8, 3) + (A)$)--($(-2.8, 1.2) + (A)$);

\draw[red] ($(1.6, 3) + (A)$)--($(1.6, 1.2) + (A)$);

\draw ($(1.9, 3) + (A)$)--($(1.9, 2.23) + (A)$);
\draw ($(1.9, 1.2) + (A)$)--($(1.9, 1.77) + (A)$);

\draw ($(2.3, 3) + (A)$)--($(2.3, 2.23) + (A)$);
\draw ($(2.3, 1.2) + (A)$)--($(2.3, 1.77) + (A)$);

\draw[red] ($(2.6, 3) + (A)$)--($(2.6, 1.2) + (A)$);

\draw [red] ($(-1, 3) + (A)$) arc (180:360:0.4);
\draw ($(0.1, 3) + (A)$) arc (0:-164:0.7);
\draw ($(0.5, 3) + (A)$) arc (0:-167:1.1);
\draw [red] ($(-2, 3) + (A)$) arc (180:360:1.4);

\draw [red] ($(0.8, 3) + (A)$) arc (180:0:0.4);
\draw ($(0.5, 3) + (A)$) arc (180:0:0.7);
\draw ($(0.1, 3) + (A)$) arc (180:90:1.1);
\draw [red] ($(-0.2, 3) + (A)$) arc (180:90:1.4);

\draw [red] ($(1.2, 4.4) + (A)$) arc (-90:0:0.4);
\draw ($(1.2, 4.1) + (A)$) arc (-90:0:0.7);

\draw ($(2.3, 4.8) + (A)$) arc (180:270:0.7);
\draw [red] ($(2.6, 4.8) + (A)$) arc (180:270:0.4);

\draw ($(3, 3.7) + (A)$) arc (90:180:0.7);
\draw [red] ($(3, 3.4) + (A)$) arc (90:180:0.4);

\coordinate (A) at (-1.5, 3);

\draw ($(0.3, 0) + (A)$) arc (0:360:0.3);
\draw ($(0.2, 0) + (A)$) arc (0:360:0.2);

\node at ($(0, 0) + (A)$) {\tiny $\theta$};

\coordinate (A) at (-1.5, -1);

\draw ($(0.3, 0) + (A)$) arc (0:360:0.3);
\draw ($(0.2, 0) + (A)$) arc (0:360:0.2);

\node at ($(0, 0) + (A)$) {\tiny $\theta$};

\coordinate (A) at (-3.3, 2);

\draw ($(0.3, 0) + (A)$) arc (0:360:0.3);
\draw ($(0.2, 0) + (A)$) arc (0:360:0.2);

\node at ($(0, 0) + (A)$) {\tiny $\theta$};

\coordinate (A) at (-3.3, -2);

\draw ($(0.3, 0) + (A)$) arc (0:360:0.3);
\draw ($(0.2, 0) + (A)$) arc (0:360:0.2);

\node at ($(0, 0) + (A)$) {\tiny $\theta$};

\coordinate (A) at (2.1, 2);

\draw ($(0.3, 0) + (A)$) arc (0:360:0.3);
\draw ($(0.2, 0) + (A)$) arc (0:360:0.2);

\node at ($(0, 0) + (A)$) {\tiny $\theta$};

\coordinate (A) at (2.1, -2);

\draw ($(0.3, 0) + (A)$) arc (0:360:0.3);
\draw ($(0.2, 0) + (A)$) arc (0:360:0.2);

\node at ($(0, 0) + (A)$) {\tiny $\phi$};

\draw[ultra thick, densely dashed, brown, opacity=0.5] (-4.2, -0.1)--(3, -0.1);
\draw[ultra thick, densely dashed, brown, opacity=0.5] (-4.2, 3.9)--(3, 3.9);
\draw[ultra thick, densely dashed, brown, opacity=0.5] (-3.3, -2.8)--(-3.3, 4.8);
\draw[ultra thick, densely dashed, brown, opacity=0.5] (2.1, -2.8)--(2.1, 4.8);

\end{tikzpicture}
} \,\, , 
\end{equation}
which is the partition function of the Ising model defined on the dual lattice. 

In our derivation of the KW duality, we also have recovered the ``half translation,'' a key element in the duality~\cite{seiberg2024non} that arises from our regrouping of Majoranas when forming tensor legs. 

\subsection{Star-Triangle Relation}
A third application of the 2D Quon language is a derivation of the star-triangle relation of the two-dimensional Ising model. While the fact that the star-triangle relation is yet another form of the Yang-Baxter equation is rather well-known, our derivation using the 2D Quon language illustrates its pictorial nature. 

The star-triangle (or Y-$\Delta$) relation corresponds to the local transformation of a ``star'' (or Y) shaped region into a ``triangle'' (or $\Delta$) shaped region: 
\begin{equation}
\label{eq:star-triangle-ising}
\raisebox{-1.2cm}{
\begin{tikzpicture}
\draw[ultra thick, gray] (-1, 0)--(1, 0);
\draw[ultra thick, gray] (0, {sqrt(3)})--(1, 0);
\draw[ultra thick, gray] (0, {sqrt(3)})--(-1, 0);

\draw[ultra thick, gray] (0, {sqrt(3)})--(0, 2.3);
\draw[ultra thick, gray] (1, 0)--({1 + 0.5*(2.3 - sqrt(3))*sqrt(3)}, {-0.5*(2.3 - sqrt(3))});
\draw[ultra thick, gray] (-1, 0)--({-1 - 0.5*(2.3 - sqrt(3))*sqrt(3)}, {-0.5*(2.3 - sqrt(3))});

\draw [thick, gray, fill=white] (1.2, 0) arc (0:360:0.2);
\node at (1, 0) {$P$};

\draw [thick, gray, fill=white] (-0.8, 0) arc (0:360:0.2);
\node at (-1, 0) {$P$};

\draw [thick, gray, fill=white] (0, {sqrt(3) + 0.2}) arc (90:450:0.2);
\node at (0, {sqrt(3)}) {$P$};

\draw [thick, gray, fill=zx_green] (-0.3, {0.5*sqrt(3) - 0.2}) rectangle (-0.7, {0.5*sqrt(3) + 0.2});
\node at (-0.5, {0.5*sqrt(3)}) {$u_1$};

\draw [thick, gray, fill=zx_green] (-0.2, -0.2) rectangle (0.2, 0.2);
\node at (0, 0) {$u_2$};

\draw [thick, gray, fill=zx_green] (0.3, {0.5*sqrt(3) - 0.2}) rectangle (0.7, {0.5*sqrt(3) + 0.2});
\node at (0.5, {0.5*sqrt(3)}) {$u_3$};
\end{tikzpicture}
} = R \raisebox{-1cm}{
\begin{tikzpicture}
\draw[ultra thick, gray] (0, 0.5)--(0, 2.1);
\draw[ultra thick, gray] (0, 0.5)--({1 + 0.5*(2.1 - sqrt(3))*sqrt(3)}, {-0.5*(2.1 - sqrt(3))});
\draw[ultra thick, gray] (0, 0.5)--({-1 - 0.5*(2.1 - sqrt(3))*sqrt(3)}, {-0.5*(2.1 - sqrt(3))});

\draw [thick, gray, fill=white] (0.2, 0.5) arc (0:360:0.2);
\node at (0, 0.5) {$P$};

\draw [thick, gray, fill=zx_green] ({0.3 + 0.25*(2.1 - sqrt(3))*sqrt(3)}, {0.05-0.25*(2.1 - sqrt(3))}) rectangle ++(0.4, 0.4);
\node at ({0.5 + 0.25*(2.1 - sqrt(3))*sqrt(3)}, {0.25-0.25*(2.1 - sqrt(3))}) {$v_1$};

\draw [thick, gray, fill=zx_green] (0.2, 1.3-0.2) rectangle (-0.2, 1.3+0.2);
\node at (0, 1.3) {$v_2$};

\draw [thick, gray, fill=zx_green] ({-0.3 - 0.25*(2.1 - sqrt(3))*sqrt(3)}, {0.05-0.25*(2.1 - sqrt(3))}) rectangle ++(-0.4, 0.4);
\node at ({-0.5 - 0.25*(2.1 - sqrt(3))*sqrt(3)}, {0.25-0.25*(2.1 - sqrt(3))}) {$v_3$};
\end{tikzpicture}
} \,\, , 
\end{equation}
where in addition to the three-leg parity tensor $P$, we introduce the following (matchgate) tensor:
\begin{equation}
\raisebox{-0.45cm}{
\begin{tikzpicture}
\draw[ultra thick, gray] (0, 0)--(0, 1);

\draw [thick, gray, fill=zx_green] (-0.2, 0.3) rectangle (0.2, 0.7);
\node at (0, 0.5) {$u$};
\end{tikzpicture}
} \,\, = \openone + u Z , 
\end{equation}
with $Z$ being the Pauli-Z operator and $u_1$, $u_2$, $u_3$, $v_1$, $v_2$, $v_3$, and $R$ being complex numbers. The star-triangle relation tells us that, given a tuple of complex numbers $(u_1, u_2, u_3)$, except at a measure-zero set of points, one can always find complex numbers $(v_1, v_2, v_3; R)$ such that the relation Eq.~\eqref{eq:star-triangle-ising} holds. Conversely, given a tuple of complex numbers $(v_1, v_2, v_3)$, except at a measure-zero set of points, one can always find complex numbers $(u_1, u_2, u_3; R)$ such that the relation Eq.~\eqref{eq:star-triangle-ising} holds. 

We now give the proof of the star-triangle relation Eq.~\eqref{eq:star-triangle-ising} using the 2D Quon language. In order to avoid unnecessary complications, we suppress constant multiplication factors in the following. So, strictly speaking, the equalities below should be understood as $\propto$ sign. We first note that the LHS can be rewritten in terms of the following quon diagram: 
\begin{equation}
\raisebox{-1.6cm}{
\begin{tikzpicture}
\draw[ultra thick, gray] (0, 0)--(0, 2.2);
\draw[ultra thick, gray] (0, 0)--({-0.4*sqrt(3)}, -0.4);
\draw[ultra thick, gray] ({1.2*sqrt(3)}, 0)--({1.6*sqrt(3)}, -0.4);

\draw[ultra thick, gray] (0, 1.2) arc (150:-150:1.2);

\draw [thick, gray, fill=white] (0, 1.4) arc (90:450:0.2);
\node at (0, 1.2) {$P$};

\draw [thick, gray, fill=white] (0, 0.2) arc (90:450:0.2);
\node at (0, 0) {$P$};

\draw [thick, gray, fill=white] ({1.2*sqrt(3)}, 0.2) arc (90:450:0.2);
\node at ({1.2*sqrt(3)}, 0) {$P$};

\draw [thick, gray, fill=zx_green] (-0.2, 0.4) rectangle ++(0.4, 0.4);
\node at (0, 0.6) {$u_1$};

\draw [thick, gray, fill=zx_green] (0.4, -0.7) rectangle ++(0.4, 0.4);
\node at (0.6, -0.5) {$u_2$};

\draw [thick, gray, fill=zx_green] (0.4, 1.5) rectangle ++(0.4, 0.4);
\node at (0.6, 1.7) {$u_3$};

\end{tikzpicture}
} = \,\,  
\raisebox{-1.6cm}{
\begin{tikzpicture}
\begin{scope}[even odd rule]
\clip (-0.6, -1.6) rectangle (0, 1.6)  ({-0.5 + 1.4*(1-cos(50))}, {1.4*sin(50)}) arc (-50:-310:0.05);
\clip (-0.6, -1.6) rectangle (0, 1.6)  ({-0.5 + 1.4*(1-cos(50))}, {-1.4*sin(50)}) arc (50:310:0.05);

\clip (-0.6, -1.6) rectangle (0, 1.6)  ({-0.5 + 1.4*(1-cos(50)) - 0.05 - 0.05*cos(50)}, {1.4*sin(50) + 0.05*sin(50)}) rectangle (0, 1.6);
\clip (-0.6, -1.6) rectangle (0, 1.6)  ({-0.5 + 1.4*(1-cos(50)) - 0.05 - 0.05*cos(50)}, {-1.4*sin(50) - 0.05*sin(50)}) rectangle (0, -1.6);

\fill [blue!10] (-0.6, -1.6) rectangle (0, 1.6);
\end{scope}

\begin{scope}[even odd rule]
\clip (2.3, -1.6) rectangle ({-0.5 + 1.4*(1+cos(50))}, 0)  ({-0.5 + 1.4*(1+cos(50))}, {-1.4*sin(50)}) arc (130:-130:0.05);
\clip (2.3, -1.6) rectangle ({-0.5 + 1.4*(1+cos(50))}, 0)  ({-0.5 + 1.4*(1+cos(50))}, -1.1) rectangle ({-0.5 + 1.4*(1+cos(50)) + 0.05 + 0.05*cos(50)}, -1.6);

\fill [blue!10] (2.3, -1.6) rectangle ({-0.5 + 1.4*(1+cos(50))}, 0);
\end{scope}

\begin{scope}[even odd rule]
\clip (-0.5, 0) arc (180:-180:1.4)  (0.5, 0) arc (180:-180:0.4);

\fill [blue!10] (-0.5, 0) arc (180:-180:1.4);
\end{scope}

\draw [thick, blue] (-0.6, -1.6)--(-0.6, 1.6);
\draw [red] (-0.5, -1.6)--(-0.5, 1.6);

\draw [thick, blue] (0.5, 0) arc (180:-180:0.4);
\draw [red] (0.4, 0) arc (180:-180:0.5);

\draw ({0.05 + 0.85*(1-cos(19)}, {0.85*sin(19)}) arc (161:136:0.85);
\draw ({0.05 + 0.85*(1-cos(83)}, {0.85*sin(83)}) arc (97:-97:0.85);
\draw ({0.05 + 0.85*(1-cos(19)}, {-0.85*sin(19)}) arc (-161:-136:0.85);

\draw [looseness=1.2] ({-0.15 + 1.05*(1-cos(15.5))}, {1.05*sin(15.5)}) to[out=120, in=-90] (-0.4, 1.6);
\draw [looseness=1.5] ({-0.15 + 1.05*(1-cos(48.2))}, {1.05*sin(48.2)}) to[out=230, in=-90] (-0.3, 1.6);

\draw [looseness=1.2] ({-0.15 + 1.05*(1-cos(15.5))}, {-1.05*sin(15.5)}) to[out=-120, in=90] (-0.4, -1.6);
\draw [looseness=1.5] ({-0.15 + 1.05*(1-cos(48.2))}, {-1.05*sin(48.2)}) to[out=-230, in=90] (-0.3, -1.6);

\draw [looseness=1.2] ({-0.15 + 1.05*(1-cos(78.5))}, {-1.05*sin(78.5)}) to[out=-60, in=90] (2.05, -1) to[out=-90, in=90] (2.05, -1.6);

\draw [looseness=1.2] ({-0.15 + 1.05*(1-cos(78.5))}, {1.05*sin(78.5)}) to[out=30, in=90] (2.15, -0.5) to[out=-90, in=90] (2.15, -1.6);

\draw [thick, blue] ({-0.5 + 1.4*(1-cos(50))}, {1.4*sin(50)}) arc (130:0:1.4);
\draw [thick, blue] (2.3, 0.0)--(2.3, -1.6);

\draw [red] ({-0.4 + 1.3*(1-cos(50))}, {1.3*sin(50)}) arc (130:0:1.3);
\draw [red] (2.2, 0)--(2.2, -1.6);

\draw [thick, blue] ({-0.5 + 1.4*(1-cos(50))}, {-1.4*sin(50)}) arc (-130:-50:1.4);
\draw [red] ({-0.4 + 1.3*(1-cos(50))}, {-1.3*sin(50)}) arc (-130:-50:1.3);

\draw [thick, blue] ({-0.5 + 1.4*(1+cos(50))}, {-1.4*sin(50)}) arc (130:0:0.05);
\draw [thick, blue] ({-0.5 + 1.4*(1+cos(50)) + 0.05 + 0.05*cos(50)}, {-1.4*sin(50) - 0.05*sin(50)})--({-0.5 + 1.4*(1+cos(50)) + 0.05 + 0.05*cos(50)}, -1.6);

\draw [looseness=1.5, red] ({-0.4 + 1.3*(1+cos(50))}, {-1.3*sin(50)}) to[out=40, in=90] (2, -1.6);

\draw [thick, blue] ({-0.5 + 1.4*(1-cos(50))}, {-1.4*sin(50)}) arc (50:180:0.05);
\draw [thick, blue] ({-0.5 + 1.4*(1-cos(50)) - 0.05 - 0.05*cos(50)}, {-1.4*sin(50) - 0.05*sin(50)})--({-0.5 + 1.4*(1-cos(50)) - 0.05 - 0.05*cos(50)}, -1.6);

\draw [looseness=1.5, red] ({-0.4 + 1.3*(1-cos(50))}, {-1.3*sin(50)}) to[out=140, in=90] (-0.2, -1.6);

\draw [thick, blue] ({-0.5 + 1.4*(1-cos(50))}, {1.4*sin(50)}) arc (-50:-180:0.05);
\draw [thick, blue] ({-0.5 + 1.4*(1-cos(50)) - 0.05 - 0.05*cos(50)}, {1.4*sin(50) + 0.05*sin(50)})--({-0.5 + 1.4*(1-cos(50)) - 0.05 - 0.05*cos(50)}, 1.6);

\draw [looseness=1.5, red] ({-0.4 + 1.3*(1-cos(50))}, {1.3*sin(50)}) to[out=-140, in=-90] (-0.2, 1.6);

\draw (0.8, 0.8) arc (0:360:0.3);
\draw (0.7, 0.8) arc (0:360:0.2);

\node at (0.5, 0.8) {\scriptsize $\phi_3$};

\draw (0.3, 0) arc (0:360:0.3);
\draw (0.2, 0) arc (0:360:0.2);

\node at (0, 0) {\scriptsize $\phi_1$};

\draw (0.8, -0.8) arc (0:360:0.3);
\draw (0.7, -0.8) arc (0:360:0.2);

\node at (0.5, -0.8) {\scriptsize $\phi_2$};

\end{tikzpicture}
} \,\, , 
\end{equation}
where $\tanh \phi_i = u_i$. Now, we first apply the string-genus relation Eq.~\eqref{eq:string-genus}, then perform the space-time duality on $\phi_2$ and $\phi_3$, and apply the Yang-Baxter equation (TABLE~\ref{tab:majorana-rewriting-rules}): 
\begin{equation}
\raisebox{-1.6cm}{
\begin{tikzpicture}
\begin{scope}[even odd rule]
\clip (-0.6, -1.6) rectangle (1.2, 1.6)  (-0.1, -1.6) arc (180:0:0.4);

\fill [blue!10] (-0.6, -1.6) rectangle (1.2, 1.6);
\end{scope}

\draw [thick, blue] (-0.6, -1.6)--(-0.6, 1.6);
\draw [red] (-0.5, -1.6)--(-0.5, 1.6);

\draw (0.1, 0.28)--(0.15, 0.39);
\draw (0.45, 0.39) to[out=-50, in=50] (0.45, -0.39);
\draw (0.1, -0.28)--(0.15, -0.39);

\draw [looseness=1.2] ({-0.15 + 1.05*(1-cos(15.5))}, {1.05*sin(15.5)}) to[out=120, in=-90] (-0.3, 1.6);
\draw [looseness=1.4] (0.15, 0.91) to[out=130, in=-90] (0, 1.6);

\draw [looseness=1.2] ({-0.15 + 1.05*(1-cos(15.5))}, {-1.05*sin(15.5)}) to[out=-120, in=90] (-0.4, -1.6);
\draw [looseness=1.4] (0.15, -0.91) to[out=-130, in=90] (-0.3, -1.6);

\draw [looseness=1.2] (0.45, -0.91) to[out=-60, in=90] (0.9, -1.6);
\draw [looseness=1] (0.45, 0.91) to[out=60, in=90] (1, -1.6);

\draw [thick, blue] (1.2, -1.6)--(1.2, 1.6);
\draw [red] (1.1, -1.6)--(1.1, 1.6);

\draw [thick, blue] (-0.1, -1.6) arc (180:0:0.4);
\draw [red] (-0.2, -1.6) arc (180:0:0.5);

\draw (0.6, 0.65) arc (0:360:0.3);
\draw (0.5, 0.65) arc (0:360:0.2);

\node at (0.3, 0.65) {\scriptsize $\phi_3^*$};

\draw (0.3, 0) arc (0:360:0.3);
\draw (0.2, 0) arc (0:360:0.2);

\node at (0, 0) {\scriptsize $\phi_1$};

\draw (0.6, -0.65) arc (0:360:0.3);
\draw (0.5, -0.65) arc (0:360:0.2);

\node at (0.3, -0.65) {\scriptsize $\phi_2^*$};

\end{tikzpicture}
} = \raisebox{-1.4cm}{
\begin{tikzpicture}
\begin{scope}[even odd rule]
\clip (-0.6, -1.6) rectangle (1.2, 1.2)  (-0.1, -1.6) arc (180:0:0.4);

\fill [blue!10] (-0.6, -1.6) rectangle (1.2, 1.2);
\end{scope}

\draw [thick, blue] (-0.6, -1.6)--(-0.6, 1.2);
\draw [red] (-0.5, -1.6)--(-0.5, 1.2);

\draw (0.2, 0.28)--(0.15, 0.39);
\draw (-0.15, 0.39) to[out=230, in=-230] (-0.15, -0.39);
\draw (0.2, -0.28)--(0.15, -0.39);

\draw [looseness=1.2] (-0.15, 0.91) to[out=120, in=-90] (-0.3, 1.2);
\draw [looseness=1] (0.15, 0.91) to[out=50, in=-90] (0.2, 1.2);

\draw [looseness=1.2] (-0.15, -0.91) to[out=-120, in=90] (-0.4, -1.6);
\draw [looseness=1] (0.15, -0.91) to[out=-50, in=90] (-0.3, -1.4) to[out=-90, in=90] (-0.3, -1.6);

\draw [looseness=1.2] (0.4, -0.28) to[out=-60, in=90] (0.9, -1.6);
\draw [looseness=1] (0.4, 0.28) to[out=60, in=240] (0.6, 0.8) to[out=60, in=90] (1, -1.6);

\draw [thick, blue] (1.2, -1.6)--(1.2, 1.2);
\draw [red] (1.1, -1.6)--(1.1, 1.2);

\draw [thick, blue] (-0.1, -1.6) arc (180:0:0.4);
\draw [red] (-0.2, -1.6) arc (180:0:0.5);

\draw (0.3, 0.65) arc (0:360:0.3);
\draw (0.2, 0.65) arc (0:360:0.2);

\node at (0, 0.65) {\scriptsize $\theta_2$};

\draw (0.6, 0) arc (0:360:0.3);
\draw (0.5, 0) arc (0:360:0.2);

\node at (0.3, 0) {\scriptsize $\theta_1$};

\draw (0.3, -0.65) arc (0:360:0.3);
\draw (0.2, -0.65) arc (0:360:0.2);

\node at (0, -0.65) {\scriptsize $\theta_3$};

\end{tikzpicture}
} . 
\end{equation}
Finally, we apply the space-time duality on $\theta_1$ (TABLE~\ref{tab:majorana-rewriting-rules}): 
\begin{equation}
\raisebox{-1.6cm}{
\begin{tikzpicture}
\begin{scope}[even odd rule]
\clip (-0.6, -0.45) rectangle (0.4, 1.6)  ({2.8-1.7 - 1.7*cos(65)}, {-1.6+1.15 + 1.7*sin(65)}) arc (-65:-65+360:0.04);
\clip (-0.6, -0.45) rectangle (0.4, 1.6)  ({2.8-1.7 - 1.7*cos(65) - 0.04 - 0.04*cos(65)}, 1.1) rectangle (0.4, 1.6);

\fill [blue!10] (-0.6, -0.45) rectangle (0.4, 1.6);
\end{scope}

\begin{scope} [even odd rule]
\clip (2.8, -1.6+1.15) arc (0:180:1.7)  (0.4, -1.6+1.15) arc (180:0:0.85);

\fill [blue!10] (2.8, -1.6+1.15) arc (0:180:1.7);
\end{scope}

\draw [thick, blue] (-0.6, -0.45)--(-0.6, 1.6);
\draw [red] (-0.5, -0.45)--(-0.5, 1.6);

\draw (-0.2, 0.87) to[out=230, in=-230] (-0.2, 0.28);
\draw [looseness=1.5] (0, 0.87) to[out=-50, in=200] (0.37, 0.6);
\draw [looseness = 1.2] (0, 0.28) to[out=50, in=230] (0.43, 0.5);

\draw [looseness=1.2] (-0.2, 1.43) to[out=120, in=-90] (-0.3, 1.6);
\draw [looseness=1] (0, 1.43) to[out=50, in=-90] (0.1, 1.6);

\draw [looseness=1.2] (-0.2, -1.43+1.15) to[out=-120, in=90] (-0.3, -1.6+1.15);
\draw [looseness=1] (0, -1.43+1.15) to[out=-50, in=90] (0.1, -1.6+1.15);


\draw [thick, blue] (2.8, -1.6+1.15) arc (0:115:1.7);
\draw [thick, blue] ({2.8-1.7 - 1.7*cos(65)}, {-1.6+1.15 + 1.7*sin(65)}) arc (-65:-180:0.04);
\draw [thick, blue] ({2.8-1.7 - 1.7*cos(65) - 0.04 - 0.04*cos(65)}, {-1.6+1.15 + 1.7*sin(65) + 0.04*sin(65)})--({2.8-1.7 - 1.7*cos(65) - 0.04 - 0.04*cos(65)}, 1.6);

\draw [red] (2.7, -1.6+1.15) arc (0:115:1.6);
\draw [looseness=1.1, red] ({2.7 - 1.6 - 1.6*cos(65)}, {-1.6+1.15 + 1.6*sin(65)}) to[out=205, in=-90] (0.25, 1.6);

\draw [thick, blue] (0.4, -1.6+1.15) arc (180:0:0.85);
\draw [red] (0.3, -1.6+1.15) arc (180:0:0.95);

\draw (2.4, -1.6+1.15) arc (0:100:1.25);
\draw (2.5, -1.6+1.15) arc (0:100.5:1.4);

\draw (0.2, 1.15) arc (0:360:0.3);
\draw (0.1, 1.15) arc (0:360:0.2);

\node at (-0.1, 1.15) {\scriptsize $\theta_2$};

\draw (0.95, 0.7) arc (0:360:0.3);
\draw (0.85, 0.7) arc (0:360:0.2);

\node at (0.65, 0.7) {\scriptsize $\theta_1^*$};

\draw (0.2, 0) arc (0:360:0.3);
\draw (0.1, 0) arc (0:360:0.2);

\node at (-0.1, 0) {\scriptsize $\theta_3$};

\end{tikzpicture}
} = \raisebox{-1.6cm}{
\begin{tikzpicture}
\draw[ultra thick, gray] (0, 0.5)--(0, 2.1);
\draw[ultra thick, gray] (0, 0.5)--({-1 - 0.5*(2.1 - sqrt(3))*sqrt(3)}, {-0.5*(2.1 - sqrt(3))});

\draw[ultra thick, gray] (0, 0.5) to[out=30, in=110, looseness=1.5] (+{1 + 0.5*(2.1 - sqrt(3))*sqrt(3)}, {-0.5*(2.1 - sqrt(3))});

\draw [thick, gray, fill=white] (0.2, 0.5) arc (0:360:0.2);
\node at (0, 0.5) {$P$};

\draw [thick, gray, fill=zx_green] (0.4, 0.5) rectangle ++(0.4, 0.4);
\node at (0.6, 0.7) {$v_1$};

\draw [thick, gray, fill=zx_green] (0.2, 1.3-0.2) rectangle (-0.2, 1.3+0.2);
\node at (0, 1.3) {$v_2$};

\draw [thick, gray, fill=zx_green] ({-0.3 - 0.25*(2.1 - sqrt(3))*sqrt(3)}, {0.05-0.25*(2.1 - sqrt(3))}) rectangle ++(-0.4, 0.4);
\node at ({-0.5 - 0.25*(2.1 - sqrt(3))*sqrt(3)}, {0.25-0.25*(2.1 - sqrt(3))}) {$v_3$};
\end{tikzpicture}
} \,\, , 
\end{equation}
which is the RHS of Eq.~\eqref{eq:star-triangle-ising}.

\section{Discussion and Outlook}
\label{sec:outlook}
The Quon language has numerous potential applications beyond the ones we discussed in Sec.~\ref{sec:applications}. Below, we list a few examples.

\textit{Variational ansatz---}In Sec.~\ref{sec:quon-universal}, we introduce new classes of tractable tensor networks, namely punctured matchgate tensor networks and hybrid Clifford-matchgate-MPS. In Sec.~\ref{sec:applications}, we further describe how to construct factories for tractable networks, which are efficient methods for generating new tractable networks from existing ones. It is desirable to use these tractable networks as variational ansatze for solving quantum many-body problems, including quantum chemistry and optimization problems. Given that these ansatz networks generally exhibit high non-Cliffordness, high non-matchgateness, and large bipartite entanglement, they possess enhanced expressive power compared to conventional tractable networks. It would be particularly interesting to evaluate the performance of our ansatz states in quantum many-body systems with high correlations. 

\textit{Quantum circuit optimization and classical simulation---}In quantum circuits, resource minimization is one of the most important task. Resource minimizations arise in various contexts, such as reducing circuit depth, minimizing $T$-gate count minimization, and lowering error correction overhead. In many cases, the optimization process is captured in terms of gadgets. Gadgets help a quantum circuit transform into a form that is more amenable to classical simulation or fault-tolerant implementation. They allow one to ``translate'' a difficult gate into a combination of gates that are easier to implement or that possess favorable error-correcting properties. Because the diagrammatic compilation of a quantum circuit into a Quon diagram is efficient, we can ask whether further performing resource minimization via diagrammatic manipulation within the Quon language is possible, or even whether new gadgets can be discovered. Although a simplified Quon diagram often represents a tensor network rather than a quantum circuit, any simplification would be desirable, particularly for the classical simulation of quantum circuits. For example, suppose we have a $T$-gate-doped Clifford circuit, where the $T$-gate count often controls the classical simulability. Since the $T$-gate compiles into a scattering element in a Quon diagram, one can ask whether the number of scattering elements can be reduced by applying Quon diagrammatic rules, thereby lowering the simulation cost. Similarly, for $\textrm{SWAP}$-gate-doped matchgate quantum circuits, $\textrm{SWAP}$ gates introduces holes in the Quon diagrams, and the number of holes in a Quon diagram serves as a measure of the computational hardness of classical simulation. Moreover, unlike other graphical calculi like the ZX-calculus, which can be efficiently implemented through standard graph rewriting systems~\cite{kissingerPatternGraphRewrite2014,kissingerQuantomaticProofAssistant2015a,kissingerPyZXLargeScale2020a}, the presence of fractionalized degrees of freedom in Quon diagrams makes it challenging to represent and rewrite. Developing an efficient software implementation of the Quon language remains a valuable area of research and could provide a foundation for large-scale quantum circuit optimization and simulation.

\textit{Quantum circuit/Tensor network obfuscation---}It is often desirable to transform a given circuit or quantum circuit into a functionally equivalent one while hiding its underlying structure. This process, also known as \textit{obfuscation}, has many potential applications, especially in cryptography~\cite{alagic2016quantum, alagic2021impossibility, broadbent2021constructions, bartusek2022indistinguishability, bartusek2023obfuscation, coladangelo2024use, bartusek2024quantum}. By utilizing the Quon language, one can employ diagrammatic manipulations to transform an original circuit or tensor network into its obfuscated counterpart. For example, the three moves introduced in Sec.~\ref{sec:applications} can serve as ``obfuscators'' by being slightly modified to preserve tensor components. By keeping the details of the manipulations performed in each move secret, the tractability of the network, encoded as a 2D Quon diagram, is expected to be lost after several rounds, since no systematic recipe exists for handling a complex-looking 2D Quon diagram. Note that when obfuscating a unitary quantum circuit, the diagrammatic moves produces a 2D Quon diagram that can often be compiled efficiently into a conventional tensor network rather than a unitary quantum circuit. 

Here, we briefly explain how to modify the switching moves to preserve tensor components, as modifying the stretching and overlaying moves is straightforward. During a switching move, a constant factor may appear, which is assumed to be absorbed into the overall scalar of the 2D Quon diagram. In a modified switching move, as before, we consider a local region containing solely a braid, a generic scattering element, a Majorana string, or three scattering elements. If the local region contains a braid, we either switch the braid type according to braid type switching (TABLE~\ref{tab:majorana-rewriting-rules}) at the expense of introducing two dots, or split it into two scattering elements via the scattering Reidemeister move II (TABLE~\ref{tab:majorana-rewriting-rules}), choosing two angles whose sum equals the original braid angle. If the local region contains a generic scattering element, we split the scattering element into two scattering elements using the scattering Reidemeister move I from TABLE~\ref{tab:majorana-rewriting-rules}. If the local region contains a Majorana string, we introduce a scattering element with an arbitrary scattering angle using the scattering Reidemeister move I. Finally, if the local region contains three scattering elements, we apply the Yang-Baxter equation (TABLE~\ref{tab:majorana-rewriting-rules}), only when such a local reconfiguration is possible. Additionally, one may further enrich the move, for instance, by replacing two parallel Majorana strings with four braids through two successive applications of a rule from TABLE~\ref{tab:majorana-rewriting-rules}. 

\textit{Soundness and Completeness of Quon language---}Given intuitive yet comprehensive diagrammatic rewriting rules of Quon, one might ask, as a graphical language, whether it is \textit{sound} and \textit{complete}. A graphical language is called sound if a series of graphical rewritings never lead to a contradiction. The soundness of Quon can easily be shown by observing that when a Quon diagram is viewed as a quantum process, the two terms---on the LHS and RHS---of each rewriting rule represent the identical quantum process. The completeness of a graphical language asks whether two diagrams representing the same quantum process can be transformed one another using only a series of rewriting rules. In the ZX-calculus, it has been shown that the set of ``axioms'' is complete. In fact, using the rewriting rules of 2D Quon, one can derive all the axioms of the ZX-calculus, thereby demonstrating that the 2D Quon language is complete. Nevertheless, it would be desirable to present a completeness proof for Quon that does not relying on the completeness result of the ZX-calculus. 

\textit{Algebraic formulation of Quon language---}While the pictorial nature of Quon offers significant flexibility, developing an algebraic formulation would be highly beneficial---for instance, for implementing the Quon language in computer programming code. By algebraic formulation, we mean assigning an algebraic object, such as a polynomial, to every tensor network, so that tensor operations performed on the networks can equivalently be performed by algebraic operations on the corresponding algebraic objects. For every Clifford tensor network, we can essentially associate a quadratic polynomial in binary variables~\cite{cai2018clifford}, and for every matchgate tensor network, we can essentially associate a quadratic polynomial in Grassmann variable~\cite{PhysRevA.59.1538, bravyi2009contraction}. It would be interesting to integrate these two approaches in a manner that is compatible with Quon diagrammatic representations. 

\textit{Physical implementation---}It is possible to realize Majorana zero modes (MZMs) in real materials, such as in hybrid semiconductor-superconductor heterostructure~\cite{lutchyn2018majorana}. In MZM-based architectures, logical qubits are often realized using the \textit{tetron} or \textit{hexon} architecture, where logical quantum operations are performed using the parity measurements of MZMs~\cite{PhysRevB.95.235305, 10.21468/SciPostPhys.8.6.091, PhysRevLett.128.180504, aasen2025roadmap}. As the quantum operations on these architectures can be captured by anyon diagrams using the Ising topological quantum field theory, it would be interesting to establish a dictionary between the anyon diagrams for quantum processes in tetrons and hexons and Quon diagrams. 

\textit{Integrable models using Quon---}As discussed in Sec.~\ref{sec:applications}, the Quon language provides succint pictorial proofs of the Kramers-Wannier duality and the star-triangle relation for the two-dimensional Ising model, one of the simplest integrable models. Given its flexibility, it is compelling to explore whether Quon can also encapsulate other two-dimensional integrable models, such as those solvable by the algebraic Bethe ansatz. A central aspect of the algebraic Bethe ansatz is the Yang-Baxter equation, which can be expressed in terms of a tensor network~\cite{PhysRevB.86.045125}. It would be interesting to investigate whether Quon can give a simple derivation of the Yang-Baxter equation for integrable models beyond Ising, notably the result from a mathematics literature on generalizing exactly solvable models based on the Yang-Baxter equation~\cite{jones2007and}. Moreover, exploring whether the Quon language can generate new solvable models would be a valuable direction for future research. 

\textit{Note added}---Please refer to Ref.~\onlinecite{quonAlgoPaper}, which presents other mathematical properties of Quon diagrams and applications to classical simulation. 

\begin{acknowledgments}
We thank Julian Bender, Garnet Kin-Lic Chan, Bryan O'Gorman, Jeongwan Haah, Dominik Hangleiter, Isaac H. Kim, Joonho Lee, and Kyungjoo Noh for illuminating discussions. BK acknowledges discussions with Patrick A. Lee. BK is supported by DOE office of Basic Sciences Grant No. DE-FG02-03ER46076 (theory). BK and SC acknowledge the support from the Center for Ultracold Atoms, an NSF Physics Frontiers Center (Grant No. 2317134). CZ acknowledges support from the Wellcome Leap Quantum for Bio program. ZL would like to thank Arthur Jaffe and Yunxiang Ren for early discussions. ZL is supported by Beijing Natural Science Foundation Key Program (Grant No.Z220002). ZL was supported by Templeton Religion Trust (TRT159).
\end{acknowledgments}

\appendix
\section{General Qubit Encoding}
\label{app:encoding-general}
When $4$ Majorana lines terminate on an open interval of the boundary of the background manifold or on a basis encoder, they encode $1$ qubit. In this section, we discuss how to systematically encode qubits when more than $4$ Majoranas terminate. Specifically, when $2+2p$ Majoranas terminate on an open interval or on a basis encoder, the corresponding Hilbert space, after imposing the global parity-even projection, has dimension $2^p$. In the case of a basis encoder, it always contains the \textit{pairing data}, which captures how pairs of Majoranas are created from the vacuum, as represented by a Majorana diagram with braids including braiding processes that occur during pair creations. When the pairing data is provided, the logical qubits can be encoded in a way that respects this pairing data. We first present a method for qubit encoding without the pairing data (applicable to open intervals), and then show how to incorporate the pairing data (applicable to basis encoders). 

\subsection{Encoding without pairing data}
If there are $2+2p$ Majorana lines terminate on an open interval, a logical encoding can be achieved by by gluing the following isomorphism: 
\begin{equation}
\label{eq:encoding-map-2+2p}
\hat{\Phi}_{2+2p} := \frac{1}{(\sqrt{2})^{p+1}} \raisebox{-0.2cm}{
\tikz{
\begin{scope}[even odd rule]
\clip (-0.2, 0.4) rectangle (4.1, 1)  (0.8, 1) arc (0:-180:0.2)  (1.7, 1) arc (0:-180:0.2)  (3.5, 1) arc (0:-180:0.2);

\fill [blue!10] (-0.2, 0.4) rectangle (4.1, 1);
\end{scope}

\draw[red] (0,0.4)--(0,1);
\draw (0.1,0.4)--(0.1,1);
\draw (0.2,0.4)--(0.2,1);

\draw (1,0.4)--(1,1);
\draw (1.1,0.4)--(1.1,1);

\draw (1.9,0.4)--(1.9,1);
\draw (2,0.4)--(2,1);

\draw (2.8,0.4)--(2.8,1);
\draw (2.9,0.4)--(2.9,1);

\draw (3.7,0.4)--(3.7,1);
\draw (3.8,0.4)--(3.8,1);
\draw[red] (3.9,0.4)--(3.9,1);

\draw[red] (0.9, 1) arc (0:-180:0.3);
\draw[thick,blue] (0.8, 1) arc (0:-180:0.2);

\draw[red] (1.8, 1) arc (0:-180:0.3);
\draw[thick,blue] (1.7, 1) arc (0:-180:0.2);

\draw[red] (3.6, 1) arc (0:-180:0.3);
\draw[thick,blue] (3.5, 1) arc (0:-180:0.2);

\draw[thick,blue] (-0.2,0.4)--(-0.2,1);
\draw[thick,blue] (4.1,0.4)--(4.1,1);

\node at (2.45, 0.7) {$\cdots$};
}} \,\, , 
\end{equation}
where $2+2p$ Majoranas terminate on the lower part, while $p$ groups of $4$ Majoranas terminate on $p$ disjoint open intervals on the upper part. For clarity of presentation, we color the boundary-tracking Majoranas in red. Using the string-genus relation Eq.~\eqref{eq:string-genus}, one can indeed show that $\hat{\Phi}_{2 + 2p}$ is an isomorphism between two Hilbert spaces, each of dimension $2^p$. Note that $\hat{\Phi}_2$ and $\hat{\Phi}_4$ are merely identity maps: 
\begin{equation}
\hat{\Phi}_{2} := \raisebox{-0.2cm}{
\tikz{
\fill [blue!10] (-0.2, 0.4) rectangle (0.4, 1);

\draw[red] (0,0.4)--(0,1);
\draw[red] (0.2,0.4)--(0.2,1);

\draw[thick,blue] (-0.2,0.4)--(-0.2,1);
\draw[thick,blue] (0.4,0.4)--(0.4,1);
}} \quad \textrm{and} \quad 
\hat{\Phi}_{4} := \raisebox{-0.2cm}{
\tikz{
\fill [blue!10] (-0.2, 0.4) rectangle (0.8, 1);

\draw[red] (0, 0.4)--(0,1);
\draw (0.2, 0.4)--(0.2, 1);
\draw (0.4, 0.4)--(0.4, 1);
\draw[red] (0.6,0.4)--(0.6,1);

\draw[thick,blue] (-0.2,0.4)--(-0.2,1);
\draw[thick,blue] (0.8,0.4)--(0.8,1);
}} 
\end{equation}
on a $1$-dimensional Hilbert space and a $2$-dimensional Hilbert space, respectively. 

Using Eq.~\eqref{eq:encoding-map-2+2p}, we can construct the (logical) computational basis states. When $p=0$, there is a unique normalized state given by 
\begin{equation}
\vert 0 \rangle = \frac{1}{\sqrt{2}} \raisebox{-0.1cm}{
\tikz{
\fill [blue!10] (0.4, 1) arc (0:180:0.4);

\draw[red] (0.15,1) arc (0:180:0.15);

\draw[thick,blue] (0.4, 1) arc (0:180:0.4);
}} \,\, . 
\end{equation}
When $p=1$, the logical $\vert 0 \rangle$ and $\vert 1 \rangle$ states are presented in Eq.~\eqref{eq:logical-qubit-4-majoranas}. When $p=2$, for example, the logical $\vert 0, 1 \rangle$ state is given by 
\begin{equation}
\vert 0, 1 \rangle = \frac{1}{(\sqrt{2})^3} \raisebox{-0.4cm}{
\begin{tikzpicture}[scale=0.7]
\draw[thick, blue, fill=blue!10] (0, 0.2)--(0, 1) arc (180:0:0.75) arc (-180:0:0.3) arc (180:0:0.75)--(3.6, 0.2);

\draw[red] (0.3, 0.2)--(0.3, 1) arc (180:0:0.45) arc (-180:0:0.6) arc (180:0:0.45)--(3.3, 0.2);
\draw (0.6, 0.2)--(0.6, 1) arc (180:0:0.15)--(0.9, 0.2);
\draw (2.7, 0.2)--(2.7, 1) arc (180:0:0.15)--(3, 0.2);

\node at (3, 1) [circle, fill, inner sep=1pt] {};
\node at (3.3, 1) [circle, fill, inner sep=1pt] {};
\end{tikzpicture}
} . 
\end{equation}
More generally, the logical $\vert b_1, b_2, \ldots, b_p \rangle$ state is given by 
\begin{equation}
\frac{1}{(\sqrt{2})^p} \,\, \raisebox{-0.6cm}{
\begin{tikzpicture}[scale=0.6]
\fill[blue!10] (0, 0.2)--(0, 1) arc (180:0:0.75) arc (-180:0:0.3) arc (180:0:0.75)--(3.6, 0.2);
\fill[blue!10] (4.8, 0.2)--(4.8, 1) arc (180:0:0.75) arc (-180:0:0.3) arc (180:0:0.75)--(8.4, 0.2);

\fill[gray, fill opacity=0.7] (0, 1) arc (180:0:0.75);
\fill[gray, fill opacity=0.7] (2.1, 1) arc (180:0:0.75);
\fill[gray, fill opacity=0.7] (4.8, 1) arc (180:0:0.75);
\fill[gray, fill opacity=0.7] (6.9, 1) arc (180:0:0.75);

\draw[ultra thick, gray] (0.75, 1.75)--(0.75, 2.2);
\draw[ultra thick, gray] (2.85, 1.75)--(2.85, 2.2);
\draw[ultra thick, gray] (5.55, 1.75)--(5.55, 2.2);
\draw[ultra thick, gray] (7.65, 1.75)--(7.65, 2.2);

\node at (0.75, 2.55) {$b_1$};
\node at (2.85, 2.55) {$b_2$};
\node at (5.55, 2.55) {$b_{p-1}$};
\node at (7.65, 2.55) {$b_p$};

\draw[thick, blue] (0, 0.2)--(0, 1) arc (180:0:0.75) arc (-180:0:0.3) arc (180:0:0.75);

\draw[red] (0.3, 0.2)--(0.3, 1) arc (180:0:0.45) arc (-180:0:0.6) arc (180:0:0.45) arc (-180:-120:0.6);
\draw (0.6, 0.2)--(0.6, 1) arc (180:0:0.15)--(0.9, 0.2);
\draw (2.7, 0.2)--(2.7, 1) arc (180:0:0.15)--(3, 0.2);

\node at (4.3, 0.55) {$\cdots$};

\draw[thick, blue] (4.8, 1) arc (180:0:0.75) arc (-180:0:0.3) arc (180:0:0.75)--(8.4, 0.2);

\draw[red] (8.1, 0.2)--(8.1, 1) arc (0:180:0.45) arc (0:-180:0.6) arc (0:180:0.45) arc (0:-60:0.6);
\draw (5.4, 0.2)--(5.4, 1) arc (180:0:0.15)--(5.7, 0.2);
\draw (7.5, 0.2)--(7.5, 1) arc (180:0:0.15)--(7.8, 0.2);
\end{tikzpicture}
} \,\, , 
\end{equation}
where $b_i \in \{0, 1\}$, and we used the basis encoder Eq.~\eqref{eq:basis-encoder}. Finally, the resolution of identity can be represented as
\begin{equation}
\raisebox{-1.2cm}{
\begin{tikzpicture}
\fill[blue!10] (0, 0) rectangle (2, 1.5);

\draw[thick, blue] (0, 0)--(0, 1.5);
\draw[thick, blue] (2, 0)--(2, 1.5);

\draw[red] (0.2, 0)--(0.2, 1.5);
\draw (0.4, 0)--(0.4, 1.5);
\draw (0.6, 0)--(0.6, 1.5);

\node at (1.05, 0.75) {$\cdots$};

\draw (1.4, 0)--(1.4, 1.5);
\draw (1.6, 0)--(1.6, 1.5);
\draw[red] (1.8, 0)--(1.8, 1.5);

\draw [decorate, decoration = {brace, mirror}] (0.15, -0.05)--(1.85, -0.05);
\node at (1, -0.35) {\scriptsize $2+2p$};
\end{tikzpicture}
} = \frac{1}{(\sqrt{2})^{3p-1}} \,\, \raisebox{-0.85cm}{
\begin{tikzpicture}[scale=0.5]
\fill[blue!10] (0, 0.2)--(0, 1) arc (180:0:0.75) arc (-180:0:0.3) arc (180:0:0.75)--(3.6, 0.2);
\fill[blue!10] (0, 3.8)--(0, 3) arc (-180:0:0.75) arc (180:0:0.3) arc (-180:0:0.75)--(3.6, 3.8);
\fill[blue!10] (4.8, 0.2)--(4.8, 1) arc (180:0:0.75) arc (-180:0:0.3) arc (180:0:0.75)--(8.4, 0.2);
\fill[blue!10] (4.8, 3.8)--(4.8, 3) arc (-180:0:0.75) arc (180:0:0.3) arc (-180:0:0.75)--(8.4, 3.8);

\fill[gray, fill opacity=0.7] (0, 1) arc (180:0:0.75);
\fill[gray, fill opacity=0.7] (2.1, 1) arc (180:0:0.75);
\fill[gray, fill opacity=0.7] (4.8, 1) arc (180:0:0.75);
\fill[gray, fill opacity=0.7] (6.9, 1) arc (180:0:0.75);

\fill[gray, fill opacity=0.7] (0, 3) arc (-180:0:0.75);
\fill[gray, fill opacity=0.7] (2.1, 3) arc (-180:0:0.75);
\fill[gray, fill opacity=0.7] (4.8, 3) arc (-180:0:0.75);
\fill[gray, fill opacity=0.7] (6.9, 3 ) arc (-180:0:0.75);

\draw[ultra thick, gray] (0.75, 1.75)--(0.75, 2.25);
\draw[ultra thick, gray] (2.85, 1.75)--(2.85, 2.25);
\draw[ultra thick, gray] (5.55, 1.75)--(5.55, 2.25);
\draw[ultra thick, gray] (7.65, 1.75)--(7.65, 2.25);

\draw[thick, blue] (0, 0.2)--(0, 1) arc (180:0:0.75) arc (-180:0:0.3) arc (180:0:0.75);
\draw[thick, blue] (0, 3.8)--(0, 3) arc (-180:0:0.75) arc (180:0:0.3) arc (-180:0:0.75);
\draw[thick, blue] (4.8, 1) arc (180:0:0.75) arc (-180:0:0.3) arc (180:0:0.75)--(8.4, 0.2);
\draw[thick, blue] (4.8, 3) arc (-180:0:0.75) arc (180:0:0.3) arc (-180:0:0.75)--(8.4, 3.8);

\node at (4.3, 0.55) {$\cdots$};
\node at (4.3, 3.45) {$\cdots$};

\draw[red] (0.3, 0.2)--(0.3, 1) arc (180:0:0.45) arc (-180:0:0.6) arc (180:0:0.45) arc (-180:-120:0.6);
\draw (0.6, 0.2)--(0.6, 1) arc (180:0:0.15)--(0.9, 0.2);
\draw (2.7, 0.2)--(2.7, 1) arc (180:0:0.15)--(3, 0.2);

\draw[red] (0.3, 3.8)--(0.3, 3) arc (-180:0:0.45) arc (180:0:0.6) arc (-180:0:0.45) arc (180:120:0.6);
\draw (0.6, 3.8)--(0.6, 3) arc (-180:0:0.15)--(0.9, 3.8);
\draw (2.7, 3.8)--(2.7, 3) arc (-180:0:0.15)--(3, 3.8);

\draw[red] (8.1, 0.2)--(8.1, 1) arc (0:180:0.45) arc (0:-180:0.6) arc (0:180:0.45) arc (0:-60:0.6);
\draw (5.4, 0.2)--(5.4, 1) arc (180:0:0.15)--(5.7, 0.2);
\draw (7.5, 0.2)--(7.5, 1) arc (180:0:0.15)--(7.8, 0.2);

\draw[red] (8.1, 3.8)--(8.1, 3) arc (0:-180:0.45) arc (0:180:0.6) arc (0:-180:0.45) arc (0:60:0.6);
\draw (5.4, 3.8)--(5.4, 3) arc (-180:0:0.15)--(5.7, 3.8);
\draw (7.5, 3.8)--(7.5, 3) arc (-180:0:0.15)--(7.8, 3.8);
\end{tikzpicture}
} \,\, , 
\end{equation}
which generalizes Eq.~\eqref{eq:resolution-of-id-2}. 

When using $\Phi_{2+2p}$ Eq.~\eqref{eq:encoding-map-2+2p} to encode logical qubits, we make a particular choice in the pairing data. We will see below that, our particular pairing choice can also be presented in terms of the following basis encoder: 
\begin{equation}
\label{eq:basis-encoder-base}
\raisebox{-0.6cm}{
\begin{tikzpicture}[scale=0.6]
\fill[blue!10] (0, 0) arc (180:0:1.6);
\fill[gray, fill opacity=0.7] (0, 0) arc (180:0:1.6);
\draw[thick, blue] (0, 0) arc (180:0:1.6);

\draw[red] (0.4, 0) arc (180:0:1.2);
\draw (0.6, 0) arc (180:0:0.3);
\draw (2, 0) arc (180:0:0.3);

\node at (1.65, 0.2) {$\cdots$};

\draw[ultra thick, gray] ({1.6-0.8*sqrt(2)}, {0.8*sqrt(2)})--({1.6-1.2*sqrt(2)}, {1.2*sqrt(2)});
\draw[ultra thick, gray] ({1.6+0.8*sqrt(2)}, {0.8*sqrt(2)})--({1.6+1.2*sqrt(2)}, {1.2*sqrt(2)});

\node at (-0.25, 2.05) {$b_1$};
\node at (3.45, 2.05) {$b_p$};

\node at (1.65, 2.05) {$\cdots$};
\end{tikzpicture}
} := 
\raisebox{-0.6cm}{
\begin{tikzpicture}[scale=0.6]
\fill[blue!10] (0, 0.2)--(0, 1) arc (180:0:0.75)--(1.5, 0.2);
\fill[blue!10] (2.8, 0.2)--(2.8, 1) arc (180:0:0.75)--(4.3, 0.2);

\fill[gray, fill opacity=0.7] (0, 1) arc (180:0:0.75);
\fill[gray, fill opacity=0.7] (2.8, 1) arc (180:0:0.75);

\draw[ultra thick, gray] (0.75, 1.75)--(0.75, 2.2);
\draw[ultra thick, gray] (3.55, 1.75)--(3.55, 2.2);

\node at (0.75, 2.55) {$b_1$};
\node at (3.55, 2.55) {$b_p$};

\draw[thick, blue] (0, 0.2)--(0, 1) arc (180:0:0.75);

\draw[red] (0.3, 0.2)--(0.3, 1) arc (180:0:0.45) arc (-180:-120:0.6);
\draw (0.6, 0.2)--(0.6, 1) arc (180:0:0.15)--(0.9, 0.2);

\node at (2.2, 0.55) {$\cdots$};

\draw[thick, blue] (2.8, 1) arc (180:0:0.75)--(4.3, 0.2);

\draw[red] (4, 0.2)--(4, 1) arc (0:180:0.45) arc (0:-60:0.6);
\draw (3.4, 0.2)--(3.4, 1) arc (180:0:0.15)--(3.7, 0.2);
\end{tikzpicture}
} \,\, , 
\end{equation}
where the pairing data is shown in the shaded region on the LHS. 

\subsection{Encoding with the pairing data}
When the pairing data is provided, we can encode the logical qubits in a manner that respects the pairing data. For example, consider the following basis encoder supporting $8$ Majoranas, where the pairing data is explicitly shown: 
\begin{equation}
\raisebox{-0.5cm}{
\begin{tikzpicture}[scale=0.6]
\fill[blue!10] (0, 0) arc (180:0:1.55);
\fill[gray, fill opacity=0.7] (0, 0) arc (180:0:1.55);

\draw (0.4, 0) arc (180:80:0.4);
\draw (1.2, 0) arc (0:60:0.4);

\draw (0.8, 0) arc (180:50:0.7);
\draw (2.2, 0) arc (0:30:0.7);

\draw (1.4, 0) arc (180:0:0.3);

\draw (1.7, 0) arc (180:160:0.5);
\draw (2.7, 0) arc (0:140:0.5);

\draw[ultra thick, gray] ({1.55-0.775*sqrt(2)}, {0.775*sqrt(2)})--({1.55-1.2*sqrt(2)}, {1.2*sqrt(2)});
\draw[ultra thick, gray] (1.55, 1.55)--(1.55, 2);
\draw[ultra thick, gray] ({1.55+0.775*sqrt(2)}, {0.775*sqrt(2)})--({1.55+1.2*sqrt(2)}, {1.2*sqrt(2)});

\node at (-0.25, 2.05) {$b_1$};
\node at (1.55, 2.3) {$b_2$};
\node at (3.4, 2) {$b_3$};

\draw[thick, blue] (0, 0) arc (180:0:1.55);
\end{tikzpicture}
} = \raisebox{-0.5cm}{
\begin{tikzpicture}[scale=0.6]
\fill[blue!10] (0, 0) arc (180:0:1.55);
\fill[gray, fill opacity=0.7] (0, 0) arc (180:0:1.55);

\draw (0.4, 0.2) rectangle (2.7, 0.8);

\node at (1.55, 0.5) {\small $p$};

\draw (0.5, 0)--(0.5, 0.2);
\draw (0.8, 0)--(0.8, 0.2);
\draw (1.1, 0)--(1.1, 0.2);
\draw (1.4, 0)--(1.4, 0.2);
\draw (1.7, 0)--(1.7, 0.2);
\draw (2, 0)--(2, 0.2);
\draw (2.3, 0)--(2.3, 0.2);
\draw (2.6, 0)--(2.6, 0.2);

\draw[ultra thick, gray] ({1.55-0.775*sqrt(2)}, {0.775*sqrt(2)})--({1.55-1.2*sqrt(2)}, {1.2*sqrt(2)});
\draw[ultra thick, gray] (1.55, 1.55)--(1.55, 2);
\draw[ultra thick, gray] ({1.55+0.775*sqrt(2)}, {0.775*sqrt(2)})--({1.55+1.2*sqrt(2)}, {1.2*sqrt(2)});

\node at (-0.25, 2.05) {$b_1$};
\node at (1.55, 2.3) {$b_2$};
\node at (3.4, 2) {$b_3$};

\draw[thick, blue] (0, 0) arc (180:0:1.55);
\end{tikzpicture}
} , 
\end{equation}
where $p$ denotes the Majorana diagram for the pairing data. From the pairing data, we can determine the Majorana pairings, which in this example are $\{(1, 3), (2, 7), (4, 6), (5, 8)\}$, where we used the ordered tuples. The logical qubit encoding can be obtained by adding dots to the leftmost Majoranas of each Majorana pair, which in our previous example are $1$, $2$, $4$, and $5$: 
\begin{equation}
\label{eq:basis-encoder-general}
\raisebox{-0.5cm}{
\begin{tikzpicture}[scale=0.6]
\fill[blue!10] (0, 0) arc (180:0:1.55);
\fill[gray, fill opacity=0.7] (0, 0) arc (180:0:1.55);

\draw (0.4, 0) arc (180:80:0.4);
\draw (1.2, 0) arc (0:60:0.4);

\draw (0.8, 0) arc (180:50:0.7);
\draw (2.2, 0) arc (0:30:0.7);

\draw (1.4, 0) arc (180:0:0.3);

\draw (1.7, 0) arc (180:160:0.5);
\draw (2.7, 0) arc (0:140:0.5);

\draw[ultra thick, gray] ({1.55-0.775*sqrt(2)}, {0.775*sqrt(2)})--({1.55-1.2*sqrt(2)}, {1.2*sqrt(2)});
\draw[ultra thick, gray] (1.55, 1.55)--(1.55, 2);
\draw[ultra thick, gray] ({1.55+0.775*sqrt(2)}, {0.775*sqrt(2)})--({1.55+1.2*sqrt(2)}, {1.2*sqrt(2)});

\node at (-0.25, 2.05) {$b_1$};
\node at (1.55, 2.3) {$b_2$};
\node at (3.4, 2) {$b_3$};

\draw[thick, blue] (0, 0) arc (180:0:1.55);
\end{tikzpicture}
} = \raisebox{-0.5cm}{
\begin{tikzpicture}[scale=0.8]
\draw[thick, blue, fill=blue!10] (-0.5, -0.5) arc (180:0:2);

\draw (0.1, 0.3) rectangle (2.9, 0.8);
\node at (1.5, 0.55) {$p$};

\draw (0.1, -0.3) rectangle (2.9, 0.1);
\node at (1.5, -0.1) {$q$};

\draw (0.45, -0.5)--(0.45, -0.3);
\draw (0.45, 0.1)--(0.45, 0.3);

\draw (0.75, -0.5)--(0.75, -0.3);
\draw (0.75, 0.1)--(0.75, 0.3);

\draw (1.05, -0.5)--(1.05, -0.3);
\draw (1.05, 0.1)--(1.05, 0.3);

\draw (1.35, -0.5)--(1.35, -0.3);
\draw (1.35, 0.1)--(1.35, 0.3);

\draw (1.65, -0.5)--(1.65, -0.3);
\draw (1.65, 0.1)--(1.65, 0.3);

\draw (1.95, -0.5)--(1.95, -0.3);
\draw (1.95, 0.1)--(1.95, 0.3);

\draw (2.25, -0.5)--(2.25, -0.3);
\draw (2.25, 0.1)--(2.25, 0.3);

\draw (2.55, -0.5)--(2.55, -0.3);
\draw (2.55, 0.1)--(2.55, 0.3);
\end{tikzpicture}
} , 
\end{equation}
where $p$ is the Majorana diagram for the pairing data and $q$ is the Majorana diagram with dots: 
\begin{equation}
\raisebox{-0.5cm}{
\begin{tikzpicture}
\draw (0.1, -0.3) rectangle (2.9, 0.1);
\node at (1.5, -0.1) {$q$};

\draw (0.45, -0.7)--(0.45, -0.3);
\draw (0.45, 0.1)--(0.45, 0.5);

\draw (0.75, -0.7)--(0.75, -0.3);
\draw (0.75, 0.1)--(0.75, 0.5);

\draw (1.05, -0.7)--(1.05, -0.3);
\draw (1.05, 0.1)--(1.05, 0.5);

\draw (1.35, -0.7)--(1.35, -0.3);
\draw (1.35, 0.1)--(1.35, 0.5);

\draw (1.65, -0.7)--(1.65, -0.3);
\draw (1.65, 0.1)--(1.65, 0.5);

\draw (1.95, -0.7)--(1.95, -0.3);
\draw (1.95, 0.1)--(1.95, 0.5);

\draw (2.25, -0.7)--(2.25, -0.3);
\draw (2.25, 0.1)--(2.25, 0.5);

\draw (2.55, -0.7)--(2.55, -0.3);
\draw (2.55, 0.1)--(2.55, 0.5);
\end{tikzpicture}
} = i^{(b_\textrm{tot})^2} \raisebox{-0.5cm}{
\begin{tikzpicture}
\draw (0.45, -0.7)--(0.45, 0.5);
\draw (0.75, -0.7)--(0.75, 0.5);
\draw (1.05, -0.7)--(1.05, 0.5);
\draw (1.35, -0.7)--(1.35, 0.5);
\draw (1.65, -0.7)--(1.65, 0.5);
\draw (1.95, -0.7)--(1.95, 0.5);
\draw (2.25, -0.7)--(2.25, 0.5);
\draw (2.55, -0.7)--(2.55, 0.5);

\node at (0.45, -0.5) [circle, fill, inner sep=1pt] {};
\node at (0.75, -0.2) [circle, fill, inner sep=1pt] {};
\node at (1.35, 0.1) [circle, fill, inner sep=1pt] {};
\node at (1.65, 0.3) [circle, fill, inner sep=1pt] {};

\node at (0.2, -0.5) {\scriptsize $b_\textrm{tot}$};
\node at (0.6, -0.2) {\scriptsize $b_1$};
\node at (1.2, 0.1) {\scriptsize $b_2$};
\node at (1.5, 0.3) {\scriptsize $b_3$};
\end{tikzpicture}
} , 
\end{equation}
with $b_\textrm{tot} = b_1 + b_2 + b_3$ (the summation is over $\mathbb{Z}$). The scalar coefficient $i^{(b_\textrm{tot})^2}$ is chosen such that Eq.~\eqref{eq:basis-encoder-general} is consistent with Eq.~\eqref{eq:basis-encoder-base}. Note that two basis encoders Eq.~\eqref{eq:basis-encoder-base} and Eq.~\eqref{eq:basis-encoder-general} \textcolor{black}{differ by a unitary that is both Clifford and matchgate}. 

\section{Time-direction and Dots}
\label{app:dots-and-time-direction}
Here, we provide some technical remarks on the necessity of the time direction when dots are present in a 2D Quon diagram. As previously explained, Majorana and Quon diagrams can be interpreted as a quantum processes, with the time direction either implicitly assumed to flow from top to bottom or explicitly indicated by a time-arrow. On the other hand, given the Reidemeister moves for braids, as well as their proper generalizations for scatterings, both of which prsented in TABLE~\ref{tab:majorana-rewriting-rules}, one might ask whether it is possible to completely disregard the time direction and instead select an arbitrary time direction only when needed. This situation can be compared with the case of knot (or link) invariants. Given a knot diagram and its rotated diagram in the plane, the two diagrams evaluate to the same (framed) knot invariant because they are related by a series of Reidemeister moves (and share the same writhe). Similarly, if a closed Majorana or Quon diagram contains braids and scattering elements but \textit{no} dots, the evaluation of the diagram is identical to that of its rotated counterpart, assuming that the time direction is fixed to flow from top to bottom in both cases. 

However, when a Majorana or Quon diagram contains dots, the time direction becomes essential. Changing the time direction alters the relative ordering between dots, and correspondingly, the the scalar coefficient of the diagram changes via the anti-commutation relation. For example, consider the following closed Majorana diagram and its rotated counterpart: 
\begin{equation}
\raisebox{-0.66cm}{
\begin{tikzpicture}
\draw (0.1, -0.32) to[out=26, in=-90] (0.4, 0) arc (0:180:0.4) to[out=-90, in=90] (0.4, -0.7) arc (0:-180:0.4) to[out=90, in=-160] (-0.1, -0.39);

\node at (-0.3, -0.2) [circle, fill, inner sep=1pt] {};
\node at (0.3, -0.2) [circle, fill, inner sep=1pt] {};
\end{tikzpicture}
} \,\, \ne \raisebox{-0.45cm}{
\begin{tikzpicture}[rotate=70]
\draw (0.1, -0.32) to[out=26, in=-90] (0.4, 0) arc (0:180:0.4) to[out=-90, in=90] (0.4, -0.7) arc (0:-180:0.4) to[out=90, in=-160] (-0.1, -0.39);

\node at (-0.3, -0.2) [circle, fill, inner sep=1pt] {};
\node at (0.3, -0.2) [circle, fill, inner sep=1pt] {};
\end{tikzpicture}
} \,\, , 
\end{equation}
where the evaluations of the two diagrams differ by a factor of $i$. In general, two closed Majorana or Quon diagrams that are related by a rotation in the plane evaluate to complex numbers differing by a factor of a power of $i$. On the other hand, a pair of dots can be ``pair-up'' and replaced with braids using the following trick: 
\begin{equation}
\label{eq:pair-of-dots-to-braid}
\raisebox{-0.43cm}{
\tikz{
\draw (-1.2, 0)--(-1.2, 1);
\draw (-0.8, 0)--(-0.8, 1);

\draw (0, 0)--(0, 1);
\draw (0.4, 0)--(0.4, 1);

\node at (-0.35, 0.5) {$\cdots$};

\node at (-1.2, 0.5) [circle, fill, inner sep=1pt] {};
\node at (0.4, 0.5) [circle, fill, inner sep=1pt] {};
}} \,\, = \raisebox{-0.43cm}{
\tikz{
\draw (-1.2, 0)--(-1.2, 1);

\draw (-0.8, 0)--(-0.8, 1);
\draw (0, 0)--(0, 1);

\draw (0.4, 0) to[out=120, in=-17, looseness=0.5] (0.05, 0.16);
\draw (-0.05, 0.18) to[out=170, in=-20] (-0.15, 0.2);
\draw (-0.65, 0.29) to[out=170, in=-20] (-0.75, 0.32);
\draw (-0.85, 0.35) to[out=160, in=-160, looseness=1.8] (-0.85, 0.65);
\draw (-0.75, 0.68) to[out=20, in=-170] (-0.65, 0.71);
\draw (-0.15, 0.8) to[out=20, in=-170] (-0.05, 0.82);
\draw (0.05, 0.84) to[out=17, in=-120, looseness=0.5] (0.4, 1);

\node at (-0.35, 0.5) {$\cdots$};

\node at (-1.2, 0.5) [circle, fill, inner sep=1pt] {};
\node at (-1, 0.5) [circle, fill, inner sep=1pt] {};
}} \,\, = e^{-i \frac{\pi}{4}} \raisebox{-0.43cm}{
\tikz{
\draw (-0.8, 0)--(-0.8, 0.6);
\draw (-0.8, 0.75)--(-0.8, 1);

\draw (-0.6, 0)--(-0.6, 1);
\draw (0, 0)--(0, 1);

\draw (0.4, 0) to[out=120, in=-17, looseness=0.5] (0.05, 0.16);
\draw (-0.05, 0.18) to[out=170, in=-20] (-0.15, 0.2);
\draw (-0.45, 0.25) to[out=170, in=-20] (-0.55, 0.27);
\draw (-0.65, 0.29) to[out=170, in=-20] (-0.75, 0.32);
\draw (-0.85, 0.35) to[out=160, in=-160, looseness=1.8] (-0.85, 0.65) to[out=20, in=-170] (-0.65, 0.71);
\draw (-0.55, 0.73) to[out=20, in=-170] (-0.45, 0.75);
\draw (-0.15, 0.8) to[out=20, in=-170] (-0.05, 0.82);
\draw (0.05, 0.84) to[out=17, in=-120, looseness=0.5] (0.4, 1);

\node at (-0.3, 0.5) {\scriptsize $\cdots$};
}} \,\, , 
\end{equation}
where we replaced two parallel dots with two braids following TABLE~\ref{tab:majorana-rewriting-rules} in the last equality. Therefore, when a Majorana or Quon diagram contains an even number of dots, this trick can be applied to completely replace the dots from the diagram, allowing the time-direction to be omitted. 

Finally, we note that, an isolated dot on any time slice is not permitted in a 2D Quon diagram, because such a diagram evaluates to zero due to the parity-even projection: 
\begin{equation}
\raisebox{-0.4cm}{\tikz{
\fill [blue!10] (-0.2, 0) rectangle (1.4, 1);

\draw (0, 0)--(0, 1);
\draw (0.2, 0)--(0.2, 1);
\draw (0.4, 0)--(0.4, 1);
\draw (1.2, 0)--(1.2, 1);

\draw[thick,blue] (-0.2, 0)--(-0.2, 1);
\draw[thick,blue] (1.4, 0)--(1.4, 1);

\node at (0.2, 0.5) [circle, fill, inner sep=1pt] {};

\node at (0.85, 0.5) {$\cdots$};
}} \,\, = \raisebox{-0.4cm}{\tikz{
\draw (0, 0)--(0, 1);
\draw (0.2, 0)--(0.2, 1);
\draw (0.4, 0)--(0.4, 1);
\draw (1.2, 0)--(1.2, 1);

\node at (0.2, 0.5) [circle, fill, inner sep=1pt] {};

\draw[thick,blue] (-0.2, 0.3) --(1.4, 0.3);
\draw[thick,blue] (-0.2, 0.7) --(1.4, 0.7);

\node at (0.85, 0.5) {$\cdots$};
}} \,\, = 0 . 
\end{equation}

\section{Matchgates}
\label{app:matchgate}
In this section, we explain the notion of \textit{matchgate} in the context of tensors, tensor networks, quantum gates, and quantum circuits. The term ``matchgate'' originates from the process of counting the number of dimer coverings, hence matchings, of a graph, and its tractability on a planar graph due to the renowned Fisher-Kasteleyn-Temperley (FKT) algorithm~\cite{kasteleyn1961statistics, temperley1961dimer}. 

\subsection{Matchgate tensor and tensor network}
An $n$-leg tensor $T$ is called \textit{matchgate}, if it satisfies the matchgate identity~\cite{valiant2001quantum, cai2009theory}: 
\begin{equation}
\label{eq:MGI}
\sum_{a: x_a \ne y_a} T(x \oplus e^a) T(y \oplus e^a) (-1)^{x_1 + \ldots + x_{a-1} + y_1 + \ldots + y_{a-1}} = 0
\end{equation}
for all $x, y \in \{0, 1\}^n$, where $T(x) = T_{x_1, x_2, \ldots, x_n}$ denotes the component of $T$, and the summation in Eq.~\eqref{eq:MGI} is over indices $a$ at which $x_a$ and $y_a$ differ. While the non-linear constraints may appear complicated, they essentially capture the fact that the components are given by the pfaffians, and that the matchgate identity is, in essence, a pfaffian identity in disguise, reflecting the gaussianity of fermions. 

A planar tensor network is called matchgate, if every tensor in the network is a matchgate tensor. Moreover, matchgate tensor networks remain closed under various planar tensor operations, as explained in Sec.~\ref{sec:planar-TN}. 

\subsection{Matchgate quantum gates and quantum circuits}
Consider a quantum circuit in which qubits are arranged linearly with a well-defined ordering, i.e., they lie along an interval rather than forming a circle. If quantum gates in a quantum circuit are either $1$-qubit gates or nearest-neighbor $2$-qubit gates that satisfy the matchgate identity Eq.~\eqref{eq:MGI}, then the quantum circuit can be viewed as a planar matchgate tensor network, with $1$-qubit gates corresponding to degree-2 vertices and $2$-qubit gates corresponding to degree-4 vertices. One can indeed define matchgate quantum gates and quantum circuits in this manner. Moreover, one can further show that the set $\{X\} \cup \{ e^{i \frac{\theta}{2}} e^{-i \frac{\theta}{2} Z}, e^{i \frac{\theta}{2}} e^{-i \frac{\theta}{2} X \otimes X} \}_{\theta \in [0, 2\pi)}$ forms a universal generating set for matchgate quantum gates and circuits. However, in the conventional definition of a matchgate quantum circuit, the allowed quantum gates are presented somewhat differently. In particular, only parity-preserving gates are considered, so the $X$ gate is explicitly excluded. Since including the $X$ gate in a matchgate quantum circuit does not compromise the tractability of the tensor network, in the main text we have included the $X$ gates in matchgate quantum circuits. In the following, we demonstrate that the conventional definition of the matchgate quantum circuit is compatible with the aforementioned definition in terms of matchgate tensor network, by showing that the set $\{ e^{i \frac{\theta}{2}} e^{-i \frac{\theta}{2} Z}, e^{i \frac{\theta}{2}} e^{-i \frac{\theta}{2} X \otimes X} \}_{\theta \in [0, 2\pi)}$ generates gates denoted by $G(A, B)$ below. 

In the conventional definition, an $n$-qubit quantum circuit is called a \textit{matchgate} circuit if it is generated by nearest-neighbor $2$-qubit gates of the following form: 
\begin{equation}
\label{eq:matchgate-G-A-B}
G(A, B) := \left( \begin{array}{cccc}
A_{11} & 0 & 0 & A_{12} \\ 
0 & B_{11} & B_{12} & 0 \\ 
0 & B_{21} & B_{22} & 0 \\ 
A_{21} & 0 & 0 & A_{22} 
\end{array} \right) 
\end{equation}
where $A$ and $B$ are unitaries satisfying $\det(A) = \det (B)$, and $A$ and $B$ may differ from gate to gate. Note that, as a tensor, $G(A, B)$ satisfies the matchgate identity Eq.~\eqref{eq:MGI} if and only if $\det(A) = \det (B)$. By multiplying an overall $\mathrm{U}(1)$ factor, one can further assume that $A, B \in \mathrm{SU}(2)$. 

We now we prove that the set of $2$-qubit gates of the form $G(A, B)$ with $A, B \in \mathrm{SU}(2)$ is equal to the set of $2$-qubit gates generated by the following parity-preserving matchgate generating set 
\begin{equation}
\label{eq:pp-matchgate-generating-set}
\{ e^{i \frac{\theta}{2}} e^{-i \frac{\theta}{2} Z}, e^{i \frac{\theta}{2}} e^{-i \frac{\theta}{2} X \otimes X} \}_{\theta \in [0, 2\pi)} .
\end{equation}
It is immediate to show that the gates in the generating set, including the $1$-qubit gates, can always be written as $G(A, B)$ for some $A, B \in \mathrm{SU}(2)$. Therefore, it suffices to show that an arbitrary $G(A, B)$ with $A, B \in \mathrm{SU}(2)$ can be generated by the generating set. 

First, observe that $G(A, B) = G(A, I) G(I, B)$, where $I$ is the $2 \times 2$ identity matrix. Second, we demonstrate that $G(A, I)$ for $A \in \mathrm{SU}(2)$ is generated by the parity-preserving matchgate generating set Eq.~\eqref{eq:pp-matchgate-generating-set}. An arbitrary $A \in \mathrm{SU}(2)$ can be parametrized as 
\begin{equation}
A = \left( \begin{array}{cc}
\alpha & -\beta^* \\ 
\beta & \alpha^* 
\end{array} \right) , \nonumber 
\end{equation}
where $\alpha, \beta \in \mathbb{C}$ satisfying $|\alpha|^2 + |\beta|^2 = 1$. We further parametrize $\alpha = e^{i \phi_1} \cos \frac{\theta}{2}$ and $\beta = e^{i \phi_2} \sin \frac{\theta}{2}$, where $\theta, \phi_1, \phi_2 \in [0, 2\pi)$. After a tedious calculation, one can show that 
\begin{align}
&G(A, I) = \Big( e^{i ( - \frac{\pi}{4} + \frac{\phi_1}{2} - \frac{\phi_2}{4} ) Z_1} \otimes e^{-i \frac{\phi_2}{4} Z_2} \Big) e^{-i \frac{\theta}{4} X_1 X_2} \nonumber \\ 
& \quad \times \Big( e^{i \frac{\pi}{4} Z_1} \otimes e^{i \frac{\pi}{4} Z_2} \Big) e^{i \frac{\theta}{4} X_1 X_2} \Big( e^{i \frac{\phi_2}{4} Z_1} \otimes e^{i (-\frac{\pi}{4} + \frac{\phi_1}{2} + \frac{\phi_2}{4}) Z_2} \Big) , 
\end{align}
i.e., $G(A, I)$ is generated by the parity-preserving matchgate generating set Eq.~\eqref{eq:pp-matchgate-generating-set}. Lastly, we demonstrate that $G(I, B)$ for $B \in \mathrm{SU}(2)$ is generated by the parity-preserving matchgate generating set Eq.~\eqref{eq:pp-matchgate-generating-set}. An arbitrary $\mathrm{SU}(2)$ matrix $B$ can be parametrized as 
\begin{equation}
B = \left( \begin{array}{cc}
e^{i \phi'_1} \cos \frac{\theta'}{2} & -e^{-i \phi'_2} \sin \frac{\theta'}{2} \\ 
e^{i \phi'_2} \sin \frac{\theta'}{2} & e^{-i \phi'_1} \cos \frac{\theta'}{2}
\end{array} \right) , \nonumber 
\end{equation}
where $\theta', \phi'_1, \phi'_2 \in [0, 2\pi)$. After a tedious calculation, one can show that 
\begin{align} 
&G(I, B) = \Big( e^{i ( - \frac{\pi}{4} + \frac{\phi_1}{2} - \frac{\phi_2}{4} ) Z_1} \otimes e^{i \frac{\phi_2}{4} Z_2} \Big) e^{-i \frac{\theta}{4} X_1 X_2} \nonumber \\ 
& \quad \times \Big( e^{i \frac{\pi}{4} Z_1} \otimes e^{-i \frac{\pi}{4} Z_2} \Big) e^{i \frac{\theta}{4} X_1 X_2} \Big( e^{i \frac{\phi_2}{4} Z_1} \otimes e^{i (\frac{\pi}{4} - \frac{\phi_1}{2} - \frac{\phi_2}{4}) Z_2} \Big) 
\end{align}
holds and thus it is generated by the generating set Eq.~\eqref{eq:pp-matchgate-generating-set}. In sum, all necessary gates have been generated, thereby concluding the proof. 

\subsection{Two ways toward universality}
It is known that any non-matchgate quantum state, a state that cannot be generated by a matchgate quantum circuit from a computational basis state, can serve as a \textit{magic state} for matchgate computation~\cite{PhysRevLett.123.080503}, thereby promoting matchgate computation to universal quantum computation. With this general result in mind, we focus here on two ways to acheiving universal quantum computation. 

The first way is to employ the Hadamard gate. In this case, the Hadamard gate plays a role analogous to that of the $T$-gate for Clifford circuits. For a given $H$-gate-doped matchgate circuit (or tensor network), where $n_H$ denotes the number of Hadamard gates, we assign a bit variable $b = 0, 1$ to each Hadamard gate. Using the decomposition of the Hadamard gate
\begin{equation}
H = \frac{1}{2} \big( X + Z \big) , 
\end{equation}
we substitute $X$ for $H$ when the bit variable $b$ equals $0$, and substitute $Z$ otherwise. Then the overall computation can be expressed as a sum of $2^{n_H}$ terms, each of which is efficiently evaluated via matchgate computation. 

The second way is to employ the $\textrm{SWAP}$ gate. In this case, rather than considering a $\textrm{SWAP}$-doped-matchgate circuit, we consider a quantum circuit $\mathcal{C}$ generated by arbitrary $1$-qubit gates and nearest-neighbor $\textrm{CZ}$-gates, and simulate $\mathcal{C}$ by embedding it into a $\textrm{SWAP}$-gate-doped matchgate circuit $M_\mathcal{C}$. To construct $M_\mathcal{C}$, we double the number of qubits and embed the $i$th qubit in $\mathcal{C}$ into the parity-even subspace spanned by qubits $2i-1$ and $2i$ in $M_\mathcal{C}$, reminiscent of a logical qubit encoded in the parity-even subspace of $4$-Majorana Hilbert space in Quon. We refer to the embedded $i$th qubit as the ``logical'' qubit. 

The initial $n$-qubit computational basis state $\vert b_1, \ldots, b_n \rangle$ in $\mathcal{C}$ is mapped to a $2n$-qubit state $\vert b_1, b_1, \ldots, b_n, b_n \rangle$ in $M_\mathcal{C}$. A logical $1$-qubit gate $U$ can be implemented in the parity-even sector of the matchgate $G(U, U)$. A logical $\textrm{CZ}$-gate between two neighboring logical qubits can be implemented using a single $\textrm{SWAP}$ gate via the following $\textrm{CZ}$ \textit{gadget}~\cite{PhysRevA.84.022310}: 
\begin{equation}
\begin{quantikz}
\lstick[2]{$1_\textrm{L}$} 1 & \qw & \qw & \qw & \qw & \qw \\
2 & \gate[2]{G_H} & \gate[swap]{} & \gate[2]{G_X} & \gate[2]{G_H} & \qw \\
\lstick[2]{$2_\textrm{L}$} 3 & & & & & \qw \\ 
4 & \qw & \qw & \qw & \qw & \qw 
\end{quantikz} \,\, , 
\end{equation}
where we denote the $i$th logical qubit as $i_\textrm{L}$ and set $G_A := G(A, A)$. In fact, viewing $M_\mathcal{C}$ as a \textit{surface} matchgate tensor network, the $\textrm{SWAP}$ gate in the $\textrm{CZ}$ gadget can be implemented by the following \textit{handle-attaching} gadget: 
\begin{equation}
\raisebox{-0.8cm}{
\begin{tikzpicture}
\draw[dashed, thick, teal, fill=teal!20, fill opacity=0.4] (0,0) ellipse (2cm and 0.7cm);

\draw[ultra thick, gray!50] (-0.86, 0.55) to[out=-20, in=160] (-0.57, 0.44);

\draw[teal, fill=teal, fill opacity=0.2, shade] (-1, -0.2) arc (-160:-20:0.2) to[out=100, in=90, looseness=2.5] (0.6, 0.2) arc (-160:-20:0.18) to[out=90, in=90, looseness=2.2] (-1, -0.2);

\draw[teal, dashed] (-1, -0.2) arc (160:20:0.2);
\draw[teal, dashed] (0.6, 0.2) arc (160:20:0.18);

\draw[ultra thick, gray] (-1.5, 0.8) to[out=-20, in=160] (-0.86, 0.55);
\draw[ultra thick, gray] (-0.57, 0.44) to[out=-20, in=160] (1.8, -0.8);

\draw[ultra thick, gray] (-1.5, -0.9) to[out=40, in=-130] (-0.83, -0.33) to[out=90, in=100, looseness=2.4] (0.83, 0.1) to[out=-10, in=-140] (1.8, 0.7);
\end{tikzpicture}
} \,\, , 
\end{equation}
which introduces a \textit{genus} into the surface. Therefore, it is desirable to develop a genus-minimization algorithm for a surface matchgate tensor network. See Refs.~\onlinecite{cimasoni2007dimers, cimasoni2008dimers, cimasoni2009dimers, Dijkgraaf2009dimer, bravyi2009contraction} for an earlier work related to this direction.

\bibliography{ref}

\begin{thebibliography}{119}%
\makeatletter
\providecommand \@ifxundefined [1]{%
 \@ifx{#1\undefined}
}%
\providecommand \@ifnum [1]{%
 \ifnum #1\expandafter \@firstoftwo
 \else \expandafter \@secondoftwo
 \fi
}%
\providecommand \@ifx [1]{%
 \ifx #1\expandafter \@firstoftwo
 \else \expandafter \@secondoftwo
 \fi
}%
\providecommand \natexlab [1]{#1}%
\providecommand \enquote  [1]{``#1''}%
\providecommand \bibnamefont  [1]{#1}%
\providecommand \bibfnamefont [1]{#1}%
\providecommand \citenamefont [1]{#1}%
\providecommand \href@noop [0]{\@secondoftwo}%
\providecommand \href [0]{\begingroup \@sanitize@url \@href}%
\providecommand \@href[1]{\@@startlink{#1}\@@href}%
\providecommand \@@href[1]{\endgroup#1\@@endlink}%
\providecommand \@sanitize@url [0]{\catcode `\\12\catcode `\$12\catcode
  `\&12\catcode `\#12\catcode `\^12\catcode `\_12\catcode `\%12\relax}%
\providecommand \@@startlink[1]{}%
\providecommand \@@endlink[0]{}%
\providecommand \url  [0]{\begingroup\@sanitize@url \@url }%
\providecommand \@url [1]{\endgroup\@href {#1}{\urlprefix }}%
\providecommand \urlprefix  [0]{URL }%
\providecommand \Eprint [0]{\href }%
\providecommand \doibase [0]{https://doi.org/}%
\providecommand \selectlanguage [0]{\@gobble}%
\providecommand \bibinfo  [0]{\@secondoftwo}%
\providecommand \bibfield  [0]{\@secondoftwo}%
\providecommand \translation [1]{[#1]}%
\providecommand \BibitemOpen [0]{}%
\providecommand \bibitemStop [0]{}%
\providecommand \bibitemNoStop [0]{.\EOS\space}%
\providecommand \EOS [0]{\spacefactor3000\relax}%
\providecommand \BibitemShut  [1]{\csname bibitem#1\endcsname}%
\let\auto@bib@innerbib\@empty
\bibitem [{\citenamefont {Dalzell}\ \emph {et~al.}(2023)\citenamefont
  {Dalzell}, \citenamefont {McArdle}, \citenamefont {Berta}, \citenamefont
  {Bienias}, \citenamefont {Chen}, \citenamefont {Gily{\'e}n}, \citenamefont
  {Hann}, \citenamefont {Kastoryano}, \citenamefont {Khabiboulline},
  \citenamefont {Kubica} \emph {et~al.}}]{dalzell2023quantum}%
  \BibitemOpen
  \bibfield  {author} {\bibinfo {author} {\bibfnamefont {A.~M.}\ \bibnamefont
  {Dalzell}}, \bibinfo {author} {\bibfnamefont {S.}~\bibnamefont {McArdle}},
  \bibinfo {author} {\bibfnamefont {M.}~\bibnamefont {Berta}}, \bibinfo
  {author} {\bibfnamefont {P.}~\bibnamefont {Bienias}}, \bibinfo {author}
  {\bibfnamefont {C.-F.}\ \bibnamefont {Chen}}, \bibinfo {author}
  {\bibfnamefont {A.}~\bibnamefont {Gily{\'e}n}}, \bibinfo {author}
  {\bibfnamefont {C.~T.}\ \bibnamefont {Hann}}, \bibinfo {author}
  {\bibfnamefont {M.~J.}\ \bibnamefont {Kastoryano}}, \bibinfo {author}
  {\bibfnamefont {E.~T.}\ \bibnamefont {Khabiboulline}}, \bibinfo {author}
  {\bibfnamefont {A.}~\bibnamefont {Kubica}}, \emph {et~al.},\ }\bibfield
  {title} {\bibinfo {title} {Quantum algorithms: A survey of applications and
  end-to-end complexities},\ }\href {https://arxiv.org/abs/2310.03011}
  {\bibfield  {journal} {\bibinfo  {journal} {arXiv:2310.03011}\ } (\bibinfo
  {year} {2023})}\BibitemShut {NoStop}%
\bibitem [{\citenamefont {McArdle}\ \emph {et~al.}(2020)\citenamefont
  {McArdle}, \citenamefont {Endo}, \citenamefont {Aspuru-Guzik}, \citenamefont
  {Benjamin},\ and\ \citenamefont {Yuan}}]{RevModPhys.92.015003}%
  \BibitemOpen
  \bibfield  {author} {\bibinfo {author} {\bibfnamefont {S.}~\bibnamefont
  {McArdle}}, \bibinfo {author} {\bibfnamefont {S.}~\bibnamefont {Endo}},
  \bibinfo {author} {\bibfnamefont {A.}~\bibnamefont {Aspuru-Guzik}}, \bibinfo
  {author} {\bibfnamefont {S.~C.}\ \bibnamefont {Benjamin}},\ and\ \bibinfo
  {author} {\bibfnamefont {X.}~\bibnamefont {Yuan}},\ }\bibfield  {title}
  {\bibinfo {title} {Quantum computational chemistry},\ }\href
  {https://doi.org/10.1103/RevModPhys.92.015003} {\bibfield  {journal}
  {\bibinfo  {journal} {Rev. Mod. Phys.}\ }\textbf {\bibinfo {volume} {92}},\
  \bibinfo {pages} {015003} (\bibinfo {year} {2020})}\BibitemShut {NoStop}%
\bibitem [{\citenamefont {Gottesman}(1998)}]{gottesman1998heisenberg}%
  \BibitemOpen
  \bibfield  {author} {\bibinfo {author} {\bibfnamefont {D.}~\bibnamefont
  {Gottesman}},\ }\bibfield  {title} {\bibinfo {title} {The heisenberg
  representation of quantum computers},\ }\href
  {https://arxiv.org/abs/quant-ph/9807006} {\bibfield  {journal} {\bibinfo
  {journal} {quant-ph/9807006}\ } (\bibinfo {year} {1998})}\BibitemShut
  {NoStop}%
\bibitem [{\citenamefont {Nielsen}\ and\ \citenamefont
  {Chuang}(2001)}]{nielsen2001quantum}%
  \BibitemOpen
  \bibfield  {author} {\bibinfo {author} {\bibfnamefont {M.~A.}\ \bibnamefont
  {Nielsen}}\ and\ \bibinfo {author} {\bibfnamefont {I.~L.}\ \bibnamefont
  {Chuang}},\ }\href@noop {} {\emph {\bibinfo {title} {Quantum computation and
  quantum information}}},\ Vol.~\bibinfo {volume} {2}\ (\bibinfo  {publisher}
  {Cambridge university press Cambridge},\ \bibinfo {year} {2001})\BibitemShut
  {NoStop}%
\bibitem [{\citenamefont {Valiant}(2001)}]{valiant2001quantum}%
  \BibitemOpen
  \bibfield  {author} {\bibinfo {author} {\bibfnamefont {L.~G.}\ \bibnamefont
  {Valiant}},\ }\bibfield  {title} {\bibinfo {title} {Quantum computers that
  can be simulated classically in polynomial time},\ }in\ \href
  {https://doi.org/10.1145/380752.380785} {\emph {\bibinfo {booktitle}
  {Proceedings of the thirty-third annual ACM symposium on Theory of
  computing}}}\ (\bibinfo {year} {2001})\ pp.\ \bibinfo {pages}
  {114--123}\BibitemShut {NoStop}%
\bibitem [{\citenamefont {Terhal}\ and\ \citenamefont
  {DiVincenzo}(2002)}]{PhysRevA.65.032325}%
  \BibitemOpen
  \bibfield  {author} {\bibinfo {author} {\bibfnamefont {B.~M.}\ \bibnamefont
  {Terhal}}\ and\ \bibinfo {author} {\bibfnamefont {D.~P.}\ \bibnamefont
  {DiVincenzo}},\ }\bibfield  {title} {\bibinfo {title} {Classical simulation
  of noninteracting-fermion quantum circuits},\ }\href
  {https://doi.org/10.1103/PhysRevA.65.032325} {\bibfield  {journal} {\bibinfo
  {journal} {Phys. Rev. A}\ }\textbf {\bibinfo {volume} {65}},\ \bibinfo
  {pages} {032325} (\bibinfo {year} {2002})}\BibitemShut {NoStop}%
\bibitem [{\citenamefont {Schollw{\"o}ck}(2011)}]{schollwock2011density}%
  \BibitemOpen
  \bibfield  {author} {\bibinfo {author} {\bibfnamefont {U.}~\bibnamefont
  {Schollw{\"o}ck}},\ }\bibfield  {title} {\bibinfo {title} {The density-matrix
  renormalization group in the age of matrix product states},\ }\href
  {https://doi.org/10.1016/j.aop.2010.09.012} {\bibfield  {journal} {\bibinfo
  {journal} {Annals of physics}\ }\textbf {\bibinfo {volume} {326}},\ \bibinfo
  {pages} {96} (\bibinfo {year} {2011})}\BibitemShut {NoStop}%
\bibitem [{\citenamefont {White}(1992)}]{PhysRevLett.69.2863}%
  \BibitemOpen
  \bibfield  {author} {\bibinfo {author} {\bibfnamefont {S.~R.}\ \bibnamefont
  {White}},\ }\bibfield  {title} {\bibinfo {title} {Density matrix formulation
  for quantum renormalization groups},\ }\href
  {https://doi.org/10.1103/PhysRevLett.69.2863} {\bibfield  {journal} {\bibinfo
   {journal} {Phys. Rev. Lett.}\ }\textbf {\bibinfo {volume} {69}},\ \bibinfo
  {pages} {2863} (\bibinfo {year} {1992})}\BibitemShut {NoStop}%
\bibitem [{\citenamefont {Jaffe}\ \emph {et~al.}(2018)\citenamefont {Jaffe},
  \citenamefont {Liu},\ and\ \citenamefont
  {Wozniakowski}}]{jaffe2018holographic}%
  \BibitemOpen
  \bibfield  {author} {\bibinfo {author} {\bibfnamefont {A.}~\bibnamefont
  {Jaffe}}, \bibinfo {author} {\bibfnamefont {Z.}~\bibnamefont {Liu}},\ and\
  \bibinfo {author} {\bibfnamefont {A.}~\bibnamefont {Wozniakowski}},\
  }\bibfield  {title} {\bibinfo {title} {Holographic software for quantum
  networks},\ }\href {https://doi.org/10.1007/s11425-017-9207-3} {\bibfield
  {journal} {\bibinfo  {journal} {Science China Mathematics}\ }\textbf
  {\bibinfo {volume} {61}},\ \bibinfo {pages} {593} (\bibinfo {year}
  {2018})}\BibitemShut {NoStop}%
\bibitem [{\citenamefont {Liu}\ \emph {et~al.}(2017)\citenamefont {Liu},
  \citenamefont {Wozniakowski},\ and\ \citenamefont {Jaffe}}]{liu2017quon}%
  \BibitemOpen
  \bibfield  {author} {\bibinfo {author} {\bibfnamefont {Z.}~\bibnamefont
  {Liu}}, \bibinfo {author} {\bibfnamefont {A.}~\bibnamefont {Wozniakowski}},\
  and\ \bibinfo {author} {\bibfnamefont {A.~M.}\ \bibnamefont {Jaffe}},\
  }\bibfield  {title} {\bibinfo {title} {{Quon 3D language for quantum
  information}},\ }\href {https://doi.org/10.1073/pnas.1621345114} {\bibfield
  {journal} {\bibinfo  {journal} {Proceedings of the National Academy of
  Sciences}\ }\textbf {\bibinfo {volume} {114}},\ \bibinfo {pages} {2497}
  (\bibinfo {year} {2017})}\BibitemShut {NoStop}%
\bibitem [{\citenamefont {Jaffe}\ and\ \citenamefont
  {Liu}(2017)}]{jaffe2017planar}%
  \BibitemOpen
  \bibfield  {author} {\bibinfo {author} {\bibfnamefont {A.}~\bibnamefont
  {Jaffe}}\ and\ \bibinfo {author} {\bibfnamefont {Z.}~\bibnamefont {Liu}},\
  }\bibfield  {title} {\bibinfo {title} {{Planar Para Algebras, Reflection
  Positivity}},\ }\href {https://doi.org/10.1007/s00220-016-2779-4} {\bibfield
  {journal} {\bibinfo  {journal} {Communications in Mathematical Physics}\
  }\textbf {\bibinfo {volume} {352}},\ \bibinfo {pages} {95} (\bibinfo {year}
  {2017})}\BibitemShut {NoStop}%
\bibitem [{\citenamefont {Liu}(2019{\natexlab{a}})}]{liu2019quon}%
  \BibitemOpen
  \bibfield  {author} {\bibinfo {author} {\bibfnamefont {Z.}~\bibnamefont
  {Liu}},\ }\bibfield  {title} {\bibinfo {title} {{Quon Language: Surface
  Algebras and Fourier Duality}},\ }\href
  {https://doi.org/10.1007/s00220-019-03361-3} {\bibfield  {journal} {\bibinfo
  {journal} {Communications in Mathematical Physics}\ }\textbf {\bibinfo
  {volume} {366}},\ \bibinfo {pages} {865} (\bibinfo {year}
  {2019}{\natexlab{a}})}\BibitemShut {NoStop}%
\bibitem [{\citenamefont {Liu}\ \emph {et~al.}(2023{\natexlab{a}})\citenamefont
  {Liu}, \citenamefont {Ming}, \citenamefont {Wang},\ and\ \citenamefont
  {Wu}}]{liu20233}%
  \BibitemOpen
  \bibfield  {author} {\bibinfo {author} {\bibfnamefont {Z.}~\bibnamefont
  {Liu}}, \bibinfo {author} {\bibfnamefont {S.}~\bibnamefont {Ming}}, \bibinfo
  {author} {\bibfnamefont {Y.}~\bibnamefont {Wang}},\ and\ \bibinfo {author}
  {\bibfnamefont {J.}~\bibnamefont {Wu}},\ }\bibfield  {title} {\bibinfo
  {title} {{3-Alterfolds and Quantum Invariants}},\ }\href
  {https://arxiv.org/abs/2307.12284} {\bibfield  {journal} {\bibinfo  {journal}
  {arXiv:2307.12284}\ } (\bibinfo {year} {2023}{\natexlab{a}})}\BibitemShut
  {NoStop}%
\bibitem [{\citenamefont {Liu}\ \emph {et~al.}(2023{\natexlab{b}})\citenamefont
  {Liu}, \citenamefont {Ming}, \citenamefont {Wang},\ and\ \citenamefont
  {Wu}}]{liu2023alterfold}%
  \BibitemOpen
  \bibfield  {author} {\bibinfo {author} {\bibfnamefont {Z.}~\bibnamefont
  {Liu}}, \bibinfo {author} {\bibfnamefont {S.}~\bibnamefont {Ming}}, \bibinfo
  {author} {\bibfnamefont {Y.}~\bibnamefont {Wang}},\ and\ \bibinfo {author}
  {\bibfnamefont {J.}~\bibnamefont {Wu}},\ }\bibfield  {title} {\bibinfo
  {title} {{Alterfold Topological Quantum Field Theory}},\ }\href
  {https://arxiv.org/abs/2312.06477} {\bibfield  {journal} {\bibinfo  {journal}
  {arXiv:2312.06477}\ } (\bibinfo {year} {2023}{\natexlab{b}})}\BibitemShut
  {NoStop}%
\bibitem [{\citenamefont {Liu}(2024)}]{liu2024functional}%
  \BibitemOpen
  \bibfield  {author} {\bibinfo {author} {\bibfnamefont {Z.}~\bibnamefont
  {Liu}},\ }\bibfield  {title} {\bibinfo {title} {{Functional Integral
  Construction of Topological Quantum Field Theory}},\ }\href
  {https://arxiv.org/abs/2409.17103} {\bibfield  {journal} {\bibinfo  {journal}
  {arXiv:2409.17103}\ } (\bibinfo {year} {2024})}\BibitemShut {NoStop}%
\bibitem [{\citenamefont {Liu}\ \emph {et~al.}(2024)\citenamefont {Liu},
  \citenamefont {Ming}, \citenamefont {Wang},\ and\ \citenamefont
  {Wu}}]{liu2024alterfold}%
  \BibitemOpen
  \bibfield  {author} {\bibinfo {author} {\bibfnamefont {Z.}~\bibnamefont
  {Liu}}, \bibinfo {author} {\bibfnamefont {S.}~\bibnamefont {Ming}}, \bibinfo
  {author} {\bibfnamefont {Y.}~\bibnamefont {Wang}},\ and\ \bibinfo {author}
  {\bibfnamefont {J.}~\bibnamefont {Wu}},\ }\bibfield  {title} {\bibinfo
  {title} {{Alterfold Theory and Topological Modular Invariance}},\ }\href
  {https://arxiv.org/abs/2412.12702} {\bibfield  {journal} {\bibinfo  {journal}
  {arXiv:2412.12702}\ } (\bibinfo {year} {2024})}\BibitemShut {NoStop}%
\bibitem [{\citenamefont {Liu}(2019{\natexlab{b}})}]{liu2019quantized}%
  \BibitemOpen
  \bibfield  {author} {\bibinfo {author} {\bibfnamefont {Z.}~\bibnamefont
  {Liu}},\ }\bibfield  {title} {\bibinfo {title} {{Quantized Graphs and Quantum
  Error Correction}},\ }\href {https://arxiv.org/abs/1910.12065} {\bibfield
  {journal} {\bibinfo  {journal} {arXiv:1910.12065}\ } (\bibinfo {year}
  {2019}{\natexlab{b}})}\BibitemShut {NoStop}%
\bibitem [{\citenamefont {Shao}\ \emph {et~al.}(2024)\citenamefont {Shao},
  \citenamefont {Wei}, \citenamefont {Wei},\ and\ \citenamefont
  {Liu}}]{shao2024variational}%
  \BibitemOpen
  \bibfield  {author} {\bibinfo {author} {\bibfnamefont {Y.}~\bibnamefont
  {Shao}}, \bibinfo {author} {\bibfnamefont {F.}~\bibnamefont {Wei}}, \bibinfo
  {author} {\bibfnamefont {Z.}~\bibnamefont {Wei}},\ and\ \bibinfo {author}
  {\bibfnamefont {Z.}~\bibnamefont {Liu}},\ }\bibfield  {title} {\bibinfo
  {title} {{Variational Graphical Quantum Error Correction Codes}},\ }\href
  {https://arxiv.org/abs/2410.02608} {\bibfield  {journal} {\bibinfo  {journal}
  {arXiv:2410.02608}\ } (\bibinfo {year} {2024})}\BibitemShut {NoStop}%
\bibitem [{\citenamefont {Shi}\ \emph {et~al.}(2018)\citenamefont {Shi},
  \citenamefont {Demler},\ and\ \citenamefont {Cirac}}]{shi2018variational}%
  \BibitemOpen
  \bibfield  {author} {\bibinfo {author} {\bibfnamefont {T.}~\bibnamefont
  {Shi}}, \bibinfo {author} {\bibfnamefont {E.}~\bibnamefont {Demler}},\ and\
  \bibinfo {author} {\bibfnamefont {J.~I.}\ \bibnamefont {Cirac}},\ }\bibfield
  {title} {\bibinfo {title} {{Variational study of fermionic and bosonic
  systems with non-Gaussian states: Theory and applications}},\ }\href
  {https://doi.org/10.1016/j.aop.2017.11.014} {\bibfield  {journal} {\bibinfo
  {journal} {Annals of Physics}\ }\textbf {\bibinfo {volume} {390}},\ \bibinfo
  {pages} {245} (\bibinfo {year} {2018})}\BibitemShut {NoStop}%
\bibitem [{\citenamefont {Mishmash}\ \emph {et~al.}(2023)\citenamefont
  {Mishmash}, \citenamefont {Gujarati}, \citenamefont {Motta}, \citenamefont
  {Zhai}, \citenamefont {Chan},\ and\ \citenamefont
  {Mezzacapo}}]{mishmash2023hierarchical}%
  \BibitemOpen
  \bibfield  {author} {\bibinfo {author} {\bibfnamefont {R.~V.}\ \bibnamefont
  {Mishmash}}, \bibinfo {author} {\bibfnamefont {T.~P.}\ \bibnamefont
  {Gujarati}}, \bibinfo {author} {\bibfnamefont {M.}~\bibnamefont {Motta}},
  \bibinfo {author} {\bibfnamefont {H.}~\bibnamefont {Zhai}}, \bibinfo {author}
  {\bibfnamefont {G.~K.-L.}\ \bibnamefont {Chan}},\ and\ \bibinfo {author}
  {\bibfnamefont {A.}~\bibnamefont {Mezzacapo}},\ }\bibfield  {title} {\bibinfo
  {title} {{Hierarchical Clifford Transformations to Reduce Entanglement in
  Quantum Chemistry Wave Functions}},\ }\href
  {https://doi.org/10.1021/acs.jctc.3c00228} {\bibfield  {journal} {\bibinfo
  {journal} {Journal of chemical theory and computation}\ }\textbf {\bibinfo
  {volume} {19}},\ \bibinfo {pages} {3194} (\bibinfo {year}
  {2023})}\BibitemShut {NoStop}%
\bibitem [{\citenamefont {Qian}\ \emph {et~al.}(2024)\citenamefont {Qian},
  \citenamefont {Huang},\ and\ \citenamefont {Qin}}]{PhysRevLett.133.190402}%
  \BibitemOpen
  \bibfield  {author} {\bibinfo {author} {\bibfnamefont {X.}~\bibnamefont
  {Qian}}, \bibinfo {author} {\bibfnamefont {J.}~\bibnamefont {Huang}},\ and\
  \bibinfo {author} {\bibfnamefont {M.}~\bibnamefont {Qin}},\ }\bibfield
  {title} {\bibinfo {title} {{Augmenting Density Matrix Renormalization Group
  with Clifford Circuits}},\ }\href
  {https://doi.org/10.1103/PhysRevLett.133.190402} {\bibfield  {journal}
  {\bibinfo  {journal} {Phys. Rev. Lett.}\ }\textbf {\bibinfo {volume} {133}},\
  \bibinfo {pages} {190402} (\bibinfo {year} {2024})}\BibitemShut {NoStop}%
\bibitem [{\citenamefont {Projansky}\ \emph {et~al.}(2024)\citenamefont
  {Projansky}, \citenamefont {Necaise},\ and\ \citenamefont
  {Whitfield}}]{projansky2024extending}%
  \BibitemOpen
  \bibfield  {author} {\bibinfo {author} {\bibfnamefont {A.~M.}\ \bibnamefont
  {Projansky}}, \bibinfo {author} {\bibfnamefont {J.}~\bibnamefont {Necaise}},\
  and\ \bibinfo {author} {\bibfnamefont {J.~D.}\ \bibnamefont {Whitfield}},\
  }\bibfield  {title} {\bibinfo {title} {{Gaussianity and Simulability of
  Cliffords and Matchgates}},\ }\href {https://arxiv.org/abs/2410.10068}
  {\bibfield  {journal} {\bibinfo  {journal} {arXiv:2410.10068}\ } (\bibinfo
  {year} {2024})}\BibitemShut {NoStop}%
\bibitem [{\citenamefont {Vazirani}(1989)}]{vazirani1989nc}%
  \BibitemOpen
  \bibfield  {author} {\bibinfo {author} {\bibfnamefont {V.~V.}\ \bibnamefont
  {Vazirani}},\ }\bibfield  {title} {\bibinfo {title} {{NC algorithms for
  computing the number of perfect matchings in $K_{3,3}$-free graphs and
  related problems}},\ }\href {https://doi.org/10.1016/0890-5401(89)90017-5}
  {\bibfield  {journal} {\bibinfo  {journal} {Information and computation}\
  }\textbf {\bibinfo {volume} {80}},\ \bibinfo {pages} {152} (\bibinfo {year}
  {1989})}\BibitemShut {NoStop}%
\bibitem [{\citenamefont {Straub}\ \emph {et~al.}(2016)\citenamefont {Straub},
  \citenamefont {Thierauf},\ and\ \citenamefont {Wagner}}]{straub2016counting}%
  \BibitemOpen
  \bibfield  {author} {\bibinfo {author} {\bibfnamefont {S.}~\bibnamefont
  {Straub}}, \bibinfo {author} {\bibfnamefont {T.}~\bibnamefont {Thierauf}},\
  and\ \bibinfo {author} {\bibfnamefont {F.}~\bibnamefont {Wagner}},\
  }\bibfield  {title} {\bibinfo {title} {{Counting the Number of Perfect
  Matchings in $K_5$-Free Graphs}},\ }\href
  {https://doi.org/10.1007/s00224-015-9645-1} {\bibfield  {journal} {\bibinfo
  {journal} {Theory of Computing Systems}\ }\textbf {\bibinfo {volume} {59}},\
  \bibinfo {pages} {416} (\bibinfo {year} {2016})}\BibitemShut {NoStop}%
\bibitem [{\citenamefont {Curticapean}(2014)}]{curticapean2014counting}%
  \BibitemOpen
  \bibfield  {author} {\bibinfo {author} {\bibfnamefont {R.}~\bibnamefont
  {Curticapean}},\ }\bibfield  {title} {\bibinfo {title} {Counting perfect
  matchings in graphs that exclude a single-crossing minor},\ }\href
  {https://arxiv.org/abs/1406.4056} {\bibfield  {journal} {\bibinfo  {journal}
  {arXiv:1406.4056}\ } (\bibinfo {year} {2014})}\BibitemShut {NoStop}%
\bibitem [{\citenamefont {Likhosherstov}\ \emph {et~al.}(2020)\citenamefont
  {Likhosherstov}, \citenamefont {Maximov},\ and\ \citenamefont
  {Chertkov}}]{likhosherstov2020tractable}%
  \BibitemOpen
  \bibfield  {author} {\bibinfo {author} {\bibfnamefont {V.}~\bibnamefont
  {Likhosherstov}}, \bibinfo {author} {\bibfnamefont {Y.}~\bibnamefont
  {Maximov}},\ and\ \bibinfo {author} {\bibfnamefont {M.}~\bibnamefont
  {Chertkov}},\ }\bibfield  {title} {\bibinfo {title} {Tractable minor-free
  generalization of planar zero-field ising models},\ }\href
  {https://doi.org/10.1088/1742-5468/abcaf1} {\bibfield  {journal} {\bibinfo
  {journal} {Journal of Statistical Mechanics: Theory and Experiment}\ }\textbf
  {\bibinfo {volume} {2020}},\ \bibinfo {pages} {124007} (\bibinfo {year}
  {2020})}\BibitemShut {NoStop}%
\bibitem [{\citenamefont {Sarma}\ \emph {et~al.}(2015)\citenamefont {Sarma},
  \citenamefont {Freedman},\ and\ \citenamefont {Nayak}}]{sarma2015majorana}%
  \BibitemOpen
  \bibfield  {author} {\bibinfo {author} {\bibfnamefont {S.~D.}\ \bibnamefont
  {Sarma}}, \bibinfo {author} {\bibfnamefont {M.}~\bibnamefont {Freedman}},\
  and\ \bibinfo {author} {\bibfnamefont {C.}~\bibnamefont {Nayak}},\ }\bibfield
   {title} {\bibinfo {title} {Majorana zero modes and topological quantum
  computation},\ }\href {https://doi.org/10.1038/npjqi.2015.1} {\bibfield
  {journal} {\bibinfo  {journal} {npj Quantum Information}\ }\textbf {\bibinfo
  {volume} {1}},\ \bibinfo {pages} {1} (\bibinfo {year} {2015})}\BibitemShut
  {NoStop}%
\bibitem [{\citenamefont {Kitaev}(2006)}]{kitaev2006anyons}%
  \BibitemOpen
  \bibfield  {author} {\bibinfo {author} {\bibfnamefont {A.}~\bibnamefont
  {Kitaev}},\ }\bibfield  {title} {\bibinfo {title} {Anyons in an exactly
  solved model and beyond},\ }\href {https://doi.org/10.1016/j.aop.2005.10.005}
  {\bibfield  {journal} {\bibinfo  {journal} {Annals of Physics}\ }\textbf
  {\bibinfo {volume} {321}},\ \bibinfo {pages} {2} (\bibinfo {year}
  {2006})}\BibitemShut {NoStop}%
\bibitem [{\citenamefont {Bravyi}\ and\ \citenamefont
  {Kitaev}(2005)}]{PhysRevA.71.022316}%
  \BibitemOpen
  \bibfield  {author} {\bibinfo {author} {\bibfnamefont {S.}~\bibnamefont
  {Bravyi}}\ and\ \bibinfo {author} {\bibfnamefont {A.}~\bibnamefont
  {Kitaev}},\ }\bibfield  {title} {\bibinfo {title} {{Universal quantum
  computation with ideal Clifford gates and noisy ancillas}},\ }\href
  {https://doi.org/10.1103/PhysRevA.71.022316} {\bibfield  {journal} {\bibinfo
  {journal} {Phys. Rev. A}\ }\textbf {\bibinfo {volume} {71}},\ \bibinfo
  {pages} {022316} (\bibinfo {year} {2005})}\BibitemShut {NoStop}%
\bibitem [{\citenamefont {Brod}\ and\ \citenamefont
  {Galv\~ao}(2011)}]{PhysRevA.84.022310}%
  \BibitemOpen
  \bibfield  {author} {\bibinfo {author} {\bibfnamefont {D.~J.}\ \bibnamefont
  {Brod}}\ and\ \bibinfo {author} {\bibfnamefont {E.~F.}\ \bibnamefont
  {Galv\~ao}},\ }\bibfield  {title} {\bibinfo {title} {Extending matchgates
  into universal quantum computation},\ }\href
  {https://doi.org/10.1103/PhysRevA.84.022310} {\bibfield  {journal} {\bibinfo
  {journal} {Phys. Rev. A}\ }\textbf {\bibinfo {volume} {84}},\ \bibinfo
  {pages} {022310} (\bibinfo {year} {2011})}\BibitemShut {NoStop}%
\bibitem [{\citenamefont {Hebenstreit}\ \emph {et~al.}(2019)\citenamefont
  {Hebenstreit}, \citenamefont {Jozsa}, \citenamefont {Kraus}, \citenamefont
  {Strelchuk},\ and\ \citenamefont {Yoganathan}}]{PhysRevLett.123.080503}%
  \BibitemOpen
  \bibfield  {author} {\bibinfo {author} {\bibfnamefont {M.}~\bibnamefont
  {Hebenstreit}}, \bibinfo {author} {\bibfnamefont {R.}~\bibnamefont {Jozsa}},
  \bibinfo {author} {\bibfnamefont {B.}~\bibnamefont {Kraus}}, \bibinfo
  {author} {\bibfnamefont {S.}~\bibnamefont {Strelchuk}},\ and\ \bibinfo
  {author} {\bibfnamefont {M.}~\bibnamefont {Yoganathan}},\ }\bibfield  {title}
  {\bibinfo {title} {{All Pure Fermionic Non-Gaussian States Are Magic States
  for Matchgate Computations}},\ }\href
  {https://doi.org/10.1103/PhysRevLett.123.080503} {\bibfield  {journal}
  {\bibinfo  {journal} {Phys. Rev. Lett.}\ }\textbf {\bibinfo {volume} {123}},\
  \bibinfo {pages} {080503} (\bibinfo {year} {2019})}\BibitemShut {NoStop}%
\bibitem [{\citenamefont {Cimasoni}\ and\ \citenamefont
  {Reshetikhin}(2007)}]{cimasoni2007dimers}%
  \BibitemOpen
  \bibfield  {author} {\bibinfo {author} {\bibfnamefont {D.}~\bibnamefont
  {Cimasoni}}\ and\ \bibinfo {author} {\bibfnamefont {N.}~\bibnamefont
  {Reshetikhin}},\ }\bibfield  {title} {\bibinfo {title} {{Dimers on surface
  graphs and spin structures. I}},\ }\href
  {https://doi.org/10.1007/s00220-007-0302-7} {\bibfield  {journal} {\bibinfo
  {journal} {Communications in Mathematical Physics}\ }\textbf {\bibinfo
  {volume} {275}},\ \bibinfo {pages} {187} (\bibinfo {year}
  {2007})}\BibitemShut {NoStop}%
\bibitem [{\citenamefont {Cimasoni}\ and\ \citenamefont
  {Reshetikhin}(2008)}]{cimasoni2008dimers}%
  \BibitemOpen
  \bibfield  {author} {\bibinfo {author} {\bibfnamefont {D.}~\bibnamefont
  {Cimasoni}}\ and\ \bibinfo {author} {\bibfnamefont {N.}~\bibnamefont
  {Reshetikhin}},\ }\bibfield  {title} {\bibinfo {title} {{Dimers on surface
  graphs and spin structures. II}},\ }\href
  {https://doi.org/10.1007/s00220-008-0488-3} {\bibfield  {journal} {\bibinfo
  {journal} {Communications in mathematical physics}\ }\textbf {\bibinfo
  {volume} {281}},\ \bibinfo {pages} {445} (\bibinfo {year}
  {2008})}\BibitemShut {NoStop}%
\bibitem [{\citenamefont {Cimasoni}(2009)}]{cimasoni2009dimers}%
  \BibitemOpen
  \bibfield  {author} {\bibinfo {author} {\bibfnamefont {D.}~\bibnamefont
  {Cimasoni}},\ }\bibfield  {title} {\bibinfo {title} {{Dimers on Graphs in
  Non-Orientable Surfaces}},\ }\href
  {https://doi.org/10.1007/s11005-009-0299-2} {\bibfield  {journal} {\bibinfo
  {journal} {Letters in Mathematical Physics}\ }\textbf {\bibinfo {volume}
  {87}},\ \bibinfo {pages} {149} (\bibinfo {year} {2009})}\BibitemShut
  {NoStop}%
\bibitem [{\citenamefont {Dijkgraaf}\ \emph {et~al.}(2009)\citenamefont
  {Dijkgraaf}, \citenamefont {Orlando},\ and\ \citenamefont
  {Reffert}}]{Dijkgraaf2009dimer}%
  \BibitemOpen
  \bibfield  {author} {\bibinfo {author} {\bibfnamefont {R.}~\bibnamefont
  {Dijkgraaf}}, \bibinfo {author} {\bibfnamefont {D.}~\bibnamefont {Orlando}},\
  and\ \bibinfo {author} {\bibfnamefont {S.}~\bibnamefont {Reffert}},\
  }\bibfield  {title} {\bibinfo {title} {Dimer models, free fermions and super
  quantum mechanics},\ }\href {https://doi.org/10.4310/ATMP.2009.v13.n5.a1}
  {\bibfield  {journal} {\bibinfo  {journal} {Adv. Theor. Math. Phys.}\
  }\textbf {\bibinfo {volume} {13}},\ \bibinfo {pages} {1255} (\bibinfo {year}
  {2009})}\BibitemShut {NoStop}%
\bibitem [{\citenamefont {Bravyi}(2009)}]{bravyi2009contraction}%
  \BibitemOpen
  \bibfield  {author} {\bibinfo {author} {\bibfnamefont {S.}~\bibnamefont
  {Bravyi}},\ }\bibfield  {title} {\bibinfo {title} {Contraction of matchgate
  tensor networks on non-planar graphs},\ }\href@noop {} {\bibfield  {journal}
  {\bibinfo  {journal} {Contemp. Math}\ }\textbf {\bibinfo {volume} {482}},\
  \bibinfo {pages} {179} (\bibinfo {year} {2009})}\BibitemShut {NoStop}%
\bibitem [{\citenamefont {Kramers}\ and\ \citenamefont
  {Wannier}(1941)}]{PhysRev.60.252}%
  \BibitemOpen
  \bibfield  {author} {\bibinfo {author} {\bibfnamefont {H.~A.}\ \bibnamefont
  {Kramers}}\ and\ \bibinfo {author} {\bibfnamefont {G.~H.}\ \bibnamefont
  {Wannier}},\ }\bibfield  {title} {\bibinfo {title} {{Statistics of the
  Two-Dimensional Ferromagnet. Part I}},\ }\href
  {https://doi.org/10.1103/PhysRev.60.252} {\bibfield  {journal} {\bibinfo
  {journal} {Phys. Rev.}\ }\textbf {\bibinfo {volume} {60}},\ \bibinfo {pages}
  {252} (\bibinfo {year} {1941})}\BibitemShut {NoStop}%
\bibitem [{\citenamefont {Baxter}(2016)}]{baxter2016exactly}%
  \BibitemOpen
  \bibfield  {author} {\bibinfo {author} {\bibfnamefont {R.~J.}\ \bibnamefont
  {Baxter}},\ }\href@noop {} {\emph {\bibinfo {title} {{Exactly Solved Models
  in Statistical Mechanics}}}}\ (\bibinfo  {publisher} {Elsevier},\ \bibinfo
  {year} {2016})\BibitemShut {NoStop}%
\bibitem [{\citenamefont {Seiberg}\ \emph {et~al.}(2024)\citenamefont
  {Seiberg}, \citenamefont {Seifnashri},\ and\ \citenamefont
  {Shao}}]{seiberg2024non}%
  \BibitemOpen
  \bibfield  {author} {\bibinfo {author} {\bibfnamefont {N.}~\bibnamefont
  {Seiberg}}, \bibinfo {author} {\bibfnamefont {S.}~\bibnamefont
  {Seifnashri}},\ and\ \bibinfo {author} {\bibfnamefont {S.-H.}\ \bibnamefont
  {Shao}},\ }\bibfield  {title} {\bibinfo {title} {{Non-invertible symmetries
  and LSM-type constraints on a tensor product Hilbert space}},\ }\href
  {https://doi.org/10.21468/SciPostPhys.16.6.154} {\bibfield  {journal}
  {\bibinfo  {journal} {SciPost Physics}\ }\textbf {\bibinfo {volume} {16}},\
  \bibinfo {pages} {154} (\bibinfo {year} {2024})}\BibitemShut {NoStop}%
\bibitem [{\citenamefont {Onsager}(1944)}]{PhysRev.65.117}%
  \BibitemOpen
  \bibfield  {author} {\bibinfo {author} {\bibfnamefont {L.}~\bibnamefont
  {Onsager}},\ }\bibfield  {title} {\bibinfo {title} {{Crystal Statistics. I. A
  Two-Dimensional Model with an Order-Disorder Transition}},\ }\href
  {https://doi.org/10.1103/PhysRev.65.117} {\bibfield  {journal} {\bibinfo
  {journal} {Phys. Rev.}\ }\textbf {\bibinfo {volume} {65}},\ \bibinfo {pages}
  {117} (\bibinfo {year} {1944})}\BibitemShut {NoStop}%
\bibitem [{\citenamefont {Nayak}\ and\ \citenamefont
  {Wilczek}(1996)}]{nayak19962n}%
  \BibitemOpen
  \bibfield  {author} {\bibinfo {author} {\bibfnamefont {C.}~\bibnamefont
  {Nayak}}\ and\ \bibinfo {author} {\bibfnamefont {F.}~\bibnamefont
  {Wilczek}},\ }\bibfield  {title} {\bibinfo {title} {2n-quasihole states
  realize $2^{n-1}$-dimensional spinor braiding statistics in paired quantum
  hall states},\ }\href {https://doi.org/10.1016/0550-3213(96)00430-0}
  {\bibfield  {journal} {\bibinfo  {journal} {Nuclear Physics B}\ }\textbf
  {\bibinfo {volume} {479}},\ \bibinfo {pages} {529} (\bibinfo {year}
  {1996})}\BibitemShut {NoStop}%
\bibitem [{\citenamefont {Freedman}\ \emph {et~al.}(2003)\citenamefont
  {Freedman}, \citenamefont {Kitaev}, \citenamefont {Larsen},\ and\
  \citenamefont {Wang}}]{freedman2003topological}%
  \BibitemOpen
  \bibfield  {author} {\bibinfo {author} {\bibfnamefont {M.}~\bibnamefont
  {Freedman}}, \bibinfo {author} {\bibfnamefont {A.}~\bibnamefont {Kitaev}},
  \bibinfo {author} {\bibfnamefont {M.}~\bibnamefont {Larsen}},\ and\ \bibinfo
  {author} {\bibfnamefont {Z.}~\bibnamefont {Wang}},\ }\bibfield  {title}
  {\bibinfo {title} {Topological quantum computation},\ }\href
  {https://doi.org/10.1090/S0273-0979-02-00964-3} {\bibfield  {journal}
  {\bibinfo  {journal} {Bulletin of the American Mathematical Society}\
  }\textbf {\bibinfo {volume} {40}},\ \bibinfo {pages} {31} (\bibinfo {year}
  {2003})}\BibitemShut {NoStop}%
\bibitem [{\citenamefont {Nayak}\ \emph {et~al.}(2008)\citenamefont {Nayak},
  \citenamefont {Simon}, \citenamefont {Stern}, \citenamefont {Freedman},\ and\
  \citenamefont {Das~Sarma}}]{RevModPhys.80.1083}%
  \BibitemOpen
  \bibfield  {author} {\bibinfo {author} {\bibfnamefont {C.}~\bibnamefont
  {Nayak}}, \bibinfo {author} {\bibfnamefont {S.~H.}\ \bibnamefont {Simon}},
  \bibinfo {author} {\bibfnamefont {A.}~\bibnamefont {Stern}}, \bibinfo
  {author} {\bibfnamefont {M.}~\bibnamefont {Freedman}},\ and\ \bibinfo
  {author} {\bibfnamefont {S.}~\bibnamefont {Das~Sarma}},\ }\bibfield  {title}
  {\bibinfo {title} {Non-abelian anyons and topological quantum computation},\
  }\href {https://doi.org/10.1103/RevModPhys.80.1083} {\bibfield  {journal}
  {\bibinfo  {journal} {Rev. Mod. Phys.}\ }\textbf {\bibinfo {volume} {80}},\
  \bibinfo {pages} {1083} (\bibinfo {year} {2008})}\BibitemShut {NoStop}%
\bibitem [{\citenamefont {Wang}(2010)}]{wang2010topological}%
  \BibitemOpen
  \bibfield  {author} {\bibinfo {author} {\bibfnamefont {Z.}~\bibnamefont
  {Wang}},\ }\href@noop {} {\emph {\bibinfo {title} {Topological quantum
  computation}}},\ \bibinfo {number} {112}\ (\bibinfo  {publisher} {American
  Mathematical Soc.},\ \bibinfo {year} {2010})\BibitemShut {NoStop}%
\bibitem [{\citenamefont {Read}\ and\ \citenamefont
  {Green}(2000)}]{PhysRevB.61.10267}%
  \BibitemOpen
  \bibfield  {author} {\bibinfo {author} {\bibfnamefont {N.}~\bibnamefont
  {Read}}\ and\ \bibinfo {author} {\bibfnamefont {D.}~\bibnamefont {Green}},\
  }\bibfield  {title} {\bibinfo {title} {Paired states of fermions in two
  dimensions with breaking of parity and time-reversal symmetries and the
  fractional quantum hall effect},\ }\href
  {https://doi.org/10.1103/PhysRevB.61.10267} {\bibfield  {journal} {\bibinfo
  {journal} {Phys. Rev. B}\ }\textbf {\bibinfo {volume} {61}},\ \bibinfo
  {pages} {10267} (\bibinfo {year} {2000})}\BibitemShut {NoStop}%
\bibitem [{\citenamefont {Kitaev}(2001)}]{kitaev2001unpaired}%
  \BibitemOpen
  \bibfield  {author} {\bibinfo {author} {\bibfnamefont {A.~Y.}\ \bibnamefont
  {Kitaev}},\ }\bibfield  {title} {\bibinfo {title} {{Unpaired Majorana
  fermions in quantum wires}},\ }\href
  {https://doi.org/10.1070/1063-7869/44/10S/S29} {\bibfield  {journal}
  {\bibinfo  {journal} {Physics-uspekhi}\ }\textbf {\bibinfo {volume} {44}},\
  \bibinfo {pages} {131} (\bibinfo {year} {2001})}\BibitemShut {NoStop}%
\bibitem [{\citenamefont {Motrunich}\ \emph {et~al.}(2001)\citenamefont
  {Motrunich}, \citenamefont {Damle},\ and\ \citenamefont
  {Huse}}]{PhysRevB.63.224204}%
  \BibitemOpen
  \bibfield  {author} {\bibinfo {author} {\bibfnamefont {O.}~\bibnamefont
  {Motrunich}}, \bibinfo {author} {\bibfnamefont {K.}~\bibnamefont {Damle}},\
  and\ \bibinfo {author} {\bibfnamefont {D.~A.}\ \bibnamefont {Huse}},\
  }\bibfield  {title} {\bibinfo {title} {{Griffiths effects and quantum
  critical points in dirty superconductors without spin-rotation invariance:
  One-dimensional examples}},\ }\href
  {https://doi.org/10.1103/PhysRevB.63.224204} {\bibfield  {journal} {\bibinfo
  {journal} {Phys. Rev. B}\ }\textbf {\bibinfo {volume} {63}},\ \bibinfo
  {pages} {224204} (\bibinfo {year} {2001})}\BibitemShut {NoStop}%
\bibitem [{\citenamefont {Karzig}\ \emph {et~al.}(2017)\citenamefont {Karzig},
  \citenamefont {Knapp}, \citenamefont {Lutchyn}, \citenamefont {Bonderson},
  \citenamefont {Hastings}, \citenamefont {Nayak}, \citenamefont {Alicea},
  \citenamefont {Flensberg}, \citenamefont {Plugge}, \citenamefont {Oreg},
  \citenamefont {Marcus},\ and\ \citenamefont {Freedman}}]{PhysRevB.95.235305}%
  \BibitemOpen
  \bibfield  {author} {\bibinfo {author} {\bibfnamefont {T.}~\bibnamefont
  {Karzig}}, \bibinfo {author} {\bibfnamefont {C.}~\bibnamefont {Knapp}},
  \bibinfo {author} {\bibfnamefont {R.~M.}\ \bibnamefont {Lutchyn}}, \bibinfo
  {author} {\bibfnamefont {P.}~\bibnamefont {Bonderson}}, \bibinfo {author}
  {\bibfnamefont {M.~B.}\ \bibnamefont {Hastings}}, \bibinfo {author}
  {\bibfnamefont {C.}~\bibnamefont {Nayak}}, \bibinfo {author} {\bibfnamefont
  {J.}~\bibnamefont {Alicea}}, \bibinfo {author} {\bibfnamefont
  {K.}~\bibnamefont {Flensberg}}, \bibinfo {author} {\bibfnamefont
  {S.}~\bibnamefont {Plugge}}, \bibinfo {author} {\bibfnamefont
  {Y.}~\bibnamefont {Oreg}}, \bibinfo {author} {\bibfnamefont {C.~M.}\
  \bibnamefont {Marcus}},\ and\ \bibinfo {author} {\bibfnamefont {M.~H.}\
  \bibnamefont {Freedman}},\ }\bibfield  {title} {\bibinfo {title} {{Scalable
  designs for quasiparticle-poisoning-protected topological quantum computation
  with Majorana zero modes}},\ }\href
  {https://doi.org/10.1103/PhysRevB.95.235305} {\bibfield  {journal} {\bibinfo
  {journal} {Phys. Rev. B}\ }\textbf {\bibinfo {volume} {95}},\ \bibinfo
  {pages} {235305} (\bibinfo {year} {2017})}\BibitemShut {NoStop}%
\bibitem [{\citenamefont {Tran}\ \emph {et~al.}(2020)\citenamefont {Tran},
  \citenamefont {Bocharov}, \citenamefont {Bauer},\ and\ \citenamefont
  {Bonderson}}]{10.21468/SciPostPhys.8.6.091}%
  \BibitemOpen
  \bibfield  {author} {\bibinfo {author} {\bibfnamefont {A.}~\bibnamefont
  {Tran}}, \bibinfo {author} {\bibfnamefont {A.}~\bibnamefont {Bocharov}},
  \bibinfo {author} {\bibfnamefont {B.}~\bibnamefont {Bauer}},\ and\ \bibinfo
  {author} {\bibfnamefont {P.}~\bibnamefont {Bonderson}},\ }\bibfield  {title}
  {\bibinfo {title} {{Optimizing Clifford gate generation for measurement-only
  topological quantum computation with Majorana zero modes}},\ }\href
  {https://doi.org/10.21468/SciPostPhys.8.6.091} {\bibfield  {journal}
  {\bibinfo  {journal} {SciPost Phys.}\ }\textbf {\bibinfo {volume} {8}},\
  \bibinfo {pages} {091} (\bibinfo {year} {2020})}\BibitemShut {NoStop}%
\bibitem [{\citenamefont {Brown}\ \emph {et~al.}(2017)\citenamefont {Brown},
  \citenamefont {Laubscher}, \citenamefont {Kesselring},\ and\ \citenamefont
  {Wootton}}]{PhysRevX.7.021029}%
  \BibitemOpen
  \bibfield  {author} {\bibinfo {author} {\bibfnamefont {B.~J.}\ \bibnamefont
  {Brown}}, \bibinfo {author} {\bibfnamefont {K.}~\bibnamefont {Laubscher}},
  \bibinfo {author} {\bibfnamefont {M.~S.}\ \bibnamefont {Kesselring}},\ and\
  \bibinfo {author} {\bibfnamefont {J.~R.}\ \bibnamefont {Wootton}},\
  }\bibfield  {title} {\bibinfo {title} {{Poking Holes and Cutting Corners to
  Achieve Clifford Gates with the Surface Code}},\ }\href
  {https://doi.org/10.1103/PhysRevX.7.021029} {\bibfield  {journal} {\bibinfo
  {journal} {Phys. Rev. X}\ }\textbf {\bibinfo {volume} {7}},\ \bibinfo {pages}
  {021029} (\bibinfo {year} {2017})}\BibitemShut {NoStop}%
\bibitem [{\citenamefont {Aasen}\ \emph {et~al.}(2019)\citenamefont {Aasen},
  \citenamefont {Lake},\ and\ \citenamefont {Walker}}]{aasen2019fermion}%
  \BibitemOpen
  \bibfield  {author} {\bibinfo {author} {\bibfnamefont {D.}~\bibnamefont
  {Aasen}}, \bibinfo {author} {\bibfnamefont {E.}~\bibnamefont {Lake}},\ and\
  \bibinfo {author} {\bibfnamefont {K.}~\bibnamefont {Walker}},\ }\bibfield
  {title} {\bibinfo {title} {Fermion condensation and super pivotal
  categories},\ }\href {https://doi.org/10.1063/1.5045669} {\bibfield
  {journal} {\bibinfo  {journal} {Journal of Mathematical Physics}\ }\textbf
  {\bibinfo {volume} {60}},\ \bibinfo {pages} {121901} (\bibinfo {year}
  {2019})}\BibitemShut {NoStop}%
\bibitem [{\citenamefont {Luo}\ and\ \citenamefont
  {Clark}(2019)}]{PhysRevLett.122.226401}%
  \BibitemOpen
  \bibfield  {author} {\bibinfo {author} {\bibfnamefont {D.}~\bibnamefont
  {Luo}}\ and\ \bibinfo {author} {\bibfnamefont {B.~K.}\ \bibnamefont
  {Clark}},\ }\bibfield  {title} {\bibinfo {title} {Backflow transformations
  via neural networks for quantum many-body wave functions},\ }\href
  {https://doi.org/10.1103/PhysRevLett.122.226401} {\bibfield  {journal}
  {\bibinfo  {journal} {Phys. Rev. Lett.}\ }\textbf {\bibinfo {volume} {122}},\
  \bibinfo {pages} {226401} (\bibinfo {year} {2019})}\BibitemShut {NoStop}%
\bibitem [{\citenamefont {Acevedo}\ \emph {et~al.}(2020)\citenamefont
  {Acevedo}, \citenamefont {Curry}, \citenamefont {Joshi}, \citenamefont
  {Leroux},\ and\ \citenamefont {Malaya}}]{acevedo2020vandermonde}%
  \BibitemOpen
  \bibfield  {author} {\bibinfo {author} {\bibfnamefont {A.}~\bibnamefont
  {Acevedo}}, \bibinfo {author} {\bibfnamefont {M.}~\bibnamefont {Curry}},
  \bibinfo {author} {\bibfnamefont {S.~H.}\ \bibnamefont {Joshi}}, \bibinfo
  {author} {\bibfnamefont {B.}~\bibnamefont {Leroux}},\ and\ \bibinfo {author}
  {\bibfnamefont {N.}~\bibnamefont {Malaya}},\ }\bibfield  {title} {\bibinfo
  {title} {Vandermonde wave function ansatz for improved variational monte
  carlo},\ }in\ \href {https://doi.org/10.1109/DLS51937.2020.00010} {\emph
  {\bibinfo {booktitle} {2020 IEEE/ACM Fourth Workshop on Deep Learning on
  Supercomputers (DLS)}}}\ (\bibinfo {organization} {IEEE},\ \bibinfo {year}
  {2020})\ pp.\ \bibinfo {pages} {40--47}\BibitemShut {NoStop}%
\bibitem [{\citenamefont {von Glehn}\ \emph {et~al.}(2022)\citenamefont {von
  Glehn}, \citenamefont {Spencer},\ and\ \citenamefont {Pfau}}]{von2022self}%
  \BibitemOpen
  \bibfield  {author} {\bibinfo {author} {\bibfnamefont {I.}~\bibnamefont {von
  Glehn}}, \bibinfo {author} {\bibfnamefont {J.~S.}\ \bibnamefont {Spencer}},\
  and\ \bibinfo {author} {\bibfnamefont {D.}~\bibnamefont {Pfau}},\ }\bibfield
  {title} {\bibinfo {title} {A self-attention ansatz for ab-initio quantum
  chemistry},\ }\href {https://arxiv.org/abs/2211.13672} {\bibfield  {journal}
  {\bibinfo  {journal} {arXiv:2211.13672}\ } (\bibinfo {year}
  {2022})}\BibitemShut {NoStop}%
\bibitem [{\citenamefont {Chen}\ and\ \citenamefont
  {Lu}(2023)}]{chen2023exact}%
  \BibitemOpen
  \bibfield  {author} {\bibinfo {author} {\bibfnamefont {Z.}~\bibnamefont
  {Chen}}\ and\ \bibinfo {author} {\bibfnamefont {J.}~\bibnamefont {Lu}},\
  }\bibfield  {title} {\bibinfo {title} {Exact and efficient representation of
  totally anti-symmetric functions},\ }\href {https://arxiv.org/abs/2311.05064}
  {\bibfield  {journal} {\bibinfo  {journal} {arXiv:2311.05064}\ } (\bibinfo
  {year} {2023})}\BibitemShut {NoStop}%
\bibitem [{\citenamefont {Geier}\ \emph {et~al.}(2025)\citenamefont {Geier},
  \citenamefont {Nazaryan}, \citenamefont {Zaklama},\ and\ \citenamefont
  {Fu}}]{geier2025attention}%
  \BibitemOpen
  \bibfield  {author} {\bibinfo {author} {\bibfnamefont {M.}~\bibnamefont
  {Geier}}, \bibinfo {author} {\bibfnamefont {K.}~\bibnamefont {Nazaryan}},
  \bibinfo {author} {\bibfnamefont {T.}~\bibnamefont {Zaklama}},\ and\ \bibinfo
  {author} {\bibfnamefont {L.}~\bibnamefont {Fu}},\ }\bibfield  {title}
  {\bibinfo {title} {Is attention all you need to solve the correlated electron
  problem?},\ }\href {https://arxiv.org/abs/2502.05383} {\bibfield  {journal}
  {\bibinfo  {journal} {arXiv:2502.05383}\ } (\bibinfo {year}
  {2025})}\BibitemShut {NoStop}%
\bibitem [{\citenamefont {Biamonte}\ \emph {et~al.}(2011)\citenamefont
  {Biamonte}, \citenamefont {Clark},\ and\ \citenamefont
  {Jaksch}}]{biamonte2011categorical}%
  \BibitemOpen
  \bibfield  {author} {\bibinfo {author} {\bibfnamefont {J.~D.}\ \bibnamefont
  {Biamonte}}, \bibinfo {author} {\bibfnamefont {S.~R.}\ \bibnamefont
  {Clark}},\ and\ \bibinfo {author} {\bibfnamefont {D.}~\bibnamefont
  {Jaksch}},\ }\bibfield  {title} {\bibinfo {title} {{Categorical Tensor
  Network States}},\ }\href {https://doi.org/10.1063/1.3672009} {\bibfield
  {journal} {\bibinfo  {journal} {AIP Advances}\ }\textbf {\bibinfo {volume}
  {1}},\ \bibinfo {pages} {042172} (\bibinfo {year} {2011})}\BibitemShut
  {NoStop}%
\bibitem [{\citenamefont {Coecke}\ and\ \citenamefont
  {Duncan}(2007)}]{coecke2007graphical}%
  \BibitemOpen
  \bibfield  {author} {\bibinfo {author} {\bibfnamefont {B.}~\bibnamefont
  {Coecke}}\ and\ \bibinfo {author} {\bibfnamefont {R.}~\bibnamefont
  {Duncan}},\ }\bibfield  {title} {\bibinfo {title} {A graphical calculus for
  quantum observables},\ }\href
  {https://www.cs.ox.ac.uk/people/bob.coecke/GreenRed.pdf} {\bibfield
  {journal} {\bibinfo  {journal} {Preprint}\ } (\bibinfo {year}
  {2007})}\BibitemShut {NoStop}%
\bibitem [{\citenamefont {Coecke}\ and\ \citenamefont
  {Duncan}(2008)}]{coecke2008interacting}%
  \BibitemOpen
  \bibfield  {author} {\bibinfo {author} {\bibfnamefont {B.}~\bibnamefont
  {Coecke}}\ and\ \bibinfo {author} {\bibfnamefont {R.}~\bibnamefont
  {Duncan}},\ }\bibfield  {title} {\bibinfo {title} {{Interacting Quantum
  Observables}},\ }in\ \href {https://doi.org/10.1007/978-3-540-70583-3_25}
  {\emph {\bibinfo {booktitle} {International Colloquium on Automata,
  Languages, and Programming}}}\ (\bibinfo {organization} {Springer},\ \bibinfo
  {year} {2008})\ pp.\ \bibinfo {pages} {298--310}\BibitemShut {NoStop}%
\bibitem [{\citenamefont {Coecke}\ and\ \citenamefont
  {Duncan}(2011)}]{coecke2011interacting}%
  \BibitemOpen
  \bibfield  {author} {\bibinfo {author} {\bibfnamefont {B.}~\bibnamefont
  {Coecke}}\ and\ \bibinfo {author} {\bibfnamefont {R.}~\bibnamefont
  {Duncan}},\ }\bibfield  {title} {\bibinfo {title} {Interacting quantum
  observables: categorical algebra and diagrammatics},\ }\href
  {https://doi.org/10.1088/1367-2630/13/4/043016} {\bibfield  {journal}
  {\bibinfo  {journal} {New Journal of Physics}\ }\textbf {\bibinfo {volume}
  {13}},\ \bibinfo {pages} {043016} (\bibinfo {year} {2011})}\BibitemShut
  {NoStop}%
\bibitem [{\citenamefont {Coecke}\ and\ \citenamefont
  {Kissinger}(2010)}]{coecke2010compositional}%
  \BibitemOpen
  \bibfield  {author} {\bibinfo {author} {\bibfnamefont {B.}~\bibnamefont
  {Coecke}}\ and\ \bibinfo {author} {\bibfnamefont {A.}~\bibnamefont
  {Kissinger}},\ }\bibfield  {title} {\bibinfo {title} {{The Compositional
  Structure of Multipartite Quantum Entanglement}},\ }in\ \href
  {https://doi.org/10.1007/978-3-642-14162-1_25} {\emph {\bibinfo {booktitle}
  {International Colloquium on Automata, Languages, and Programming}}}\
  (\bibinfo {organization} {Springer},\ \bibinfo {year} {2010})\ pp.\ \bibinfo
  {pages} {297--308}\BibitemShut {NoStop}%
\bibitem [{\citenamefont {van~de Wetering}(2020)}]{van2020zx}%
  \BibitemOpen
  \bibfield  {author} {\bibinfo {author} {\bibfnamefont {J.}~\bibnamefont
  {van~de Wetering}},\ }\bibfield  {title} {\bibinfo {title} {{ZX-calculus for
  the working quantum computer scientist}},\ }\href
  {https://arxiv.org/abs/2012.13966} {\bibfield  {journal} {\bibinfo  {journal}
  {arXiv:2012.13966}\ } (\bibinfo {year} {2020})}\BibitemShut {NoStop}%
\bibitem [{\citenamefont {Ng}\ and\ \citenamefont
  {Wang}(2017)}]{ng2017universal}%
  \BibitemOpen
  \bibfield  {author} {\bibinfo {author} {\bibfnamefont {K.~F.}\ \bibnamefont
  {Ng}}\ and\ \bibinfo {author} {\bibfnamefont {Q.}~\bibnamefont {Wang}},\
  }\bibfield  {title} {\bibinfo {title} {{A universal completion of the
  ZX-calculus}},\ }\href {https://arxiv.org/abs/1706.09877} {\bibfield
  {journal} {\bibinfo  {journal} {arXiv:1706.09877}\ } (\bibinfo {year}
  {2017})}\BibitemShut {NoStop}%
\bibitem [{\citenamefont {Ng}\ and\ \citenamefont
  {Wang}(2018)}]{ng2018completeness}%
  \BibitemOpen
  \bibfield  {author} {\bibinfo {author} {\bibfnamefont {K.~F.}\ \bibnamefont
  {Ng}}\ and\ \bibinfo {author} {\bibfnamefont {Q.}~\bibnamefont {Wang}},\
  }\bibfield  {title} {\bibinfo {title} {{Completeness of the ZX-calculus for
  Pure Qubit Clifford+T Quantum Mechanics}},\ }\href
  {https://arxiv.org/abs/1801.07993} {\bibfield  {journal} {\bibinfo  {journal}
  {arXiv:1801.07993}\ } (\bibinfo {year} {2018})}\BibitemShut {NoStop}%
\bibitem [{\citenamefont {Hadzihasanovic}\ \emph {et~al.}(2018)\citenamefont
  {Hadzihasanovic}, \citenamefont {Ng},\ and\ \citenamefont
  {Wang}}]{hadzihasanovic2018two}%
  \BibitemOpen
  \bibfield  {author} {\bibinfo {author} {\bibfnamefont {A.}~\bibnamefont
  {Hadzihasanovic}}, \bibinfo {author} {\bibfnamefont {K.~F.}\ \bibnamefont
  {Ng}},\ and\ \bibinfo {author} {\bibfnamefont {Q.}~\bibnamefont {Wang}},\
  }\bibfield  {title} {\bibinfo {title} {Two complete axiomatisations of
  pure-state qubit quantum computing},\ }in\ \href
  {https://doi.org/10.1145/3209108.3209128} {\emph {\bibinfo {booktitle}
  {Proceedings of the 33rd annual ACM/IEEE symposium on logic in computer
  science}}}\ (\bibinfo {year} {2018})\ pp.\ \bibinfo {pages}
  {502--511}\BibitemShut {NoStop}%
\bibitem [{\citenamefont {Jeandel}\ \emph {et~al.}(2020)\citenamefont
  {Jeandel}, \citenamefont {Perdrix},\ and\ \citenamefont
  {Vilmart}}]{jeandel2020completeness}%
  \BibitemOpen
  \bibfield  {author} {\bibinfo {author} {\bibfnamefont {E.}~\bibnamefont
  {Jeandel}}, \bibinfo {author} {\bibfnamefont {S.}~\bibnamefont {Perdrix}},\
  and\ \bibinfo {author} {\bibfnamefont {R.}~\bibnamefont {Vilmart}},\
  }\bibfield  {title} {\bibinfo {title} {{Completeness of the ZX-Calculus}},\
  }\href {https://doi.org/10.23638/LMCS-16(2:11)2020} {\bibfield  {journal}
  {\bibinfo  {journal} {Logical Methods in Computer Science}\ }\textbf
  {\bibinfo {volume} {16}},\ \bibinfo {pages} {2} (\bibinfo {year}
  {2020})}\BibitemShut {NoStop}%
\bibitem [{\citenamefont {Wang}(2022)}]{wang2022completeness}%
  \BibitemOpen
  \bibfield  {author} {\bibinfo {author} {\bibfnamefont {Q.}~\bibnamefont
  {Wang}},\ }\bibfield  {title} {\bibinfo {title} {{Completeness of the
  ZX-calculus}},\ }\href {https://arxiv.org/abs/2209.14894} {\bibfield
  {journal} {\bibinfo  {journal} {arXiv:2209.14894}\ } (\bibinfo {year}
  {2022})}\BibitemShut {NoStop}%
\bibitem [{\citenamefont {Backens}(2014)}]{backens2014zx}%
  \BibitemOpen
  \bibfield  {author} {\bibinfo {author} {\bibfnamefont {M.}~\bibnamefont
  {Backens}},\ }\bibfield  {title} {\bibinfo {title} {{The ZX-calculus is
  complete for stabilizer quantum mechanics}},\ }\href
  {https://doi.org/10.1088/1367-2630/16/9/093021} {\bibfield  {journal}
  {\bibinfo  {journal} {New Journal of Physics}\ }\textbf {\bibinfo {volume}
  {16}},\ \bibinfo {pages} {093021} (\bibinfo {year} {2014})}\BibitemShut
  {NoStop}%
\bibitem [{\citenamefont
  {Hadzihasanovic}(2015)}]{hadzihasanovic2015diagrammatic}%
  \BibitemOpen
  \bibfield  {author} {\bibinfo {author} {\bibfnamefont {A.}~\bibnamefont
  {Hadzihasanovic}},\ }\bibfield  {title} {\bibinfo {title} {{A Diagrammatic
  Axiomatisation for Qubit Entanglement}},\ }in\ \href
  {https://doi.org/10.1109/LICS.2015.59} {\emph {\bibinfo {booktitle} {2015
  30th Annual ACM/IEEE Symposium on Logic in Computer Science}}}\ (\bibinfo
  {organization} {IEEE},\ \bibinfo {year} {2015})\ pp.\ \bibinfo {pages}
  {573--584}\BibitemShut {NoStop}%
\bibitem [{\citenamefont {Carette}\ \emph {et~al.}(2023)\citenamefont
  {Carette}, \citenamefont {Moutot}, \citenamefont {Perez},\ and\ \citenamefont
  {Vilmart}}]{carette2023compositionality}%
  \BibitemOpen
  \bibfield  {author} {\bibinfo {author} {\bibfnamefont {T.}~\bibnamefont
  {Carette}}, \bibinfo {author} {\bibfnamefont {E.}~\bibnamefont {Moutot}},
  \bibinfo {author} {\bibfnamefont {T.}~\bibnamefont {Perez}},\ and\ \bibinfo
  {author} {\bibfnamefont {R.}~\bibnamefont {Vilmart}},\ }\bibfield  {title}
  {\bibinfo {title} {Compositionality of planar perfect matchings},\ }\href
  {https://arxiv.org/abs/2302.08767} {\bibfield  {journal} {\bibinfo  {journal}
  {arXiv:2302.08767}\ } (\bibinfo {year} {2023})}\BibitemShut {NoStop}%
\bibitem [{\citenamefont {Koch}(2022)}]{koch2022quantum}%
  \BibitemOpen
  \bibfield  {author} {\bibinfo {author} {\bibfnamefont {M.}~\bibnamefont
  {Koch}},\ }\bibfield  {title} {\bibinfo {title} {{Quantum Machine Learning
  using the ZXW-Calculus}},\ }\href {https://arxiv.org/abs/2210.11523}
  {\bibfield  {journal} {\bibinfo  {journal} {arXiv:2210.11523}\ } (\bibinfo
  {year} {2022})}\BibitemShut {NoStop}%
\bibitem [{\citenamefont {Po{\'o}r}\ \emph {et~al.}(2023)\citenamefont
  {Po{\'o}r}, \citenamefont {Wang}, \citenamefont {Shaikh}, \citenamefont
  {Yeh}, \citenamefont {Yeung},\ and\ \citenamefont
  {Coecke}}]{poor2023completeness}%
  \BibitemOpen
  \bibfield  {author} {\bibinfo {author} {\bibfnamefont {B.}~\bibnamefont
  {Po{\'o}r}}, \bibinfo {author} {\bibfnamefont {Q.}~\bibnamefont {Wang}},
  \bibinfo {author} {\bibfnamefont {R.~A.}\ \bibnamefont {Shaikh}}, \bibinfo
  {author} {\bibfnamefont {L.}~\bibnamefont {Yeh}}, \bibinfo {author}
  {\bibfnamefont {R.}~\bibnamefont {Yeung}},\ and\ \bibinfo {author}
  {\bibfnamefont {B.}~\bibnamefont {Coecke}},\ }\bibfield  {title} {\bibinfo
  {title} {{Completeness for arbitrary finite dimensions of ZXW-calculus, a
  unifying calculus}},\ }in\ \href
  {https://doi.org/10.1109/LICS56636.2023.10175672} {\emph {\bibinfo
  {booktitle} {2023 38th Annual ACM/IEEE Symposium on Logic in Computer Science
  (LICS)}}}\ (\bibinfo {organization} {IEEE},\ \bibinfo {year} {2023})\ pp.\
  \bibinfo {pages} {1--14}\BibitemShut {NoStop}%
\bibitem [{\citenamefont {Bravyi}\ and\ \citenamefont
  {Kitaev}(2002)}]{bravyi2002fermionic}%
  \BibitemOpen
  \bibfield  {author} {\bibinfo {author} {\bibfnamefont {S.~B.}\ \bibnamefont
  {Bravyi}}\ and\ \bibinfo {author} {\bibfnamefont {A.~Y.}\ \bibnamefont
  {Kitaev}},\ }\bibfield  {title} {\bibinfo {title} {{Fermionic Quantum
  Computation}},\ }\href {https://doi.org/10.1006/aphy.2002.6254} {\bibfield
  {journal} {\bibinfo  {journal} {Annals of Physics}\ }\textbf {\bibinfo
  {volume} {298}},\ \bibinfo {pages} {210} (\bibinfo {year}
  {2002})}\BibitemShut {NoStop}%
\bibitem [{\citenamefont {Bravyi}(2006)}]{PhysRevA.73.042313}%
  \BibitemOpen
  \bibfield  {author} {\bibinfo {author} {\bibfnamefont {S.}~\bibnamefont
  {Bravyi}},\ }\bibfield  {title} {\bibinfo {title} {{Universal quantum
  computation with the $\nu=5/2$ fractional quantum Hall state}},\ }\href
  {https://doi.org/10.1103/PhysRevA.73.042313} {\bibfield  {journal} {\bibinfo
  {journal} {Phys. Rev. A}\ }\textbf {\bibinfo {volume} {73}},\ \bibinfo
  {pages} {042313} (\bibinfo {year} {2006})}\BibitemShut {NoStop}%
\bibitem [{\citenamefont {Ahlbrecht}\ \emph {et~al.}(2009)\citenamefont
  {Ahlbrecht}, \citenamefont {Georgiev},\ and\ \citenamefont
  {Werner}}]{PhysRevA.79.032311}%
  \BibitemOpen
  \bibfield  {author} {\bibinfo {author} {\bibfnamefont {A.}~\bibnamefont
  {Ahlbrecht}}, \bibinfo {author} {\bibfnamefont {L.~S.}\ \bibnamefont
  {Georgiev}},\ and\ \bibinfo {author} {\bibfnamefont {R.~F.}\ \bibnamefont
  {Werner}},\ }\bibfield  {title} {\bibinfo {title} {{Implementation of
  Clifford gates in the Ising-anyon topological quantum computer}},\ }\href
  {https://doi.org/10.1103/PhysRevA.79.032311} {\bibfield  {journal} {\bibinfo
  {journal} {Phys. Rev. A}\ }\textbf {\bibinfo {volume} {79}},\ \bibinfo
  {pages} {032311} (\bibinfo {year} {2009})}\BibitemShut {NoStop}%
\bibitem [{\citenamefont {Jones}(2021)}]{jones2021planar}%
  \BibitemOpen
  \bibfield  {author} {\bibinfo {author} {\bibfnamefont {V.}~\bibnamefont
  {Jones}},\ }\bibfield  {title} {\bibinfo {title} {Planar algebras},\ }\href
  {https://doi.org/10.53733/172} {\bibfield  {journal} {\bibinfo  {journal}
  {New Zealand Journal of Mathematics}\ }\textbf {\bibinfo {volume} {52}},\
  \bibinfo {pages} {1} (\bibinfo {year} {2021})}\BibitemShut {NoStop}%
\bibitem [{\citenamefont {Cai}\ \emph {et~al.}(2009)\citenamefont {Cai},
  \citenamefont {Choudhary},\ and\ \citenamefont {Lu}}]{cai2009theory}%
  \BibitemOpen
  \bibfield  {author} {\bibinfo {author} {\bibfnamefont {J.-Y.}\ \bibnamefont
  {Cai}}, \bibinfo {author} {\bibfnamefont {V.}~\bibnamefont {Choudhary}},\
  and\ \bibinfo {author} {\bibfnamefont {P.}~\bibnamefont {Lu}},\ }\bibfield
  {title} {\bibinfo {title} {{On the Theory of Matchgate Computations}},\
  }\href {https://doi.org/10.1109/CCC.2007.22} {\bibfield  {journal} {\bibinfo
  {journal} {Theory of Computing Systems}\ }\textbf {\bibinfo {volume} {45}},\
  \bibinfo {pages} {108} (\bibinfo {year} {2009})}\BibitemShut {NoStop}%
\bibitem [{\citenamefont {Wan}\ \emph {et~al.}(2023)\citenamefont {Wan},
  \citenamefont {Huggins}, \citenamefont {Lee},\ and\ \citenamefont
  {Babbush}}]{wan2023matchgate}%
  \BibitemOpen
  \bibfield  {author} {\bibinfo {author} {\bibfnamefont {K.}~\bibnamefont
  {Wan}}, \bibinfo {author} {\bibfnamefont {W.~J.}\ \bibnamefont {Huggins}},
  \bibinfo {author} {\bibfnamefont {J.}~\bibnamefont {Lee}},\ and\ \bibinfo
  {author} {\bibfnamefont {R.}~\bibnamefont {Babbush}},\ }\bibfield  {title}
  {\bibinfo {title} {{Matchgate Shadows for Fermionic Quantum Simulation}},\
  }\href {https://doi.org/10.1007/s00220-023-04844-0} {\bibfield  {journal}
  {\bibinfo  {journal} {Communications in Mathematical Physics}\ }\textbf
  {\bibinfo {volume} {404}},\ \bibinfo {pages} {629} (\bibinfo {year}
  {2023})}\BibitemShut {NoStop}%
\bibitem [{\citenamefont {Schuch}\ \emph {et~al.}(2007)\citenamefont {Schuch},
  \citenamefont {Wolf}, \citenamefont {Verstraete},\ and\ \citenamefont
  {Cirac}}]{PhysRevLett.98.140506}%
  \BibitemOpen
  \bibfield  {author} {\bibinfo {author} {\bibfnamefont {N.}~\bibnamefont
  {Schuch}}, \bibinfo {author} {\bibfnamefont {M.~M.}\ \bibnamefont {Wolf}},
  \bibinfo {author} {\bibfnamefont {F.}~\bibnamefont {Verstraete}},\ and\
  \bibinfo {author} {\bibfnamefont {J.~I.}\ \bibnamefont {Cirac}},\ }\bibfield
  {title} {\bibinfo {title} {{Computational Complexity of Projected Entangled
  Pair States}},\ }\href {https://doi.org/10.1103/PhysRevLett.98.140506}
  {\bibfield  {journal} {\bibinfo  {journal} {Phys. Rev. Lett.}\ }\textbf
  {\bibinfo {volume} {98}},\ \bibinfo {pages} {140506} (\bibinfo {year}
  {2007})}\BibitemShut {NoStop}%
\bibitem [{\citenamefont {Haferkamp}\ \emph {et~al.}(2020)\citenamefont
  {Haferkamp}, \citenamefont {Hangleiter}, \citenamefont {Eisert},\ and\
  \citenamefont {Gluza}}]{PhysRevResearch.2.013010}%
  \BibitemOpen
  \bibfield  {author} {\bibinfo {author} {\bibfnamefont {J.}~\bibnamefont
  {Haferkamp}}, \bibinfo {author} {\bibfnamefont {D.}~\bibnamefont
  {Hangleiter}}, \bibinfo {author} {\bibfnamefont {J.}~\bibnamefont {Eisert}},\
  and\ \bibinfo {author} {\bibfnamefont {M.}~\bibnamefont {Gluza}},\ }\bibfield
   {title} {\bibinfo {title} {Contracting projected entangled pair states is
  average-case hard},\ }\href
  {https://doi.org/10.1103/PhysRevResearch.2.013010} {\bibfield  {journal}
  {\bibinfo  {journal} {Phys. Rev. Res.}\ }\textbf {\bibinfo {volume} {2}},\
  \bibinfo {pages} {013010} (\bibinfo {year} {2020})}\BibitemShut {NoStop}%
\bibitem [{\citenamefont {Bampounis}\ and\ \citenamefont
  {Soares~Barbosa}(2024)}]{matchgatehierarchiy}%
  \BibitemOpen
  \bibfield  {author} {\bibinfo {author} {\bibfnamefont {A.}~\bibnamefont
  {Bampounis}}\ and\ \bibinfo {author} {\bibfnamefont {d.~S.~N.}\ \bibnamefont
  {Soares~Barbosa}, \bibfnamefont {Rui~and}},\ }\bibfield  {title} {\bibinfo
  {title} {{Matchgate hierarchy: A Clifford-like hierarchy for deterministic
  gate teleportation in matchgate circuits}},\ }\href
  {https://arxiv.org/abs/2410.01887} {\bibfield  {journal} {\bibinfo  {journal}
  {arXiv:2410.01887}\ } (\bibinfo {year} {2024})}\BibitemShut {NoStop}%
\bibitem [{\citenamefont {Janzing}\ \emph {et~al.}(2005)\citenamefont
  {Janzing}, \citenamefont {Wocjan},\ and\ \citenamefont
  {Beth}}]{janzing2005non}%
  \BibitemOpen
  \bibfield  {author} {\bibinfo {author} {\bibfnamefont {D.}~\bibnamefont
  {Janzing}}, \bibinfo {author} {\bibfnamefont {P.}~\bibnamefont {Wocjan}},\
  and\ \bibinfo {author} {\bibfnamefont {T.}~\bibnamefont {Beth}},\ }\bibfield
  {title} {\bibinfo {title} {{``NON-IDENTITY-CHECK'' IS QMA-COMPLETE}},\ }\href
  {https://doi.org/10.1142/S0219749905001067} {\bibfield  {journal} {\bibinfo
  {journal} {International Journal of Quantum Information}\ }\textbf {\bibinfo
  {volume} {3}},\ \bibinfo {pages} {463} (\bibinfo {year} {2005})}\BibitemShut
  {NoStop}%
\bibitem [{\citenamefont {Ji}\ and\ \citenamefont {Wu}(2009)}]{ji2009non}%
  \BibitemOpen
  \bibfield  {author} {\bibinfo {author} {\bibfnamefont {Z.}~\bibnamefont
  {Ji}}\ and\ \bibinfo {author} {\bibfnamefont {X.}~\bibnamefont {Wu}},\
  }\bibfield  {title} {\bibinfo {title} {{Non-Identity Check Remains
  QMA-Complete for Short Circuits}},\ }\href {https://arxiv.org/abs/0906.5416}
  {\bibfield  {journal} {\bibinfo  {journal} {arXiv:0906.5416}\ } (\bibinfo
  {year} {2009})}\BibitemShut {NoStop}%
\bibitem [{\citenamefont {Tarabunga}\ \emph {et~al.}(2024)\citenamefont
  {Tarabunga}, \citenamefont {Tirrito}, \citenamefont {Ba\~nuls},\ and\
  \citenamefont {Dalmonte}}]{PhysRevLett.133.010601}%
  \BibitemOpen
  \bibfield  {author} {\bibinfo {author} {\bibfnamefont {P.~S.}\ \bibnamefont
  {Tarabunga}}, \bibinfo {author} {\bibfnamefont {E.}~\bibnamefont {Tirrito}},
  \bibinfo {author} {\bibfnamefont {M.~C.}\ \bibnamefont {Ba\~nuls}},\ and\
  \bibinfo {author} {\bibfnamefont {M.}~\bibnamefont {Dalmonte}},\ }\bibfield
  {title} {\bibinfo {title} {{Nonstabilizerness via Matrix Product States in
  the Pauli Basis}},\ }\href {https://doi.org/10.1103/PhysRevLett.133.010601}
  {\bibfield  {journal} {\bibinfo  {journal} {Phys. Rev. Lett.}\ }\textbf
  {\bibinfo {volume} {133}},\ \bibinfo {pages} {010601} (\bibinfo {year}
  {2024})}\BibitemShut {NoStop}%
\bibitem [{\citenamefont {Hangleiter}\ and\ \citenamefont
  {Gullans}(2024)}]{PhysRevLett.133.020601}%
  \BibitemOpen
  \bibfield  {author} {\bibinfo {author} {\bibfnamefont {D.}~\bibnamefont
  {Hangleiter}}\ and\ \bibinfo {author} {\bibfnamefont {M.~J.}\ \bibnamefont
  {Gullans}},\ }\bibfield  {title} {\bibinfo {title} {{Bell Sampling from
  Quantum Circuits}},\ }\href {https://doi.org/10.1103/PhysRevLett.133.020601}
  {\bibfield  {journal} {\bibinfo  {journal} {Phys. Rev. Lett.}\ }\textbf
  {\bibinfo {volume} {133}},\ \bibinfo {pages} {020601} (\bibinfo {year}
  {2024})}\BibitemShut {NoStop}%
\bibitem [{\citenamefont {Leone}\ \emph {et~al.}(2024)\citenamefont {Leone},
  \citenamefont {Oliviero},\ and\ \citenamefont {Hamma}}]{leone2024learning}%
  \BibitemOpen
  \bibfield  {author} {\bibinfo {author} {\bibfnamefont {L.}~\bibnamefont
  {Leone}}, \bibinfo {author} {\bibfnamefont {S.~F.}\ \bibnamefont
  {Oliviero}},\ and\ \bibinfo {author} {\bibfnamefont {A.}~\bibnamefont
  {Hamma}},\ }\bibfield  {title} {\bibinfo {title} {Learning t-doped stabilizer
  states},\ }\href {https://doi.org/10.22331/q-2024-05-27-1361} {\bibfield
  {journal} {\bibinfo  {journal} {Quantum}\ }\textbf {\bibinfo {volume} {8}},\
  \bibinfo {pages} {1361} (\bibinfo {year} {2024})}\BibitemShut {NoStop}%
\bibitem [{\citenamefont {Lyu}\ and\ \citenamefont
  {Bu}(2024)}]{lyu2024fermionic}%
  \BibitemOpen
  \bibfield  {author} {\bibinfo {author} {\bibfnamefont {X.}~\bibnamefont
  {Lyu}}\ and\ \bibinfo {author} {\bibfnamefont {K.}~\bibnamefont {Bu}},\
  }\bibfield  {title} {\bibinfo {title} {{Fermionic Gaussian Testing and
  Non-Gaussian Measures via Convolution}},\ }\href
  {https://arxiv.org/abs/2409.08180} {\bibfield  {journal} {\bibinfo  {journal}
  {arXiv:2409.08180}\ } (\bibinfo {year} {2024})}\BibitemShut {NoStop}%
\bibitem [{\citenamefont {Coffman}\ \emph {et~al.}(2025)\citenamefont
  {Coffman}, \citenamefont {Smith},\ and\ \citenamefont
  {Gao}}]{coffman2025measuring}%
  \BibitemOpen
  \bibfield  {author} {\bibinfo {author} {\bibfnamefont {L.}~\bibnamefont
  {Coffman}}, \bibinfo {author} {\bibfnamefont {G.}~\bibnamefont {Smith}},\
  and\ \bibinfo {author} {\bibfnamefont {X.}~\bibnamefont {Gao}},\ }\bibfield
  {title} {\bibinfo {title} {{Measuring Non-Gaussian Magic in Fermions:
  Convolution, Entropy, and the Violation of Wick's Theorem and the Matchgate
  Identity}},\ }\href {https://arxiv.org/abs/2501.06179} {\bibfield  {journal}
  {\bibinfo  {journal} {arXiv:2501.06179}\ } (\bibinfo {year}
  {2025})}\BibitemShut {NoStop}%
\bibitem [{\citenamefont {Alagic}\ and\ \citenamefont
  {Fefferman}(2016)}]{alagic2016quantum}%
  \BibitemOpen
  \bibfield  {author} {\bibinfo {author} {\bibfnamefont {G.}~\bibnamefont
  {Alagic}}\ and\ \bibinfo {author} {\bibfnamefont {B.}~\bibnamefont
  {Fefferman}},\ }\bibfield  {title} {\bibinfo {title} {{On Quantum
  Obfuscation}},\ }\href {https://arxiv.org/abs/1602.01771} {\bibfield
  {journal} {\bibinfo  {journal} {arXiv:1602.01771}\ } (\bibinfo {year}
  {2016})}\BibitemShut {NoStop}%
\bibitem [{\citenamefont {Alagic}\ \emph {et~al.}(2021)\citenamefont {Alagic},
  \citenamefont {Brakerski}, \citenamefont {Dulek},\ and\ \citenamefont
  {Schaffner}}]{alagic2021impossibility}%
  \BibitemOpen
  \bibfield  {author} {\bibinfo {author} {\bibfnamefont {G.}~\bibnamefont
  {Alagic}}, \bibinfo {author} {\bibfnamefont {Z.}~\bibnamefont {Brakerski}},
  \bibinfo {author} {\bibfnamefont {Y.}~\bibnamefont {Dulek}},\ and\ \bibinfo
  {author} {\bibfnamefont {C.}~\bibnamefont {Schaffner}},\ }\bibfield  {title}
  {\bibinfo {title} {{Impossibility of Quantum Virtual Black-Box Obfuscation of
  Classical Circuits}},\ }in\ \href
  {https://doi.org/10.1007/978-3-030-84242-0_18} {\emph {\bibinfo {booktitle}
  {Advances in Cryptology--CRYPTO 2021: 41st Annual International Cryptology
  Conference, CRYPTO 2021, Virtual Event, August 16--20, 2021, Proceedings,
  Part I 41}}}\ (\bibinfo {organization} {Springer},\ \bibinfo {year} {2021})\
  pp.\ \bibinfo {pages} {497--525}\BibitemShut {NoStop}%
\bibitem [{\citenamefont {Broadbent}\ and\ \citenamefont
  {Kazmi}(2021)}]{broadbent2021constructions}%
  \BibitemOpen
  \bibfield  {author} {\bibinfo {author} {\bibfnamefont {A.}~\bibnamefont
  {Broadbent}}\ and\ \bibinfo {author} {\bibfnamefont {R.~A.}\ \bibnamefont
  {Kazmi}},\ }\bibfield  {title} {\bibinfo {title} {{Constructions for Quantum
  Indistinguishability Obfuscation}},\ }in\ \href
  {https://doi.org/10.1007/978-3-030-88238-9_2} {\emph {\bibinfo {booktitle}
  {International Conference on Cryptology and Information Security in Latin
  America}}}\ (\bibinfo {organization} {Springer},\ \bibinfo {year} {2021})\
  pp.\ \bibinfo {pages} {24--43}\BibitemShut {NoStop}%
\bibitem [{\citenamefont {Bartusek}\ and\ \citenamefont
  {Malavolta}(2022)}]{bartusek2022indistinguishability}%
  \BibitemOpen
  \bibfield  {author} {\bibinfo {author} {\bibfnamefont {J.}~\bibnamefont
  {Bartusek}}\ and\ \bibinfo {author} {\bibfnamefont {G.}~\bibnamefont
  {Malavolta}},\ }\bibfield  {title} {\bibinfo {title} {{Indistinguishability
  Obfuscation of Null Quantum Circuits and Applications}},\ }in\ \href
  {https://doi.org/10.4230/LIPIcs.ITCS.2022.15} {\emph {\bibinfo {booktitle}
  {13th Innovations in Theoretical Computer Science Conference (ITCS 2022)}}},\
  \bibinfo {series} {Leibniz International Proceedings in Informatics
  (LIPIcs)}, Vol.\ \bibinfo {volume} {215},\ \bibinfo {editor} {edited by\
  \bibinfo {editor} {\bibfnamefont {M.}~\bibnamefont {Braverman}}}\ (\bibinfo
  {publisher} {Schloss Dagstuhl -- Leibniz-Zentrum f{\"u}r Informatik},\
  \bibinfo {address} {Dagstuhl, Germany},\ \bibinfo {year} {2022})\ pp.\
  \bibinfo {pages} {15:1--15:13}\BibitemShut {NoStop}%
\bibitem [{\citenamefont {Bartusek}\ \emph {et~al.}(2023)\citenamefont
  {Bartusek}, \citenamefont {Kitagawa}, \citenamefont {Nishimaki},\ and\
  \citenamefont {Yamakawa}}]{bartusek2023obfuscation}%
  \BibitemOpen
  \bibfield  {author} {\bibinfo {author} {\bibfnamefont {J.}~\bibnamefont
  {Bartusek}}, \bibinfo {author} {\bibfnamefont {F.}~\bibnamefont {Kitagawa}},
  \bibinfo {author} {\bibfnamefont {R.}~\bibnamefont {Nishimaki}},\ and\
  \bibinfo {author} {\bibfnamefont {T.}~\bibnamefont {Yamakawa}},\ }\bibfield
  {title} {\bibinfo {title} {{Obfuscation of Pseudo-Deterministic Quantum
  Circuits}},\ }in\ \href {https://doi.org/10.1145/3564246.3585179} {\emph
  {\bibinfo {booktitle} {Proceedings of the 55th Annual ACM Symposium on Theory
  of Computing}}}\ (\bibinfo {year} {2023})\ pp.\ \bibinfo {pages}
  {1567--1578}\BibitemShut {NoStop}%
\bibitem [{\citenamefont {Coladangelo}\ and\ \citenamefont
  {Gunn}(2024)}]{coladangelo2024use}%
  \BibitemOpen
  \bibfield  {author} {\bibinfo {author} {\bibfnamefont {A.}~\bibnamefont
  {Coladangelo}}\ and\ \bibinfo {author} {\bibfnamefont {S.}~\bibnamefont
  {Gunn}},\ }\bibfield  {title} {\bibinfo {title} {{How to Use Quantum
  Indistinguishability Obfuscation}},\ }in\ \href
  {https://doi.org/10.1145/3618260.3649779} {\emph {\bibinfo {booktitle}
  {Proceedings of the 56th Annual ACM Symposium on Theory of Computing}}}\
  (\bibinfo {year} {2024})\ pp.\ \bibinfo {pages} {1003--1008}\BibitemShut
  {NoStop}%
\bibitem [{\citenamefont {Bartusek}\ \emph {et~al.}(2024)\citenamefont
  {Bartusek}, \citenamefont {Brakerski},\ and\ \citenamefont
  {Vaikuntanathan}}]{bartusek2024quantum}%
  \BibitemOpen
  \bibfield  {author} {\bibinfo {author} {\bibfnamefont {J.}~\bibnamefont
  {Bartusek}}, \bibinfo {author} {\bibfnamefont {Z.}~\bibnamefont
  {Brakerski}},\ and\ \bibinfo {author} {\bibfnamefont {V.}~\bibnamefont
  {Vaikuntanathan}},\ }\bibfield  {title} {\bibinfo {title} {{Quantum State
  Obfuscation from Classical Oracles}},\ }in\ \href
  {https://doi.org/10.1145/3618260.3649673} {\emph {\bibinfo {booktitle}
  {Proceedings of the 56th Annual ACM Symposium on Theory of Computing}}}\
  (\bibinfo {year} {2024})\ pp.\ \bibinfo {pages} {1009--1017}\BibitemShut
  {NoStop}%
\bibitem [{\citenamefont {Amy}\ \emph {et~al.}(2014)\citenamefont {Amy},
  \citenamefont {Maslov},\ and\ \citenamefont {Mosca}}]{amy2014polynomial}%
  \BibitemOpen
  \bibfield  {author} {\bibinfo {author} {\bibfnamefont {M.}~\bibnamefont
  {Amy}}, \bibinfo {author} {\bibfnamefont {D.}~\bibnamefont {Maslov}},\ and\
  \bibinfo {author} {\bibfnamefont {M.}~\bibnamefont {Mosca}},\ }\bibfield
  {title} {\bibinfo {title} {{Polynomial-Time T-Depth Optimization of
  Clifford+T Circuits Via Matroid Partitioning}},\ }\href
  {https://doi.org/10.1109/TCAD.2014.2341953} {\bibfield  {journal} {\bibinfo
  {journal} {IEEE Transactions on Computer-Aided Design of Integrated Circuits
  and Systems}\ }\textbf {\bibinfo {volume} {33}},\ \bibinfo {pages} {1476}
  (\bibinfo {year} {2014})}\BibitemShut {NoStop}%
\bibitem [{\citenamefont {Heyfron}\ and\ \citenamefont
  {Campbell}(2018)}]{heyfron2018efficient}%
  \BibitemOpen
  \bibfield  {author} {\bibinfo {author} {\bibfnamefont {L.~E.}\ \bibnamefont
  {Heyfron}}\ and\ \bibinfo {author} {\bibfnamefont {E.~T.}\ \bibnamefont
  {Campbell}},\ }\bibfield  {title} {\bibinfo {title} {{An efficient quantum
  compiler that reduces T count}},\ }\href
  {https://doi.org/10.1088/2058-9565/aad604} {\bibfield  {journal} {\bibinfo
  {journal} {Quantum Science and Technology}\ }\textbf {\bibinfo {volume}
  {4}},\ \bibinfo {pages} {015004} (\bibinfo {year} {2018})}\BibitemShut
  {NoStop}%
\bibitem [{\citenamefont {Ruiz}\ \emph {et~al.}(2025)\citenamefont {Ruiz},
  \citenamefont {Laakkonen}, \citenamefont {Bausch}, \citenamefont {Balog},
  \citenamefont {Barekatain}, \citenamefont {Heras}, \citenamefont {Novikov},
  \citenamefont {Fitzpatrick}, \citenamefont {Romera-Paredes}, \citenamefont
  {van~de Wetering} \emph {et~al.}}]{ruiz2025quantum}%
  \BibitemOpen
  \bibfield  {author} {\bibinfo {author} {\bibfnamefont {F.~J.}\ \bibnamefont
  {Ruiz}}, \bibinfo {author} {\bibfnamefont {T.}~\bibnamefont {Laakkonen}},
  \bibinfo {author} {\bibfnamefont {J.}~\bibnamefont {Bausch}}, \bibinfo
  {author} {\bibfnamefont {M.}~\bibnamefont {Balog}}, \bibinfo {author}
  {\bibfnamefont {M.}~\bibnamefont {Barekatain}}, \bibinfo {author}
  {\bibfnamefont {F.~J.}\ \bibnamefont {Heras}}, \bibinfo {author}
  {\bibfnamefont {A.}~\bibnamefont {Novikov}}, \bibinfo {author} {\bibfnamefont
  {N.}~\bibnamefont {Fitzpatrick}}, \bibinfo {author} {\bibfnamefont
  {B.}~\bibnamefont {Romera-Paredes}}, \bibinfo {author} {\bibfnamefont
  {J.}~\bibnamefont {van~de Wetering}}, \emph {et~al.},\ }\bibfield  {title}
  {\bibinfo {title} {{Quantum circuit optimization with AlphaTensor}},\ }\href
  {https://doi.org/10.1038/s42256-025-01001-1} {\bibfield  {journal} {\bibinfo
  {journal} {Nature Machine Intelligence}\ ,\ \bibinfo {pages} {1}} (\bibinfo
  {year} {2025})}\BibitemShut {NoStop}%
\bibitem [{\citenamefont {Kissinger}\ and\ \citenamefont {van~de
  Wetering}(2020{\natexlab{a}})}]{PhysRevA.102.022406}%
  \BibitemOpen
  \bibfield  {author} {\bibinfo {author} {\bibfnamefont {A.}~\bibnamefont
  {Kissinger}}\ and\ \bibinfo {author} {\bibfnamefont {J.}~\bibnamefont {van~de
  Wetering}},\ }\bibfield  {title} {\bibinfo {title} {{Reducing the number of
  non-Clifford gates in quantum circuits}},\ }\href
  {https://doi.org/10.1103/PhysRevA.102.022406} {\bibfield  {journal} {\bibinfo
   {journal} {Phys. Rev. A}\ }\textbf {\bibinfo {volume} {102}},\ \bibinfo
  {pages} {022406} (\bibinfo {year} {2020}{\natexlab{a}})}\BibitemShut
  {NoStop}%
\bibitem [{\citenamefont {de~Beaudrap}\ \emph {et~al.}(2020)\citenamefont
  {de~Beaudrap}, \citenamefont {Bian},\ and\ \citenamefont
  {Wang}}]{de2020fast}%
  \BibitemOpen
  \bibfield  {author} {\bibinfo {author} {\bibfnamefont {N.}~\bibnamefont
  {de~Beaudrap}}, \bibinfo {author} {\bibfnamefont {X.}~\bibnamefont {Bian}},\
  and\ \bibinfo {author} {\bibfnamefont {Q.}~\bibnamefont {Wang}},\ }\bibfield
  {title} {\bibinfo {title} {{Fast and Effective Techniques for T-Count
  Reduction via Spider Nest Identities}},\ }in\ \href
  {https://doi.org/10.4230/LIPIcs.TQC.2020.11} {\emph {\bibinfo {booktitle}
  {15th Conference on the Theory of Quantum Computation, Communication and
  Cryptography}}}\ (\bibinfo {year} {2020})\BibitemShut {NoStop}%
\bibitem [{\citenamefont {Backens}\ \emph {et~al.}(2021)\citenamefont
  {Backens}, \citenamefont {Miller-Bakewell}, \citenamefont {de~Felice},
  \citenamefont {Lobski},\ and\ \citenamefont {van~de
  Wetering}}]{backens2021there}%
  \BibitemOpen
  \bibfield  {author} {\bibinfo {author} {\bibfnamefont {M.}~\bibnamefont
  {Backens}}, \bibinfo {author} {\bibfnamefont {H.}~\bibnamefont
  {Miller-Bakewell}}, \bibinfo {author} {\bibfnamefont {G.}~\bibnamefont
  {de~Felice}}, \bibinfo {author} {\bibfnamefont {L.}~\bibnamefont {Lobski}},\
  and\ \bibinfo {author} {\bibfnamefont {J.}~\bibnamefont {van~de Wetering}},\
  }\bibfield  {title} {\bibinfo {title} {{There and back again: A circuit
  extraction tale}},\ }\href {https://doi.org/10.22331/q-2021-03-25-421}
  {\bibfield  {journal} {\bibinfo  {journal} {Quantum}\ }\textbf {\bibinfo
  {volume} {5}},\ \bibinfo {pages} {421} (\bibinfo {year} {2021})}\BibitemShut
  {NoStop}%
\bibitem [{\citenamefont {Cudby}\ and\ \citenamefont
  {Strelchuk}(2023)}]{cudby2023gaussian}%
  \BibitemOpen
  \bibfield  {author} {\bibinfo {author} {\bibfnamefont {J.}~\bibnamefont
  {Cudby}}\ and\ \bibinfo {author} {\bibfnamefont {S.}~\bibnamefont
  {Strelchuk}},\ }\bibfield  {title} {\bibinfo {title} {Gaussian decomposition
  of magic states for matchgate computations},\ }\href
  {https://arxiv.org/abs/2307.12654} {\bibfield  {journal} {\bibinfo  {journal}
  {arXiv:2307.12654}\ } (\bibinfo {year} {2023})}\BibitemShut {NoStop}%
\bibitem [{\citenamefont {Terhal}\ and\ \citenamefont
  {DiVincenzo}(2004)}]{terhal2004adaptive}%
  \BibitemOpen
  \bibfield  {author} {\bibinfo {author} {\bibfnamefont {B.~M.}\ \bibnamefont
  {Terhal}}\ and\ \bibinfo {author} {\bibfnamefont {D.~P.}\ \bibnamefont
  {DiVincenzo}},\ }\bibfield  {title} {\bibinfo {title} {{Adaptive Quantum
  Computation, Constant Depth Quantum Circuits and Arthur-Merlin Games}},\
  }\href {https://doi.org/10.26421/QIC4.2-5} {\bibfield  {journal} {\bibinfo
  {journal} {Quantum Information \& Computation}\ }\textbf {\bibinfo {volume}
  {4}},\ \bibinfo {pages} {134} (\bibinfo {year} {2004})}\BibitemShut {NoStop}%
\bibitem [{\citenamefont {Jones}\ and\ \citenamefont
  {Sunder}(1997)}]{jones1997introduction}%
  \BibitemOpen
  \bibfield  {author} {\bibinfo {author} {\bibfnamefont {V.~F.}\ \bibnamefont
  {Jones}}\ and\ \bibinfo {author} {\bibfnamefont {V.~S.}\ \bibnamefont
  {Sunder}},\ }\href@noop {} {\emph {\bibinfo {title} {Introduction to
  subfactors}}},\ Vol.\ \bibinfo {volume} {234}\ (\bibinfo  {publisher}
  {Cambridge University Press},\ \bibinfo {year} {1997})\BibitemShut {NoStop}%
\bibitem [{\citenamefont {Evans}\ and\ \citenamefont
  {Kawahigashi}(1998)}]{evans1998quantum}%
  \BibitemOpen
  \bibfield  {author} {\bibinfo {author} {\bibfnamefont {D.~E.}\ \bibnamefont
  {Evans}}\ and\ \bibinfo {author} {\bibfnamefont {Y.}~\bibnamefont
  {Kawahigashi}},\ }\href@noop {} {\emph {\bibinfo {title} {Quantum symmetries
  on operator algebras}}}\ (\bibinfo  {publisher} {Oxford University Press},\
  \bibinfo {year} {1998})\BibitemShut {NoStop}%
\bibitem [{\citenamefont {Jiang}\ \emph {et~al.}(2019)\citenamefont {Jiang},
  \citenamefont {Liu},\ and\ \citenamefont {Wu}}]{jiang2019block}%
  \BibitemOpen
  \bibfield  {author} {\bibinfo {author} {\bibfnamefont {C.}~\bibnamefont
  {Jiang}}, \bibinfo {author} {\bibfnamefont {Z.}~\bibnamefont {Liu}},\ and\
  \bibinfo {author} {\bibfnamefont {J.}~\bibnamefont {Wu}},\ }\bibfield
  {title} {\bibinfo {title} {Block maps and fourier analysis},\ }\href
  {https://doi.org/10.1007/s11425-017-9263-7} {\bibfield  {journal} {\bibinfo
  {journal} {Science China Mathematics}\ }\textbf {\bibinfo {volume} {62}},\
  \bibinfo {pages} {1585} (\bibinfo {year} {2019})}\BibitemShut {NoStop}%
\bibitem [{\citenamefont {Kissinger}\ \emph {et~al.}(2014)\citenamefont
  {Kissinger}, \citenamefont {Merry},\ and\ \citenamefont
  {Soloviev}}]{kissingerPatternGraphRewrite2014}%
  \BibitemOpen
  \bibfield  {author} {\bibinfo {author} {\bibfnamefont {A.}~\bibnamefont
  {Kissinger}}, \bibinfo {author} {\bibfnamefont {A.}~\bibnamefont {Merry}},\
  and\ \bibinfo {author} {\bibfnamefont {M.}~\bibnamefont {Soloviev}},\
  }\bibfield  {title} {\bibinfo {title} {Pattern graph rewrite systems},\
  }\href {https://doi.org/10.4204/EPTCS.143.5} {\bibfield  {journal} {\bibinfo
  {journal} {Electronic Proceedings in Theoretical Computer Science}\ }\textbf
  {\bibinfo {volume} {143}},\ \bibinfo {pages} {54} (\bibinfo {year}
  {2014})}\BibitemShut {NoStop}%
\bibitem [{\citenamefont {Kissinger}\ and\ \citenamefont
  {Zamdzhiev}(2015)}]{kissingerQuantomaticProofAssistant2015a}%
  \BibitemOpen
  \bibfield  {author} {\bibinfo {author} {\bibfnamefont {A.}~\bibnamefont
  {Kissinger}}\ and\ \bibinfo {author} {\bibfnamefont {V.}~\bibnamefont
  {Zamdzhiev}},\ }\bibfield  {title} {\bibinfo {title} {{Quantomatic: A Proof
  Assistant for Diagrammatic Reasoning}},\ }in\ \href
  {https://doi.org/10.1007/978-3-319-21401-6_22} {\emph {\bibinfo {booktitle}
  {Automated Deduction-CADE-25: 25th International Conference on Automated
  Deduction, Berlin, Germany, August 1-7, 2015, Proceedings 25}}}\ (\bibinfo
  {year} {2015})\ pp.\ \bibinfo {pages} {326--336}\BibitemShut {NoStop}%
\bibitem [{\citenamefont {Kissinger}\ and\ \citenamefont {van~de
  Wetering}(2020{\natexlab{b}})}]{kissingerPyZXLargeScale2020a}%
  \BibitemOpen
  \bibfield  {author} {\bibinfo {author} {\bibfnamefont {A.}~\bibnamefont
  {Kissinger}}\ and\ \bibinfo {author} {\bibfnamefont {J.}~\bibnamefont {van~de
  Wetering}},\ }\bibfield  {title} {\bibinfo {title} {{PyZX: Large Scale
  Automated Diagrammatic Reasoning}},\ }\href
  {https://doi.org/10.4204/EPTCS.318.14} {\bibfield  {journal} {\bibinfo
  {journal} {Electronic Proceedings in Theoretical Computer Science}\ }\textbf
  {\bibinfo {volume} {318}},\ \bibinfo {pages} {229} (\bibinfo {year}
  {2020}{\natexlab{b}})}\BibitemShut {NoStop}%
\bibitem [{\citenamefont {Cai}\ \emph {et~al.}(2018)\citenamefont {Cai},
  \citenamefont {Guo},\ and\ \citenamefont {Williams}}]{cai2018clifford}%
  \BibitemOpen
  \bibfield  {author} {\bibinfo {author} {\bibfnamefont {J.-Y.}\ \bibnamefont
  {Cai}}, \bibinfo {author} {\bibfnamefont {H.}~\bibnamefont {Guo}},\ and\
  \bibinfo {author} {\bibfnamefont {T.}~\bibnamefont {Williams}},\ }\bibfield
  {title} {\bibinfo {title} {{Clifford gates in the Holant framework}},\ }\href
  {https://doi.org/10.1016/j.tcs.2018.06.010} {\bibfield  {journal} {\bibinfo
  {journal} {Theoretical Computer Science}\ }\textbf {\bibinfo {volume}
  {745}},\ \bibinfo {pages} {163} (\bibinfo {year} {2018})}\BibitemShut
  {NoStop}%
\bibitem [{\citenamefont {Cahill}\ and\ \citenamefont
  {Glauber}(1999)}]{PhysRevA.59.1538}%
  \BibitemOpen
  \bibfield  {author} {\bibinfo {author} {\bibfnamefont {K.~E.}\ \bibnamefont
  {Cahill}}\ and\ \bibinfo {author} {\bibfnamefont {R.~J.}\ \bibnamefont
  {Glauber}},\ }\bibfield  {title} {\bibinfo {title} {Density operators for
  fermions},\ }\href {https://doi.org/10.1103/PhysRevA.59.1538} {\bibfield
  {journal} {\bibinfo  {journal} {Phys. Rev. A}\ }\textbf {\bibinfo {volume}
  {59}},\ \bibinfo {pages} {1538} (\bibinfo {year} {1999})}\BibitemShut
  {NoStop}%
\bibitem [{\citenamefont {Lutchyn}\ \emph {et~al.}(2018)\citenamefont
  {Lutchyn}, \citenamefont {Bakkers}, \citenamefont {Kouwenhoven},
  \citenamefont {Krogstrup}, \citenamefont {Marcus},\ and\ \citenamefont
  {Oreg}}]{lutchyn2018majorana}%
  \BibitemOpen
  \bibfield  {author} {\bibinfo {author} {\bibfnamefont {R.~M.}\ \bibnamefont
  {Lutchyn}}, \bibinfo {author} {\bibfnamefont {E.~P.}\ \bibnamefont
  {Bakkers}}, \bibinfo {author} {\bibfnamefont {L.~P.}\ \bibnamefont
  {Kouwenhoven}}, \bibinfo {author} {\bibfnamefont {P.}~\bibnamefont
  {Krogstrup}}, \bibinfo {author} {\bibfnamefont {C.~M.}\ \bibnamefont
  {Marcus}},\ and\ \bibinfo {author} {\bibfnamefont {Y.}~\bibnamefont {Oreg}},\
  }\bibfield  {title} {\bibinfo {title} {Majorana zero modes in
  superconductor--semiconductor heterostructures},\ }\href
  {https://doi.org/10.1038/s41578-018-0003-1} {\bibfield  {journal} {\bibinfo
  {journal} {Nature Reviews Materials}\ }\textbf {\bibinfo {volume} {3}},\
  \bibinfo {pages} {52} (\bibinfo {year} {2018})}\BibitemShut {NoStop}%
\bibitem [{\citenamefont {McLauchlan}\ and\ \citenamefont
  {B\'eri}(2022)}]{PhysRevLett.128.180504}%
  \BibitemOpen
  \bibfield  {author} {\bibinfo {author} {\bibfnamefont {C.~K.}\ \bibnamefont
  {McLauchlan}}\ and\ \bibinfo {author} {\bibfnamefont {B.}~\bibnamefont
  {B\'eri}},\ }\bibfield  {title} {\bibinfo {title} {{Fermion-Parity-Based
  Computation and Its Majorana-Zero-Mode Implementation}},\ }\href
  {https://doi.org/10.1103/PhysRevLett.128.180504} {\bibfield  {journal}
  {\bibinfo  {journal} {Phys. Rev. Lett.}\ }\textbf {\bibinfo {volume} {128}},\
  \bibinfo {pages} {180504} (\bibinfo {year} {2022})}\BibitemShut {NoStop}%
\bibitem [{\citenamefont {Aasen}\ \emph {et~al.}(2025)\citenamefont {Aasen},
  \citenamefont {Aghaee}, \citenamefont {Alam}, \citenamefont {Andrzejczuk},
  \citenamefont {Antipov}, \citenamefont {Astafev}, \citenamefont {Avilovas},
  \citenamefont {Barzegar}, \citenamefont {Bauer}, \citenamefont {Becker} \emph
  {et~al.}}]{aasen2025roadmap}%
  \BibitemOpen
  \bibfield  {author} {\bibinfo {author} {\bibfnamefont {D.}~\bibnamefont
  {Aasen}}, \bibinfo {author} {\bibfnamefont {M.}~\bibnamefont {Aghaee}},
  \bibinfo {author} {\bibfnamefont {Z.}~\bibnamefont {Alam}}, \bibinfo {author}
  {\bibfnamefont {M.}~\bibnamefont {Andrzejczuk}}, \bibinfo {author}
  {\bibfnamefont {A.}~\bibnamefont {Antipov}}, \bibinfo {author} {\bibfnamefont
  {M.}~\bibnamefont {Astafev}}, \bibinfo {author} {\bibfnamefont
  {L.}~\bibnamefont {Avilovas}}, \bibinfo {author} {\bibfnamefont
  {A.}~\bibnamefont {Barzegar}}, \bibinfo {author} {\bibfnamefont
  {B.}~\bibnamefont {Bauer}}, \bibinfo {author} {\bibfnamefont
  {J.}~\bibnamefont {Becker}}, \emph {et~al.},\ }\bibfield  {title} {\bibinfo
  {title} {Roadmap to fault tolerant quantum computation using topological
  qubit arrays},\ }\href {https://arxiv.org/abs/2502.12252} {\bibfield
  {journal} {\bibinfo  {journal} {arXiv:2502.12252}\ } (\bibinfo {year}
  {2025})}\BibitemShut {NoStop}%
\bibitem [{\citenamefont {Murg}\ \emph {et~al.}(2012)\citenamefont {Murg},
  \citenamefont {Korepin},\ and\ \citenamefont
  {Verstraete}}]{PhysRevB.86.045125}%
  \BibitemOpen
  \bibfield  {author} {\bibinfo {author} {\bibfnamefont {V.}~\bibnamefont
  {Murg}}, \bibinfo {author} {\bibfnamefont {V.~E.}\ \bibnamefont {Korepin}},\
  and\ \bibinfo {author} {\bibfnamefont {F.}~\bibnamefont {Verstraete}},\
  }\bibfield  {title} {\bibinfo {title} {{Algebraic Bethe ansatz and tensor
  networks}},\ }\href {https://doi.org/10.1103/PhysRevB.86.045125} {\bibfield
  {journal} {\bibinfo  {journal} {Phys. Rev. B}\ }\textbf {\bibinfo {volume}
  {86}},\ \bibinfo {pages} {045125} (\bibinfo {year} {2012})}\BibitemShut
  {NoStop}%
\bibitem [{\citenamefont {Jones}(2007)}]{jones2007and}%
  \BibitemOpen
  \bibfield  {author} {\bibinfo {author} {\bibfnamefont {V.~F.}\ \bibnamefont
  {Jones}},\ }\bibfield  {title} {\bibinfo {title} {In and around the origin of
  quantum groups},\ }\href {https://doi.org/10.1090/conm/437/08427} {\bibfield
  {journal} {\bibinfo  {journal} {Contemporary Mathematics}\ }\textbf {\bibinfo
  {volume} {437}},\ \bibinfo {pages} {101} (\bibinfo {year}
  {2007})}\BibitemShut {NoStop}%
\bibitem [{\citenamefont {Feng}\ \emph {et~al.}(2025)\citenamefont {Feng},
  \citenamefont {Liu}, \citenamefont {Lu},\ and\ \citenamefont
  {Wang}}]{quonAlgoPaper}%
  \BibitemOpen
  \bibfield  {author} {\bibinfo {author} {\bibfnamefont {Z.}~\bibnamefont
  {Feng}}, \bibinfo {author} {\bibfnamefont {Z.}~\bibnamefont {Liu}}, \bibinfo
  {author} {\bibfnamefont {F.}~\bibnamefont {Lu}},\ and\ \bibinfo {author}
  {\bibfnamefont {N.}~\bibnamefont {Wang}},\ }\bibfield  {title} {\bibinfo
  {title} {{Quon Classical Simulation Unifying Cliffords and Matchgates}},\
  }\href@noop {} {\bibfield  {journal} {\bibinfo  {journal} {arXiv preprint, to
  appear in the same arXiv posting}\ } (\bibinfo {year} {2025})}\BibitemShut
  {NoStop}%
\bibitem [{\citenamefont {Kasteleyn}(1961)}]{kasteleyn1961statistics}%
  \BibitemOpen
  \bibfield  {author} {\bibinfo {author} {\bibfnamefont {P.~W.}\ \bibnamefont
  {Kasteleyn}},\ }\bibfield  {title} {\bibinfo {title} {{The statistics of
  dimers on a lattice: I. The number of dimer arrangements on a quadratic
  lattice}},\ }\href {https://doi.org/10.1016/0031-8914(61)90063-5} {\bibfield
  {journal} {\bibinfo  {journal} {Physica}\ }\textbf {\bibinfo {volume} {27}},\
  \bibinfo {pages} {1209} (\bibinfo {year} {1961})}\BibitemShut {NoStop}%
\bibitem [{\citenamefont {Temperley}\ and\ \citenamefont
  {Fisher}(1961)}]{temperley1961dimer}%
  \BibitemOpen
  \bibfield  {author} {\bibinfo {author} {\bibfnamefont {H.~N.}\ \bibnamefont
  {Temperley}}\ and\ \bibinfo {author} {\bibfnamefont {M.~E.}\ \bibnamefont
  {Fisher}},\ }\bibfield  {title} {\bibinfo {title} {Dimer problem in
  statistical mechanics-an exact result},\ }\href
  {https://doi.org/10.1080/14786436108243366} {\bibfield  {journal} {\bibinfo
  {journal} {Philosophical Magazine}\ }\textbf {\bibinfo {volume} {6}},\
  \bibinfo {pages} {1061} (\bibinfo {year} {1961})}\BibitemShut {NoStop}%
\end{thebibliography}%

\end{document}